\documentclass[12pt,a4paper, oneside, leqno]{amsart}
\usepackage{xcolor}
\usepackage{color}
 \usepackage{a4wide}
 \usepackage{amsaddr}
 \usepackage{setspace}
 \usepackage[font=small,labelfont=bf]{caption}
  \usepackage[backend=bibtex, style=alphabetic, isbn=false, url=false]{biblatex}
  \bibliography{refs-2}
 \usepackage{geometry}
\usepackage{amsmath, amsfonts, amsthm, amssymb, amscd}
\usepackage{tensor}
\usepackage[multiple]{footmisc}
\usepackage{amsmath, amsthm, slashed}
\usepackage{amsfonts, slashed, cancel}
\usepackage{amssymb}
\usepackage{rotating}\usepackage{epsfig}
\usepackage{verbatim}
\usepackage{rotating}
\usepackage{graphicx}
\usepackage{amsmath,amsthm,amssymb,verbatim,a4}
\usepackage{mathrsfs,stmaryrd}
\usepackage{marginnote}

\newcounter{mnotecount}[section]

\renewcommand{\themnotecount}{\thesection.\arabic{mnotecount}}

\newcommand{\mnote}[1]
{\protect{\stepcounter{mnotecount}}$^{\mbox{\footnotesize
$\!\!\!\!\!\!\quad\bullet$\themnotecount}}
$\marginpar{\raggedright\em
$\!\!\!\!\!\!\quad\bullet$\themnotecount: #1} }

\usepackage{chngcntr}

\counterwithin{dummy}{section}

\newcommand{\eq}[1]{\begin{equation}#1\end{equation}}
\newcommand{\zz}{\mathfrak z}

\usepackage{dsfont}
\newtheorem{theorem}{Theorem}[section]
\newtheorem{proposition}[theorem]{Proposition}
\newtheorem{lemma}[theorem]{Lemma}
\newtheorem{corollary}[theorem]{Corollary}
\newtheorem{definition}[theorem]{Definition}
\newtheorem{remark}[theorem]{Remark}

\numberwithin{equation}{section}

\newcommand \Tcal {\mathcal T}
\newcommand \Ecal {\mathcal E}

\newcommand \Kcal {\mathcal K}

\newcommand \del \partial
\newcommand \delu {\underline{\del}}
\newcommand \vu {\underline{v}}
\newcommand \wu {\underline{w}}

\newcommand \Su {\underline{S}}

\newcommand \Hu {\underline{H}}
\newcommand \hu {\underline{h}}

\newcommand {\thetau}{\underline{\theta}}

\newcommand {\Pu}{\underline{P}}
\newcommand {\gu}{\underline{g}}

\let\oldmarginpar\marginpar
\renewcommand\marginpar[1]{\- \oldmarginpar[\raggedleft\footnotesize #1]%
{\raggedright\footnotesize #1}}

\newcommand \alg[1]{\begin{aligned}#1\end{aligned}}

\makeatletter
\def\blx@maxline{77}
\makeatother

\begin{document}

\author{ David Fajman}
\address{Gravitational Physics,
Faculty of Physics,
University of Vienna,\\
Boltzmanngasse 5,
1090 Wien,
Austria.}
\email{david.fajman@univie.ac.at}
\author{J\'er\'emie Joudioux}
\address{Gravitational Physics,
Faculty of Physics,
University of Vienna,\\
Boltzmanngasse 5,
1090 Wien, Austria.}
\email{jeremie.joudioux@univie.ac.at}
\author{Jacques Smulevici}
\address{Laboratoire de Math\'ematiques, Univ. Paris-Sud,\\ CNRS, Universit\'e Paris-Saclay, 91405 Orsay\\
 and\\
  D\'epartement de math\'ematiques et applications, \'Ecole Normale Sup\'erieure, CNRS, PSL Research University, 75005 Paris, France.}
\email{jacques.smulevici@math.u-psud.fr}

\date{\today}

\title[Stability of the Minkowski space for the Einstein--Vlasov system]{\large The Stability of the Minkowski space\\\vspace{0.2cm} for the Einstein--Vlasov system}

\maketitle
\footnotetext{\bfseries{Preprint number}: \mdseries UwThPh-17-06-15 (University of Vienna).}
\begin{abstract}
We prove the global stability of the Minkowski space viewed as the trivial solution of the Einstein-Vlasov system. To estimate the Vlasov field, we use the vector field and modified vector field techniques  developed in \cite{fjs:vfm,fjs:savn}. In particular, the initial support in the velocity variable does not need to be compact.
To control the effect of the large velocities, we identify and exploit several structural properties of the Vlasov equation to prove that the worst non-linear terms in the Vlasov equation either enjoy a form of the null condition or can be controlled using the wave coordinate gauge. The basic propagation estimates for the Vlasov field are then obtained using only weak interior decay for the metric components. Since some of the error terms are not time-integrable, several hierarchies in the commuted equations are exploited to close the top order estimates. For the Einstein equations, we use wave coordinates and the main new difficulty arises from the commutation of the energy-momentum tensor, which needs to be rewritten using the modified vector fields.
\end{abstract}

\begin{spacing}{1.2}
\tableofcontents
\newpage


\section{Introduction}
This paper establishes the stability of the Minkowski space $(\mathbb{R}^{1+3}, \eta)$ viewed as the trivial solution to the Einstein-Vlasov system
\begin{eqnarray} \label{eq:evs1}
Ric(g)-\frac{1}{2}g R(g)&=& T[f], \\
T_g(f)&=&0.\label{eq:evs2}
\end{eqnarray}
Here $g$ is a Lorentzian metric on a $4$-dimensional manifold, $Ric(g)$ and $R(g)$ the Ricci and scalar curvatures of $g$, $f$ is a massive Vlasov field, $T[f]$ its energy-momentum tensor and $T_g$ is the geodesic spray vector field. The Minkowski space is then the simplest solution to these equations with $f=0$.

We refer to sections \ref{se:vfcf} and \ref{se:evs} as well as \cite{MR2098917, Ehlers1973,  MR3186493} for a presentation of the equations and the terminology used here. We recall that the system \eqref{eq:evs1}-\eqref{eq:evs2} admits an initial value problem formulation which is, at least when the initial data enjoy sufficient regularity, locally well-posed \cite{MR0337248, MR3186493}. We are therefore interested in the global Cauchy problem for this system, that is to say we want to understand the asymptotics of the solutions.

The Einstein-Vlasov system is actively used in astrophysics and cosmology. It describes a statistical ensemble of self-gravitating particles which interact only indirectly through the Einstein equations. It is in fact the natural fully general relativistic analogue of the Vlasov-Poisson system\footnote{We refer to the classical \cite{MR1379589} for a presentation of the Vlasov-Poisson and related kinetic systems.}, replacing Newtonian mechanics by general relativity.

\subsection{The vacuum problem}
At least from a PDE perspective, there are two fundamental differences between the Poisson and the Einstein equations. The Poisson equation is elliptic and linear in the gravitational potential, while the Einstein equations are (after a suitable gauge choice) hyperbolic and non-linear in the metric components. In particular, in Newtonian mechanics, no matter source implies no gravitational force, while in general relativity, there are plenty of non-trivial, vacuum solutions to the Einstein equations $Ric(g)=0$. Thus, the stability of the Minkowski space for the \emph{vacuum} Einstein equations is a necessary starting point. This problem was solved in full generality by Christodoulou and Klainerman \cite{ck:nlsms} (see also \cite{lr:gsmhg, MR2531716} and \cite{Lindblad:2005ex, MR727195}).

Let us recall some of the main features of the problem and its proof. The vacuum Einstein equations can be recast in so-called \emph{wave coordinates} as a system of quasilinear wave equations and the stability of the Minkowski space then corresponds to a small data global existence result for this system. For any quasilinear system of wave equations, controlling the non-linearities requires higher order estimates of the solutions, that is to say estimates obtained after commutation of the equations with well chosen vector fields, so as to control a high number of derivatives of the solutions. These vector fields typically arise  from the symmetries of the linearised equations. For the wave equation, they are thus the Killing and (some of the) conformal Killing fields of Minkowski space.  One then combines these higher order estimates with weighted Sobolev inequalities linked to the equations (this is the vector field method of Klainerman \cite{sk:udelicw}) to prove decay estimates for the non-linear terms\footnote{There are many other ways to establish decay estimates, though the vector field method is certainly the most robust one. In particular, we stress that a standard strategy for quasilinear wave equations consists in using the basic vector field method and energy estimates to obtain first, rough, decay estimates for the solutions under weak assumptions, and, only in a second step, use another method, for instance, integral estimates and representation formulas, to obtain improved decay estimates.}.

It is well known that the case of three spatial dimensions is critical for this type of questions. In higher spatial dimensions, linear waves enjoy stronger decay properties, so that such small global existence results always hold for general quasilinear wave equations \cite{sk:udelicw}, while in dimension $3$, small data global existence is linked to structural properties of the equations and blow up is known to occur in some cases \cite{MR600571}. A general criterion that guarantees small data global existence is the null condition of Klainerman \cite{MR837683}. Essentially, for a solution to the free wave equation, it is well known that derivatives tangential to the light-cone decay faster than the transversal ones and the null condition ensures that each non-linear product contains derivatives tangential to the light-cone.

The null condition is however not satisfied in the wave coordinate formulation of the Einstein equations \cite{MR1770106}. Thus, the strategy of \cite{ck:nlsms} exploits another formulation of the Einstein equations, where the main energy estimates control the curvature rather than the metric itself. Another key element of \cite{ck:nlsms} is the construction of an optical function, and then vector fields, which are tied to the characteristics of the spacetime, or equivalently, to the null cones of the metric\footnote{Interestingly, the null cones of the constructed spacetimes eventually diverge logarithmically compared to the cones of Minkowski space, as a remnant of the failure of the null condition.}.

Even though the Einstein equations in wave coordinates do not satisfy the null condition, the stability of the Minkowski space was subsequently obtained in this gauge in \cite{Lindblad:2005ex, lr:gsmhg}. The key observation is that the Einstein equations still enjoy a weak version of the null condition, of which a trivial example is provided by the system
\begin{eqnarray}
\square u &=& \partial_t v \cdot \partial_t v, \label{eq:wnc}\\
\square v &=&0, \nonumber
\end{eqnarray}
where $u$ and $v$ are two scalar functions defined on $\mathbb{R}_t \times \mathbb{R}^3$. The second equation is linear, thus the first is simply a linear inhomogenous wave equation and obviously the solutions of this system do not blow up.

In the case of the Einstein equations, this trivial example is replaced by a hierarchy of wave equations for the metric components. Moreover, in order to control the non-linear terms not satisfying the null condition, the wave gauge condition $\square_g x^\alpha=0$ for the coordinates $x^\alpha$ is used extensively.

\subsection{The mass problem}
Recall that an initial data set for the vacuum Einstein equation $Ric(g)=0$ is given by $(\Sigma, g_0,k)$, where $(\Sigma, g_0)$ is a smooth Riemannian manifold and $k$ is a symmetric $2$-tensor field,
such that $(\Sigma, g_0, k)$ solves the constraint equations
\begin{eqnarray*}
R(g_0) -|k|^2 +tr_{g_0}(k)^2 &=& 0,\\
\mathrm{div} k - d(tr_{g_0} k ) &=& 0,
\end{eqnarray*}
where $R(g_0)$ is the scalar curvature of $(\Sigma, g_0)$, $|k|^2=k_{ij}k^{ij}$, $tr_{g_0} k=k_{ij} g_0^{ij}$, $[\mathrm{div} k]_j = \tensor[^{(g_0)}]{\nabla}{^i} k_{ij}$, with $\tensor[^{(g_0)}]{\nabla}{}$ the Levi-Civita connection of $g_0$.

In the case of perturbations of the Minkowski space, one considers initial data such that $\Sigma=\mathbb{R}^3$ and the data are asymptotically flat i.e.~$g_0$ tends to the Euclidean metric and $k$ tends to $0$ as $|x|=r\rightarrow\infty$. The positive mass theorem \cite{MR526976, MR626707} then implies that $g_0=\delta_E(1+ 2m/r)+ o(r^{-1-\rho})$, where $\delta_E$ is the Euclidean metric, $\rho > 0$ and where $m>0$ unless the initial data correspond to an initial data set induced by the Minkowski space, in which case the solution of the evolution problem must naturally coincide with Minkowski space.

The positive mass theorem limits the possible radial decay of the initial data. In particular, one cannot consider compact initial data, for which the metric perturbations would be all contained in some ball of finite radius. The closest one can get from those are initial data corresponding to the Schwarzschild metric outside from some compact set. We refer to \cite{MR1794269, MR1902228, cs:aveec} for general methods leading to the construction of such data.

For a solution of the linear wave equation in Minkowski space $\square \psi=0$, the interior decay of $\psi$, i.e.~estimates of the form $|\psi(t,x)| \lesssim C(R) \frac{1}{(1+t)^p}$ for $|x| < R$, is directly related to the amount of radial decay of the initial data for $\psi$. The stronger the decay, the higher the value of $p$. In view of the $r^{-1}$ behavior of the initial data for the perturbations, this implies that, even at the linearised level, we cannot expect interior decay faster than $t^{-1}$ for the metric perturbations $|g-\eta|$ and $t^{-2}$ for their first derivatives $|\partial g|$.

\subsection{Einstein-matter systems}
Consider now a coupled system of the form
\begin{eqnarray}
Ric(g)-\frac{1}{2}g R(g)&=& T[\psi], \label{eq:ekg}\\
N_g(\psi)&=&0, \label{eq:mf}
\end{eqnarray}
where $T[\psi]$ is the energy-momentum tensor of some matter field $\psi$, itself subject to an evolution equation depending on the metric $g$, which we write schematically as \eqref{eq:mf}.

If $\psi$ solves a wave equation, as in the scalar field case $\square_g \psi=0$ or the Maxwell equations, the matter equation \eqref{eq:mf} can be treated by the same methods as the Einstein equations themselves. In particular, one can commute the Einstein equations and the matter equations by the same vector fields and thus extend the vacuum stability results to these cases \cite{lr:gsmhg, MR2531716,MR2582443,MR3012654}.

One of the simplest models for which this approach does not readily work is the Einstein-Klein-Gordon system, where the matter field $\psi$ is a scalar function solving the Klein-Gordon equation
\begin{equation} \label{eq:kg}
\square_g \psi-\psi=0,
\end{equation}
and where the energy momentum tensor is given by $$T[\psi]=d\psi\otimes d\psi- \frac{1}{2}g \left(g( \nabla \psi, \nabla \psi)+\psi^2 \right).$$

The Klein-Gordon equation shares many properties with the wave equation, but it has less symmetries. In particular, it enjoys poor commutation properties with respect to the scaling vector field $S=x^\alpha \partial_{x^\alpha}$. Moreover, in dimension $3$, the interior decay for $\partial \psi$ is limited by $\frac{1}{t^{3/2}}$, which is weaker than the maximal interior decay one can obtain for the first derivatives of the metric components in the vacuum case. It does enjoy on the other hand stronger decay near the light cone than a pure wave. Finally, the classical vector field method of Klainerman for Klein-Gordon fields \cite{MR1199196} typically requires the use of a \emph{hyperboloidal foliation}, while the analysis of the vacuum Einstein equations as in \cite{ck:nlsms, lr:gsmhg} uses only a foliation by standard, asymptotically flat, spacelike hypersurfaces as well as a foliation by null cones.

 In view of (or despite) the above difficulties, the stability of the Minkowski space for the Einstein-Klein-Gordon system was only recently obtained in \cite{lm:gsmkg} (see also \cite{qw:ihaekg}). In some sense, this is the first stability result (in three spatial dimensions, without symmetry or cosmological constant) for an Einstein-matter system which cannot be obtained by a direct extension of the methodology of the vacuum case.
\subsection{The Einstein-Vlasov system}
The Einstein-Vlasov system \eqref{eq:evs1}-\eqref{eq:evs2} couples the Einstein equations to kinetic theory. For particles of mass $m_p$, the Vlasov field $f$ is a non-negative function defined on the submanifold $\mathcal{P}$ of the tangent bundle\footnote{Since we can use the metric to identify the tangent and cotangent bundles, we can also consider $f$ as a function on a submanifold of the cotangent bundle. While this is perhaps less common, we shall actually use this formulation in this paper, see Section \ref{se:vfcf}.} corresponding to future-directed causal vectors normalized to $-m_p^2$. The Vlasov field $f$ then is, at each point of $\mathcal{P}$, the density of particles with given position and velocity (or momentum). The Vlasov equation $T_g(f)=0$ is the conservation of this particle density by the geodesic flow. The local Cauchy theory for the Einstein-Vlasov system was treated in \cite{MR0337248} (see also \cite{MR3186493}, Chapter 6). In particular, to any given appropriate initial data set $(\Sigma, g_0, k,f_0)$\footnote{See Section \ref{se:evs} for a presentation of the initial data for the Einstein-Vlasov system.}, one can associate a unique (up to diffeomorphism) maximal Cauchy development $(M,g,f)$, where $(M,g)$ is a Lorentzian manifold and $f$ a Vlasov field. 

The stability of the Minkowski space for the spherically symmetric Einstein-Vlasov system in dimension $3+1$ has been treated in \cite{rr:gesssvsssd, rr:err} for the massive case and in \cite{md:ncsdsgm} for the massless case with compactly supported initial data. A proof of stability for the massless case without spherical symmetry and with compact support in both $x$ and $v$ has been given in \cite{mt:smsmevs}. As in \cite{md:ncsdsgm}, the compact support assumptions and the fact that the particles are massless are important as they allow to reduce the proof to that of the vacuum case outside from a strip going to null infinity.

In this paper, we prove the stability of the Minkowski space for the Einstein-Vlasov system in the case of a Vlasov field corresponding to massive particles. For simplicity, we assume that all particles have the same mass $m_p$ and we later fix $m_p=1$.
\subsection{Statement of the results}
The main result can then be stated as follows.
\begin{theorem} \label{th:main}
Let $(\Sigma=\mathbb{R}^3,g_0,k,f_0)$ be an initial data set for the Einstein-Vlasov system which coincides with a Schwarzschild initial data set of mass $m \ge 0$ outside from a ball of radius $R > 0$.

Let $(M,g,f)$ be the unique\footnote{As usual, by uniqueness, we mean uniqueness up to diffeomorphism.} maximal globally hyperbolic development of the given initial data set and denote by $i: \Sigma \rightarrow M$ the embedding of $\Sigma$ into a Cauchy hypersurface of $M$ given by the local existence theorem.

Let $N \ge 14$, $q\ge3$ and $\epsilon >0$. Assume that
\begin{eqnarray*}
&&||h-\delta_E ||_{H^N(|x| < R)}+ ||k||_{H^{N-1}(|x| < R)}+m^2\\
&&\quad+ || (1+|v|^2)^q f_0 ||_{W^{N+3,1}\left(T^\star \mathbb{R}^3\right)}+|| (1+|v|^2)^{q+2} f_0 ||_{W^{N-2,1}\left(T^\star \mathbb{R}^3\right)} \le \epsilon,
\end{eqnarray*}
where $\delta_E$ is the Euclidean metric and $m$ the mass of the Schwarzschild metric for $|x| \ge R$.

Then, there exists $\epsilon_0(R) > 0$ such that if $\epsilon\le \epsilon_0(R)$, there exists a global system of wave coordinates $(t,x)$ on $\mathbb{R}^4\simeq M$ such that, $t$ is a temporal function, $i(\Sigma)=\{ t=2 \}$ and with $\Kcal := \{(t, x) \, / \, |x| <t-1\}$, for $(t,x) \in J^+(i(\Sigma))\setminus \Kcal$, $g$ coincides with the Schwarzschild metric of mass $m$, while for $(t,x) \in \Kcal \cap J^+(i(\Sigma))$, we have
\begin{eqnarray*}
\Ecal_{N}[g](\rho) &\le& D_N \epsilon \rho^{D_N\epsilon^{1/2}}, \\
E_{N-2,q+2}[f](\rho) & \le & D_N \epsilon \rho^{D_N\epsilon^{1/2}}, \\
E_{N,q}[f](\rho) & \le & D_N \epsilon \rho^{D_N\epsilon^{1/2}}, \\
||g(t,x)-\eta||_{L^\infty_x} &\le& D_N \epsilon^{1/2}(1+t)^{-1+D_N \epsilon^{1/2}},
\end{eqnarray*}
where $\rho=\sqrt{t^2-|x|^2}$ denotes a hyperboloidal time function, $D_N$ is a constant depending only on $N$ and $\Ecal_{N}[g]$, $E_{N,q}[f]$ are energy norms depending on up to $N$ derivatives of $f$ and $g$.

In particular, $(M,g)$ is future causal geodesically complete.
\end{theorem}

\begin{remark}
A similar statement holds for the past of $i(\Sigma)$. Moreover, redefining some of the coordinates, we can shift the slice $\{t =2\}$ to any other $t=const$ slice.
\end{remark}
\begin{remark}
The norms $\Ecal_{N}(g)$ and $E_{N,q}(f)$ are defined in \eqref{def:hon} and \eqref{def:enq}, respectively. The index $q$ in $E_{N,q}(f)$ refers to the number of additional $v$ weights, so that $E_{0,0}(f)$ correspond to the natural energy norm of $f$. The norms $|| . ||_{H^N}$ and $|| . ||_{W^{N+3,1}}$ which we use for the initial data are standard Sobolev norms. \end{remark}

\begin{remark} When $|\alpha| \ge N-2$, we prove $L^2$ decay estimates for source terms of the form $T[ \widehat{K}^\alpha (f)]$ arising in the wave equations. These require more regularity for the Vlasov field, hence three extra derivatives are required for the initial datum of $f$.
\end{remark}

\begin{remark}
We refer to the body of the proof for many extra details concerning the asymptotics of the solutions. For instance, for $q' \ge 0$ sufficiently small (in particular $q'=0$, corresponding to the basic energy norm), we prove bounds $E_{N-2,q'}(f) \lesssim \epsilon$ without growth. Moreover, we obtain sharp pointwise decay estimates on the components of the energy-momentum tensor $T[f]$ as well as its derivatives.
\end{remark}\label{rem-geodcpl}
\begin{remark}The geodesic completeness is a direct consequence of the asymptotics of the metric and its derivatives. See for instance \cite{Lindblad:2005ex}, Section 16.
\end{remark}

\begin{remark}
For simplicity, we have considered initial data which coincides with the Schwarzschild data outside of a compact set. Since we use a hyperboloidal foliation, our results do not extend immediately to more general data that would allow the Vlasov field to have non-compact support in the $x$ variable. We note however that the method of this paper are readily applicable (with slightly different asymptotics) for initial data such that $f$ is initially supported in $\mathcal{B}_x \times \mathbb{R}^3_v$, for some compact set $\mathcal{B}_x$, and the data for the metric is such that the analysis of \cite{ck:nlsms} or \cite{lr:gsmhg} is applicable. Indeed, in that case, using a standard domain of dependence argument, the solution is vacuum outside from the domain of influence of $\mathcal{B}_x$ and we can repeat the analysis of \cite{ck:nlsms} or \cite{lr:gsmhg} in that region. As is clear from the proof of our theorem, the techniques of this paper do not depend on the exact nature of the asymptotics of the metric at spatial and null infinity. In particular, we prove our main propagation estimates for the Vlasov field using only weak interior decay for the metric coefficients $(|\partial g|(t,x) \lesssim t^{-3/2+\delta}$ for $|x| < \frac{t}{2})$.
\end{remark}
\begin{remark}A similar stability result has been obtained independently by Lindblad and Taylor \cite{lt}.
\end{remark}

\begin{remark}
In this paper, we consider only massive particles but a large part of our analysis can be extended to the massless case. In particular, the estimates of sections \ref{se:pave} and \ref{se:cve} do not make use of the strict positivity of the mass of the particles and one could easily modify the remainder of the paper to cover the massless case as well. In fact, an important simplification in the massless case arises from the strong interior decay for velocity averages of massless Vlasov fields.
\end{remark}

\subsection{Key elements of the proof and main difficulties}

\subsubsection{The coupling}
At the linear level, the Vlasov equation is given by the free transport equation on Minkowski space
$$
v^\alpha \partial_{x^\alpha} f=0,
$$
for $f:=f(t,x,v)$, $(x^\alpha)=(t,x) \in \mathbb{R}^{1+3}$, $(v^\alpha)=\left( v^0, v^i\right)$, with $v^0=\sqrt{m_p^2+|v|^2}$, $(v^i) \in \mathbb{R}^3$. In particular, for massive particles, $m_p>0$, the characteristics of the Vlasov equation are the timelike geodesics, while for massless particles $m_p=0$, the characteristics are the null geodesics, as for the wave or the Einstein equations.

As in the case of the Einstein-Klein-Gordon equations \eqref{eq:ekg}-\eqref{eq:kg}, the coupling is non-trivial.
\begin{itemize}
\item Kinetic equations such as the Vlasov equation are intrinsically of different nature compared to wave equations. The domain of definition of the unknown $f$ is a different manifold ($\mathcal{P}$) and the coupling through the energy-momentum tensor $T[f]$ takes the form of velocity averages of $f$, i.e.~(weighted) integrals in $v$ of $f$. In fact, the Einstein-Vlasov system is a system of integro-partial-differential equations and not a pure PDE system.
\item For massive particles, the characteristics are different from those of the wave equation. For the free transport operator, they are given by the timelike curves $\left( t, t \frac{v^i}{\sqrt{1+|v|^2}}\right)$. Note that for high velocities $|v| \rightarrow +\infty$, these curves approach the null curves $(t,\omega^i t)$, where $\omega^i= \frac{v^i}{|v|} \in \mathbb{S}^2$. For low $|v|$, we do expect to face the difficulty of an equation that does not share the characteristics of the wave equation. On the other hand, for large $|v|$, we expect the difficulties associated with pure wave equations, such as the slow decay of transversal derivatives to the light-cone, to also be an issue. This difficulty naturally disappears for distributions which are of compact $v$ support initially, but we treat here initial data which are merely integrable in $v$ against a measure $(1+|v|^2)^{k/2} dv$.
\end{itemize}

\subsubsection{Commuting the Vlasov equation using complete lifts}
Another important difficulty arise from commuting the Vlasov equation.

Recall that we cannot expect to control the behaviour of the metric components without commuting the Einstein equations. In view of the coupling, this implies that we must estimate $K^N T[f]$, where $K^N$ is a differential operator of order $N$. In flat space, where $g$ is the Minkowski metric $\eta$, we have
$$
T_{\alpha \beta}[f]= \int_{v \in \mathbb{R}^3} f v_\alpha v_\beta \frac{dv}{\sqrt{m_p^2+|v|^2}},$$
so that, for any vector field $K=K^\alpha \partial_{x^\alpha}$,
$$
KT[f]= T[K(f)].  
$$
The vector fields $K$ are those that commute with the flat wave operator, i.e.~the Killing and conformal Killing fields\footnote{We note that they are many variants of these methods. In particular, one can commute only with a subalgebra of the full algebra of Killing and conformal Killing fields (see for instance \cite{MR1374174}), or, in another setting, one can commute with vector fields containing only radial weights, as in \cite{dr:npsdw}.} of Minkowski space.

In general, if $K$ is a Killing vector field on a Lorentzian manifold, $K(f)$ does not readily make sense, since $f$ is defined on a different manifold. Thus, one needs first to lift $K$ to $\mathcal{P}$. There are several such possible lifts, but the one we consider here is the \emph{complete lift} of $K$, denoted $\widehat{K}$. Complete lifts have the following properties.
\begin{itemize}
\item The complete lift operation lifts vector fields on $M$ to vector fields on $TM$.
\item If $K$ is Killing, then $\widehat{K}$ is tangent to the submanifold $\mathcal{P}$ of $TM$. In particular, for any regular distribution $f$ defined on $\mathcal{P}$, $\widehat{K}(f)$ is well-defined.
\item If $K$ is Killing, $\widehat{K}$ commutes with the geodesic spray vector field $T_g$.
\item If $K$ is Killing, $\mathcal{L}_K T[f]=T[ \widehat{K}(f)]$, where $\mathcal{L}_K$ is the Lie derivative in the direction of $K$.
\end{itemize}
In \cite{fjs:vfm}, we exploited such a geometric treatment of the commutation properties of the Vlasov equation to extend the traditional vector field method of Klainerman for wave equations to the class of transport equations of Vlasov type\footnote{See \cite{ww:cvfa} for an extension of these methods to other dispersive PDEs.}. In particular, we established Klainerman-Sobolev inequalities for velocity averages of Vlasov fields and gave an illustration of our method to obtain (almost) sharp asymptotics for the $3$-dimensional massless and the $n\ge4$ massive Vlasov-Nordstr\"om systems.

\subsubsection{Non-integrable decay and the modified vector fields}
While it seems that working with complete lifts would thus solve the difficulties involved with commuting the Vlasov equation and the energy-momentum tensor, for a general perturbation of Minkowski space, one should not expect any of the original Killing fields to remain Killing, so that none of the above properties can be directly applied. As a first step, one can write the Vlasov equation in coordinates, and then commutes the Vlasov equation with coordinate equivalents of the original vector fields of Minkowski space. For instance, let us write schematically the Vlasov equation as
$$
T_g(f)=v^\alpha \partial_{x^\alpha} f+Q(\partial g, v, v) \partial_v f,
$$
for some multi-linear form $Q$, and consider a Lorentz boost $Z_i= t \partial_{x^i} + x^i \partial_t$. In Minkowski space, the restriction to $\mathcal{P}$ of its complete lift would be given by $\widehat{Z}_i= t \partial_{x^i}+x^i\partial_t + v^0 \partial_{v^i}$. Commuting the above equation, we obtain
$$
T_g ( \widehat{Z}_i f ) = - [\widehat{Z}_i , Q(\partial g, v, v) \partial_v] f.
$$
Neglecting the $v$ components, the right-hand side leads to error terms of the form $\partial Z(g)\cdot \partial_v f$. On the other hand, for a solution of the free transport operator, $\partial_v f$ behaves essentially like $t \partial_{x^\alpha} f$. If we expect to prove boundedness for some norms of $|\partial_{x^\alpha} f|$, then $t \partial Z(g)$ needs to be time integrable in order to control $\widehat{Z}f$. Assuming that the interior decay for the metric components\footnote{For some metric components, there is already a logarithmic divergence for the interior decay estimates, so that the expected behaviour is in fact worse that the one presented here. Moreover, as the name indicates, the interior decay estimates are only valid in the interior and the fact that they are not global is another source of difficulty that we neglected in this informal discussion, linked with the null structure of the equations.} can readily be used, in three spatial dimensions, $|  t \partial Z(g)| \lesssim \frac{1}{t}$ leads to logarithmic divergences. On other hand, any loss in the Vlasov estimates would limit the interior decay for the metric even further, since in order to obtain sharp or almost sharp interior decay, one already needs sharp estimates on the source terms $K^N T[f]$ arising in the equations for the metric components.

This interesting issue is in fact already present in the much simpler Vlasov-Poisson system, where it was solved in \cite{MR3595457} by modifying the commutation vector fields, replacing the lifted vector fields $\widehat{Z}$ by some $Y= \widehat{Z}+ \Phi^i \partial_{x^i}$, where the coefficients $\Phi^i$ are functions in the variable $(t,x,v)$, depending on the solution and constructed in order to cancel the worst error terms in the commutator formulas\footnote{See also \cite{hwang11} for previous results concerning sharp asymptotics for solutions of the Vlasov-Poisson system based on the method of characteristics.}. The method of modified vector fields was adapted to a basic model of wave/kinetic interaction, namely the 3-dimensional Vlasov-Nordstr\"om system, in \cite{fjs:savn}. Many of the difficulties presented above are in fact present for this system. In particular, important strutural properties of the system where used in \cite{fjs:savn} in order to account for difficulties arising for large $v$.

In this paper, we thus also consider commuting the Vlasov equation with modified vector fields. The use of modified vector fields is however not without drawbacks. Since the coefficients of these vector fields depend on the solution itself, they need to be estimated. Moreover, these coefficients depend on $(t,x,v)$ and as a consequence, these modified vector fields cannot be used in return in the wave equations. Thus, an important effort is made to rewrite source terms of the form $K^N T[f]$, that arise after commuting the wave equations, in terms of the modified vector fields. Here, the integration in $v$ present in the definition of $T[f]$ is crucially used. Finally, the Klainerman-Sobolev inequalities must also be rewritten using the modified vector fields.

\subsubsection{Hierarchy of equations and the null structure}
We first prove energy estimates for the Vlasov field assuming weak interior decay for the metric components. With only these weak estimates for the metric coefficients at our disposal, some of the error terms in the commuted Vlasov equation fail to be time-integrable. To close the estimates, we exploit a hierarchy in the commuted equations. More specifically, we first find replacements for the spatial translations that enjoy improved commutation properties with the Vlasov equation. These vector fields\footnote{We already used a version of these vector fields in \cite{fjs:savn}.}, denoted $X_i$, are simply given by $X_i =\partial_{x^i} +\frac{v^i}{\sqrt{1+|v|^2}}\partial_t$ and the improvement results from the identity
$$
X_i=\frac{Z_i}{t}+ \frac{\vu_i}{w^0}\partial_t,
$$
where $\vu_i= v_i -\frac{x^i}{t}v^0$. Using this identity, one can prove that a product of the form $X_i(\psi)\cdot k$ for $\psi$ a solution to the wave equation and $k$ a solution to the Vlasov equation enjoys better decay properties compared to an arbitrary product $\partial \psi\cdot k$.

Assuming weak bounds on the first order energy, we then prove that commuting with $X_i$ only produces integrable error terms, and thus obtain estimates for $E[X_i(f)]$. We then consider commuting with $\partial_t$. Only one term is not time-integrable (because of a lack of null structure), but it can be estimated using the bounds on $E[X_i(f)]$ and thus produces only a mild growth $\rho^\delta$, for some small $\delta> 0$. We then commute with the modified vector fields $Y$ and again find that the terms which are not time-integrable only depend on $E[X_i(f)]$ and $E[\partial_t(f)]$, which allows us to close the first order estimates. This hierarchy is then extended to the higher order estimates. It is in fact very reminiscent of similar hierarchies present in the context of the weak null condition, as in the system \eqref{eq:wnc}.

Once the basic energy estimates for the Vlasov field have been established, one can propagate stronger weighted norms, which then imply, together with the improved decay for the metric coefficients, energy and decay estimates for the Vlasov field without loss.

\subsubsection{The Einstein equations and the top order estimate}
The analysis of the Einstein equations in wave coordinates is now classical and we follow the approach of \cite{Lindblad:2005ex, lr:gsmhg} and its adaptation to the hyperboloidal foliation in \cite{lm:gsmkg}. The major new difficulty consists in rewriting and estimating the source terms coming from the Vlasov field in terms of the modified vector fields without any hard\footnote{We can afford a $\rho^{D\delta}$ for $\delta>0$ small enough in these estimates and $D$ being a positive constant.} loss of decay. However, the key step to avoid loss of decay involves an integration by parts in $v$, which, in turn, implies a loss of regularity. At top order, we therefore must allow for some hard loss of decay. The worst source term in the top order estimate for the metric coefficients then implies another source of small growth at top order.

\subsection{Related works}
We present here some previous works to put the results of this paper in context.
\subsubsection{Stability problems for Vlasov systems without sharp decay}
There is a large number of results concerning small data global existence for various systems of Vlasov type, as in \cite{bd:gevp, MR919231, sf:gssvns}. In these works, the gravitational or electromagnetic fields satisfy a linear, inhomogeneous equation, whose source term is given by velocity averages of the Vlasov field. The linear aspect of the field equations implies that one can control the system at a much lower level of regularity than for a system of quasilinear wave equations. Moreover, these systems typically exhibit a gain of regularity, either because of the elliptic nature of the Poisson equation, or using a non-resonant phenomenon due to the difference between the characteristics of the waves and that of the massive particles. This allows to close the estimates without understanding sharp decays for the velocity averages of the Vlasov field and its derivatives.

\subsubsection{Sharp decay for derivatives}
The first work establishing sharp decay for derivatives of velocity averages of the Vlasov field is \cite{hwang11}. The question was revisited using vector field techniques in \cite{MR3595457}. In \cite{fjs:vfm} and \cite{fjs:savn}, we developed and tested a vector field approach to derive sharp asymptotics for the Vlasov-Nordstr\"om system. The techniques of \cite{hwang11} have also been extended to the so-called Poisson-Yukawa system in dimension $2$ \cite{MR2737852}.

\subsubsection{Non-trivial stationary states and further stability results}
The strongest results concerning the stability of non-trivial stationary solutions of the gravitational Vlasov-Poisson system have been obtained in \cite{lmr:ossgm}. They are not based on decay estimates but on a variational characterisation of the stationary solutions. On the other hand, this type of method does not provide asymptotic stability of the solutions but orbital stability. It is likely that any result addressing the question of asymptotic stability will need to go back to an appropriate linearization of the equations combined with robust decay estimates\footnote{See for instance \cite{gr:nvanssd} for some stability results using the linearization approach in the case of the spherically-symmetric King model.}.

There is a large literature concerning the construction of stationary states for the Einstein-Vlasov system \cite{MR1254978,MR1269939,MR2842969,MR3210151,MR3369063,AFT17,MR1816475}. We refer to the living review \cite{MR1966539}, Section 5, for a detailed discussion of those results. Naturally, it would be interesting to understand the stability properties of any of these stationary solutions.

The vector field method has also been extended to the Kerr background to
prove Morawetz estimates for massless Vlasov fields, see
\cite{abj:hsdvk}. The approach relies on the use of multiplicative
symmetries for massless fields.

\subsubsection{The cosmological case}
There is also a large amount of works concerning solutions to the Einstein-Vlasov system arising from initial data given on a compact manifold. Let us mention in particular the work of Ringstr\"om \cite{MR3186493}, concerning the study of expanding solutions with de-Sitter like asymptotics, as well as the stability result \cite{Fajman2017}, where the slower expansion only provides polynomial decay for perturbations.

\subsubsection{Coupled systems}
There are many recent works involving coupled systems of equations for which the coupling is non-trivial, beyond the Einstein-Klein-Gordon system already mentioned. Let us mention in particular \cite{MR3450481, MR3283401, MR2976318} concerning coupled systems of equations with different characteristics.

\subsubsection{Introductory materials on kinetic theory in general relativity}
There are many such materials but we would like to mention the classical texts \cite{Ehlers1973, MR2098917, Stewart1971} as well as the elegant geometric treatment of the Vlasov equation in \cite{Sarbach:2013vy}.

\subsection{Structure of the paper}
Section \ref{se:pre} contains preliminaries about the Einstein-Vlasov system, basic definitions and notations that we use throughout the text. In particular, Section \ref{se:ln} contains a list of notations and can be used as a reference.
In Section \ref{se:basp}, we set up the bootstrap assumptions and describe, towards the end of the section, the different steps required in order to establish the main result. Section \ref{se:bcwgc} contains direct consequences of the bootstrap assumptions and of the wave gauge condition. Sections \ref{se:pave} to \ref{se:hovf} are devoted to the proof of energy estimates for the Vlasov field, under the weak decay assumptions for the metric coefficients. In Section \ref{se:hocc}, we estimate the coefficients appearing in the modified vector fields. In Section \ref{se:comt}, we consider the source terms arising from the energy-momentum tensor in the Einstein equations after commutation and explain how they can be rewritten in terms of the modified vector fields. Sections \ref{sec:aree} to \ref{se:ien1} concern the analysis of the Einstein equations. In Section \ref{se:ievf}, we improve the estimates on the Vlasov field and the $C$ coefficients, so as to eventually close the top order energy estimates for the metric components in Section \ref{se:toemc}. Finally, the last two sections concern $L^\infty$- and $L^2$-decay estimates for the velocity averages of the Vlasov fied.

\subsection{Acknowledgements}We thank P.G.~LeFloch and Y. Ma for helpfully answering questions about their work \cite{lm:gsmkg}. D.F.~acknowledges support of the Austrian Science Fund (FWF) through the START- Project Y963-N35 of M.~Eichmair as well as through the project \emph{Geometric transport equations and the non-vacuum Einstein flow (P 29900-N27)}. J.J.~and J.S.~ are supported in part by the ANR grant AARG, "Asymptotic Analysis in General Relativity" (ANR-12-BS01-012-01).
J.S.~acknowledges funding from the European Research Council under the European Union's Horizon 2020 research and innovation program (project GEOWAKI, grant agreement 714408).
D.F., J.J.~and J.S.~ acknowledge the partial support of the Erwin Schr\"odinger International Institute for Mathematics and Physics during the workshop \emph{Geometric Transport equations in General Relativity}, where part of this work has been written.
\section{Preliminaries} \label{se:pre}

\subsection{Vlasov fields in the cotangent bundle formulation} \label{se:vfcf}
Let $(M,g)$ be a smooth time-oriented, oriented, $4$-dimensional Lorentzian manifold.

We denote by $\mathcal{P}$ the mass-shell. While it is generally considered as a submanifold of the tangent bundle $TM$, we shall, equivalently, consider here $\mathcal{P}$ as a subset of the cotangent bundle\footnote{This formulation is linked with the Hamiltonian property of the equations, cf~\cite{Ehlers1973}.} $T^\star M$, defined by
$$
\mathcal{P}: =\left\{(x,v) \in T^\star M\; :\, g^{-1}_x(v,v)=-1\,\,\mathrm{and}\,\,v\,\,\mathrm{future\,\,oriented} \right\}.
$$

Given a coordinate system on $M$, $(U,x^\alpha)$, for any $x \in U \subset M$, any $v \in T_x^\star M$ can be written as
$$
v= v_\alpha [d x^\alpha]_x
$$
and the functions $v \rightarrow v_\alpha$ can be used to define a coordinate system on $T^\star_x M$ called \emph{conjugates} to the coordinates $(x^\alpha)$. In the following, we consider such coordinate systems even if it is not stated explicitly. We denote by $\pi$ the canonical projection
$$
\pi: \mathcal{P} \rightarrow M.
$$

For $x \in M$, we define a metric on $T_x^\star M$ by $$g^{-1}_{T_x^\star M}=g^{\alpha \beta}dv_\alpha dv_\beta, $$
where $g^{\alpha \beta}$ are the components of $g^{-1}$ in a local coordinate system $\left(U, x^\alpha\right)$ and $v_\alpha$ are conjugate to the $x^\alpha$. Let $d\mu_{T_x^\star M}$ be the associated volume form, i.e.~$d\mu_{T_x^\star M}=\sqrt{-g^{-1}} dv_0 \wedge dv_1 \wedge dv_2 \wedge dv_3$ and let $q_x$ be the map
\begin{eqnarray*}
q_x: T^\star_x M &\rightarrow& \mathbb{R}, \\
 v &\mapsto& g^{-1}_{T_x^\star M}(v,v).
 \end{eqnarray*}
Let $dq$ be its differential (in $v$). Since $\pi^{-1}(x)=q^{-1}\left( \{-1 \}\right)$ is a level set of $q$, $dq= 2 dv_\alpha v_\beta g^{\alpha \beta}$ is normal to $\pi^{-1}(x)$ and on $\pi^{-1}(x)$, there is a unique volume form denoted $d\mu_{\pi^{-1}(x)}$ such that
 $$
 d\mu_{T_x^\star M}= \frac12 dq \wedge d\mu_{\pi^{-1}(x)}.
 $$
 We assume that there exist local coordinates such that $x^0=t$ is a smooth temporal function, i.e.~it is strictly increasing along any future causal curve and its gradient is past directed and timelike\footnote{The fact that the gradient of $t$ is timelike is equivalent to $g^{00}<0$ and the property of being strictly increasing along any future causal curve implies that the induced metric on each level set of $t$ has to be positive.}. In that case, the algebraic equation
 $$v_\alpha v_\beta g^{\alpha \beta}=-1  \mathrm{\,\,and\,\,} v_\alpha \mathrm{\,\,future\,\,directed}\vspace{0.15cm}$$
can be solved for $v_0$ by
\begin{equation} \label{eq:v0}
v_0= - (g^{00})^{-1} \left(  g^{0j}v_j - \sqrt{(g^{0j}v_j)^2+(-g^{00})(1+g^{ij}v_iv_j)} \right).
\end{equation}

 It follows that $(x^\alpha, v_i)$, $1 \le i \le 3$ are smooth coordinates on $\mathcal{P}$ and for any $x \in M$, $(v_i)$, $1 \le i \le 3$ are smooth coordinates on $\pi^{-1}(x)$. With respect to these coordinates, the volume form $d\mu_{\pi^{-1}(x)}$ reads
$$
d\mu_{\pi^{-1}(x)} = \frac{\sqrt{-g^{-1}}}{v_\beta g^{\beta 0}} dv_1\wedge dv_2 \wedge dv_3.
$$

For any sufficiently regular\footnote{By "sufficiently regular", we mean that $f$ is smooth enough and decays in $v$ sufficiently fast  so that $T[f]$ is well-defined and the necessary integration by parts in $v$ can be performed. Later, we also perform integration in $x$, so we also require the regular distribution function $f$ to obey decay in the $x$ variable along each hyperboloid. In any case, one can assume for simplicity that all distribution functions are smooth and compactly supported for all computations to hold. } distribution function $f: \mathcal{P} \rightarrow M$, we define its energy-momentum tensor as the tensor field
\eq{\label{eq-enmt}
T_{\alpha \beta}[f](x)= \int_{\pi^{-1}(x)} v_\alpha v_\beta f  d\mu_{\pi^{-1}(x)}.
}

In the following, to simplify the notation, we write
$$\int_{\pi^{-1}(x)} \mbox{ as } \int_v$$
and
$$d\mu_v \mbox{ for the measure } d\mu_{\pi^{-1}(x)}.$$
 Even on a curved spacetime, we use another reference measure, namely that corresponding to the Minkowski space $\frac{dv}{\sqrt{1+|v|^2}}$. When we do so, we write the measure explicitly.\\

The Vlasov field $f$ is required to solve the \emph{Vlasov equation}, which can be written in the $(x^\alpha, v_i)$ coordinate system as
\begin{equation} \label{eq:vl}
T_g(f):= g^{\alpha \beta} v_\alpha \partial_{x^\beta} f - \frac{1}{2}v_\alpha v_\beta \partial_{x^i} g^{\alpha \beta} \partial_{v_i} f=0.
\end{equation}

It follows from the Vlasov equation that the energy-momentum tensor is divergence free for solutions of the Vlasov equation. More generally, for any sufficiently regular distribution function $k: \mathcal{P}\rightarrow \mathbb{R}$,
$$
\nabla^\alpha T_{\alpha \beta}[k]= \int_v T_g[k] v_\beta d\mu_v.
$$

\subsection{The Einstein-Vlasov system and the initial value problem} \label{se:evs}

We consider the Einstein equations
\begin{equation}
\label{eq:evs}
Ric(g) -  \frac{1}{2} g R = T[f],
\end{equation}
where $Ric(g)$ denotes the Ricci curvature of $(M,g)$, $R=g^{\alpha \beta} R_{\alpha \beta}$ its scalar curvature and $T[f]$ is the energy momentum tensor of a Vlasov field $f$ satisfying the Vlasov equation \eqref{eq:vl}. 


Contracting the Einstein equations, we get that
$$
- R(g) = \int_v  f v_\alpha v_\beta g^{\alpha \beta} d \mu_v = -\int_v  f d \mu_v .
$$
Thus, the Einstein equations can be rewritten as

$$
Ric(g)= T [f] + \frac{1}{2} g \int_{\pi^{-1}(x)}  f d \mu_{\pi^{-1}(x)}.
$$

The initial value problem for the Einstein-Vlasov system was first considered in \cite{MR0337248}. We refer to \cite{MR3186493} for a thorough analysis of the local well-posedness of this system.

Recall that an \emph{initial data set} for the Einstein-Vlasov system is given by $(\Sigma, g_0,k,f_0)$, where $\Sigma$ is a smooth manifold, $g_0$ is a Riemannian metric on $\Sigma$, $k$ is a symmetric $2$-tensor field on $\Sigma$ and $f_0$ is a real-valued function (taken non-negative in physics) and defined on the cotangent bundle\footnote{Again, it is more standard to consider the initial data for the Vlasov field as a function on the tangent bundle, but one can naturally move from one representation to the other without any difficulty.} of $\Sigma$, such that $(\Sigma, g_0, k, f_0)$ solves the constraint equations
\begin{eqnarray}
R(g_0) -|k|^2 +tr_{g_0}(k)^2 &=& 2 \rho[f_0], \label{eq:cons1}\\
\mathrm{div} k - d(tr_{g_0} k ) &=& - j[f_0], \label{eq:cons2}
\end{eqnarray}
where $R(g_0)$ is the scalar curvature of $g_0$, $|k|^2=k_{ij}k^{ij}$, $tr_{g_0} k=k_{ ij} {g_0}^{ij}$, $[\mathrm{div} k]_j = \tensor[^{(g_0)}]{\nabla}{^i} k_{ ij}$, with $\tensor[^{(g_0)}]{\nabla}{}$ the Levi-Civita connection of $g_0$ and where the source terms $\rho[f_0]$ and $j[f_0]$ are given, for any $x \in \Sigma$, by
\begin{eqnarray*}
\rho[f_0](x)&=& \int_{p \in T_x^\star \Sigma} f_0(x,p)  \frac{d\mu_{T_x^\star \Sigma}}{\sqrt{1+ g_0(p,p)}}, \\
j_i[f_0](x) &=& \int_{p \in T_x^\star \Sigma} f_0(x,p)p_i  \frac{d\mu_{T_x^\star \Sigma}}{\sqrt{1+ g_0(p,p)}}, \\
\end{eqnarray*}
where the volume form $d\mu_{T^\star_x \Sigma}$ is the one associated with the metric $(g_0)^{-1}_{T^\star_x \Sigma}(p,p)=g_0^{ij}p_i p_j$.

For simplicity, we assume smooth initial data with sufficiently strong decay  in both $x,v$. For any such data, it follows from standard results that there exists a unique maximal Cauchy development up to diffeomorphism \cite{MR0337248, MR3186493}.

\subsection{The wave gauge and the reduced Einstein equations}
We consider the Einstein equations in wave coordinates, following the hyperboloidal foliation formulation\footnote{The first global results for the Einstein equations in wave coordinates are of course that of \cite{lr:gsmhg, Lindblad:2005ex}, while the introduction of an hyperboloidal foliation for the study of Klein-Gordon fields goes back to \cite{MR803252}.} of  \cite{lm:gsmkg}.
We thus consider coordinates $x^\alpha$ satisfying the wave equation
\begin{eqnarray} \label{eq:wc}
\square_g x^\alpha = 0,
\end{eqnarray}
or equivalently
$$
g^{\alpha \beta} \Gamma^\gamma_{\alpha \beta}=0,
$$
where $\Gamma^\gamma_{\alpha \beta}$ are the Christoffel symbols of the metric $g$. It is well known that the Ricci curvature then simplifies to a second order non-linear wave operator acting on each metric coefficient. The Einstein equations \eqref{eq:evs} then transform to
\begin{eqnarray*}
\label{eq:ree}
\widetilde{\square}_g h_{\alpha \beta} = F_{\alpha \beta}(h, \partial h) -  \int_v f \left(2v_\alpha v_\beta + g_{\alpha \beta}\right) d \mu_v,
\end{eqnarray*}
where $\widetilde{\square}_g := g^{\alpha \beta} \del_{\alpha} \del_{\beta}$ is the reduced wave operator,
$$
h_{\alpha \beta} := g_{\alpha \beta} - \eta_{\alpha \beta}
$$
denotes the deviation of the metric $g$ from the Minkowski metric $\eta$ and the nonlinear terms $F_{\alpha \beta}(h, \del h)$ are quadratic in the first order derivatives of the metric.

As usual, any solution to the reduced equations arising from initial data satisfying the constraint equations is a solution of the original Einstein-Vlasov in view of the propagation of the wave gauge condition (see for instance \cite{MR3186493}, p. 370, for the Einstein-Vlasov case).

For future reference, we define the tensor $S[f]$ by
\begin{eqnarray} \label{def:sf}
S[f]_{\alpha \beta}= \int_v f \left(2v_\alpha v_\beta + g_{\alpha \beta}\right) d \mu_v.
\end{eqnarray}

Let us also recall that schematically, the non-linear terms $F(h, \partial h)$ are all of the form
\begin{equation} \label{eq:F}
g^{-1}\cdot g^{-1}\cdot \partial h\cdot \partial h.
\end{equation}

\subsection{Convention for raising and lowering indices}

Given the metric $g$, and its components $g_{\alpha \beta}$, the components of its inverse are denoted as usual by $g^{\alpha \beta}$. We also define
$$
H^{\alpha \beta}:=g^{\alpha \beta}-\eta^{\alpha \beta}.
$$

For a tensor field $T$ with components $T_{\alpha \beta}$, we define
$$
T^{\alpha \beta}:=T_{\gamma \rho}\eta^{\alpha \gamma} \eta^{\beta \rho},
$$ i.e.~we raise indices using the Minkowski metric $\eta$. In view of the definition of $g^{\alpha \beta}$, this convention applies to all tensor fields but the metric $g$ itself.

Similarly, we lower indices using also the metric $\eta$.

Note that
$$
g^{\alpha \beta}:=(g^{-1})^{\alpha \beta}= \eta^{\alpha \beta} + H^{\alpha \beta}=\eta^{\alpha \beta}-h^{\alpha \beta} +O(h^2),
$$
so that within the small data regime of this paper, we can switch from $H$ or $h$ without any difficulty.


\subsection{The hyperboloidal foliation}

Fix global Cartesian coordinates $(t,x^i)$, $1 \le i \le 3$ on $\mathbb{R}^{3+1}$. For any $\rho > 0$, define $H_\rho$ by
$$
H_{\rho}=\left\{ (t,x)\,\, \big |\,\, t \ge |x|\,\, \mathrm{and}\,\, t^2-|x|^2= \rho^2 \right\}.
$$

We denote by $\Kcal := \{(t, x) \, / \, |x| <t-1\}$ the chronological future of the point $(1,0,0,0)$, see Figure \ref{fig:coord}.

\begin{center}
  \begin{figure}[h]
    \hspace{-2cm}
\begin{picture}(0,0)%
\includegraphics[width=12cm]{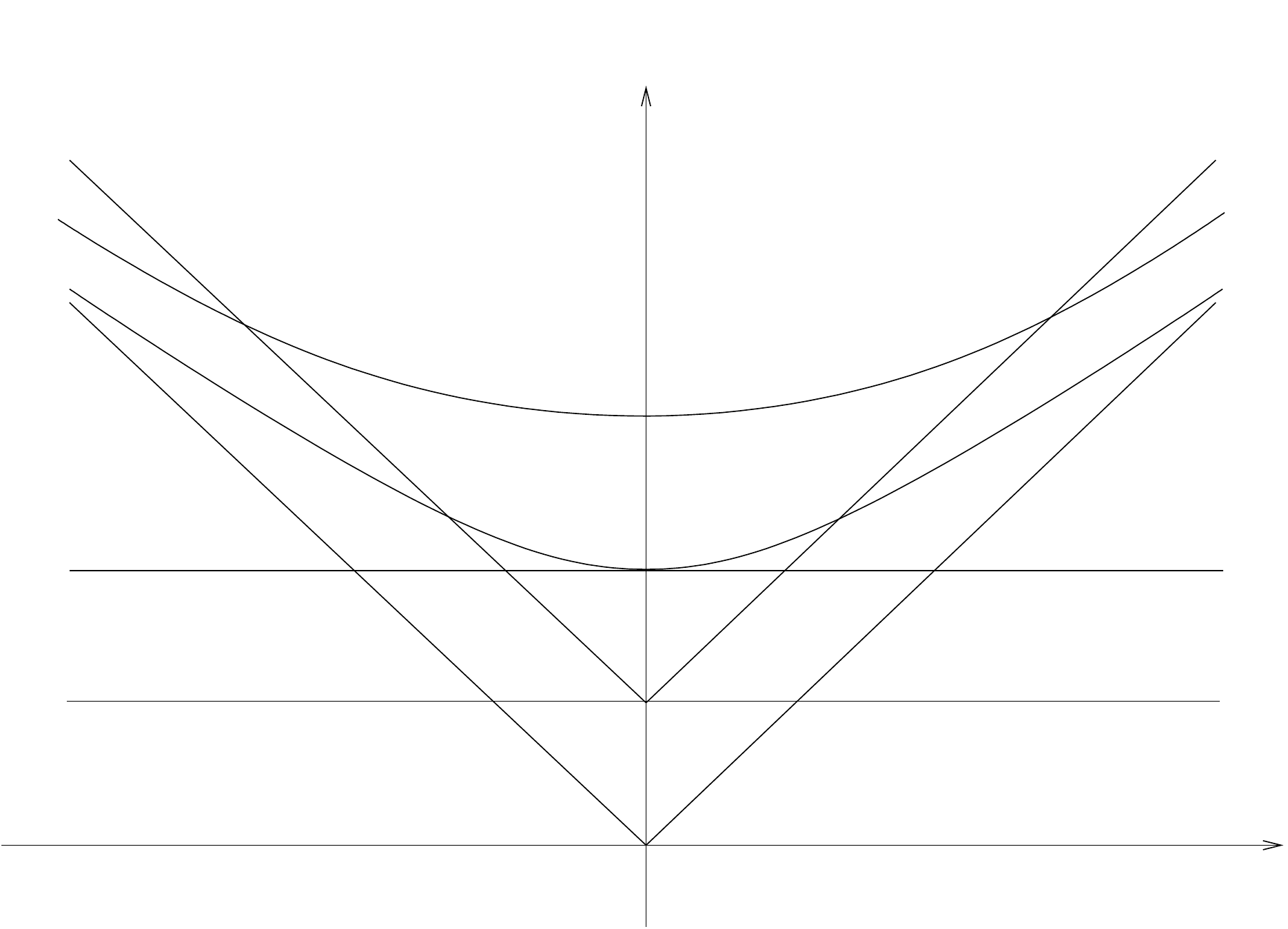}%
\end{picture}%
\setlength{\unitlength}{2500sp}%
\begingroup\makeatletter\ifx\SetFigFont\undefined%
\gdef\SetFigFont#1#2#3#4#5{%
  \reset@font\fontsize{#1}{#2pt}%
  \fontfamily{#3}\fontseries{#4}\fontshape{#5}%
  \selectfont}%
\fi\endgroup%
\begin{picture}(8529,6144)(3139,-8713)
\put(10800,-5281){\makebox(0,0)[lb]{\smash{{\SetFigFont{12}{14.4}{\rmdefault}{\mddefault}{\updefault}{\color[rgb]{0,0,0}$u=0$}%
}}}}
\put(11700,-3450){\makebox(0,0)[lb]{\smash{{\SetFigFont{12}{14.4}{\rmdefault}{\mddefault}{\updefault}{\color[rgb]{0,0,0}$u=1$}%
}}}}
\put(7876,-4900){\makebox(0,0)[lb]{\smash{{\SetFigFont{12}{14.4}{\rmdefault}{\mddefault}{\updefault}{\color[rgb]{0,0,0}$H_{\rho}$}%
}}}}
\put(7846,-5940){\makebox(0,0)[lb]{\smash{{\SetFigFont{12}{14.4}{\rmdefault}{\mddefault}{\updefault}{\color[rgb]{0,0,0}$H_2$}%
}}}}
\put(7850,-3536){\makebox(0,0)[lb]{\smash{{\SetFigFont{12}{14.4}{\rmdefault}{\mddefault}{\updefault}{\color[rgb]{0,0,0}$\Kcal := \{(t, x) \, / \, |x| <t-1\}$}
}}}}
\put(7300,-3100){\makebox(0,0)[lb]{\smash{{\SetFigFont{12}{14.4}{\rmdefault}{\mddefault}{\updefault}{\color[rgb]{0,0,0}$t$}
}}}}
\put(11700,-8650){\makebox(0,0)[lb]{\smash{{\SetFigFont{12}{14.4}{\rmdefault}{\mddefault}{\updefault}{\color[rgb]{0,0,0}$|x|$}
}}}}
\put(11680,-7230){\makebox(0,0)[lb]{\smash{{\SetFigFont{12}{14.4}{\rmdefault}{\mddefault}{\updefault}{\color[rgb]{0,0,0}$t=1$}%
}}}}
\put(11680,-6290){\makebox(0,0)[lb]{\smash{{\SetFigFont{12}{14.4}{\rmdefault}{\mddefault}{\updefault}{\color[rgb]{0,0,0}$t=2$}%
}}}}
\end{picture}
\caption{Initial hyperboloid -- hyperboloidal foliation}
\label{fig:coord}
\end{figure}
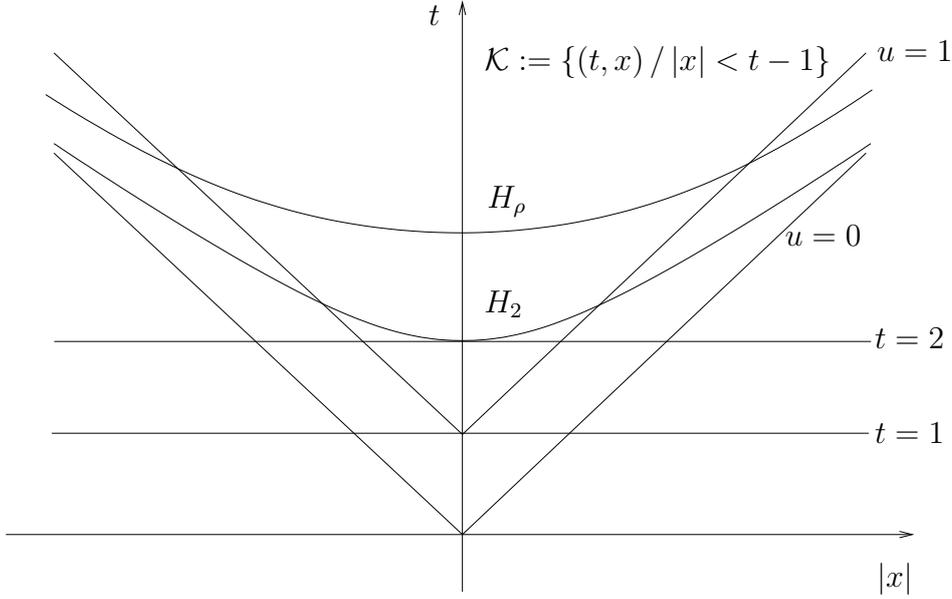
\end{center}

On $\Kcal$, we use as an alternative to the Cartesian coordinates the following two other sets of coordinates.

\begin{itemize}
\item {Spherical coordinates}\vspace{0.2cm}\\
We first consider spherical coordinates $(r,\omega)$ on $\mathbb{R}^3_x$, where $\omega$ denotes spherical coordinates on the $2$-dimensional spheres and $r=|x|$. Then, $(\rho:=\sqrt{t^2-|x|^2},r,\omega)$ defines a coordinate system on $\Kcal$. These new coordinates are defined globally on $\Kcal$ apart from the usual degeneration of spherical coordinates and at $r=0$.
\vspace{0.2cm}

\item {Pseudo-Cartesian coordinates}\vspace{0.2cm}\\
These are the coordinates $(y^0,y^j):=(\rho,x^j)$, which are also defined globally on $\Kcal$.
\end{itemize}

For any function defined on (some part of) $\Kcal$, we move freely between these three sets of coordinates.
\subsubsection{The semi-hyperboloidal frame}
As in \cite{lm:gsmkg}, we make use of special vector fields adapted to the hyperboloidal foliation.
More precisely, given the coordinate system $(x^\alpha)$, we denote the
$$
\mbox{translations by}\quad \left\{\partial_{x^\alpha}\,\big| \,\alpha\in\{0,\hdots,3\}\right\}.
$$ We also consider the
$$
\mbox{Lorentz boosts}\quad \left\{ Z_i = t \partial_{x^i} + x^i \partial_{t}\,\big |\,i\in\{1,2,3\}\right\}
$$
 and the
$$\mbox{rescaled Lorentz boosts} \quad \left\{\delu_i=\frac{Z_i}{t}\,\big|\, i\in\{1,2,3\}\right\}.$$

With these notations, the {\sl semi-hyperboloidal frame} is, by definition,
\begin{equation}
\delu_0 := \del_t,
\qquad \delu_a:= \frac{x^a}{t} \del_t + \del_a, \qquad a=1,2,3.
\end{equation}

We also consider the scaling vector field $S=x^\alpha \partial_{x^\alpha}$ and denote by $Z$ any of the homogeneous vector fields (i.e.~any of the Lorentz boosts or scaling vector field).

We use standard multi-index notations so that, for instance, for any multi-index $\alpha$ of lenghth $|\alpha|$, $Z^\alpha$ denotes a combination of $|\alpha|$ vector fields among the Lorent boosts or the scaling vector field.

Associated with the semi-hyperboloidal frame, one has the dual frame $\thetau^0:= dt - \frac{x^a}{t} \, dx^a$, $\thetau^a: = dx^a$ and therefore the relations
$$
\aligned
&\delu_{\alpha} = \Phi_{\alpha}^{\alpha'} \del_{\alpha'},
\quad \del_{\alpha}
= \Psi_{\alpha}^{\alpha'} \delu_{\alpha'},
\qquad
\thetau^{\alpha} = \Psi_{\alpha'}^{\alpha} \, dx^{\alpha'},
\quad
dx^{\alpha} = \Phi^{\alpha}_{\alpha'} \thetau^{\alpha'},
\endaligned
$$
in which the transition matrix $\left(\Phi_{\alpha}^{\beta} \right)$ and its inverse $\left(\Psi_{\alpha}^{\beta} \right)$ are
$$
\big(\Phi_{\alpha}^{\beta} \big)
=
\left(
\aligned
&1 &&0 &&&0 &&&&0
\\
&x^1/t &&1 &&&0 &&&&0
\\
&x^2/t &&0 &&&1 &&&&0
\\
&x^3/t &&0 &&&0 &&&&1
\endaligned
\right),
\qquad
\qquad
\big(\Psi_\alpha^{\beta} \big)
=
\left(
\aligned
&1 &&0 &&&0 &&&&0
\\
-&x^1/t &&1 &&&0 &&&&0
\\
-&x^2/t &&0 &&&1 &&&&0
\\
-&x^3/t &&0 &&&0 &&&&1
\endaligned
\right).
$$

In the above matrix notation, the up index labels the columns and the down index labels the lines.

With this notation, the rules for transforming tensors are as follows: for any two-tensor $h_{\alpha \beta} \, dx^\alpha \otimes dx^{\beta} = \underline{h}_{\alpha \beta} \thetau^{\alpha} \otimes \thetau^{\beta}$, we can write
\begin{eqnarray*}
\hu_{\alpha \beta} &=& h_{\alpha'\beta'} \Phi_\alpha^{\alpha'}
 \Phi_{\beta}^{\beta'}, \\
h_{\alpha \beta} &=& \hu_{\alpha'\beta'} \Psi_\alpha^{\alpha'} \Psi_{\beta}^{\beta'},\\
h^{\mu \nu}  &=& \Phi^{\mu}_{\mu'}\Phi^{\nu}_{\nu'} \underline{h}^{\mu'\nu'},\\
\underline{h}^{\mu \nu}  &=& \Psi^{\mu}_{\mu'}\Psi^{\nu}_{\nu'} h^{\mu'\nu'}.
\end{eqnarray*}
Similarly, for a velocity $v_\alpha$,
\begin{eqnarray*}
v_{\mu} &=& \Psi_{\mu}^{\mu'} \underline{v}_{\mu'},\\
\underline{v}_{\mu} &=& \Phi_{\mu}^{\mu'} {v}_{\mu'}.
\end{eqnarray*}

With these notations, note that

$$
\underline{v}_0=v_0, \quad \underline{v}_a= \frac{x^a}{t}v_0 + v_a.
$$

 As a consequence, for any symmetric $2$-tensor $G$,

$$
G^{\alpha \beta} v_\alpha v_\beta = \underline{G}^{00} (v_0)^2 + 2 \underline{G}^{a 0} v_0 \underline{v}_a +  \underline{G}^{a b} \underline{v}_a \underline{v}_b.
$$

\subsection{Some standard classes of functions} \label{se:scf}
Later in the analysis, we compute various commutators. The expressions that we find are linear combinations with coefficients which are smooth homogeneous functions of $(t,x)$, sometimes depending also on $v$, and which behave well with respect to differentiation.
More precisely, let $\mathcal{F}_x$ be the set of functions $g:=g(t,x)$ defined in $\Kcal$ such that, for any multi-index $\alpha$
$$
| \partial_{x,t}^\alpha g| \le C_{\alpha} t^{-|\alpha|}
$$
and similarly, let $\mathcal{F}_{x,v}$ be the set of functions $g:=g(t,x,v)$ defined in $\Kcal \times \mathbb{R}^3$ such that, for any pair of multi-indices $\alpha$, $\beta$
$$
| \partial_{x,t}^\alpha \partial_{v}^\beta g| \le C_{\alpha} t^{-|\alpha|} \sqrt{1+|v|^2}^{\,-|\beta|}.
$$

In the remainder of the article, unless specified otherwise, any linear combination is taken over $\mathcal{F}_x$ when we consider functions of $(t,x)$ only or $\mathcal{F}_{x,v}$ when we consider functions of $(t,x,v)$.

\subsection{Commutators of the frame vector fields}

We recall the following standard formulae.

\begin{lemma}
One has
$$
[\del_t, \delu_a] = - \frac{x^a}{t^2} \del_t, \quad [\delu_a, \delu_b] = 0.
$$
There exist constants $a_{\alpha}^\beta$ such that
$$
[\del_{x^\alpha}, Z] = \sum_\beta a_{\alpha}^\beta \del_{x^\beta}.
$$

For a Lorentz boost $Z_a$,
$$
[Z_a, \delu_{b}]
 = \frac{x^b}{t} \delu_{x^a},
$$
while for the scaling vector field $S$
$$
[S, \delu_b ]= -\delu_b.
$$
The above commutators can then be iterated straightforwardly. For instance, for any multi-index $\gamma$,
$$
[\del_{x^\alpha}^\gamma, \delu_c] = \sum_{ | \beta | \le | \gamma|} c_\beta \partial_{x^\alpha}^\beta
$$
holds, where $c_\beta \in \frac{1}{t}\mathcal{F}_x$.
\end{lemma}

\subsection{From the geometric initial data to the PDE data at $t=2$}
We recall that it follows from the constraint equations and the structure of the data (see for instance \cite[Section 4]{Lindblad:2005ex} and \cite[Chapter VI]{MR3186493} for the additional Vlasov field) that given geometric data $(\Sigma,g_0, k, f_0)$ satisfying the constraint equations \eqref{eq:cons1}-\eqref{eq:cons2}, there exists PDE data, $g_{t=0}$, $\partial_t g_{t=0}$, $f_{t=0}$, of the reduced Einstein-Vlasov system, where $g_{t=0}$, $\partial_t g_{t=0}$ are defined on $\mathbb{R}^3$ and $f_{t=0}$ is defined on $\mathbb{R}^3 \times \mathbb{R}^3$. Moreover, if $(\Sigma,g_0, k, f_0)$ coincides with the data of a Schwarzschild spacetime outside from a ball of radius $R$ and verify $$||g_0-\delta||_{H^{N}}+||k||_{H^{N-1}} + || (1+|v|^2)^{q/2}f_0 ||_{W^{N+3,1}}+ || (1+|v|^2)^{(q+2)/2}f_0 ||_{W^{N-2,1}}\le \epsilon,$$where $\delta$ is the Euclidean metric, then we can arrange for the PDE data to verify the following.
\begin{itemize}
\item The following estimates hold
\begin{eqnarray*}
||g_{t=0}-\delta||_{H^N} , || \partial_t g_{t=0} ||_{H^{N-1}} &\le& C(R) \epsilon,  \\
|| (1+|v|^2)^{q/2}f_{t=0} ||_{W^{N+3,1}}+|| (1+|v|^2)^{(q+2)/2}f_0 ||_{W^{N-2,1}} &\le& C(R)\epsilon.
\end{eqnarray*}
\item $\partial_t g_{t=0}$ and $f_{t=0}$ vanish for $|x| > R$.
\item $g$ coincides with the metric of a Schwarzschild metric in wave coordinates for $|x| > R$.
\end{itemize}

Moreover, we can shift the initial slice $t=0$ to any other initial time, and we assume here that the pde data is given at time $t=2$.

By rescaling, we may as well assume that $R=1$ and we shall do so in the remainder of the paper. To summarize, we assume that the initial data is given on $\{t=2\}$  and coincide with Schwarzschild data outside of the ball $\{ |x| < 1 \}$.

We recall that the Schwarzschild metric $g_m$ in standard wave coordinates $(t=x^0, x^1, x^2, x^3)$ takes the form (see for instance \cite{Asanov1989})
\begin{eqnarray*}
{g_m}_{\,00} &=& - \frac{r-m}{r+m},\\
{g_m}_{\,ab} &=& \frac{r+m}{r-m} \omega_a \omega_b + \frac{(r+m)^2}{r^2}(\delta_{ab} - \omega_a \omega_b),
\end{eqnarray*}
with $\omega_a := x_a/|x|$ and $r=|x|$.

\subsection{Energy norms for the metric components}

Let
$$
H_\rho^\star := H_\rho \cap \Kcal
$$
be the intersection of the hyperboloid of constant $\rho$ and the interior of the future light-cone $\Kcal$. For any sufficiently regular function $\psi:=\psi(t,x)$, we define its energy norms on $H_\rho$ and $H_\rho^\star$, first using the perturbed metric $g$,
\begin{eqnarray*}
\Ecal_{g}[\psi](\rho)
& = & \int_{H_\rho} \Big(-g^{00}|\del_t \psi|^2 + g^{ab} \del_a\psi\del_bu + \sum_a \frac{2x^a}{t}g^{a \beta} \del_{\beta}\psi\del_t \psi\Big) \, dx,
\\
\Ecal_{g}^\star[\psi](\rho)
& = &  \int_{H_\rho^\star} \Big(-g^{00}|\del_t \psi|^2 + g^{ab} \del_a\psi\del_b\psi + \sum_a \frac{2x^a}{t}g^{a \beta} \del_{\beta}\psi\del_t \psi\Big) \, dx
\end{eqnarray*}
and then, using the flat metric $\eta$,

\eq{\alg{
\Ecal[\psi](\rho)=\Ecal_{\eta}[\psi](\rho)
& =  \int_{H_\rho} \Big(|\del_t \psi|^2 + \sum_a|\del_a\psi|^2 + \sum_a \frac{2x^a}{t} \del_a\psi\del_t \psi \Big) \, dx,
\\
\Ecal^\star[\psi](\rho)=\Ecal^\star_{\eta}[\psi](\rho)
& =
 \int_{H_\rho^\star} \Big(|\del_t \psi|^2 + \sum_a|\del_a\psi|^2 + \sum_a \frac{2x^a}{t} \del_a\psi\del_t \psi \Big) \, dx.
}}
Note that in the above norms, we have used the standard measure $dx$ on $H_\rho \simeq \mathbb{R}^3$, but the energy norm
$\Ecal[\psi](\rho)$ actually coincides with the standard energy associated with the energy-momentum tensor of a wave $u$, the multiplier vector field $\partial_t$ and the induced volume form on $H_\rho$. Indeed, the volume form induced by the Minkowski metric on each $H_\rho$ is $\frac{\rho}{t}dx$, but the normal to $H_\rho$ also various weights cancelling the $\frac{\rho}{t}$ to give the above expression for $\Ecal[\psi]$ and $\Ecal^\star[\psi]$.

We recall the following coercivity properties of the energy $\Ecal[\psi](\rho)$
\begin{eqnarray}
\Ecal[\psi](\rho)& = & \int_{H_\rho} \Big((\rho/t)^2 |\del_t \psi|^2 + \sum_a|\delu_a \psi|^2 \Big) \, dx. \label{eq:ewc}
\end{eqnarray}

Finally, we define the higher order energy norms
\begin{eqnarray} \label{def:hon}
\Ecal_N[\psi] &=& \sum_{| \alpha | \le N} \Ecal[ K^\alpha \psi ],\\
\Ecal^\star_N[\psi] &=& \sum_{| \alpha | \le N} \Ecal^{\star}[ K^\alpha \psi ],
\end{eqnarray}
where $K^\alpha$ denotes a combination of $|\alpha|$ vector fields among the translations, hyperbolic rotations and scaling vector fields.
\subsection{Energy norms for the Vlasov field}
We define in the following the fundamental energies for the distribution function. We therefore consider the one-form
$$
n:=t dt - x_i dx^i
$$ normal to $H_\rho$ with
$$
-\eta^{-1} (n,n)= \rho^2.
$$
We define similarly,
\eq{ \label{s-def}
s^2:=-g^{-1} \left( t dt - x_i dx^i, t dt - x_i dx^i \right).
}
For later purposes, we state the identity
\eq{ \label{id:dsr}
s^2-\rho^2=-H^{\alpha\beta}x_\alpha x_\beta.
}

Let us define two energy densities based on the energy-momentum tensor of $f$ (cf.~\eqref{eq-enmt}) in the perturbed and flat case, respectively.
\eq{\alg{
\chi_g(f)&:= -T_{0\beta}[f] s^{-1} n_\alpha g^{\alpha \beta}
}}
\eq{\alg{
\chi(f)&:=\int_{v}\sqrt{1+|v|^2}f\frac{ \sqrt{1+|v|^2} t-x_i v^i }{\rho}\frac{dv}{\sqrt{1+|v|^2}}
}}
Note that $\chi$ without the $g$ index is $\chi_g$ evaluated for $g=\eta$ and that
\eq{
\nu_g:=-s^{-1} n_\alpha g^{\alpha \beta}= -s^{-1}\left(t g^{\alpha 0} \partial_\alpha-{x_i}{}g^{i\alpha}\partial_\alpha\right)
}
is the future unit normal vector field to $H_\rho$ for the metric  $g$. Let us define accordingly the two energy norms on $H_\rho$
\begin{eqnarray}
E_{g}[f]&:=&\int_{H_{\rho}} \chi_g(|f|)d\mu_{H_\rho,g}, \label{def:egf}\\
E[f]=E_{\eta}[f]&:=& \int_{H_{\rho}} \chi(|f|)d\mu_{H_\rho,\eta}, \label{def:ef}
\end{eqnarray}
where $d\mu_{H_\rho,g}$ and $d\mu_{H_\rho,\eta}(=\frac{\rho}{t}dx)$ are the induced volume forms on $H_\rho$  respectively for the $g$ and $\eta$ metrics. Note that in the above energy densities, we replaced $f$ by $|f|$, as the positivity property of the distribution function is not necessarily preserved by commutation.

We compute

\eq{\begin{aligned}
\chi_g(|f|)
&=-T_{0\beta}[|f|]s^{-1}n_\alpha g^{\alpha \beta}\\
&=s^{-1}\left(T_{00}[|f|] (-tg^{00}+x_ig^{i0}) + T_{i0}[|f|] (-tg^{0i}+x_jg^{ji})\right)\\
&= s^{-1}\left[\int_{v}|f|v_0\left(v_0 (-tg^{00}+x_ig^{i0})+v_i (-tg^{0i}+x_jg^{ji})\right)\mu_{\pi^{-1}(x)}\right]\\
&= s^{-1}\left[\int_{v}|f|v_0\left(v_0 (-tg^{00}+x_ig^{i0})+v_i (-tg^{0i}+x_jg^{ji})\right)\frac{\sqrt{-g^{-1}}}{v_\beta g^{\beta 0}}dv\right].
\end{aligned}
}
This implies

\eq{\alg{
&E_g[f]\\
&= \int_{H_\rho} \left[\int_{dv}\frac{|f|v_0\left(v_0 (-tg^{00}+x_ig^{i0})+v_i (-tg^{0i}+x_jg^{ji})\right)}s\frac{\sqrt{-g^{-1}}}{v_\beta g^{\beta 0}}dv\right]
\sqrt{-g} \frac s{g^{\alpha 0}x_\alpha} dx
}}
and for the energy with respect to the unperturbed metric,
\eq{\label{bck-en}
E_\eta[f]=\int_{H_\rho} \int_v \frac{|f|\left(\sqrt{1+|v|^2}t-v_ix^i\right)}{\rho}\frac{\rho}{t}dvdx.
}

\subsubsection{Coercivity of the energy $E[f]$} \label{se:coeref}

Recall that given mass-shell coordinates $(x^\alpha, v_i)$, we compute $v_0$ by solving the mass-shell equation $v_\alpha v_\beta g^{\alpha \beta}=-1$. It will be useful to also define coordinates on the corresponding Minkowski mass shell, denoted $(x^\alpha, w_{\alpha})$ with
\eq{\alg{
w_{0}&:= - \sqrt{1+\eta^{ij}v_iv_j} = - \sqrt{1+ |v|^2},\\
w_i &: =  v_i.
}}
Thus, for instance, $w^\alpha \partial_{x^\alpha}$ denotes the free transport operator for the Minkowski space.

Consider
$$\nu_\rho = \rho^{-1} x^\alpha \partial_{x^{\alpha}},$$
the normalized normal vector field, for the Minkowski metric, to the hyperboloidal foliation.
\begin{lemma} The following identities hold
  \begin{eqnarray*}
 \nu^{\alpha}_\rho w_\alpha w_0 &=&   \frac{t}{2 \rho} \left(\frac{\rho^2}{t^2} (w^0)^2 +1+ \sum_{i=1}^3(\vu_i)^2 \right)\\
 &=& \frac{t}{2 \rho} \left(\frac{\rho^2}{t^2} (w^0)^2 +1+ \sum_{i=1}^3\mathfrak{z}_i^2 t^{-2} \right)\\
&=&  \frac{t}{2\rho} \left(\left(w_ 0  +  \frac{x^i}{t}w_ i \right)^2 + \sum_{1\leq i< j \leq 3}  \left( \frac{ x^i w_j -  x^jw_i}{t} \right)^2 + \frac{\rho^2}{t^2}w_0^2 + \frac{r^2}{t^2} \right),
  \end{eqnarray*}
  where $\mathfrak{z}_i$ is the hyperbolic weight
  $
\mathfrak{z}_i =  w_0 x_i + w_i t
  $
  and $\vu_i=\frac{\mathfrak{z}_i}{t}$.
\end{lemma}

This leads to a decomposition of the energy norm $E[f]$ similar to \eqref{eq:ewc}

\begin{eqnarray}
\label{dec-bck-en}
E[f]&=&\int_{H_{\rho}} \int_{v}|f|\frac{1}{2 w^0}\left(\frac{\rho^2}{t^2}(w^0)^2+1+\sum_{i=1}^3\vu_i^2 \right)dvdx.
\end{eqnarray}

\subsection{Analysis of the support}
Since our data is exactly Schwarzschild outside from a compact ball, it follows by a standard domain of dependence argument that the solution will be exactly Schwarzschild outside from some cone.

More precisely,

\begin{proposition} \label{prop:compactsupport}
Let $(g, f)$ be a solution to the system \eqref{eq:evs1}-\eqref{eq:evs2} with initial data given on $t=2$ coinciding with the data induced by the Schwarzschild metric for $|x| \ge 1$. Then $(g - g_m)$ is supported in the region $\Kcal$ and vanishes in a neighborhood of the boundary $\{r=t-1, t\geq 2\}$.
\end{proposition}
\begin{proof}

We follow here the standard argument, as explained in \cite[Section 4.2]{lm:gsmkg} (see also \cite[Lemma 4.1]{Lindblad:2005ex}). Essentially, we rewrite the equations in terms of $q=g-g_m$ and $f$, where $g_m$ is the Schwarzschild metric. We then prove that the equations are homogeneous (i.e.~$q=0$, $f=0$ is a solution) and use the domain of dependence property. However, since the Schwarzschild metric becomes singular for small $r$, we change it in the interior of the cone $\{t-|x|=3/2\}$.

We only give a sketch of the proof since the parts concerning the wave equations are identical to that of \cite[Proposition 2.3]{lm:gsmkg} and the arguments concerning the Vlasov equation are very similar.

Let
$$
p(t,x):= \left( g_m -\eta\right) \xi(t-r)+ \eta,
$$

where $\xi$ is a smooth real function such that $0 \le \xi \le 1$, $\xi (x)=1$ for $x \le 1$, while $\xi(x) = 0$ for $x \ge 3/2$.

By construction, for $r\ge t-1$, $p$ coincides with the Schwarzschild metric while for $r\le t-3/2$, $p$ coincides with the Minkowski metric.  \\
Let
$
q := g - p
$.
Recall that the reduced Einstein equations can be schematically written as
$$
\widetilde{\Box}_g g = F(g, \partial g) -S[f],
$$
where $S[f]$ is the tensor defined in \eqref{def:sf} and $F$ has the structure \eqref{eq:F}.

The equations for $q$ can then be written as
\begin{eqnarray*}
\widetilde{\Box}_p q   & = & - \widetilde {\Box}_pp + B \big(p,p, \del p, \del p\big)
\\
&&\hbox{} + B \big(p,p, \del p, \del q\big) + B \big(p,p, \del q, \del(p+q) \big)
\\
&&\hbox{} + B \big(p,q, \del(p+q), \del(p+q) \big) + B \big(q,p+q, \del(p+q), \del(p+q) \big)
\\
&&\hbox{}-  q^{\mu\nu} \del_{\mu} \del_{\nu}g  -S[f],
\end{eqnarray*}
where $B$ is some multi-linear form.
By definition, for $r\ge t-1$, $p$
coincides with the Schwarzschild metric in wave coordinates, which is a solution to the (vacuum) reduced Einstein equation.
Thus, for $r\geq t-1$, we have
$$
\widetilde{\Box}_p p  = B(p,p, \del p, \del p).
$$

To treat the distribution function $f$, let us define $\tilde{w}_i =  v_i$ and then $\tilde{w}_0$ as the solution to
$$p^{\alpha \beta} w_\alpha w_\beta=-1, \quad \tilde{w}_0 < 0.$$

Note that it then follows from the definitions that $\tilde{w}_0 -v_0=\tilde{B}(q,v)$ for some $\tilde{B}(q,v)$ which vanishes for $q=0$.

With this definition, $f$ satisfies the following transport equation with respect to the metric $p$
\begin{eqnarray*}
\tilde{w}_{\alpha}p^{\alpha\beta}\partial_{x^\beta} f - \frac{1}{2}\tilde{w}_{\alpha}\tilde{w}_{\beta}\dfrac{\partial p^{\alpha\beta}}{\partial x^{i}}\partial_{\tilde{w}_i}f &=& -\left(v_{\alpha}-\tilde{w}_{\alpha}
\right)p^{\alpha\beta}\partial_{x^\beta} f \\
&&\hbox{}+ \frac{1}{2}\left(v_{\alpha}-\tilde{w}_{\alpha}\right)\tilde{v}_{\beta}\dfrac{\partial g^{\alpha\beta}}{\partial x^{i}}\partial_{v_i}f.
\end{eqnarray*}

We conclude that $(q,f)$ satisfies
\begin{eqnarray*}
\widetilde{\Box}_p q &=& - \widetilde{\Box}_p p + F \big(p,p, \del p, \del p\big)\\
&& \hbox{}+ \del q\cdot G_1(p, \del p, q, \del q)  \\
&&\hbox{}+ q\cdot G_2(p, \del p, \del\del p,q, \del q)
-S[f],
\\
\tilde{w}_{\alpha}p^{\alpha\beta}\partial_{x^\beta} f - \frac{1}{2}\tilde{w}_{\alpha}\tilde{w}_{\beta}\dfrac{\partial p^{\alpha\beta}}{\partial x^{i}}\partial_{\tilde{w}_i}f &=& -\left(v_{\alpha}-\tilde{w}_{\alpha}
\right)p^{\alpha\beta}\partial_{x^\beta} f \\
&&\hbox{}+ \frac{1}{2}\left(v_{\alpha}-\tilde{w}_{\alpha}\right)\tilde{v}_{\beta}\dfrac{\partial g^{\alpha\beta}}{\partial x^{i}}\partial_{v_i}f.
\end{eqnarray*}
and moreover, by assumption, $q_{t=2}$, $\partial_t q_{t=2}$, $f_{t=2}$ vanish for $|x| > 1$.

To prove that $(q,f)$ vanishes outside $\Kcal$, it remains to analyse the domain of dependence associated with the metric $p$ outside $\Kcal$. This is in fact a direct consequence of the fact that the boundary of $\Kcal$ is strictly spacelike with respect to the metric $p$ (cf Step 2 in \cite[Proposition 2.3]{lm:gsmkg} and \cite[Lemma 4.1]{Lindblad:2005ex}).

This leads to the conclusion that the domain of dependence of $(t, x) \notin \Kcal$ does not intersect $\{t=2,r\leq t-1\}$, which implies that $q$ and $f$ must vanish outside $\Kcal$.
\end{proof}

\subsection{The data induced on $H_2$} \label{se:dih2}
By standard local well-posedness, the solution $(g,f)$ to the reduced Einstein equation emanating from $t=2$ with sufficiently small initial data necessarily exists up to $t=3$.

It follows from the last section that the trace of $f$ and $g-g_m$ on $H_2$ is compactly supported and contained within the set $\{ 2 \le t \le 3 \} \cap \Kcal$. In particular, provided the initial data is small enough, we have
\begin{eqnarray*}
\Ecal^\star_{N}[h](2)&\le& 2\epsilon, \\
E_{N,q}[f](2) &\le& 2\epsilon, \\
E_{N-2,q+2}[f](2) &\le& 2\epsilon,
\end{eqnarray*}
where the norms $\Ecal^\star_N[h]$ and $E_{N,q}[f]$ are defined in \eqref{def:hon} and \eqref{def:enq} respectively\footnote{While this norm is defined using modified vector fields, it is immediate that for small enough initial data the estimate \ref{es:dv} holds. Note also that in view of the initial data assumptions on $f$, we also have bounds for $[\widehat{Z}^\alpha f]_{\rho=2}$ up to $|\alpha|=N+3$. We will only use these later in Section \ref{se:l2vf}.}.

\subsection{A list of notations} \label{se:ln}
We list here some of the notations we will use throughout the paper for reference.
\begin{itemize}
\item The different metrics
\begin{itemize}
\item $g$ the curved metric.
\item $\eta$ the Minkowski metric.
\item $\delta_E$ the Euclidean metric
\item $g_m$ the Schwarzschild metric written in standard wave coordinates.
\end{itemize}
\item Indices of tensors
\begin{itemize}
\item Greek indices $\alpha, \beta, \mu, \nu$ for tensors such as $h_{\alpha \beta}$, $T_{\mu \nu}$ etc.. These are spacetime indices running from $0$ to $3$.
\item Latin indices $i,k, a, b$ for tensors such as $h_{ab}$. These are spatial indices only, running from $1$ to $3$.
\end{itemize}
\item Notations for the velocity
\begin{itemize}
\item $w^0=\sqrt{1+|v|^2}$, \,$w_0=-\sqrt{1+|v|^2}$.
\item $w_i=w^i=v_i$.
\item $\wu_i=\vu_i=v_i+\frac{x^i}{t}w_0$, \,$\wu_0=v_0$.
\end{itemize}
\item The vector fields
\begin{itemize}
\item Translation vector fields for metric components: $\partial_{t,x}$, $\partial_{x^\alpha}$, or generically $X$.
\item Rescaled boosts: $\delu_a$ or $\delu_{x^a}$. Note also $\delu_0=\partial_t$.
\item Homogeneous vector fields for metric components: $Z_i=t\partial_{x^i}+x^i \partial_t$ for a Lorentz boost, $S=x^\alpha \partial_{x^\alpha}$ for the scaling vector field, $Z$ for a generic one.
\item Any of the above vector fields for metric components: $K$.
\item Translation vector fields for the Vlasov field: $\partial_{t}=X_0$, $X_i= \partial_{x^i}+ \frac{v_i}{w^0}\partial_t$, or generically $X$.
\item Homogeneous vector fields for Vlasov field: $\widehat{Z}_i=t\partial_{x^i}+x^i \partial_t+w^0 \partial_{v_i}$ for a Lorentz boost, $S=x^\alpha \partial_{x^\alpha}$ for the scaling vector field, $\widehat{Z}$ for a generic one.
\item Modified vector fields: $Y:= \widehat{Z}+C^\alpha X_\alpha$.
\item Any of the above vector fields for the Vlasov field: $\widehat{K}$.
\end{itemize}
\item Higher-order differential operators and multi-indices
\begin{itemize}
\item For $\alpha$ a multi-index of length $|\alpha|$, $K^\alpha$, $\widehat{K}^\alpha$, composition of $|\alpha|$ vector fields as above.
\item For $\alpha$ a multi-index of length $|\alpha|_X$, $|\alpha|=\alpha_X+\alpha_Z$, where $\alpha_X$ is the number of translations and $\alpha_Z$ the number of homogeneous vector fields. Note that this applies whether we consider a differential operator of the form $K^\alpha$ or $\widehat{K}^\alpha$.
\item $K^N$, $\widehat{K}^N$. A differential operator of the form $K^\alpha$ or $\widehat{K}^\alpha$ with $|\alpha|=N$. Similarly,  $Z^N$, $X^N$ means a differential operator composed of $N$ homogeneous vector fields or $N$ translation vector fields.
\end{itemize}
\item The energy norms
\item $\Ecal$, $\Ecal_N$ the energy norm of order $1$ or $N$ for the metric coefficients. Their star version $\Ecal^\star$, $\Ecal_N\star$ means that we only integrate over $H_\rho^\star$, i.e.~we stop when we reach the domain where the solution is exactly Schwarzschild.
\item The source terms
\begin{itemize}
\item $T_{\mu \nu}[f]$ the energy momentum tensor of $f$ with respect to the curved metric $g$.
\item $S[f]:= T[f]+1/2 g \int_v f d\mu_v$.
\item $\Tcal_{\mu \nu}[f]$ the energy momentum tensor of $f$ with respect to the Minkowski metric $\eta$.
\end{itemize}
\item The main operators
\begin{itemize}
\item $\widetilde{\square}_g$ the reduced wave operator $g^{\alpha \beta} \partial_{x^\alpha} \partial_{x^\beta}$.
\item $T_g$ the geodesic spray vector field defining the Vlasov equation.
\end{itemize}
\item Spacetimes and mass-shell functions:
\begin{itemize}
\item Cartesian coordinates $t=x^0,x^i$.
\item  Hyperbolic time $\rho=\sqrt{t^2-|x|^2}$.
\item Radial function $r=|x|$.
\item  Retarded time $u:=t-|x|$.
\item The classes $\mathcal{F}_x$, $\mathcal{F}_{x,v}$ as in Section \ref{se:scf}.
\end{itemize}
\item Constants
\begin{itemize}
\item $D_N$ a strictly positive constant depending only on $N$.
\item $A \lesssim B$ for $A \le D_N B$.
\item $\delta$ a small enough constant depending only $N$.
\item $L$ verifying $\delta <<L << 1$ a small enough constant depending only on $N$.

\end{itemize}
\end{itemize}

\section{Bootstrap assumptions and the structure of the proof} \label{se:basp}

From the previous section, it follows that the trace of the solutions, emanating from initial data on $\{t=2\}$, on $H_2$ is well-defined and we consider the Einstein-Vlasov system with initial data on the initial hyperboloid $H_2$ given by $g_{\alpha \beta}|_{H_2}$, $\del_tg_{\alpha \beta}|_{H_2}$
and $f_{H_2 \times \mathbb{R}_v^3}$.

Denote by $(g,f)$ the unique solution agreeing with the initial data and by $h=g-\eta$ the deviation from the Minkowski metric.

From Section \ref{se:dih2}, it follows that
\begin{eqnarray}
\Ecal_N^\star[h](2)
&\lesssim& \epsilon,  \label{es:dh}\\
E_{N,q}[f](2)& \lesssim & \epsilon, \label{es:dv} \\
E_{N-2,q+2}[f](2)& \lesssim & \epsilon, \label{es:dv2}
\end{eqnarray}
where $E_{N,q}[f]$ is the higher order norm defined in Section \ref{se:hovf}.

Let $\delta>0$ be a small but fixed number, depending only on $N$ and $D$ be a large number depending only on $N$ that will be fixed in course of the analysis.
\subsection{Bootstrap assumptions}
The global analysis is based on only one bootstrap assumption for the metric coefficients
\begin{equation}
\label{eq:bsm}
\Ecal_N[h](\rho)  \le  D \epsilon \rho^{\delta}.
\end{equation}
Let us consider $\rho^*$ such that
$$
\rho^* := \sup \Big\{\rho_1 \, \big| \, \text{for all }2\leq \rho \leq \rho_1, \text{the estimate \eqref{eq:bsm} holds} \Big\}.
$$

From standard well-posedness, $\rho^*> 2$.
The proof of the theorem follows if we can prove the improved estimate
\begin{proposition} \label{prop:ibs} For all $\rho \in [2, \rho^*)$,
we have the improved estimates for the metric coefficients
\begin{equation} \label{eq:ibsm}
\Ecal_N[h](\rho)  \leq \frac D2 \epsilon \rho^{D_N \epsilon^{1/2}},
\end{equation}
where $D_N>0$ is a constant depending only on $N$.
\end{proposition}

\subsection{Direct consequences of the bootstrap assumptions}

\subsubsection{The basic decay estimates for the metric coefficients}
As an immediate consenquence of the boostrap assumptions on $h$, we have, by standard Klainerman-Sobolev inequalities on hyperboloids (see for instance \cite{fjs:vfm}, Proposition 4.12), the following proposition.

\begin{proposition}
One has the decay estimates, for all $|\alpha | \le N-2$, and $\rho \ge 2$,
\begin{eqnarray} \label{es:bde}
| K^\alpha \partial h |(t,x)+| \partial Z^\alpha  h |(t,x) &\lesssim& \epsilon^{1/2} \rho^{\delta/2} \frac{1}{t(1+u)^{1/2} }, \\
&\lesssim& \epsilon^{1/2} \rho^{\delta/2} \frac{1}{t^{1/2} \rho }, \nonumber \\
| K^\alpha h |(t,x) &\lesssim& \epsilon^{1/2} \rho^{\delta/2} \frac{(1+u)^{1/2}}{t}. \label{es:bde2}
\end{eqnarray}
\end{proposition}
\begin{remark}
In Section \ref{se:ideiwe}, we improve the decay rates in $t-|x|$ of the above estimates, using sup-norm estimates and integration along characteristics. These estimates use the wave equations satisfy by the $h$ coefficients, so they require in particular some information about the source terms of these equations. An alternative approach to the global structure of the paper would have been to consider extra bootstrap assumptions, so that one has access to pointwise estimates for these source terms. We have chosen instead to assume just the simplest bootstrap assumptions and to then revisit the decay estimates for $h$. One drawback to our strategy is that we will also need to revisit our estimates for the Vlasov field (see Section \ref{se:ievf}), once we have access to the improved decay estimates on $h$.
\end{remark}

Recall also that additional translations can be converted into extra $u$ decay (cf~Appendix \ref{se:dect}), so that

\begin{proposition} \label{prop:esbde3}
One has the decay estimates, for all $|\alpha | \le N-2$, and $\rho \ge 2$,
\begin{eqnarray} \label{es:bde3}
| K^\alpha \partial h |(t,x)+| \partial Z^\alpha  h |(t,x) &\lesssim& \epsilon^{1/2} \rho^{\delta/2} \frac{1}{t(1+u)^{1/2+ \alpha_X} }, \\
&\lesssim& \epsilon^{1/2} \rho^{\delta/2} \frac{1}{t^{1/2} \rho }, \nonumber \\
| K^\alpha h |(t,x) &\lesssim& \epsilon^{1/2} \rho^{\delta/2} \frac{(1+u)^{1/2-\alpha_X}}{t}, \label{es:bde4}
\end{eqnarray}
where $\alpha_X$ denotes the number of translations in $K^\alpha$.
\end{proposition}

\subsubsection{Basic consequences of the wave gauge condition and the bootstrap assumptions} \label{se:bcwgc}
We recall here that when the gauge condition holds, some of the $h$ coefficients behave better than others (cf.~Lemma 4.6 of \cite{lm:gsmkg}).

\begin{lemma} \label{lem:wc}
Let $(g_{\alpha \beta})$ be a metric satisfying the wave gauge condition \eqref{eq:wc}.  Then $\del_t\hu^{00}$ can be written as a linear combination of
\begin{equation} \label{dec:h00}
(\rho/t)^2\del_{\alpha} \hu^{\beta \gamma}, \quad \delu_a \hu^{\beta \gamma}, \quad t^{-1} \hu^{\alpha \beta}, \quad \hu^{\alpha \beta} \del_{\gamma} \hu^{\alpha'\beta'}, \quad t^{-1}h_{\alpha \beta} \hu^{\alpha'\beta'}.
\end{equation}
\end{lemma}
\begin{remark}
Note that the first term in \eqref{dec:h00} contains in particular a $\partial_t \hu^{00}$ but that term also contains an additional good weight of $(\rho/t)^2$.
\end{remark}

This implies, in conjunction with the basic decay estimates \eqref{es:bde} and commutation,

\begin{lemma} We have the improved decay estimates, for any $|\alpha | \le N-2$,
\begin{eqnarray} \label{es:ide1}
|\partial K^\alpha \hu^{00} | \lesssim \epsilon^{1/2} \rho^{\delta/2} \frac{\rho}{t^{5/2}}.
\end{eqnarray}
Also,
\begin{eqnarray} \label{es:ide2}
| K^\alpha \hu^{00} | &\lesssim& \epsilon^{1/2} \rho^{\delta/2} \frac{\rho^{3}}{t^{7/2}}+ \frac{\epsilon^{1/2}}{t}  \\
&\lesssim&  \epsilon^{1/2} \rho^{\delta/2} \frac{u^{3/2}}{t^2}+\frac{\epsilon^{1/2}}{t}. \nonumber
\end{eqnarray}

As in Proposition \eqref{prop:esbde3}, we have in fact
$$
| K^\alpha \hu^{00} |
\lesssim  \epsilon^{1/2} \rho^{\delta/2} \frac{u^{3/2-\alpha_X}}{t^2}+\frac{\epsilon^{1/2}}{t^{1+\alpha_X}},
$$
where $\alpha_X$ denotes the number of translations in $K^\alpha$.
\end{lemma}

\begin{proof}
For \eqref{es:ide1}, recall the decomposition $\partial_a=\underline\partial_a-x^a/t\partial_t$ and then use the improved decay for $\partial_t\underline h_{00}$. For \eqref{es:ide2}, we use \eqref{es:ide1} and integrate along the integral curves of the vector field $\partial_t - \frac{x^i}{| x|} \partial_{x^i}$ using that the solution is Schwarzschild outside from $\Kcal\cap \{ t \ge 2\}$.
\end{proof}

\subsubsection{Comparison of flat and curved energy norm for the wave equation}
As in \cite{lm:gsmkg}, Lemma 7.2, we have,
\begin{lemma}
There exists an $\varepsilon>0$ such that
\begin{eqnarray*} \label{cond-equiv}
\|g-\eta\|_{\infty} &\lesssim& \varepsilon^{1/2},\\
\|\underline h^{00}\|_{\infty}&\lesssim& \varepsilon^{1/2} \frac{\rho^2}{t^2},\\
\|\underline h^{0a}\|_{\infty}&\lesssim& \varepsilon^{1/2}�\frac{\rho}{t},\\
\|\underline h^{ab}\|_{\infty}&\lesssim& \varepsilon^{1/2},
\end{eqnarray*}
imply that
\begin{enumerate}
\item $\Ecal_{g}$ and $\Ecal_{\eta}$ are equivalent, i.e.~there exists a uniform constant $C>0$ such that, for any regular field $\psi$,
$$
C^{-1}\Ecal_{g}[\psi]\leq \Ecal_{\eta}[\psi]\leq C \Ecal_{g}[\psi].
$$
\item $\Ecal^\star_{g}$ and $\Ecal^\star_{\eta}$ are equivalent, i.e.~there exists a uniform constant $C>0$ such that, for any regular field $\psi$,
$$
C^{-1}\Ecal^\star_{g}[\psi]\leq \Ecal^\star_{\eta}[\psi]\leq C \Ecal^\star_{g}[\psi].
$$
\end{enumerate}
\begin{remark} In view of the decay estimates \eqref{es:bde} and \eqref{es:ide2}, the assumptions of the lemma hold provided the bootstrap assumptions hold.
\end{remark}
\end{lemma}

\subsection{Structure of the proof}

We list here the main steps of the proof.
\begin{enumerate}
\item[Step 0] We first analyse in Section \ref{se:pave} the basic properties of the Vlasov equation. In particular, we describe the basic energy estimate for the Vlasov field and prove that, under the bootstrap assumptions, the energy norm $E_g$ is equivalent to the one in Minkowski space $E_\eta$.
\item[Step 1] In Section \ref{se:cve}, we define the algebra of commutator vector fields, systematically analyze each first order commutator and close the energy estimates after one commutation.
\item[Step 2] In the first part of Section \ref{se:hovf}, from the bootstrap assumptions only, we prove energy estimates for the Vlasov field up to $N-2$ commutations, using that we have access to the basic pointwise decay estimates on all metric related quantities $h$ (see  Proposition \ref{prop:eesvfl}). These energy estimates can in particular be weighted by powers of $v$ (see Section \ref{se:multivz}), to allow for future losses in $v$.
\item[Step 3] Using Klainerman-Sobolev type inequalities for the Vlasov field (see Section \ref{se:KSf}), we then have access to decay estimates for velocity averages of $f$ after a small enough number of commutations. This allows to push the energy estimates for the Vlasov field up to order $N$ (Proposition \ref{prop:eevfn}).
\item[Step 4] We then consider the wave equations for $h$. The error terms not involving the Vlasov field have been considered before (\cite{lr:gsmhg} and \cite{lm:gsmkg}) and are stated in Section \ref{sec:aree}.  It remains to consider the new terms coming from the Vlasov field.
\item[Step 5] In Section \ref{se:comt}, we give formulas for the commutation of the energy momentum tensor. Here, it is important to distinguish between the top order case, when the number of commutations is $N$, and the non top order case.
\item[Step 6] For a low number of commutations, we can estimate directly the source term coming from the energy-momentum tensor, using the previous commutations and the decay estimates of Section \ref{se:KSf}. For a large number of commutations up to $N-1$, we also use $L2$ decay estimates for the Vlasov field, proven in Section \ref{se:l2vf}.
\item[Step 7] It remains to estimate the contribution of the energy momentum tensor at top order. This case is different from the others, as the source term actually depends linearly on $\partial K^N h$. Moreover, it has borderline decay. In order to close the estimates, we then need to improve the estimates on the Vlasov field, to remove the losses at low order and replace the $\rho^{D \delta}$ losses by $\rho^{D\epsilon^{1/2}}$ up to $N-1$ order. This is the purpose of Section \ref{se:ievf}.
This then gives an extra $\rho^{D\epsilon^{1/2}}$ growth to the top order energy, but still improves the bootstrap assumptions and conclude the proof of the theorem.
\end{enumerate}

\section{Preliminary analysis of the Vlasov equation} \label{se:pave}

\subsection{Decomposition of the operator $v_\alpha g^{\alpha \beta} \partial_{x^\beta}$}

Recall that the transport operator $T_g$ is given by
\eq{
T_g:= v_\alpha g^{\alpha \beta} \partial_{x^\beta} - \frac{1}{2}v_\alpha v_\beta \partial_{x^i} g^{\alpha \beta} \partial_{v_i}.
}

In this section, we consider the first part of this transport operator, namely
$v_\alpha g^{\alpha \beta} \partial_{x^\beta}$. Our aim is write it as a sum of the flat transport operator
\eq{
T_\eta=w_\alpha\eta^{\alpha\beta}\partial_{x^\beta}
}
and some perturbative terms satisfying some sort of null condition.

\begin{lemma}
Let $w=(w_0, v_i)$, where $w_0=- \sqrt{1+ |v|^2}.$ We have

\begin{equation} \label{eq:Tdec}
v_\alpha g^{\alpha \beta} \partial_{x^\beta} = (v_0-w_0 ) g^{0 \beta} \partial_{x^\beta}  +w_\alpha \eta^{\alpha \beta} \partial_{x^\beta }  +\wu_\alpha \Hu^{\alpha \beta} \delu_{x^\beta }.
\end{equation}
\end{lemma}
\begin{proof}
We decompose
\begin{eqnarray*}
v_\alpha g^{\alpha \beta} \partial_{x^\beta} f&=& (v_0-w_0 ) g^{0 \beta} \partial_{x^\beta} f +w_\alpha g^{\alpha \beta} \partial_{x^\beta } f\\
&=& (v_0-w_0 ) g^{0 \beta} \partial_{x^\beta} f
+w_\alpha \eta^{\alpha \beta} \partial_{x^\beta } f
+w_\alpha H^{\alpha \beta} \partial_{x^\beta } f \\
&=& (v_0-w_0 ) g^{0 \beta} \partial_{x^\beta} f +w_\alpha \eta^{\alpha \beta} \partial_{x^\beta } f +\wu_\alpha \Hu^{\alpha \beta}\delu_{x^\beta } f.
\end{eqnarray*}
\end{proof}
Note that in \eqref{eq:Tdec}, we might expect that the last term on the right-hand side has some form of null condition, using the semi-hyperboloid frame and the improved decay on $\hu^{00}$. The middle term is the flat transport operator and therefore commutes with the standard lifted vector fields. For the first term, we need to understand a form of the null condition for $v_0-w_0$.

\subsection{Estimates for $v_0$}
In this section, we derive some preliminary estimates concerning the solution $v_0:=v_0(x^\alpha, v_i)$ of the mass-shell relation $v_\alpha v_\beta g^{\alpha \beta}=-1$. We start by analysing the difference between $v_0$ and $w_0=-\sqrt{1+|v|^2}$. First, we have
$$
v_0-w_0 = \frac{v_0^2-w_0^2}{v_0+ w_0},
$$
and thus it is sufficient to estimate only $v_0^2-w_0^2$. For this, we will need the formula

\eq{\alg{
v_0^2-w_0^2= v_\alpha v_\beta H^{\alpha \beta}&= (v_0)^2 \Hu^{00}+ 2 \vu_0 \vu_a \Hu^{0a}+ \vu_a \vu_b \Hu^{ab} \nonumber \\
&= (v_0)^2 \hu^{00}+ \vu_0 \vu_a \hu^{0a}+2 \vu_a \vu_b \hu^{ab}+ O\left(h^2 v_0^2 \right). \label{eq:v0w0dif}
}}
The first identity follows straightforwardly when evaluating $v_\alpha v_\beta H^{\alpha\beta}$.

\begin{remark}
Note in particular that the right-hand side of \eqref{eq:v0w0dif} has some null structure\footnote{In the rest of the article, we will sometimes comment informally about the "null structure" of an expression. This means that this expression, as here $v_0-w_0$, has a better behaviour than if it had been estimated naively, here by $v_0^2|h|$. Thus, in this language, any product containing a factor of $\hu^{00}$ will have the null structure, in view of the improved decay for this metric component \eqref{es:ide1}-\eqref{es:ide2}.}, since, in the semi-hyperboloidal decomposition of $v$, the worst components of $v$ is $v_0$, but it is multiplied by the coefficient $\hu^{00}$ which satisfies the improved decay estimates \eqref{es:ide1} and \eqref{es:ide2}. We recall that the components $\vu_a$ are better behaved than $\vu_0$ in view of \eqref{dec-bck-en}.
\end{remark}

This leads to
\begin{lemma} \label{lem:voan}
From \eqref{es:bde2} and \eqref{es:ide2}, we have
\begin{eqnarray}
|v_0^2-w_0^2| &\lesssim& \vu_0^2 \left( \rho^{\delta/2} \epsilon^{1/2} \frac{1+u}{t^{3/2}}+\frac{\epsilon^{1/2}}{t} \right)  + |\vu_a|^2 \frac{t}{\rho} \rho^{\delta/2} \epsilon^{1/2} \frac{1+u^{1/2}}{t}. \label{es:vwd2}
\end{eqnarray}
In particular, we have

\eq{\label{est:v0-w0}
|v_0-w_0|
 \lesssim | w_0|  \left( \rho^{\delta/2} \epsilon^{1/2} \frac{1+u}{t^{3/2}}+\frac{\epsilon^{1/2}}{t} \right)  +  \frac{|\vu_a|^2}{ | w_0|} \frac{t}{\rho} \rho^{\delta/2} \epsilon^{1/2} \frac{1+u^{1/2}}{t}.
}

as well as
$$ 1-c \epsilon^{1/2}\le \left| \frac{v_0}{w_0}\right| =\left| \frac{v_0}{\sqrt{1+|v|^2}}\right| \le 1+ C\epsilon^{1/2},$$
for some constants $c,C>0$.
Similarly, we have  $$1-c \epsilon^{1/2} \le \left| \frac{v_\alpha g^{\alpha 0}}{w_0}\right|\le 1+ C\epsilon^{1/2},$$
\end{lemma}
\begin{proof}
The estimate
$$
|v_0^2-w_0^2| \lesssim (v_0)^2 \left( \epsilon^{1/2}\rho^{\delta/2}  \frac{1+u^{3/2}}{t^2}  + \frac{\epsilon^{1/2}}{t}  \right)+  |\vu_0| |\vu_a| \rho^{\delta/2} \epsilon^{1/2} \frac{1+u^{1/2}}{t}
$$
is a direct consequence of \eqref{es:bde2} and \eqref{es:ide2}. Then \eqref{es:vwd2} just follows from
$$
2 |\vu_0| |\vu_a| \le |\vu_0|^2 \frac{\rho}{t} +  |\vu_a|^2 \frac{t}{\rho}.
$$
The final estimates of the lemma are immediate consequences.
\end{proof}

\begin{remark}
In view of the structure of the energy \eqref{dec-bck-en} and the fact that the volume form induced by $\eta$ on each $H_\rho$ is given by $\frac{\rho}{t}dx$, we will systematically try to estimate Vlasov related quantities in terms of $\frac{\rho}{t}(w^0)^2$ and $|\vu_a|^2 \frac{t}{\rho}$, or $\frac{\rho}{t}w^0$ and $\frac{|\vu_a|^2}{w^0} \frac{t}{\rho}$ if we take into account the volume form $\frac{dv}{\sqrt{1+|vl^2}}= \frac{dv}{\sqrt{w^0}}$. Thus, the estimate on $|v_0-w_0|$ can be read as
$$
|v_0-w_0|
 \lesssim  |w_0|\frac{\rho}{t}\left( \frac{\epsilon^{1/2}}{\rho}  + \rho^{\delta/2} \epsilon^{1/2} \frac{(1+u)^{1/2}}{t}\right)  + \frac{|\vu_a|^2}{ | w_0|} \frac{t}{\rho} \rho^{\delta/2} \epsilon^{1/2} \frac{(1+u)^{1/2}}{t}
$$
\end{remark}

Differentiating the mass-shell relation $v_\alpha v_\beta g^{\alpha \beta}=-1$, we obtain the following lemma.
\begin{lemma}
\begin{eqnarray}
\partial_{x^\alpha} v_ 0 &=&  - \frac{v_\mu v_\nu}{2g^{0\beta}v_\beta}\frac{\partial g^{\mu \nu}}{\partial x^\alpha }, \label{eq:dxv0}\\
\partial_{v_i}(v_0)&=& - \frac{v_\alpha g^{\alpha i}}{g^{0\beta}v_\beta}.\label{eq:dvv0}
\end{eqnarray}
\end{lemma}
\begin{remark}
Recall that we are using $\eta$ to raise and lower indices, which is why we wrote $g^{0\beta}v_\beta$ in the above formulae.
\end{remark}

From \eqref{eq:dxv0}, we obtain
\begin{lemma}[Decomposition of $\partial_{x^\gamma}v_0$]\label{lem:decdxv0}
$\partial_{x^\gamma} v_0$ can be written as a linear combination of
\begin{eqnarray*}
\frac{1}{g^{0\beta}v_\beta} \partial_t \hu^{00} \vu_0 \vu_0, \\
\frac{1}{g^{0\beta}v_\beta} \partial_t \hu^{0a} \vu_0 \vu_a, \\
\frac{1}{g^{0\beta}v_\beta} \delu_b \hu^{\alpha \beta} \vu_\alpha \vu_\beta, \\
\frac{1}{2g^{0\beta}v_\beta } \partial_{x^\gamma} \left( \Phi\cdot \Phi \right) \hu\, v_\alpha v_\beta,
\end{eqnarray*}
and cubic terms of the form $w\cdot h \cdot \partial_{x^\gamma} h$.
In particular, from \eqref{es:bde}, \eqref{es:ide1} and \eqref{es:ide2}, we have the estimate
$$
| \partial_{x^\gamma} v_0 | \lesssim \epsilon^{1/2} \rho^{-3/2+\delta/2} \left( \frac{\rho}{t} w^0 + \frac{t}{\rho} (\vu_a)^2/w^0\right ),
$$
where $w^0=\sqrt{1+|v|^2}$.
\end{lemma}
\begin{proof} We use \eqref{eq:dxv0} and simply compute
\begin{eqnarray*}
  \frac{1}{g^{0\beta}v_\beta} \partial_{x^\gamma} h^{\mu \nu} v_{\mu} v_{\nu} &=&   \frac{1}{g^{0\beta}v_\beta} \partial_{x^\gamma} \left(\Phi^\mu_{\alpha} \Phi^\nu_{\beta}\hu^{\alpha \beta} \right) v_{\mu} v_{\nu} \\
  &=&   \frac{1}{g^{0\beta}v_\beta} \partial_{x^\gamma} \left(\Phi^\mu_{\alpha} \Phi^\nu_{\beta}\right) \hu^{\alpha \beta}  v_{\mu} v_{\nu}+ \frac{1}{g^{0\beta}v_\beta} \partial_{x^\gamma} \left(\hu^{\alpha \beta}\right)  \Phi^\mu_{\alpha} \Phi^\nu_{\beta} v_{\mu} v_{\nu}\\
    &=&   \frac{1}{g^{0\beta}v_\beta} \partial_{x^\gamma} \left(\Phi^\mu_{\alpha} \Phi^\nu_{\beta}\right) \hu^{\alpha \beta}  v_{\mu} v_{\nu} + \frac{1}{g^{0\beta}v_\beta} \partial_{x^\gamma} \left(\hu^{\alpha \beta}\right)  \underline{v}_{\alpha} \underline{v}_{\beta}.
\end{eqnarray*}
Since
$$
\vert \partial_{x^\gamma} \left(\Phi^\mu_{\alpha} \Phi^\nu_{\beta}\right) \vert \lesssim t^{-1},
$$
using the decay estimates \eqref{es:bde}, \eqref{es:ide1} and \eqref{es:ide2}, one obtains
$$
  \left|  \frac{1}{g^{0\beta}v_\beta} \partial_{x^\gamma}
  h^{\mu \nu}
  v_{\mu} v_{\nu}
  \right| \lesssim
  \dfrac{\epsilon^{1/2} \rho^{\delta/2} }{t^{3/2}} \dfrac{v_0^2}{|g^{0\beta}v_\beta|} + \dfrac{\epsilon^{1/2} \rho^{\delta/2} }{\rho t^{1/2}} \dfrac{|(\underline{v}_a v_0)|}{|g^{0\beta}v_\beta|} +  \dfrac{\rho^{1+\delta/2}}{t^{5/2}} \dfrac{v_0^2}{|g^{0\beta}v_\beta|}.
$$
\end{proof}

\begin{remark}
We can summarize the above structure of $\partial_{x^\gamma}v_0$ as follows. $\partial_{x^\gamma}(v_0)$ can be written as a linear combination of the terms
\begin{eqnarray*}
&&w_\alpha w_\beta \frac{ \partial_{x^\gamma} h^{\alpha \beta}}{w^0}, \\
&& w\cdot h \cdot \partial_{x} h,
\end{eqnarray*}
in which we can also expand the first term on the semi-hyperboloidal frame.
\end{remark}

\subsection{Comparison of flat and curved energy norms for the Vlasov field}

We establish now the equivalency of the energy norms \eqref{def:egf} and \eqref{def:ef} for the distribution function.
\begin{lemma}
There exists an $\varepsilon>0$ such that
\begin{eqnarray} \label{cond-equiv-0}
\|g-\eta\|_{\infty} &\lesssim& \varepsilon^{1/2},\nonumber\\
\|\underline h^{00}\|_{\infty}&\lesssim& \varepsilon^{1/2}  \frac{\rho^2}{t^2},\\
\|\underline h^{0a}\|_{\infty}&\lesssim& \varepsilon^{1/2} \frac{\rho}{t},\nonumber\\
\|\underline h^{ab}\|_{\infty}&\lesssim& \varepsilon^{1/2},\nonumber
\end{eqnarray}
imply that $E_{g}$ and $E=E_{\eta}$ are equivalent, i.e.~there exists a uniform constant $C>0$ such that
\eq{ \label{eq:compev}
C^{-1}E_{g}[k](\rho)\leq E_{\eta}[k](\rho)\leq C E_{g}[k](\rho),
}
for all $k$ suitably regular.
\begin{remark} In view of the decay estimates \eqref{es:bde2} and \eqref{es:ide2} and using that $\frac{1}{t} \le \frac{u}{t} \lesssim \frac{\rho^2}{t^2}$, since $u \ge 1$ in $\Kcal$, the assumptions of the lemma hold provided the bootstrap assumptions hold.
\end{remark}
\end{lemma}

\begin{proof}
We denote in this proof
\eq{
D_g:=\frac{\sqrt{-g^{-1}}}{v_\beta g^{\beta 0}}\sqrt g \frac s{g^{\alpha 0}x_\alpha}
}
and
\eq{
I_g:=\left(v_0 (-tg^{00}+x_ig^{i0})+v_i (-tg^{0i}+x_jg^{ji})\right)
}
and both corresponding quantities w.r.t.~the flat metric by $D_\eta$ and $I_\eta$, respectively. Recall also the definition of $s$ given by \eqref{s-def}.

Using \eqref{id:dsr},
$$
 | s^2- \rho^2| \le \left| H^{\alpha\beta}x_\alpha x_\beta\right| =|\Hu^{00}| t^2=|\hu^{00}|t^2+O(h^2)
$$
so that, using \eqref{es:ide2}, we have the estimate
\begin{eqnarray}
\left| \frac{s^2}{\rho^2}-1 \right| &\lesssim& \epsilon^{1/2} \left( \frac{u^{3/2}}{\rho^{2-\delta/2}}+\frac{t}{\rho^2} \right) \nonumber \\
&\lesssim& \epsilon^{1/2}, \label{es:diffrs}
\end{eqnarray}
and in particular, $\frac {1}{s}\lesssim \frac 1\rho.$ Since we can freely replace $v_\alpha g^{\alpha 0}$ by $w^0$ in view of Lemma \ref{lem:voan}, it then follows that
\eq{
0 \le D_g\lesssim\frac{1}{w^0} \frac \rho t=D_\eta,\,
}

We can now compute and estimte the difference between the energies for the flat and the curved metric. By definition,
\eq{\alg{ \nonumber
E_g[f]&=\iint\frac{|f|v_0I_g}s D_gdvdx\\
&=\iint{|f|v_0I_g}(s^{-1}-\rho^{-1}) D_gdvdx\\
&\quad +\iint v_0|f|\frac{I_g-I_\eta}{\rho}D_gdv dx+\iint v_0|f|\frac{I_\eta}{\rho}(D_g-D_\eta)dv dx\\
&\quad+\iint (v_0-w_0)|f|\frac{I_\eta}{\rho}D_\eta dvdx+\iint w_0|f|\frac{I_\eta}{\rho}D_\eta dv dx \\
&=I+II+III+IV+E_{\eta}[f].
}}
We estimate the terms $I$ -- $IV$ in the following. First, we have
\begin{eqnarray*}
 |I_g-I_\eta |&\lesssim& |{v_0-w_0}||x_\beta| g^{0\beta}+|w_{\alpha}x_\beta(g^{\alpha\beta}-\eta^{\alpha\beta})| \nonumber \\
 &\lesssim& |{v_0-w_0}|t+|t(v_0\underline h^{00}+\underline v_a \underline h^{a0})|+O(t w^0 |h|^2) \nonumber  \\
  &\lesssim& w^0|\underline h^{00}| t+2|\underline v_a||\underline h^{0a}| t+\frac{|\underline v_a\underline v_b\underline h^{ab}|}{w^0}t+\frac{t^2}\rho|\frac{I_\eta}{\rho}||\underline h^{00}|+|\frac{I_\eta}{\rho}|t|\underline h^{a0}|\\
  &&+\,O(t w^0 |h|^2) \nonumber \\
  &\lesssim& \left(|\underline h^{00}| \frac {t^2}{\rho^2}+2|\underline h^{0a}| \frac{t}\rho+\frac{t^2}{\rho^2}|\underline h^{00}|+\frac{t}{\rho}|\underline h^{a0}| \right)|I_\eta|
  \end{eqnarray*}
  \begin{eqnarray*}
  &&\hbox{}+\frac t\rho\frac{|\underline v_a|^2}{w^0}|\underline h^{0a}| {t}+\frac t\rho|\frac{\underline v_a\underline v_b}{w^0}|{|\underline h^{ab}|}\rho+O(t w^0 |h|^2)\nonumber  \\
    &\lesssim& \left(\frac{t^2}{\rho^2}|\underline h^{00}|+\frac{t}{\rho}|\underline h^{a0}| +|\underline h^{ab}\right)|I_\eta|+O(t w^0 |h|^2)\nonumber  \\
   \label{sdflj}  &\lesssim&  \sqrt{\varepsilon}|I_\eta|,
\end{eqnarray*}
 where we used
 \eq{\label{wHx}
 w_{\alpha}x_\beta(g^{\alpha\beta}-\eta^{\alpha\beta})=-t(v_0\underline h^{00}+\underline v_a \underline h^{a0})+O(t w^0 |h|^2)
 }
 and
 \eq{
 \frac{\underline v_i^2}{w^0}\lesssim \frac{I_\eta}{t}, \quad \frac{\rho}{t}w^0\lesssim\frac{I_\eta}{\rho}.
 }
The above estimate in particular implies
\eq{
|I_g|\lesssim |I_\eta|.
}

With this, we evaluate
\eq{\alg{
I&= \iint |f|v_0I_g(s^{-1}-\rho^{-1}) D_gdvdx\\
&\lesssim \iint |f|w^0\frac{|I_g|}{\rho}|\frac{s^2-\rho^2}{ s(s+\rho)}| D_gdvdx\\
&\lesssim  \sqrt{\varepsilon}\iint |f|w^0\frac{|I_\eta|}{\rho} D_\eta dvdx,\\
}}
where we used \eqref{es:diffrs}.

Next we estimate $II$.
\eq{\alg{
II&\lesssim\iint |v_0| |f|\frac{|I_g-I_\eta|}{\rho}D_gdv dx\\
&\lesssim \sqrt{\varepsilon}\iint |w_0| |f|\frac{|I_\eta|}{\rho}D_\eta dv dx.
}}
For the next term we estimate
\eq{\alg{
|D_g-D_\eta|&\lesssim  \left(\frac{s^2-\rho^2}{\rho^2}+\frac{|w^0-v_\beta g^{0\beta}|}{w^0}+\frac{t-g^{\alpha 0}x_\alpha}{t}\right)D_\eta\\
&\lesssim \sqrt{\varepsilon} |D_\eta|.
}}
This implies
\eq{\alg{
III&\lesssim \iint |w_0||f|\frac{I_\eta}{\rho}|D_g-D_\eta|dv dx\\
&\lesssim  \sqrt{\varepsilon}\iint |w_0||f|\frac{I_\eta}{\rho}D_\eta dv dx
}}
and
\eq{\alg{
IV&\lesssim \iint |v_0-w_0||f|\frac{I_\eta}{\rho}D_\eta dvdx\\
&\lesssim \sqrt{\varepsilon}\iint w^0|f|\frac{I_\eta}{\rho}D_\eta dvdx.
}}

In total, we conclude that
\eq{
\left| E_g[f]-E_\eta[f] \right| \lesssim \sqrt{\varepsilon} E_\eta[f],
}
which implies the claim.
\end{proof}
\subsection{Energy estimate for Vlasov fields in a curved spacetime}

We have

\begin{lemma}
Let $k:=k(t,x,v)$ be a solution to $T_g(k)= F[k]$, then

\begin{eqnarray*}
\nabla^\alpha \left( T_{0 \alpha}[k] \right)
&=&\int_v v_0 F[k] d\mu_v \\
&& + \int_v k v^\alpha \partial_{x^\alpha}(v_0) d\mu_v +  \frac{1}{2} \int_v k  v_\alpha v_\beta\partial_{x^i} (g^{\alpha \beta} )  \frac{v_\gamma g^{\gamma i}}{v_\beta g^{\beta0}} d\mu_v.
\end{eqnarray*}
\end{lemma}
\begin{proof}
This follows from the divergence properties $T$. Alternatively,  one can use that if $N_\alpha[k]=\int_v k v_\alpha d\mu_v$ is the particle current of $k$, then we have the divergence identity
$$
\nabla^\alpha N_\alpha[k]=\int_v T_g(k) d\mu_v.
$$
On the other hand,
$$
 \nabla^\alpha \left( T_{0 \alpha}[k]\right)= \nabla^{\alpha}N_\alpha[ v_0 k],
 $$
so that, together with the divergence property of $N$ and the Vlasov equation for $v_0 k$,
\begin{eqnarray*}
\nabla^\alpha \left( T_{0 \alpha}[k] \right)&=& \int_v T_g[kv_0]  d\mu_v \\
&=& \int_v T_g[k]v_0  d\mu_v\\
&&\hbox{}+ \int_v k v^\alpha \partial_{x^\alpha}(v_0) d\mu_v - \frac{1}{2} \int_v k v_\alpha v_\beta  \partial_{x^i} (g^{\alpha \beta} ) \partial_{v_i} (v_0 ) d\mu_v \\
&=&\int_v T_g[k]v_0 d\mu_v \\
&&\hbox{}+ \int_v k v^\alpha \partial_{x^\alpha} (v_0) d\mu_v+ \frac{1}{2} \int_v k  v_\alpha v_\beta\partial_{x^i} (g^{\alpha \beta} )  \frac{v_\gamma g^{\gamma i}}{v_\beta g^{\beta0}} d\mu_v.
\end{eqnarray*}
\end{proof}

\begin{remark} Recall that if $k$ solves a transport equation of the form $T_g(k)= F[k]$, then $|k|$ solves $T_g( |k|)= \frac{k}{|k|}F[k]$ in the sense of distribution. Moreover, for sufficiently regular $k$, the above computations and lemma still make sense with $|k|$ instead of $k$ and $\frac{k}{|k|}F[k]$ instead of $F[k]$.
\end{remark}

Applying the divergence theorem on a region of spacetime bounded by two hyperboloids, we obtain for any $\rho_2 > \rho_1$, and $k:=k(t,x,v) \ge 0$,
\begin{eqnarray*}
E_g[ k](\rho_2)-E_g[k](\rho_1) &=& \int_{\rho_1 \le \rho \le \rho_2} \sqrt{-g} dx^4 \bigg(  \int_v v_0 F[k] d\mu_v \\
&& + \int_v k v^\alpha \partial_{x^\alpha}(v_0) d\mu_v + \frac{1}{2} \int_v k  v_\alpha v_\beta\partial_{x^i} (g^{\alpha \beta} )  \frac{v_\gamma g^{\gamma i}}{v_\beta g^{\beta0}} d\mu_v\bigg)
\end{eqnarray*}

Recall that in the pseudo-spherical coordinates $(\rho,r,\theta, \phi)$, the volume form induced by $\eta$ is given by $\frac{\rho}{t}r^2dr \sin(\theta)d\theta d\phi$. Moreover, since
$$
\sqrt{-g}=1-\frac{1}{2}tr(h)+O(h^2),
$$
we have, in $(\rho, r, \omega)$ coordinates$$
(1-D^{-1}\epsilon^{1/2}) \frac{\rho}{t}r^2 \sin(\theta) \le\sqrt{-g} \le (1+D\epsilon^{1/2})\frac{\rho}{t}r^2 \sin(\theta),
$$
for some $D > 0$.

Thus, we have
\begin{eqnarray*}
E_g[ k](\rho_2)-E_g[k](\rho_1) &\lesssim& \int_{\rho_1 \le \rho \le \rho_2}d\rho \int_{H_\rho}  \bigg(  \int_v v_0 |F[k]| d\mu_v \\
&& + \int_v k \left| v^\alpha \partial_{x^\alpha}(v_0)\right| d\mu_v  \\
&&+ \frac{1}{2} \int_v k  \left|  v_\alpha v_\beta\partial_{x^i} (g^{\alpha \beta} )\right|  \frac{v_\gamma g^{\gamma i}}{v_\beta g^{\beta0}} d\mu_v \bigg) \frac{\rho}{t}r^2 dr d\omega_{\mathbb{S}^2}
\end{eqnarray*}

Finally, in view of the estimates on $\partial_{x^\alpha}(v_0)$ and since
\begin{eqnarray*}
\left| v_\alpha v_\beta \partial_{x^i} (g^{\alpha \beta}) \right|&\le&
(v_0)^2 \left|\partial (\hu^{00}) \right|+ \vu_0 |\vu_a| |\partial \hu| + (\vu_0)^2 |\hu \partial ( \Phi.\Phi)|+O\left( (\vu_0)^2 h \partial_{x^\gamma} h \right)\\
&\lesssim& \rho^{-3/2+\delta} \left( \rho/t (w^0)^2 + t/\rho(\vu_a)^2 \right),
\end{eqnarray*}
we obtain, using \eqref{eq:compev},

\begin{lemma}For any regular distribution function $k$, satisfying $T_g(k)=F[k]$.
\begin{eqnarray} \label{es:eevf}
\quad E[ k](\rho_2)\lesssim E[k](\rho_1) + \int_{\rho_1 \le \rho \le \rho_2}d\rho \int_{H_\rho} \left( \int_v |F[k]|dv\right)\frac{\rho}{t}r^2 dr d\omega_{S^2}.
\end{eqnarray}
\end{lemma}

In particular, since $T_g(f)=0$, the previous lemma implies the energy bound for $f$.
\begin{corollary}
We have, for all $\rho \in [2, \rho^*)$,
$$
E[f](\rho) \lesssim \epsilon.
$$
\end{corollary}

\subsection{Weighting by $1+u$ factors}

At different places below, we weight quantities that fulfill transport equations by $(1+u)$ factors. The following lemma will then be useful.

\begin{lemma}\label{lem:weghtingwithu}
\eq{
T_g(1+u)=w^0-w_i\frac{x^i}{|x|}+E_u,
}
where
\eq{
|E_u|\lesssim | w_0|  \left( \rho^{\delta/2} \epsilon^{1/2} \frac{1+u}{t^{3/2}} + \frac{\epsilon^{1/2}}{t} \right) +  \frac{|\vu_a|^2}{ | w_0|} \frac{t}{\rho} \rho^{\delta/2} \epsilon^{1/2} \frac{1+u^{1/2}}{t}
}
and we emphasize
\eq{ \nonumber
w^0-w_i\frac{x^i}{|x|}\geq 0.
}
\end{lemma}
\begin{proof}
Compute straightforwardly
\eq{\alg{
T_g(1+u)&=v_\alpha g^{\alpha 0} - v_\alpha g^{\alpha i}\frac{x_i}{|x|}\\
&= v_\alpha\left(\frac{g^{\alpha 0}x_0}{x_0}-\frac{g^{\alpha i} x_i}{|x|}\right)\\
&=w^0-\frac{w_ix^i}{|x|}+v_\alpha \frac{x_\mu}{x_0}H^{\alpha\mu}-(v_0-w_0)+v_\alpha H^{\alpha i}\frac{x_i}{|x|}\frac{|x|-t}{t}.
}}
The first two terms give the explicit part of the identity to prove. The last three terms provide the term $E_u$ and are estimated in the following. The term $v_0-w_0$ has been estimated in \eqref{est:v0-w0}. The first term in $E_u$ can be estimated (replacing $v_0$ by $w_0$ and estimating the additional difference term as before), using \eqref{wHx}, by
\eq{
|\frac{w_\alpha x_\mu}{x_0}H^{\alpha \mu} |=t^{-1}|t(v_0 \underline h^{00}+\underline v_a\underline h^{a0})|+\mbox{cubic terms},
}
which yields the same terms as the $v_0-w_0$ term. Finally, we estimate
\eq{
\left |v_\alpha H^{\alpha i}\frac{x_i}{|x|}\frac{|x|-t}{t} \right|\lesssim w^0|h|\frac{u}{t}+\mbox{cubic terms},
}
which again yields terms of the types above and finishes the proof.
\end{proof}

\section{The commutation of the Vlasov equation} \label{se:cve}
The aim of this section is to compute the commutators of the Vlasov equation with respect to well-chosen vector fields. The choice of vector fields is constraint by several factors, in particular
\begin{enumerate}
\item the error terms in the commutators must decay sufficiently fast.
\item since we will commute the Einstein equations, we require decay estimates on quantities of the form $K^\alpha (T_{\mu\nu}[f])$, where $K^\alpha$ is a combination of $|\alpha|$ vector fields among the standard Killing fields\footnote{We actually only commute with translations, Lorentz boosts and the scaling vector fields, and do not need to commute the Einstein equations with the spatial rotations. } and the scaling vector field.
\end{enumerate}
A starting point is to consider the complete lifts, denoted here $\widehat{Z}$, of the usual Killing fields together with the standard, non-lifted scaling vector field. This was used in \cite{fjs:vfm} for the Vlasov-Nordstr\"om system with spatial dimension $n \ge 4$. In dimension $3$, the error terms arising after commutations decays too slowly for this strategy to close. Instead, we used in \cite{fjs:savn} modified vector fields. Those are essentially constructed out of the complete lift $\widehat{Z}$ plus a correction and are therefore of the form
$$
Y=\widehat{Z} + C^v\cdot\partial_v+C^{t,x}\cdot\partial_{t,x},
$$
where $C^v$ and $C^{t,x}$ are coefficients, carefully chosen for each initial vector field $Z$, so as to obtain, after commutation with the non-linear transport operator, a cancellation with the worst terms appearing in the original commutator.

In \cite{fjs:savn}, two different kinds of corrections were used, corrections containing $v$ derivatives, which appeared even for translations and corrections containing $t,x$ derivatives, which only appeared for the homogeneous vector fields. It turns out that in the co-tangent bundle formulation used here, only corrections containing $t$- and $x$-derivatives are needed. However, a difficulty in the cotangent bundle approach occurs due to the part of the transport operator containing $t$- and $x$-derivatives, which is given in the tangent bundle formulation by $v^\alpha \partial_{x^\alpha}$. This term now involves a non-linear coupling with the metric through $v_\alpha g^{\alpha \beta} \partial_{x^\alpha}$. (Note however that for the Einstein-Vlasov system, $v^0$ always depends on $g$ and therefore, the operator $v^\alpha \partial_{x^\alpha}$ always involve the metric.)

Thus, the starting point is to consider as commutation vector fields
\begin{itemize}
\item translations $\partial_{x,t}$
\item For each homogeneous vector field $Z$, a modified vector field of the form
$$
Y=\widehat{Z}+C^\alpha \partial_{x^\alpha}.
$$
for some coefficients $C^\alpha:=C^\alpha(t,x,v)$.
\end{itemize}
However, already in \cite{fjs:savn}, we observed that replacing $\partial_{x^i}$ by vector fields of the form\footnote{In \cite{fjs:savn}, the fields $X_i$ were actually defined using so-called \emph{generalized translations}. In the present work, it is sufficient to only consider translations. Again, this is due to the cotangent bundle formalism. } $$X_i=\partial_{x^i} + \frac{v^i}{\sqrt{1+|v|^2}} \partial_t$$ has the advantage that the fields $X_i$ behave better than $\partial_{x^i}$ when applied to a solution of the wave equation, which is essentially due to the decomposition
$$
X_i=\frac{Z_i}{t}+ \frac{\vu_i}{w^0}\partial_t.
$$

In this work, we will therefore consider modification of the form

\begin{equation} \label{def:mvf}
Y=\widehat{Z}+C^i X_i + C^0 \partial_t,
\end{equation}
where
\begin{equation}
X_i=\partial_{x^i} +\frac{v_i}{w^0} \partial_t. \label{def:xivf}
\end{equation}
The $C^\alpha$ are then to obey, under the bootstrap assumptions, $$|C^\alpha| \lesssim \epsilon^{1/2} \rho^{\delta/2} u^{1/2}.$$

Finally, as becoming apparent in the next section, the commutators between the translations and the non-linear transport operator $T_g$ generate borderline error terms, essentially because of a lack of null structure. It turns out that replacing the $\partial_{x^i}$ vector fields by the $X_i$ vector fields solves this issue, apart from the commutator with $\partial_t$, since there is no suitable replacement vector field in this case. We keep the $\partial_t$ vector field, but to close the estimates for the first order commuted equations, we also need to exploit a hierarchy in the equations\footnote{This hierarchy is reminiscent of the weak null condition, as in the system \eqref{eq:wnc}. Indeed, the borderline terms all arise from the commutator with $\partial_t$ but they involve products of the $X_i$ vector fields and metric coefficients, while the $X_i$ vector fields enjoy improved commutation properties with no resulting borderline terms.}.

In the following section, we systematically compute various commutators and immediately estimate the error terms. While typically, one often presents pure algebraic computations and estimates separately, we expect it is be easier to understand the structure of the proof if one understands how each term behaves asymptotically.

\subsection{Commutation with translations and $X_i$ vector fields}
We start by computing the commutators for the basic translations.

\begin{lemma} \label{lem:comtr}
\begin{eqnarray*}
[T_g, \partial_{x^\gamma}] &=& F_{\gamma 1} + F_ {\gamma 2} +F_{ \gamma 3}+F_{\gamma 4},
\end{eqnarray*}
where the error terms $F_{\gamma i}$ are given by
\begin{eqnarray*}
F_{\gamma 1}&=& - \partial_{x^\gamma}(g^{\alpha \beta}) v_\beta \partial_{x^\alpha } f , \\
F_ {\gamma 2}&=& - \partial_{x^\gamma}(v_0) g^{\alpha 0}  \partial_{x^\alpha} f, \\
F_{ \gamma 3}&=& \frac{1}{2} v_\alpha v_\beta \partial_{x^i} \partial_{x^\gamma} (g^{ \alpha \beta}) \partial_{v_i} f, \\
F_{\gamma 4}&=&\partial_{x^\gamma} (v_0) v_\beta \partial_{x^i} g^{0 \beta} \partial_{v_i} f .
\end{eqnarray*}
\end{lemma}

\begin{proof}
A straightforward computation yields
\eq{\nonumber\alg{\/
[T_g, \partial_{x^\gamma}] &= - \partial_{x^\gamma}(v_\beta g^{\alpha \beta}) \partial_{x^\alpha} f + \frac{1}{2} v_\alpha v_\beta \partial_{x^i} \partial_{x^\gamma} (g^{ \alpha \beta}) \partial_{v_i} f + \partial_{x^\gamma}(v_0) v_\beta \partial_{x^i} g^{0 \beta} \partial_{v_i} f\\
&= - \partial_{x^\gamma}(g^{\alpha \beta}) v_\beta \partial_{x^\alpha } f -  \partial_{x^\gamma}(v_0) g^{\alpha 0}  \partial_{x^\alpha} f + \frac{1}{2} v_\alpha v_\beta \partial_{x^i} \partial_{x^\gamma} (g^{ \alpha \beta}) \partial_{v_i} f \\
&\hbox{}+\partial_{x^\gamma}(v_0) v_\beta \partial_{x^i} g^{0 \beta} \partial_{v_i} f.
}}
\end{proof}

Let us describe the expected decay behaviour of each of the $F_{\gamma i}$ error terms. First note that the error term $F_{\gamma4}$ is a higher order error term, since both $\partial_{x^\gamma} (v_0)$ and $\partial_{x^i} g^{0 \beta}$ decay and therefore enjoy stronger decay properties than the other three error terms. Second, one can easily estimate $F_{\gamma2}$ using Lemma \ref{lem:decdxv0}, leading to an integrable error term.

Using the decay estimates \eqref{es:bde}, $F_{\gamma 1}$ can be estimated by a cubic term of the form $$w^0 | h \partial h \partial_{t,x} f |$$
and the main terms of the form
 $$w^0 | \partial h  \partial_{t,x} f |= w^0\frac{\rho}{t}\cdot\frac{t}{\rho} |  \partial h  \partial_{t,x} f | \lesssim \rho^{\delta/2-1}/ (1+u)^{1/2} w^0 \frac{\rho}{t} |\partial_{t,x} f|.$$
Together with the energy estimate \eqref{es:eevf} and the coercivity estimate, this leads to an estimate of the form
$$
E[ \partial_{t,x} f](\rho) - E[\partial_{t,x} f](2)\lesssim  \int_{2}^\rho \frac{( \rho')^{\delta/2-1}}{(1+u)^{1/2}} E[\partial_{t,x} f](\rho') d\rho'.
$$

Unfortunately, this fails to close, essentially because we are not able to exploit the additional $(1+u)^{1/2}$ decay. One may hope to get extra decay using the semi-hyperboloidal foliation and the improved estimate on $\hu^{00}$, but as we will see, there are still non-integrable error terms. Each of these terms is of the form $ q(v,h)X_a(f)$, where $q(v,h)$ leads to a $\rho^{\delta/2}$ loss as above. Since the $X_a$ vector fields enjoy better commutation properties than the field $\partial_t$, we are then able to close the estimates (albeit with a loss).

Let us thus consider the structure of the $F_{\gamma1}$ error term in greater details. We write $Cub$ below to denote any of the cubic term.

First, we compute
\begin{eqnarray*}
F_{\gamma 1}&=& - \partial_{x^\gamma}(\hu^{\sigma \kappa} \Phi_\sigma^\alpha \Phi_\kappa^\beta) \Psi^{\beta'}_\beta \vu_{\beta'} \Psi^{\alpha'}_\alpha \underline{\partial}_{x^{\alpha'} } f +Cub , \\
&=& - \partial_{x^\gamma}(\hu^{\sigma \kappa} )  \vu_{\sigma} \underline{\partial}_{x^{\kappa} } f \hbox{}- \hu^{\sigma \kappa} \partial_{x^\gamma}(\Phi_\sigma^\alpha \Phi_\kappa^\beta) v_{\beta}  {\partial}_{x^{\alpha} } f +Cub.
\end{eqnarray*}
The second term on the right-hand side has much better decay, by virtue of the fact that
$$| \partial_{x^\gamma}(\Phi_\sigma^\alpha \Phi_\delta^\beta) | \lesssim t^{-1}.$$

For the first term, we have
\eq{ \label{eq: E1dec}\alg{
\partial_{x^\gamma}(\hu^{\sigma \kappa} )  \vu_{\sigma} \underline{\partial}_{x^{\kappa} } f &= \partial_{x^\gamma}(\hu^{00} )  \vu_{0} \underline{\partial}_{x^{0} } f + \partial_{x^\gamma}(\hu^{0a} )  \vu_{0} \underline{\partial}_{x^{a} } f \\
&\hbox{}+  \partial_{x^\gamma}(\hu^{a \alpha} )  \vu_{a} \underline{\partial}_{x^{\alpha} } f.
}}
The first term on the right-hand side has stronger decay thanks to the improved estimate \eqref{es:ide1} on $\hu^{00}$. For the last term, we have

$$
| \partial \hu \, \vu_a \, \delu_{x^\alpha} f | \le \rho^{\delta-3/2} \left(t/\rho\frac{\vu_a^2}{w^0}+\rho/t w^0 \right)  |\partial_{t,x} f |.
$$
In view of the coercivity of the Vlasov energy \eqref{dec-bck-en}, the contribution of this error term is therefore integrable.

On the other hand, the middle term on the right-hand side of \eqref{eq: E1dec} does not have enough decay. In fact, we can write

\begin{eqnarray*}
\partial_{x^\gamma}(\hu^{0a} )  \vu_{0} \underline{\partial}_{x^{a} } (f) &=& \partial_{x^\gamma}(\hu^{0a} )  \vu_{0} \left( \partial_{x^a}+ x_a/t \partial_t \right)(f) \\
&=& \partial_{x^\gamma}(\hu^{0a} )  \vu_{0} \left( \partial_{x^a}+ v_a/w^0 \partial_t + (x_a /t - v_a/w^0) \partial_t \right) (f) \\
&=& \partial_{x^\gamma}(\hu^{0a} )  \vu_{0} \left(\partial_{x^a}+ v_a/w^0 \partial_t - \vu_a/w^0 \partial_t \right) (f) \\
&=& \partial_{x^\gamma}(\hu^{0a} )  \vu_{0} \left(  \partial_{x^a}+ v_a/w^0 \partial_t \right)(f)- \frac{\vu_0}{w^0} \partial_{x^\gamma}(\hu^{0a} )\vu_a \partial_t f.
\end{eqnarray*}
Since $\frac{|\vu_0|}{w^0} \lesssim 1$, the last term can be estimated as above, but the first term on the right-hand side is still problematic, because there is no improved decay a priori for $\left(\partial_{x^a}+ v_a/w^0 \partial_t\right)(f)$.

 Note that $\partial_{x^a}+ v_a/w^0 \partial_t$ is one of the $X_a$ vector fields introduced in \eqref{def:xivf}. Moreover, we recall there is an improved estimate for products involving $X_a(h)$, using $X_a= Z_a/t + \vu_a/w^0 \partial_t$ and that $h$ solves a wave equation, so that $\tfrac{Z_a(h)}{t}$ behaves better than $\partial_{t,x} h$.

Thus, we commute by $X_a$ instead of $\partial_{x^a}$. For the vector fields $X_a$  we have

\begin{lemma} \label{lem:Xacom}
Let $X_a= \partial_{x^a} + \frac{v_a}{w^0}\partial_t$, where $w^0=\sqrt{1+|v|^2}$. Then,
\begin{eqnarray*}
[T_g, X_a] f &=& F_{X_a1} + F_ {X_a2} +F_{X_a3}+F_{X_a4}+F_{X_a5},
\end{eqnarray*}
where the error terms $F_{X_ai}$ are given by
\begin{eqnarray*}
F_{X_a 1}&=& - X_a(g^{\alpha \beta}) v_\beta \partial_{x^\alpha } f , \\
F_ {X_a2}&=& - X_a(v_0) g^{\alpha 0}  \partial_{x^\alpha} f, \\
F_{ X_a 3}&=& \frac{1}{2} v_\alpha v_\beta \partial_{x^i} X_a (g^{ \alpha \beta}) \partial_{v_i} f, \\
F_{ X_a 4}&=& X_a(v_0) v_\beta \partial_{x^i} g^{0 \beta} \partial_{v_i} f,  \\
F_{ X_a 5}&=&-\frac{1}{2} v_\alpha v_\beta \partial_{x^i}  (g^{ \alpha \beta})\left[ \frac{\delta_{ai}}{w^0}- \frac{v_a v_i}{(w^0)^3} \right] \partial_{t} f.\
\end{eqnarray*}
\end{lemma}
\begin{proof}
A straightforward computation yields
\eq{\nonumber\alg{
\/[T_g, X_a] f &= - v_\beta X_a( g^{\alpha \beta}) \partial_{x^\alpha} f- X_a(v_0) g^{\alpha 0}  \partial_{x^\alpha} f + \frac{1}{2} v_\alpha v_\beta \partial_{x^i} X_a (g^{ \alpha \beta}) \partial_{v_i} f \\
&\quad-\frac{1}{2} v_\alpha v_\beta \partial_{x^i}  (g^{ \alpha \beta})\left[ \frac{\delta_{ai}}{w^0}- \frac{v_a v_i}{(w^0)^3} \right] \partial_{t} f+X_a(v_0) v_\beta \partial_{x^i} g^{0 \beta} \partial_{v_i} f.
}}
\end{proof}
\begin{remark}
The term $F_{ X_a 5}$ arises from the commutator $[ X_a, \partial_{v_i}]$.
\end{remark}

Our replacement for the translations is then composed of the usual translation vector field $\partial_t$ and of the $X_a$ vector fields. Note that any $\partial_{x^\gamma}$ can be rewritten as a linear combination (with coefficients in $\mathcal{F}_{x,v}$)) of $\partial_t$ and of the $X_a$ vector fields. Thus, when we write $\partial_{t,x}$ or $\partial_{x^\gamma}$, one should keep in mind the decomposition
$$
\partial_{x^\gamma}= a_0 \partial_t +a^iX_i,
$$
such that
$$
| \partial_{x^\gamma}f| \lesssim |\partial_tf| +\sum_{i} |X_if|.
$$

We now revisit the commutator $[T_g, \partial_t]$ given in Lemma \ref{lem:comtr} and the corresponding error terms $F_{0i}$ with $i=1,..,4$, using the new set of vector fields $\{\partial_t,X_1,X_2,X_3\}$.
\begin{lemma}
The error term $F_{01}$, arising in the commutator $[T_g, \partial_t]$ can be written as a linear combination of

\begin{itemize}
\item the strongly decaying terms
\begin{eqnarray*}
\hu^{\gamma \delta} \partial ( \Phi\cdot \Phi) v\cdot \partial_{t,x}(f), \\
\partial_t(\hu^{00}) \vu_0 \delu_t f, \\
\partial_t (\hu^{\alpha a}) \vu_a \delu_{x^\alpha} f, \\
\frac{\vu_0}{w^0} \partial (\hu^{\alpha a}) \vu_a \delu_{x^\alpha} f, \\
w^0 h \cdot \partial h \del_{x^\alpha} f.
\end{eqnarray*}
\item the borderline term
\begin{eqnarray*}
\partial_t (\hu^{0a}) \vu_0 X_a f.
\end{eqnarray*}
\end{itemize}
\end{lemma}
\begin{remark}
As before, the last term is the worst error term. It can be estimated as
$$
| \partial_t (\hu^{0a}) \vu_0 X_a f | \lesssim \epsilon^{1/2} \vu_0 \frac{\rho}{t} \rho^{\delta/2-1}(1+u)^{-1/2} |X_a (f)|
$$
and as before, this bad behavior essentially arise because we are not making use of the $(1+u)^{-1/2}$ decay factor.
However, since the fields $X_a$ enjoy improved commutation properties, we will still be able to close the energy estimate with a $\rho^{\delta/2}$ growth factor.
\end{remark}

For $F_{02}$, we simply use Lemma \ref{lem:decdxv0}. Each error term is easily seen to be integrable. More precisely,
\begin{lemma}
The error term $F_{02}$ can be written as a linear combination of
\begin{eqnarray*}
&&\frac{1}{g^{0\mu}v_\mu} \partial_t \hu^{00} \vu_0 \vu_0 g^{0 \gamma} \partial_{x^\gamma} f , \\
&&\frac{1}{g^{0\mu}v_\mu} \partial_t \hu^{\alpha a} \vu_\alpha \vu_a g^{0 \gamma} \partial_{x^\gamma} f, \\
&&\frac{1}{g^{0\mu}v_\mu} \partial_t \hu^{\alpha \beta} \vu_\alpha \vu_\beta g^{0 \gamma} \partial_{x^\gamma} f, \\
&&\frac{1}{g^{0\mu}v_\mu} \partial_{t} \left( \Phi\cdot \Phi \right) \hu^{\alpha\beta}\, v_\alpha v_\beta g^{0 \gamma} \partial_{x^\gamma} f, \\
&& w^0 h \cdot\partial h\cdot \partial_{x^\gamma} f,
\end{eqnarray*}
where $$|g^{0 \alpha} \partial_{x^\alpha} f| \lesssim |\partial_t f|+ \sum_a|X_af|.$$
\end{lemma}

We can summarize our analysis of $F_{01}$ and $F_{02}$ as follows.
\begin{lemma}
The error term terms $F_{01}$ and $F_{02}$ can be written as a linear combination of the terms
\begin{itemize}
\item The strongly decaying terms
\begin{eqnarray*}
w\cdot h \partial_{t,x} (\Phi\cdot \Phi ) \partial_{t,x}, \\
w\cdot \partial_{t,x}  (\hu^{00}) \partial_{t,x},  \\
\vu_a \partial_{t,x}  (\hu) \partial_{t,x}, \\
w \cdot h \cdot \partial_{t,x}  (h) \partial_{t,x}.
\end{eqnarray*}
\item The borderline terms
$$
w\cdot \partial h X_a (f).
$$
\end{itemize}
\end{lemma}

For $F_{03}$ and $F_{04}$, we need to do a bit more work, because the vector fields $\partial_{v_i}$ are not part of the algebra of vector fields we want to consider and enjoy poor commutation property even with the free transport operator $w^\alpha \partial_{x^\alpha}$.

To simplify the notation below, let us define $X_0:= \partial_t$.

We recall that for a Lorentz boost $Z_i$, its complete lift is given by $\widehat{Z}_i = Z_i + w^0 \partial_{v_i}$. According to \eqref{def:mvf}, we can then decompose $\partial_{v_i}$ as
$$
\partial_{v_i}=\frac{1}{w^0} \left(  Y_i - C_i^\alpha X_\alpha - t \partial_{x^i} - x^i \partial_t \right),
$$
where $$Y_i= \widehat{Z}_i + C_i^\alpha X_\alpha$$ is a modified Lorentz boost.

This leads to $F_{03}$ as
\begin{eqnarray} \label{eq:f03}
F_{03}&=& \frac{1}{2w^0} v_\alpha v_\beta \partial_{x^i} \partial_t (g^{\alpha \beta}) \left( Y_i - C_i^\alpha X_\alpha - t \partial_{x^i} - x^i \partial_t \right).
\end{eqnarray}
We can also decompose the first product on the semi-hyperboloidal frame.
\begin{eqnarray*}
v_\alpha v_\beta \partial_t \partial_{x^i} (g^{\alpha \beta}) &=& \vu_0 \vu_0 \partial_t  \partial_{x^i} (\hu^{00}) + 2 \vu_0 \vu_a \partial_t  \partial_{x^i} (\hu^{0a})\\
&& + \vu_a \vu_b \partial_t \partial_{x^i}  (\hu^{ab})+ v_\alpha v_\beta \partial_t \partial_{x^i}  (\Phi.\Phi)g^{\alpha' \beta'}.
\end{eqnarray*}
Assuming that\footnote{We will in fact prove the stronger bounds \eqref{es:ccu} on the $C$ coefficients using only the bootstrap assumptions.}
\begin{equation} \label{es:C}
|C^\alpha_i| \lesssim \epsilon^{1/2}\rho^{1/2+\delta/2},
\end{equation}
one then verifies, using the decay estimates \eqref{es:bde}, \eqref{es:ide2} and the above decomposition, that

$$
| v_\alpha v_\beta \partial_{x^i} \partial_t (g^{\alpha \beta})  | Y_i - C_i^\alpha X_\alpha | \lesssim \epsilon^{1/2}\rho^{-3/2+\delta/2}\left( |Yf| + |Xf| \right) \left( (w^0)^2 \rho/t+|\vu_a|^2 \frac{t}{\rho} \right).
$$

For the last term in \eqref{eq:f03}, the basic estimate leads to

\begin{eqnarray*}
\left| \frac{1}{2 w^0} v_\alpha v_\beta \partial_{x^i} \partial_t (g^{\alpha \beta}) \left( t \partial_{x^i} + x^i \partial_t \right)\right| \lesssim  \left| t   w_\alpha w_\beta \partial^2 (g^{\alpha \beta} )\partial_{t,x} f\right|+Cub.
\end{eqnarray*}

Unfortunately, (even with the decomposition on the semi-hyperboloidal frame), this is not good enough, because the $t$ loss is only compensated by a gain in $u$ decay coming from  $\partial^2 (g^{\alpha \beta} )$.

This means that $F_{03}$ should contain an extra null structure, which is the identity \eqref{id:nfk} of the following lemma.

\begin{lemma} \label{lem:ntnf}
The following identities hold. For a regular function $k:=k(t,x)$,
\eq{\alg{
\partial_{x^i} (k) \partial_{v_i}&= \delu_i k\cdot \partial_{v_i} - \frac{x^i}{t} \partial_t k\cdot \partial_{v_i}   \\
&= \delu_i k\cdot \partial_{v_i} - \frac{x^i}{t} \partial_t k\cdot \frac{1}{w^0}\widehat{Z}_i + \frac{x^i}{t} \partial_t k \frac{1}{w^0} ( t \partial_{x^i}+x^i \partial_t )  \\
&= \delu_i k\cdot \partial_{v_i} - \frac{x^i}{t} \partial_t k\cdot \frac{1}{w^0}\widehat{Z}_i +  \partial_t k \frac{1}{w^0} \left( S + \frac{|x|^2-t^2}{t} \partial_t \right). \label{id:nfk}
}}
The free transport operator can be rewritten as
\begin{eqnarray} \label{eq:tdws}
\frac{w^\gamma}{w^0} \partial_{x^\gamma}  = \frac{S}{t} + \frac{\vu_i}{w^0}\partial_{x^i}.
\end{eqnarray}
Finally, $\partial_{v_i}$ can be rewritten as
 \begin{eqnarray*}
\partial_{v_i}&=& \frac{\widehat{Z}_i}{w^0} - \frac{t}{w^0} X_i + t \frac{\vu_i}{(w^0)^2} \partial_t \\
&=&\frac{\widehat{Z}_i}{w^0} - \frac{t}{w^0} X_i + \frac{\zz_i}{(w^0)^2} \partial_t.
\end{eqnarray*}
\end{lemma}
\begin{remark} The other identities of the lemma will be used to analyse the commutator with the homogeneous vector fields in Section \ref{se:mvf}.
\end{remark}

Using identity \eqref{id:nfk}, the fact that $t \delu_{x^i}=Z_i$ and that $(t-|x|) \partial_{t,x}$ can be written as a linear combination of the homogeneous vector fields, we can rewrite $F_{03}$ as follows.

\begin{lemma}The error term $F_{03}$ can be written as a linear combination of the terms
\begin{eqnarray*}
&&\frac{1}{w^0}w_\alpha w_\beta \partial_{t,x} Z(h^{\alpha \beta} ) \partial_{t,x} f, \\
&& w\cdot h \cdot \partial_{t,x} Z(h) \partial_{t,x} f, \\
&&\frac{1}{w^0} w_\alpha w_\beta \partial^2_{t,x} (h^{\alpha \beta} ) Y(f),
\end{eqnarray*}
\begin{eqnarray*}
&& w\cdot h \cdot \partial^2_{t,x} (h) Y(f),\\
&&\frac{1}{w^0} w_\alpha w_\beta \partial^2_{t,x} (h^{\alpha \beta} ) C^\alpha X_\alpha (f), \\
&& w\cdot h \cdot \partial^2_{t,x} (h) C^\alpha X_\alpha (f),
\end{eqnarray*}
in which each expression of the form $w_\alpha w_\beta \partial_{t,x} Z(h^{\alpha \beta} )$ or $w_\alpha w_\beta \partial^2_{t,x} (h^{\alpha \beta} )$ can be expressed w.r.t.~the semi-hyperboloidal frame.
\end{lemma}

This leads to the estimate

$$
|F_{03}| \lesssim \epsilon^{1/2}\rho^{-3/2+\delta}\left( |Yf| + \sum_{a=1,2,3} |X_af|+ |\partial_t f| \right) (w^0 \frac{\rho}{t}+|\vu_a|^2 \frac{t}{w^0 \rho}).
$$

For $F_{04}$, we can proceed similarly to $F_{03}$, the generated terms being all cubic.

We can summarize the important structure of $F_{03}$ and $F_{04}$ as follows
\begin{lemma}The error terms $F_{03}$ and $F_{04}$ can be written as a linear combination of the terms
\begin{eqnarray*}
&&\frac{1}{w^0}w_\alpha w_\beta \partial_{t,x} K^\gamma (h^{\alpha \beta} )C^{k}\cdot \widehat{K}^{\delta} f, \\
&& w\cdot h\cdot \partial_{t,x} K^\gamma (h)  C^{k}\cdot \widehat{K}^{\delta} f, \\
&& w\cdot \partial_{t,x} h \cdot K^\gamma (h)  C^{k}\cdot \widehat{K}^{\delta} f,
\end{eqnarray*}
where $|\gamma|, |\delta|, k \le 1$ and $\gamma_X \ge k+\delta_Z$ and we adopt the conventions
\begin{itemize}
\item $\gamma_X=1$ if $K^\gamma$ is one of the $\partial_t, X_i$ vector fields and zero otherwise.
\item $C^k$ is one of the $C$ coefficient if $k=0$ and $1$ otherwise.
\item $\delta_Z$ is $1$ if $\widehat{K}$ is one of the $Y$ modified vector fields and zero otherwise.
\end{itemize}
\end{lemma}
\begin{remark}
In the following, we keep the same conventions for multi-indices. Thus, for any multi-index $\gamma$ and corresponding differential operator $K^\gamma$ (respectively $\widehat{K}^\gamma$) we shall denote by $\gamma_X$ the number of translations and $\gamma_Z$ the number of boosts or scaling vector fields  (respectively modified lifted boosts $Y$ or modified scaling vector fields $Y_S$).
\end{remark}

Together with the energy estimate \eqref{es:eevf}, we have thus proven that

\begin{lemma} \label{lem:esEt}
Assume that the $C$ coefficients satisfy \eqref{es:C}. Then, we have the estimate
$$
 E[ \partial_t f](\rho) \lesssim E( \partial_t f)(2)+ \epsilon^{1/2}  \int_2^\rho  \left( (\rho')^{-3/2+\delta} E[ \widehat{K} f ] + (\rho')^{\delta/2-1} \sum_{a=1,2,2} E[ X_a f ] \right) d \rho',
 $$ where $\widehat{K}$ denotes any of the commutation vector fields among $\partial_t, X_a$ and the modified vector fields of the form $Y=\widehat{Z}+C X$.
 \end{lemma}
 \begin{remark}
 The estimate \eqref{es:C} on the $C$ coefficients will be obtained in the next section and depends only on the decay \eqref{es:bde}, \eqref{es:ide1} and \eqref{es:ide2} of the metric coefficients. Thus, it is essentially a direct consequence on the bootstrap assumptions.
 \end{remark}

We now turn to the analysis of the error terms for the commutator $[T_g, X_i]$. We thus consider each of the five error terms $F_{Xij}$, $1\le j \le5$ given by Lemma \ref{lem:Xacom}. We drop the $i$ index below and simply refer to these error terms as $F_{Xj}$.

\begin{lemma}[Decomposition of $F_{X1}$]
$F_{X1}$ can be written as a linear combination of
\begin{eqnarray*}
\hu^{\gamma \delta} \partial ( \Phi\cdot \Phi) v\cdot \partial_{t,x}(f), \\
X_i(\hu^{00}) \vu_0 \delu_t f, \\
X_i (\hu^{\alpha a}) \vu_a \delu_{x^\alpha} f, \\
\frac{\vu_0}{w^0} X_i(\hu^{\alpha a}) \vu_a \delu_{x^\alpha} f, \\
w^0 h \cdot \partial h \del_{x^\alpha} f, \\
X_i (\hu^{0a}) \vu_0 X_a f.
\end{eqnarray*}
Moreover, the last term can be futher decomposed as
$$
\frac{Z_i}{t} (\hu^{0a}) \vu_0 X_a f
$$
and
$$
\partial_t(\hu^{0a}) \frac{\vu_0}{w^0}\vu_i  X_a f.
$$
\end{lemma}
\begin{proof}
The first part of the lemma follows as for the decomposition of $F_{01}$. The decomposition of $X_i (\hu^{0a}) \vu_0 X_a f$
 follows from the formula
$$
X_i =\frac{Z_i}{t} + \frac{\vu_i}{w^0} \partial_t.
$$
\end{proof}

The error term $F_{Xj}$ for $j=2,3,4$ can be decomposed as before. Since they do not contain any borderline terms, their contribution is integrable.

For $F_{X5}$, we have similarly,
\begin{lemma}[Decomposition of $F_{X5}$]
The error term $F_{X5}$ can be written as a linear combination of
\begin{eqnarray*}
\wu_0 \partial \hu^{00} \partial_t,\\
\vu_a \partial \hu^{0a} \partial_t,
\end{eqnarray*}
\begin{eqnarray*}
w\cdot h\cdot \partial \hu \partial_t, \\
w\cdot h \partial (\Phi) \partial_t.
\end{eqnarray*}
\end{lemma}
\begin{remark}
Recall from Lemma \ref{lem:Xacom} that the error term $F_{X5}$ contains a factor of the form $\left[ \frac{\delta_{ai}}{w^0}- \frac{v_a v_i}{(w^0)^3} \right]\partial_{x^i} g^{\alpha\beta}$. For the purpose of the above decomposition, we have just considered $\left(\frac{\delta_{ai}}{w^0}- \frac{v_a v_i}{(w^0)^3}\right)  \partial_{x^i}$ as the product of $\frac{1}{w^0}\partial_{x^i}$ with a function in $\mathcal{F}_{x,v}$, that is to say a coefficient that we then ignored. However, it is interesting to note that there is an additional structure here, namely $\left( \delta_{ai}- \frac{v_a v_i}{(w^0)^2}\right)\partial_{x^i}=X_a-\frac{v_a}{w^0}\frac{w^\alpha}{w^0} \partial_{x^\alpha}.$ This decomposition can be used to obtain an improved estimate for the error term $F_{X5}$, but the improvement is not necessary in order to close the energy estimate for $X_a(f)$.
\end{remark}

Note that all the terms coming from $F_{X5}$ can be estimated as above and are integrable. We have thus proven

\begin{lemma} \label{lem:esEX}
Assume that the $C$ coefficients satisfy \eqref{es:C}. Then, we have the estimate

$$
E[ Xf ] \lesssim \epsilon^{1/2} \int_2^\rho (\rho')^{-3/2+ \delta/2} E[ \widehat{K} f ] d\rho'+E[X f](2),
$$
where $\widehat{K}$ denotes any of the $\partial_t$, $X_a$ or modified vector field $Y$.
\end{lemma}

From Lemma \ref{lem:esEt} and \ref{lem:esEX}, we obtain
\begin{lemma} Assume that the $C$ coefficients satisfy \eqref{es:C} and that $$E[Yf] \lesssim \epsilon \rho^{\delta/2}, $$ for $Y$ a modified Lorentz boost or scaling vector field.
Then,
$$
E[ X f ] \lesssim \epsilon^{3/2}+ \epsilon.
$$
and
$$
E[ \partial_t f ] \lesssim \frac{\epsilon^{3/2}}{\delta} \rho^{\delta/2} + \epsilon.
$$
\end{lemma}

Let us also summarize the structure of the commutators we have computed

\begin{lemma} Let $X=X_i, \partial_t$, then the commutator $[T_g, X]$ can be written as a linear combination of the terms
\begin{itemize}
\item The strongly decaying terms
\begin{eqnarray*}
&&w\cdot h \partial_{t,x} (\Phi\cdot \Phi ) \partial_{t,x}, \\
&&w\cdot \partial_{t,x}  (\hu^{00}) \partial_{t,x},  \\
&&\vu_a\cdot \partial_{t,x}  (\hu) \partial_{t,x}, \\
&&w\cdot \delu_{x^i}  (\hu) \partial_{t,x}, \\
&&w \cdot h \cdot \partial_{t,x}  (h) \partial_{t,x}.
\end{eqnarray*}
as well as
\begin{eqnarray*}
&&\frac{1}{w^0}w_\alpha w_\beta \partial_{t,x} K^\gamma (h^{\alpha \beta} )C^{k}\cdot \widehat{K}^{\mu} f, \\
&& w\cdot h \cdot \partial_{t,x} K^\gamma (h)  C^{k}\cdot \widehat{K}^{\mu} f, \\
&& w\cdot \partial h \cdot  K^\gamma (h)  C^{k}\cdot \widehat{K}^{\mu} f,
\end{eqnarray*}
where $|\gamma|, |\mu|, k \le 1$ and $\gamma_X \ge k+\mu_Z$.
\vspace{0.5cm}
\item The borderline terms, which arise only if $X=\partial_t$.
$$
w\cdot \partial h X_a (f)
$$
\end{itemize}
\end{lemma}

\subsection{The modified vector fields} \label{se:mvf}
\subsubsection{The non-modified boosts and scaling vector fields}
We now consider commuting with complete lifts of Lorentz boosts, as well as the scaling vector field.

Let $Z_i=t \partial_{x^i} + x_i \partial_t$ be a Lorentz boost and recall that its complete lift\footnote{See for instance \cite{fjs:vfm} for a presentation of complete lifts.} $\widehat{Z}_i$ is given in the cotangent bundle formulation by
$$
\widehat{Z}_i= Z_i + w^0 \delta_{ij} \partial_{v_j}=Z_i - w_0 \delta_{ij} \partial_{v_j},
$$
where the $j$ index on the vector field $\partial_{v_j}$ is counted downstairs in the Einstein summation convention.

We shall also, as in \cite{fjs:savn}, commute with the scaling\footnote{Note that this implies that we will also commute the reduced Einstein equations with the scaling vector field, contrary to \cite{lm:gsmkg}. We use the scaling vector field because it will naturally appears in some of the decompositions below. However, using essentially formula \eqref{eq:tdws}, we could have avoided the use of the scaling of the vector field, at the cost of a slightly more complicated commutator formula. Since having the scaling vector field for the wave equations also allows for an easy improved estimate for basic derivatives, we decided to use it for simplicity in the exposition.} vector field $$S=x^\alpha \partial_{x^\alpha}.$$
We denote by $Z$ any vector field among the $Z_i$ or $S$ and by $\widehat{Z}$, either the complete lift of a Lorentz boost or again $S$ (now viewed as an operator on functions of $(t,x,v)$).

Using the decomposition \eqref{eq:Tdec} of the operator $v_\alpha g^{\alpha \beta} \partial_{x^\beta}$, we have
\begin{lemma} \label{lem:comZx}

 Let $Z_i$ be any Lorentz boost and $\widehat{Z}_i$ its complete lift.
Then,
\begin{eqnarray*}
[v_\alpha g^{\alpha \beta} \partial_{x^\beta}, \widehat{Z}_i ]&=& \left[\left( v_0 - w_0 \right) g^{0 \beta} \partial_{x^\beta}, \widehat{Z}_i \right]
+ \left[ \wu_\alpha \Hu^{\alpha \beta}\delu_{x^\beta}, \widehat{Z}_i \right] \\
&=& F_{Z_i1}+ F_{ Z_i2}.
\end{eqnarray*}
 For the scaling vector field, we have similarly,
\begin{eqnarray}
[v_\alpha g^{\alpha \beta} \partial_{x^\beta}, S ]&=&w^\alpha \partial_{x^\alpha}+ \left[\left( v_0 - w_0 \right) g^{0 \beta} \partial_{x^\beta}, S \right]
+ \left[ \wu_\alpha \Hu^{\alpha \beta}\delu_{x^\beta}, S \right] \nonumber \\
&=&w^\alpha \partial_{x^\alpha}+ F_{S1}+ F_{S2}. \label{eq:fts}
\end{eqnarray}
\end{lemma}
\begin{remark}\label{rem:s1}
The first term on the right-hand side of \eqref{eq:fts} will be added to some terms arising from the computation of $F_{S2}$, to obtain a copy of the $v_\alpha g^{\alpha \beta}\partial_{x^\beta}$ operator.
\end{remark}

The next lemma gives a decomposition of the error terms $F_{Z_i1}$ and $F_{S1}$.
\begin{lemma}
Let $Z$ be a Lorentz boost $Z_i$ or the scaling vector field $S$, then
$F_{Z1}$ can be written as a linear combination of
\begin{eqnarray*}
(v_0-w_0) g \partial_{t,x} f, \\
(v_0-w_0) Z(g) \partial_{t,x} f, \\
\widehat{Z} (v_0-w_0) g \partial_{t,x} f.
\end{eqnarray*}
\end{lemma}

To analyse the last term, recall that
$$v_0-w_0= \frac{1}{v_0+w_0} \left(  (v_0)^2 \Hu^{00}+ 2 \vu_0 \vu_a \Hu^{0a}+ \vu_a \vu_b \Hu^{ab} \right). $$

Crucially, the null structure of the right-hand side of the last line is preserved by commutation.

\begin{lemma}
Let $Z$ be a Lorentz boost or the scaling vector field, then $\widehat{Z} (v_0-w_0)$ can be expressed as linear combination of
\begin{eqnarray*}
&&\widehat{Z}\left( \frac{1}{v_0+w_0}\right) \left(  (v_0)^2 \Hu^{00}+2 \vu_0 \vu_a \Hu^{0a}+ \vu_a \vu_b \Hu^{ab} \right), \\
&&\frac{1}{v_0+w_0} \left( (v_0)^2 Z(\Hu^{00})+ 2 \vu_0 \vu_a Z(\Hu^{0a})+ \vu_a \vu_b Z(\Hu^{ab}) \right)\\
&&\frac{1}{v_0+w_0} \widehat{Z}(v_0) v_0 \Hu^{00}, \quad \frac{1}{v_0+w_0} \widehat{Z} \left(\vu_0 \right) \vu_a \Hu^{0a}, \\
&&\frac{1}{v_0+w_0} \vu_0 \widehat{Z}\left(\vu_a\right) \Hu^{0a}, \quad \frac{1}{v_0+w_0} \widehat{Z}(\vu_a) \vu_b Z(\Hu^{ab}),
\end{eqnarray*}
where
\begin{eqnarray*}
\left| \widehat{Z}\left( \frac{1}{v_0+w_0}\right)\right| &\lesssim& \frac{1}{w^0},\\
| \widehat{Z}(v_0) | &\lesssim& w^0
\end{eqnarray*}
and for $Z=Z_i= t \partial_{x^i} +x^i \partial_t$,
$$\widehat{Z}_i (\vu_a)= \frac{-x_a}{t} \vu_i,$$
while for $Z= S$, $$\widehat{Z}(\vu_a)=S(\vu_a)=0.$$
In particular, the null structure of $v_0-w_0$ is preserved by commutation and $\widehat{Z} (v_0-w_0)$ can be estimated as $v_0-w_0$, i.e.~
 \begin{eqnarray*}
 \left| \widehat{Z} (v_0-w_0) \right|
 &\lesssim& | w_0| \left(  \rho^{\delta/2} \epsilon^{1/2} \frac{1+u}{t^{3/2}} +  \frac{\epsilon^{1/2}}{t}\right)  +  \frac{|\vu_a|^2}{ | w_0|} \frac{t}{\rho} \rho^{\delta/2} \epsilon^{1/2} \frac{1+u^{1/2}}{t}.
\end{eqnarray*}

\end{lemma}

Since $\frac{1}{v_0+w_0}=\frac{1}{2w_0} + \frac{O(h)}{w^0}$, $\Hu=\hu+O(h^2)$ and $g^{-1}=\eta+O(h)$, we can summarize the structure of the error term $F_{Z1}$ as follows.

\begin{lemma} \label{lem:comZx1}
Let $Z$ be a Lorentz boost or the scaling vector field, then
$F_{Z1}$ can be written as a linear combination of
\begin{eqnarray*}
&& \frac{1}{w^0} \wu_\alpha K^\gamma (\wu_\beta) K^\delta(\Hu^{\alpha \beta}) \partial_{t,x} f, \quad |\gamma|+|\delta| \le 1\\
&& w\cdot K^\alpha (h) h \partial_{t,x} f, \quad |\alpha| \le 1.
\end{eqnarray*}

\end{lemma}
\begin{remark}Eventhough we have an explicit formula for $\widehat{Z}( \wu_\alpha)$, we often keep it as such in several formulae as it allows to keep track of the $\widehat{Z}$.
\end{remark}
\begin{remark}
Note that the second type of error terms above is cubic and can be estimated by
$$
|w\mu  K(h) h \partial_{t,x} f| \lesssim w^0 \frac{\rho^{\delta/2}u}{t^2} \le w^0\frac{\rho}{t} \frac{\rho^{\delta/2}}{t }.
$$
These cubic terms are therefore borderline terms if we only use the basic decay estimate \eqref{es:bde2}. This estimate will be sufficient nonetheless to close the first order energy estimates with a $\rho^{\delta/2}$ loss.
\end{remark}

Similarly, we have
\begin{lemma}\label{lem:comZx2}
$F_{Z2}$ can be written as a linear combination  of

\begin{eqnarray*}
\wu_\alpha Z(\Hu^{\alpha \beta}) \delu_{x^\beta}, \\
\vu_b \Hu^{a \beta} \delu_{x^\beta}, \\
\widehat{Z}( \wu_0) \Hu^{00} \partial_t, \\
\wu_\alpha \Hu^{\alpha \beta}[\widehat{Z}, \delu_{x^\beta}],
\end{eqnarray*}
where $\left| [\widehat{Z}_i, \delu_{x^a}](f) \right| \le |\delu_{x^a} f|$,$\left| [\widehat{Z}_i, \delu_{t}](f) \right| \le | \partial_{x^i} f|$
and
$[S, \delu_{x^\beta}](f)=- \delu_{x^\beta}$.
\end{lemma}

\begin{remark} \label{rem:s2}
Note that for the scaling vector field, we have
$$
v_\alpha g^{\alpha \beta} \partial_{x^\beta} \left(  x^\gamma \partial_{x^\gamma} f \right)= x^\gamma v_\alpha g^{\alpha \beta} \partial_{x^\beta} \partial_{x^\gamma} f + v_\alpha g^{\alpha \beta} \partial_{x^\beta} f.
$$
In other words, the term of the form $\wu_\alpha \Hu^{\alpha \beta}[S, \delu_{x^\beta}]=-\wu_\alpha \Hu^{\alpha \beta} \delu_{x^\beta}$ recombines with the $w^\gamma \partial_{x^\gamma}$ present in \eqref{eq:fts} to give an exact copy of $v_\alpha g^{\alpha \beta} \partial_{x^\beta}$. 
\end{remark}
\begin{remark}The worst terms in the above computation are of the form $\wu_0 Z(H^{0a}) \partial_{t,x}$ or $\wu_0 H^{0a} \partial_{t,x}$. They have a priori no null structure and will generate borderline error terms. In the next lemma, we will purposefully forget some of the structures we have found above, since these structures are not present in the worst terms.
\end{remark}

We summarize the computation of the error terms $F_{Z1}$ and $F_{Z2}$ as follows.

\begin{lemma} \label{lem:comZx}
Let $Z$ be a Lorentz boost or a scaling vector field, then $[v_\alpha g^{\alpha \beta} \partial_{x^\beta}, \widehat{Z}]$ can be written as a linear combination of three possible error terms
\begin{eqnarray*}
&&w\cdot K^\alpha(h)\cdot\partial_{t,x},  \quad |\alpha| \le 1, \\
&& w\cdot K^\alpha(h)h\cdot\partial_{t,x}, \quad |\alpha| \le 1, \\
&&v_\alpha g^{\alpha \beta} \partial_{x^\beta}.
\end{eqnarray*}
\end{lemma}

Recall that in Lemma \ref{lem:ntnf}, we proved that $\partial_{x^i} k. \partial_{v_i}$ was in fact a null form, in the sense that it behaves better than an arbitrary product $\partial k. \partial_v$.

In the next lemma, we prove that this null structure is preserved by commutation with $\widehat{Z}_j$ or $S$.

\begin{lemma} For any regular function $k:=k(t,x)$ of $(t,x)$, we have
\begin{itemize}
\item for any Lorentz boost $Z_j$,
\begin{equation} \label{eq:Zpkv}
[ \partial_{x^i} k \cdot  \partial_{v_i}, \widehat{Z}_j ]= - \partial_{x^i} Z_j(k)\cdot \partial_{v_i}+ \frac{1}{w^0}w^\alpha \partial_{x^\alpha}(k)\cdot \partial_{v_j}.
\end{equation}
\item for the scaling vector field,
\begin{equation} \label{eq:Spkv}
[ \partial_{x^i} k \cdot  \partial_{v_i}, S ]= - \partial_{x^i} S(k)\cdot \partial_{v_i}+\partial_{x^i} k \cdot  \partial_{v_i}.
\end{equation}
\end{itemize}
\end{lemma}
\begin{proof}
We have
\begin{eqnarray}
\hspace{0.3cm}[ \partial_{x^i} k \cdot  \partial_{v_i}, \widehat{Z}_j ]&=&-Z_j \left(\partial_{x^i} k \right)\cdot \partial_{v_i} + \frac{1}{w^0}w^i \partial_{x^i}(k)\cdot \partial_{v_j}  \label{eq:intp} \\
&=&  -\partial_{x^i} Z_j(k)\cdot \partial_{v_i} + [ \partial_{x^i}, Z_j] (k)\cdot \partial_{v_i}+
\frac{1}{w^0}w^i \partial_{x^i}(k)\cdot \partial_{v_j} \nonumber\\
&=& - \partial_{x^i} Z_j(k)\cdot \partial_{v_i}+ \frac{1}{w^0}w^\alpha \partial_{x^\alpha}(k)\cdot \partial_{v_j}.\nonumber
\end{eqnarray}
The second computation is similar.
\end{proof}
\begin{remark}
The point of the lemma is that, each term in the right-hand side of \eqref{eq:Zpkv} and \eqref{eq:Spkv} can be seen as a null form, in the sense that it enjoys stronger decay properties than an arbitrary product $\partial Z(k)\cdot \partial_v $. For instance, the decomposition in the right-hand side of the first line of \eqref{eq:intp} would not be sufficient.
\end{remark}
\begin{remark}Consider the commutator $[T_g, S]$. In order to compute the commutator $[-\frac{1}{2} v_\alpha v_\beta \partial_{x^i} g^{\alpha \beta} \partial_{v_i}, S],$ we use formula \eqref{eq:Spkv}. The second term on the right-hand side of \eqref{eq:Spkv} will then generate a term of the form $-\frac{1}{2} v_\alpha v_\beta \partial_{x^i} g^{\alpha \beta} \partial_{v_i}$, which can be recombined with the $v_{\alpha}g^{\alpha \beta} \partial{x^\beta}$ arising from $[v_\alpha g^{\alpha \beta} \partial_{x^\beta}, S]$ to get an exact copy of $T_g$ (cf remarks \ref{rem:s1} and \ref{rem:s2} ).
\end{remark}

Using the previous lemma, we obtain
\begin{lemma} \label{lem:comZv}
For any Lorentz boost or scaling vector field $Z$, the commutator $[-\frac{1}{2} v_\alpha v_\beta \partial_{x^i} g^{\alpha \beta} \partial_{v_i}, \widehat{Z}]$ can be written as a linear combination of
\begin{eqnarray*}
&&\widehat{Z}(w_\alpha) w_\beta \partial_{x^i} g^{\alpha \beta} \cdot \partial_{v_i}, \\
&&w_\alpha w_\beta \partial_{x^i} (K^\mu h^{\alpha \beta}) \cdot \partial_{v_i}, \quad |\mu| \le 1,  \\
&&w_\alpha w_\beta \frac{w^\gamma }{w^0} \partial_{x^\gamma} h^{\alpha \beta} \cdot \partial_{v}
\end{eqnarray*}
and the cubic terms
\begin{eqnarray*}
&&w\cdot w\cdot K^\gamma(h) \partial_{x^i}(K^\mu h) \cdot  \partial_{v^i }, \quad |\gamma|+|\mu| \le 1, \\
&&w\cdot w\cdot h \frac{w^\gamma }{w^0} \partial_{x^\gamma} h\cdot  \partial_{v}.
\end{eqnarray*}
Moreover, any of the above terms can be decomposed on the semi-hyperboloidal frame, giving as error terms for $j=0,1$,
\begin{eqnarray*}
&&\widehat{Z}^j(\wu_0) \wu_\alpha \partial_{x^i} \hu^{0\alpha} \partial_{v_i}, \\
&&\widehat{Z}^j(\wu_a) \wu_\alpha \partial_{x^i} \hu^{a\alpha} \partial_{v_i}, \\
&&\widehat{Z}(w_\alpha) w_\beta h^{\alpha' \beta'} \partial_{x^i}( \Phi\cdot \Phi) \partial_{v_i}, \\
&&w_\alpha w_\beta h^{\alpha' \beta'} \partial_{x^i}Z^j\left( \Phi\cdot \Phi \right) \partial_{v_i}, \\
&&w_\alpha w_\beta Z(h^{\alpha' \beta'})\partial_{x^i}( \Phi\cdot \Phi) \partial_{v_i}, \\
&&w_\alpha w_\beta \partial_{x^i} Z(h^{\alpha' \beta'})( \Phi\cdot \Phi) \partial_{v_i}, \\
&&\wu_\alpha \wu_\beta \frac{w^\gamma}{w^0}\partial_{x^\gamma} (\hu^{\alpha \beta}) \partial_{v_j},\\
&&\wu_\alpha \wu_\beta \hu^{\alpha' \beta'} \frac{w^\gamma}{w^0}\partial_{x^\gamma} (\Phi\cdot \Phi) \partial_{v_j}.
\end{eqnarray*}
\end{lemma}
\begin{remark} The structure of the two cubic terms could actually be forgotten for the estimates to close.
\end{remark}

\begin{proof}
This follows by straightforward computations and decompositions
 on the semi-hyperboloidal frame of the form
$$
v_\alpha v_\beta \partial_{x^i} g^{\alpha \beta} \partial_{v_i}= \vu_\alpha \vu_\beta \partial_{x^i} \hu^{\alpha \beta} \partial_{v_i} + v_\alpha v_\beta h^{\alpha' \beta'}\left(\partial_{x^i}( \Phi\cdot \Phi)\right) \partial_{v_i}.
$$
\end{proof}

Let us summarize the results of Lemma \ref{lem:comZx} and \ref{lem:comZv} as
\begin{proposition}
The commutator $[T_g, \widehat{Z}]$ can be written as a linear combination of
\begin{itemize}
\item The terms in $\partial_{t,x}$, which we denote $F_{Zx}$,
\begin{eqnarray*}
&&w\cdot K^\mu (h)\cdot\partial_{t,x},  \quad |\mu| \le 1\\
&& w\cdot K^\mu (h)h\cdot\partial_{t,x}, \quad |\mu| \le 1.
\end{eqnarray*}
\item The terms in $\partial_v$ arising from Lemma \ref{lem:comZv}, which we denote $F_{Zv}$
\begin{eqnarray*}
&&\widehat{K}(w_\alpha) w_\beta \partial_{x^i} g^{\alpha \beta} \cdot \partial_{v_i}, \\
&&w_\alpha w_\beta \partial_{x^i} (K^\mu h^{\alpha \beta}) \cdot \partial_{v_i}, \quad |\mu| \le 1, \\
&&w_\alpha w_\beta \frac{w^\gamma }{w^0} \partial_{x^\gamma} h^{\alpha \beta} \cdot \partial_{v},\\
&&w\cdot w\cdot K^\gamma(h) \partial_{x^i}(K^\mu h) \cdot  \partial_{v^i },  \quad |\gamma|+ |\mu| \le 1,  \\
&&w\cdot w\cdot h \frac{w^\gamma }{w^0} \partial_{x^\gamma} h\cdot  \partial_{v}.
\end{eqnarray*}
\item The scaling term $T_g$.
\end{itemize}
\end{proposition}
The error terms in the class $F_{Zv}$ need to be rewritten since $\partial_v$ is not part of the algebra of commuting vector fields. Moreover, some care is needed in order to exploit the null condition present in each non-linear product. For this, we use the identities of Lemma \ref{lem:ntnf}.

We can then rewrite each term in the $F_{Zv}$ class as follows
\begin{lemma} The error terms in the $F_{Zv}$ class can be written as a linear combinations of (we have suppressed some indices below for clarity in the exposition)
\begin{enumerate}
\item The good terms
\begin{eqnarray*}
&&\widehat{K}^\sigma (w_\alpha) w_\beta \partial_{t,x} K^\gamma(h^{\alpha \beta}) \frac{\widehat{Z}}{w^0}, \quad |\sigma|+|\gamma| \le 1, \\
&&\widehat{K}^\sigma (w_\alpha) w_\beta \partial_{t,x} K^\gamma(h^{\alpha \beta}) \frac{\partial_{t,x}}{w^0}, \quad |\sigma|+|\gamma| \le 1, \\
&&w\cdot K^\sigma (h) \partial_{x^i} K^\gamma(h) \widehat{Z}, \quad |\sigma|+|\gamma| \le 1.
\end{eqnarray*}
\item The bad terms obtained from expanding null forms of type $\partial_{x^i} k . \partial_{v^i}$
\begin{eqnarray*}
&&\widehat{K}^\mu (w_\alpha) w_\beta \delu_{x^k} K^\gamma (h^{\alpha \beta}) \frac{t}{w^0} X_i , \quad  |\mu|+|\gamma| \le 1, \\
&&\widehat{K}^\mu(w_\alpha) w_\beta \delu_{x^k} K^\gamma(h^{\alpha \beta}) \frac{ t \vu_a}{(w^0)^2} \partial_t , \quad  |\mu|+|\gamma| \le 1, \\
&&\widehat{K}^\mu (w_\alpha) w_\beta \del_{t} K^\gamma(h^{\alpha \beta})\frac{|x|^2-t^2}{t} \frac{\partial_t}{w^0}, \quad |\mu|+|\gamma| \le 1.
\end{eqnarray*}
\item The other bad terms obtained from expanding null forms of type $\frac{w^\gamma}{w^0} \partial_{x^\gamma}(k)\cdot \partial_{v}$
\begin{eqnarray*}
&&\widehat{K}^\sigma (w_\alpha) w_\beta \partial_{t,x} K^\gamma (h^{\alpha \beta}) \frac{t \vu_i }{(w^0)^2} X_i , \quad  |\sigma|+|\gamma| \le 1, \\
&&\widehat{K}^\sigma(w_\alpha) w_\beta \partial_{t,x} K^\gamma(h^{\alpha \beta}) \frac{ t \vu_i \vu_a}{(w^0)^3} \partial_t , \quad  |\sigma|+|\gamma| \le 1, \\
&&w\cdot K^\gamma (h)\cdot\partial_{t,x},  \quad |\gamma| \le 1.
\end{eqnarray*}
\item Borderline cubic terms (will be counted as bad terms below)
\begin{eqnarray*}
&&w\cdot K^\sigma(h)\cdot\partial_{t,x} K^\gamma (h)\cdot t \partial_{t,x},  \quad |\sigma|+|\gamma| \le 1.
\end{eqnarray*}
\end{enumerate}
\end{lemma}

Using the decay estimates \eqref{es:bde}, \eqref{es:ide1} and \eqref{es:ide2} for $h$, we expect the good terms to be integrable terms\footnote{Since we eventually replace $\widehat{Z}$ by modified vector fields of the form $Y=\widehat{Z}+C\cdot X$, estimating the good terms will be slightly more complicated and in particular, they will also generate some borderline terms generating a small growth.}, 
while the bad terms need to be canceled by suitable correction factors.

Finally, we summarize the computation of the whole commutator $[T_g, \widehat{Z}]$ as follows.

\begin{lemma} \label{lem:sumZc}
The commutator $[T_g, \widehat{Z}]$ verifies
$$
[T_g, \widehat{Z}]= F_{ZG}+ F_{ZB}
$$
where $F_{ZG}$ 
can be written as a linear combination of the good terms 
\begin{eqnarray*}
&&\widehat{K}^\sigma (w_\alpha) w_\beta \partial_{t,x} K^\gamma(h^{\alpha \beta}) \frac{\widehat{Z}}{w^0}, \quad |\sigma|+|\gamma| \le 1, \\
&&\widehat{K}^\sigma (w_\alpha) w_\beta \partial_{t,x} K^\gamma(h^{\alpha \beta} )\frac{\partial_{t,x}}{w^0}, \quad |\sigma|+|\gamma| \le 1, \\
&&w\cdot K^\sigma(h) \partial_{t,x} K^\gamma(h) \widehat{Z}, \quad |\sigma|+|\gamma| \le 1, \\
&&T_g=v_\alpha g^{\alpha \beta} \partial_{x^\beta} - \frac{1}{2} v_\alpha v_\beta \partial_{x^i} g^{\alpha \beta}\cdot \partial_{v_i}
\end{eqnarray*}

and where $F_{ZB}$ can be written as a linear combination of the bad terms

\begin{eqnarray*}
&&w\cdot K^\sigma (h)\cdot\partial_{t,x},  \quad |\sigma| \le 1, \\
&& w\cdot K^\sigma (h)h\cdot\partial_{t,x}, \quad |\sigma| \le 1, \\
&&\widehat{K}^\sigma (w_\alpha) w_\beta \delu_{x^k} K^\gamma (h^{\alpha \beta}) \frac{t}{w^0} X_i , \quad  |\sigma|+|\gamma| \le 1, \\
&&\widehat{K}^\sigma(w_\alpha) w_\beta \delu_{x^k} K^\gamma(h^{\alpha \beta} )\frac{ t \vu_a}{(w^0)^2} \partial_t , \quad  |\sigma|+|\gamma| \le 1, \\
&&\widehat{K}^\sigma (w_\alpha) w_\beta \del_{t} K^\gamma(h^{\alpha \beta})\frac{|x|^2-t^2}{t} \frac{\partial_t}{w^0}, \quad |\sigma|+|\gamma| \le 1, \\
&&\widehat{K}^\sigma (w_\alpha) w_\beta \partial_{t,x} K^\gamma (h^{\alpha \beta} )\frac{t \vu_i }{(w^0)^2} X_i , \quad  |\sigma|+|\gamma| \le 1, \\
&&\widehat{K}^\sigma(w_\alpha) w_\beta \partial_{t,x} K^\gamma(h^{\alpha \beta}) \frac{ t \vu_i \vu_a}{(w^0)^3} \partial_t , \quad  |\sigma|+|\gamma| \le 1, \\
&&w\cdot K^\sigma(h)\cdot \partial_{t,x} K^\gamma (h)\cdot t \partial_{t,x},  \quad |\sigma|+|\gamma| \le 1.
\end{eqnarray*}

\end{lemma}

\begin{definition} \label{def:compFB}
In the following, we denote by $F^i_{ZB}$ and $F^0_{ZB}$ the components of  $F_{ZB}$ in the basis $\{\partial_t, X_i\}$.
\end{definition}

\begin{remark}
The terms in $F_{ZB}$ are of two sorts. For the first one $($for instance $w\cdot K^\sigma (h)\cdot\partial_{t,x})$, one does not need to check carefully the null structure of the equations (eventhough, a careful analysis reveals that there is indeed such a structure). For the second one (for instance $\widehat{K}^\sigma (w_\alpha) w_\beta \delu_{x^k} K^\gamma (h^{\alpha \beta} )\frac{t}{w^0} X_i$), the decomposition of $\partial_v$ in terms of the commutator vector fields introduces $t$ weights. To compensate for these $t$ weights, one needs to carefully take into account the structure of the products.
\end{remark}

\subsection{The correction terms}\label{sec:corrterm}

For any homogeneous vector field $Z$, with $Z= Z_j$ for a Lorentz boost or $Z=S=Z_0$ for the scaling vector field, we now consider modified vector fields of the form
\begin{eqnarray*}
Y_\alpha&=& \widehat{Z}_\alpha + C^i_\alpha X_i +C^0_\alpha \partial_t \\
&=& \widehat{Z}_\alpha + C^\beta_\alpha X_\beta,
\end{eqnarray*}
where by definition $X_0=\partial_t$.
We have
\eq{\alg{ \label{eq:comY1}\/
[ T_g, Y_\alpha ]&= [T_g, \widehat{Z}_\alpha] \\
\hbox{}&+ T_g(C_\alpha^\beta) X_\beta+ C_\alpha^\beta [T_g, X_\beta].
}}

We then define $C_\alpha^\beta$ as the solution of the inhomogeneous problem
\begin{equation}
T_g(C^\beta_\alpha)= - F^\beta_{Z_\alpha B}, \quad C^\beta_\alpha(\rho=2)=0,
\end{equation}
 where $F^\beta_{Z_\alpha B}$ are the components of $F_{Z_\alpha B}$ in the $\partial_t, X_i$ basis as in Definition \ref{def:compFB}.



\subsection{First estimate for the $C$ coefficients}
Similar to \cite[Section 6]{fjs:savn}, we have the following estimate.

\begin{lemma} \label{lem:linft} Assume that $T_g(|C|)\le F$, then
$$
|| C(\rho) || \le \int_2^\rho || \frac F{v^\rho} ||_{L^\infty(H_{\rho'})} d\rho',
$$
where
\begin{equation} \label{def:vrho}
v^\rho:=v_\alpha g^{\alpha \beta} \partial_{x^\beta}(\rho)= \frac{v_\alpha g^{\alpha 0} t - x_j v_\alpha g^{j \alpha}}{\rho}.
\end{equation} 

\end{lemma}
\begin{proof}
Recall first that, by integration along characteristics, for any solution $U$ to $T_g( U)=0$ with initial data prescribed at $\rho=\rho'$, we have
$$
| U(\rho,x,v)| \le || U_{\rho=\rho'} ||_{L^\infty}.
$$
The lemma then follows from the Duhamel formula, which we recall below.

Let $U(\rho, \rho',x,v)$ be the solution to
\begin{eqnarray*}
T_g(U)&=&0, \\
U(\rho=\rho', \rho', x, v)&=& \frac{F}{v^\rho}\left( \rho', x, v\right).
\end{eqnarray*}
Then, if $T_g(C)=F$ and $C(\rho=2)=0$, we have
$$C=\int_2^\rho U(\rho, \rho', x, v) d\rho'.$$
\end{proof}

In view of the weights decomposition and the estimates on $v_0-w_0$, we have, similar to the analysis of Section \ref{se:coeref},

\begin{lemma} 
Under conditions \eqref{cond-equiv-0} the following estimates hold.
\begin{eqnarray*}
|\frac{w^0}{v^\rho} | \lesssim \frac{t}{\rho}, \quad | \frac{\vu_a}{v^\rho}|  \lesssim 1, \quad \frac{t}{\rho}| \frac{\vu_a \vu_b}{v^\rho w^0}| \lesssim 1.
\end{eqnarray*}
\end{lemma}
\begin{proof}
The first estimate is derived by estimating
\eq{\label{asfji}\alg{
v^\rho&=\frac{v^0t-v^ix_i}{\rho}-\frac{x_\mu H^{\mu\alpha}v_\alpha}{\rho}=w^\rho+\frac{t}{\rho}(v^0-w^0)-\frac{x_\mu H^{\mu\alpha}v_\alpha}{\rho}\\
&\lesssim (1+\sqrt{\varepsilon})w^\rho
}
}
analogous to the corresponding term estimated before \eqref{sdflj}, where $I_\eta/\rho=w^\rho$. Using this in combination with Remark 2.12 of \cite{fjs:vfm} we obtain
\eq{
\frac{w^0}{v^\rho}\lesssim\frac{w^0}{w^\rho}\lesssim\frac{t}{\rho},
}
which proves the first estimate.
The second estimate follows by comparing the integrand of \eqref{dec-bck-en} with the standard form as written for instance in \eqref{bck-en}, which yields
\eq{\label{sdfoih}
 \frac{\underline v_i^2}{w^0}\lesssim w^{\rho}\frac{\rho}{t}\lesssim v^\rho\frac{\rho}{t},
}
where we used $w^\rho\lesssim v^\rho$, which can be shown as \eqref{asfji} above.
Multiplying by $w^0$ and using the first estimate implies the second claim. Finally, the third estimate follows directly from \eqref{sdfoih}.
\end{proof}
In view of the decay estimates on $h$, this leads to the following estimate for the coefficients $C$.

\begin{lemma}  
The $C$ coefficients satisfy the estimate
$$
| C_\alpha | \lesssim \epsilon^{1/2} \rho^{\delta/2+1/2}.
$$
\end{lemma}
\begin{proof}

Consider first a source term in the equation for $C$ of the form.

$$
w\cdot K^\gamma(h),  \quad |\gamma| \le 1.
$$
This can be estimated as
$$
\left| w\cdot K^\gamma (h)\right| \lesssim \epsilon^{1/2} \frac{\rho}{t}w^0 \rho^{\delta/2} \frac{u^{1/2}}{\rho} \le \epsilon^{1/2}\frac{\rho}{t}w^0 \rho^{\delta/2} \frac{1}{\rho^{1/2}}.
$$

Consider now a source of the form
$$
\widehat{K}^\sigma (w_\alpha) w_\beta \delu_{x^k} K^\gamma(h)^{\alpha \beta} \frac{t}{w^0} , \quad |\sigma|+|\gamma| \le 1.
$$
Using that $t \delu_{x^k}=Z_k$, this can again be estimated by

$$
\epsilon^{1/2} \frac{\rho}{t}w^0 \rho^{\delta/2} \frac{u^{1/2}}{\rho} \le \epsilon^{1/2}\frac{\rho}{t}w^0 \rho^{\delta/2} \frac{1}{\rho^{1/2}}.
$$

For a term of the form $\widehat{K}^\sigma (w_\alpha) w_\beta \partial_{t,x} K^\gamma (h)^{\alpha \beta} \frac{t \vu_i }{(w^0)^2}, \quad  |\sigma|+|\gamma| \le 1$, one first applies a decomposition on the semi-hyperboloidal frame. Any term containing derivatives of the frame field $\Phi$ or $\hu^{00}$ is easily seen to satisfy an estimate similar to those already obtained. This leaves terms of the form

$$\partial_{t,x} K^\gamma (\hu) \frac{\vu_a  \vu_i t}{w^0}.$$
Those can be estimated as follows

\begin{eqnarray*}
\left| \partial_{t,x} K^\gamma (\hu) \frac{ \vu_a \vu_i t}{w^0} \right| &\lesssim& \epsilon^{1/2} \rho^{\delta/2} \frac{1}{u^{1/2}} \frac{\rho}{t} \frac{|\vu_a||\vu_i|}{w^0} \frac{t}{\rho}, \\
&\lesssim & \epsilon^{1/2} \rho^{\delta/2} \frac{u^{1/2}}{\rho} \frac{|\vu_a||\vu_i|}{w^0} \frac{t}{\rho}.
\end{eqnarray*}

All the other source terms satisfy similar estimates and the statement then follows from the previous two lemmas.
\end{proof}

We can in fact prove an improved estimate for the $C$ coefficients, reflecting the null structure of the equations.
\begin{lemma}\label{lem:ccu}
The $C$ coefficients satisfy the estimate
\begin{eqnarray} \label{es:ccu}
| C | \lesssim \epsilon^{1/2}\frac{1}{\delta} u^{1/2} \rho^{\delta/2}.
\end{eqnarray}
\end{lemma}
\begin{proof}We do the proof for the coefficients $C^0_\alpha$, the others being similar. Dropping the $\alpha$ index, we compute
$$
T_g\left( \frac{C^0}{(1+u)^{1/2}} \right)=\frac{-1}{(1+u)^{1/2} }F^0_{Z B} - \frac{1}{2}C^0 T_g(u) \frac{1}{(1+u)^{3/2} }.
$$

To estimate the right-hand side, we first need
$$| F^0_{ZB} | \lesssim \epsilon^{1/2} \frac{\rho}{t}w^0 \rho^{\delta/2} \frac{u^{1/2}}{\rho}+  \epsilon^{1/2} |\vu_a| \frac{u^{1/2}}{\rho}\rho^{\delta/2} + \epsilon^{1/2} \frac{|\vu_a|^2}{w^0}\frac{t}{\rho} \frac{u^{1/2}}{\rho}\rho^{\delta/2}.$$

Now recall that $u=t-|x| \ge 1$ in $\Kcal$ and, from Lemma \ref{lem:weghtingwithu} that
\eq{
T_g(1+u)=w^0-w_i\frac{x^i}{|x|}+E_u,
}
where
$w^0-w_i\frac{x^i}{|x|} \ge 0$
and
\begin{eqnarray*}
|E_u|&\lesssim& | w_0|  \left( \rho^{\delta/2} \epsilon^{1/2} \frac{1+u}{t^{3/2}} + \frac{\epsilon^{1/2}}{t} \right) +  \frac{|\vu_a|^2}{ | w_0|} \frac{t}{\rho} \rho^{\delta/2} \epsilon^{1/2} \frac{1+u^{1/2}}{t}\\
&\lesssim& \left(  | w_0| \frac{t}{\rho}+  \frac{|\vu_a|^2}{ | w_0|} \frac{t}{\rho}\right) \rho^{\delta/2} \epsilon^{1/2} \frac{1+u^{1/2}}{t}
+  | w_0| \frac{\epsilon^{1/2}}{t},
\end{eqnarray*}
where the contribution of the first term on the right-hand side will lead to integrable terms below while we treat the contribution of the second term separately.


Using that
$$
T_g\left( \frac{|C^0|}{(1+u)^{1/2}} \right)\le \frac{1}{(1+u)^{1/2} }|F^0_{Z B}| + \frac{1}{2} \frac{|C^0|}{ (1+u)^{1/2} }E_u \frac{1}{1+u },
$$
and the above estimates on $C^0$, $E_u$ and  $F^0_{ZB}$, we are left with

\begin{eqnarray}
T_g\left( \frac{|C^0|}{(1+u)^{1/2}} \right)&\lesssim& F +  \frac{1}{2} \frac{|C^0|}{ (1+u)^{1/2} }w^0\frac{\rho}{t} \frac{\epsilon^{1/2}}{\rho\left( 1+u\right) }, \label{eq:inttcu}\\
& \lesssim &  F +  \frac{1}{2} \frac{|C^0|}{ (1+u)^{1/2} }w^0\frac{\rho}{t} \frac{\epsilon^{1/2}}{\rho}, \nonumber
\end{eqnarray}
where
$$\left|\frac{F}{v^\rho} \right| \lesssim \epsilon^{1/2} \rho^{-1+\delta/2}.$$

so that the result follows from the $L^\infty$-estimate in Lemma \ref{lem:linft} and an application of Gronwall's lemma.
\end{proof}
\begin{remark}
In \eqref{eq:inttcu}, we could not make use of the extra $u$ decay in the term $\frac{1}{2} \frac{|C^0|}{ (1+u)^{1/2} }w^0\frac{\rho}{t} \frac{\epsilon^{1/2}}{\rho\left( 1+u\right) }$. However, note that instead of using the estimate $w^0 \frac{\rho}{t}\lesssim v^\rho$, we could have used $w^0 \le w^0 v^\rho$, lose a power of $w^0$ in the final estimate but obtain an integrable term. This means that this term will be integrable provided we can absorb an extra power of $v$ in the norm of $f$.
\end{remark}

\subsection{End of the derivation of the first order commutator formula}
It follows from the definition of the $C^\beta$ coefficients, the definition of $Y$ and equation \eqref{eq:comY1}, that we have
\begin{eqnarray*}
[T_g, Y ]&=& F_{ZG}+ C^\beta [T_g, X_\beta],
\end{eqnarray*}
where $F_{ZG}$ was defined in Lemma \ref{lem:sumZc}.
$F_{ZG}$ contains $\widehat{Z}$, which needs to be replaced by $Y=\widehat{Z}+ C^\alpha X_\alpha$ vector fields.

The result of this operation is the content of the following final lemma of this section.
\begin{lemma}
The commutator $[T_g, Y]$ can be written as a linear combination of the following terms

\begin{itemize}
\item The terms coming from $C^\beta [T_g, X_\beta]$
\begin{eqnarray*}
&&w\cdot h \partial_{t,x} (\Phi\cdot \Phi ) C \partial_{t,x}, \\
&&w\cdot \partial_{t,x}  (\hu^{00}) C \partial_{t,x},  \\
&&\vu_a \partial_{t,x}  (\hu) C \partial_{t,x}, \\
&&w\cdot \delu_{x^i}  (\hu) C \partial_{t,x}, \\
&&w \cdot h \cdot \partial_{t,x}  (h) C \partial_{t,x}, \\
&&w\cdot \partial h C\cdot X_a (f),
\end{eqnarray*}
as well as
\begin{eqnarray*}
&&\frac{1}{w^0}w_\alpha w_\beta \partial_{t,x} K^\gamma (h^{\alpha \beta} )C^{k+1}\cdot \widehat{K}^{\mu} f, \\
&& w\cdot h \cdot \partial_{t,x} K^\gamma (h)  C^{k+1}\cdot \widehat{K}^{\mu} f, \\
&& w\cdot \partial_{t,x}h \cdot K^\gamma (h)  C^{k+1}\cdot \widehat{K}^{\mu} f,
\end{eqnarray*}
where $|\gamma|, |\mu|, k \le 1$ and $\gamma_X \ge k+\mu_Z$.
\vspace{0.3cm}
\item The terms coming from $F_{ZG}$
\begin{eqnarray*}
&&\widehat{K}^\sigma (w_\alpha) w_\beta \partial_{x^\alpha} K^\gamma(h^{\alpha \beta})C^k \frac{\widehat{K}^\kappa}{w^0}, \quad |\sigma|+|\gamma| \le 1,\quad |\kappa| \le 1,\quad k+\kappa_Z \le 1, \\
&& w\cdot K^\sigma(h)\partial_{t,x}K^\gamma(h)C^k \widehat{K}^\kappa, \quad |\sigma|+|\gamma| \le 1,\quad |\kappa| \le 1,\quad k+\kappa_Z \le 1, \\
&&T_g=v_\alpha g^{\alpha \beta} \partial_{x^\beta} - \frac{1}{2} v_\alpha v_\beta \partial_{x^i} g^{\alpha \beta}\cdot \partial_{v_i}
\end{eqnarray*}
\end{itemize}
\end{lemma}
In view of the energy estimate \eqref{es:eevf}, the estimates on the $C$ coefficients and the previous estimates on each of the previous error terms, we have already proven
\begin{lemma} \label{es:Y1e}
For any modified vector field $Y$, we have, for all $\rho \in [2, \rho^{*})$,
\begin{eqnarray*}
E[Y(f)] (\rho) \lesssim E[Y(f)](2)+ \epsilon^{1/2} \sum_{X=X_i, \partial_t} \int_2^\rho \frac{1}{\rho^{1-\delta}} E[Xf] d\rho'\\
\hbox{}+ \epsilon^{1/2} \sum_{\widehat{K}=X, Y} \int_2^\rho \frac{1}{\rho^{3/2-3/2 \delta}} E[\widehat{K} f] d\rho'.
\end{eqnarray*}
As a corollary, we have, for any modified vector field, $E[Y(f)] (\rho) \lesssim \epsilon+ \epsilon^{3/2}\rho^{\delta }$.
\end{lemma}

\section{Commutators of the algebra of modified vector fields} \label{se:camvf}
In this section, we analyse commutators $[ \widehat{K}^\alpha, \widehat{K}^\beta]$ for $\widehat{K}^\alpha$, $\widehat{K}^\beta$ composed of the $X, \partial_t$ and $Y$ vector fields.
At first order, we have

\begin{lemma} The following commutators hold.
\begin{itemize}
\item For any $X_i$, $X_j$, $[X_i, X_j]=0$ and $[X_i, \partial_t]=0$.
\item For any $X=X_i, \partial_t$ and any modified vector field $Y$, $[X, Y]$ can be written as a linear combination of
$$
\partial_t,\, X_i \mbox{ and } X(C^\alpha)\cdot X_\alpha.
$$
\item For any $Y_i, Y_j$ two modified vector fields, $[Y_i, Y_j]$ can be written as a linear combination of
$$
Y(C^\alpha)\cdot X_\alpha,\, C\partial_t,\, CX_i,\, CX(C^\alpha)\cdot X_\alpha\mbox{ and } \widehat{\Omega_{ij}}, 
$$
where $\widehat{\Omega_{ij}}=x^j \partial_{x^i}-x^i \partial_{x^j} + v_j \partial_{v_i}- v_i \partial_{v_j} $ is the complete lift of a rotation vector field.
\end{itemize}
\end{lemma}
\begin{remark}
From the above lemma, and in view of the fact $|X(C)| \lesssim \epsilon^{1/2} \rho^{\delta/2}$ (cf.~Lemma \ref{lem:iexc}), while $|Y(C)|,|C| \lesssim \epsilon^{1/2} \rho^{\delta/2+1/2}$, it seems that the commutators $[X, Y]$ do not generate much growth, while those of the form $[Y_i, Y_j]$ would be problematic. As it turns out, for the present paper, we only need to use commutators of the form $[X, Y]$. However, the commutator $[Y_i, Y_j]$ in fact behaves better than what the above lemma suggest. For instance, a careful analysis of the error terms suggests that one of the error terms will be given by a modified, lifted rotation vector fields of the form $\widehat{\Omega_{ij}}+ C_i^\alpha X_\alpha  -C_j^\beta X_\beta$.
\end{remark}

We also need the following higher order version of the above lemma.
\begin{lemma} \label{lem:commvf}
For any multi-index $\alpha$ and any $X=X_i, \partial_t$, and any multi-index $\alpha$ with $\alpha_Z \ge 1$, the commutator $[\widehat{K}^\alpha, X]$ can be written as a linear combination of terms of the form
$$
P(X(C))^{k,r_Z,s_X}\cdot K^{\alpha'},
$$
where $|\alpha'|+r_Z+s_X \le |\alpha|$, $|\alpha'| \le |\alpha|$, $\alpha'_X \ge 1+\alpha_X$, $k \le |\alpha|$, $r_Z+s_X \le |\alpha|-1$ and where
we denote by $P(X(C))^{k,r_Z,s_X}$ a product of the form
$$
\prod_{i=1}^k \widehat{K}^{\rho_i}(X(C_i)),
$$
where the $C_i$s can be any of the $C$ coefficients and where the total number of $Y$ vector fields appearing on the right-hand side is less than $r_Z$ and the total number of $X$ vector fields is less than $s_X$.

\end{lemma}
\begin{proof}We have already proven the formula for $|\alpha|=1$. Let $\alpha$ be a multi-index and $\widehat{K}$ be any of the $X$ or $Y$ vector field. Assume that the formula holds for any multi-index of length $|\gamma|\le|\alpha|$ with $\gamma_Z \ge 1$.

Consider the commutator
$$[K\widehat{K}^\alpha, X]=K [\widehat{K}^\alpha, X]+[K,X]\widehat{K}^\alpha.$$
\begin{itemize}
\item If $K=X$ and $\alpha_Z \ge 1$ (otherwise the commutator vanishes) then
\begin{eqnarray*}
[X \widehat{K}^\alpha, X]&=&X[\widehat{K}^\alpha, X]\\
&=& X\left( P(X(C))^{k,r_Z,s_X}\cdot K^{\alpha'}\right),
\end{eqnarray*}
where $|\alpha'|+r_Z+s_X \le |\alpha|$, $|\alpha'| \le |\alpha|$, $\alpha'_X \ge 1+\alpha_X$, $k \le |\alpha|$, $r_Z+s_X \le |\alpha|-1$, so that distributing the $X$ vector field, we see that the formula holds.
\item If $K=Y$, then
\begin{eqnarray*}
[Y \widehat{K}^\alpha, X]&=&Y [\widehat{K}^\alpha, X]+[Y, X] \widehat{K}^\alpha.
\end{eqnarray*}
For the first term on the right-hand side, either $\alpha_Z=0$ and this term vanishes, or $\alpha_Z \ge 1$ and we can apply the commutator formula for $\alpha$ and then distribute the $Y$.
On the other hand, for the second term on the right-hand side, we use the previous lemma, and it follows that $[Y, X] \widehat{K}^\alpha$ can be written as a linear combination of
$$
X K^\alpha, \quad X(C)XK^\alpha,
$$
which are of the required form.

\end{itemize}
\end{proof}

\section{Higher order estimates for the Vlasov field} \label{se:hovf}

\subsection{Higher order commutator formula and energy norms}

We define the energy norm of order $N$ of $f$ by
$$
E_N[f](\rho):= \sum_{|\alpha| \le N } E[ \widehat{K}^\alpha f ](\rho),
$$
where $\widehat{K}$ denotes any vector fields among $\partial_t$, $X_i$ or the $Y$ vector fields.
Let us also define the weighted norm
\begin{equation} \label{def:enq}
E_{N,q}[f](\rho):= \sum_{|\alpha| \le N } E\left[(1+|v|^{2})^{q/2}]\widehat{K}^\alpha f \right](\rho).
\end{equation}

In this section, we shall first propagate bounds for $E_{N-2}[f](\rho)=E_{N-2,q=0}[f](\rho)$. Note that for $|\alpha| \le N-2$, we have access to pointwise bounds on $\partial h$.

In view of the computations of Section \ref{se:multivz}, we can actually add the extra $v$ weights and propagate bounds for any $E_{N-2,q}[f](\rho)$.
Finally, to propagate bounds for $f$ after $N$ commutations, we will need to lose $2$ powers of $v$ at lower order. Thus, to prove bounds on $E_{N}[f]$ requires bounds on $E_{N-2,2}[f]$ (this is the minimum extra $v$ weight we need for this section) and more generally, to prove bounds on $E_{N,q}[f]$ will require bounds on $E_{N-2,q+2}[f]$. \\


To compute the higher order commutators, we use the following notations.

\begin{itemize}
\item We denote by $K$ any standard boost, scaling vector field or one of the translation $\partial_{x^i},\partial_t$ vector fields.
\item We denote by $\widehat{K}$ any of the modified vector fields $Y$ or any of the translation $X=X_i,\partial_t$ vector fields.
\item For any multi-index $\alpha$ of length $|\alpha|$ and differential operator $K^\alpha$, we write
$|\alpha|=\alpha_Z+\alpha_X$, where $\alpha_Z$ denotes the number of homogeneous vector fields in $K^\alpha$ and $\alpha_X$ the number of translation vector fields.
\item Similarly, for any multi-index $\alpha$ of length $|\alpha|$ and differential operator $\widehat{K}^\alpha$, we write
$|\alpha|=\alpha_Z+\alpha_X$ , where $\alpha_Z$ denotes the number of modified vector fields in $\widehat{K}^\alpha$ and $\alpha_X$ the number of translation vector fields.
\item For multi-indices $\gamma$ and $\beta$, we denote by $[\widehat{K}^\gamma \widehat{K}^\beta]$ a differential operator composed of the same vector fields as $\widehat{K}^\gamma$ and $\widehat{K}^\beta$, but such that the order of theses vector fields is arbitrary.
\item We denote by $P(C)^{k,r_Z,s_X}$ a linear combination of products of the form
$$
\prod_{i=1}^k \widehat{K}^{\rho_i}(C_i),
$$
where the $C_i$s can be any of the $C$ coefficients and where the total number of $Y$ vector fields appearing on the right-hand side is less than $r_Z$ and the total number of $X$ vector fields is less than $s_X$.
\end{itemize}

First, by induction, we obtain easily from the first order commutator formula a general formula that does not take into account the null structure of the equations.
\begin{lemma}
With the above notations,  the commutator $[T_g, \widehat{K}^\alpha]$ can be written as a linear combination

\begin{itemize}\label{lem:mthc}
\item The main terms
\begin{eqnarray} \label{eq:mthc}
\frac{1}{t^q }P(C)^{k, r_Z,s_X} w \cdot \partial_{t,x} K^\beta(h) \widehat{K}^\sigma
\end{eqnarray}
where $r_Z+s_X+|\beta|+|\sigma| \le |\alpha|+1$, $q+|\beta| \le |\alpha|$, $1 \le |\sigma| \le |\alpha|$, $r_Z+s_X \le |\alpha|-1$ 
and the number of $C$ coefficients, $k$, satisfies either of the two conditions\\

\begin{enumerate}
\item[C1.] $k \le \beta_X+q$, \label{c1k} 
\item[C2.] $k =\beta_X+q+1$ and $\widehat{K}^\sigma=[\widehat{K}^{\sigma'} X]$ with $\sigma'_Z+r_Z+\beta_Z \le \alpha_Z-1$. \\
\end{enumerate}
\item The frame terms
$$
\frac{1}{t^q }P(C)^{k, r_Z,s_X} w \cdot  K^\gamma(h) \partial_{t,x} K^\beta(\Phi\cdot\Phi ) \widehat{K}^\sigma,
$$
where $r_Z+s_X+|\beta|+|\sigma|+|\gamma| \le |\alpha|+1$, $q+|\beta|, |\gamma|, |\sigma| \le |\alpha|$, $r_Z+s_X \le |\alpha|-1$  
and $k$ verifies either \\
\begin{enumerate}
\item[C1'.] $k \le \beta_X+q+\gamma_X$. \label{c1k} 
\item[C2'.] $k =\beta_X+q+1+\gamma_X$ and $\widehat{K}^\sigma=[\widehat{K}^{\sigma'} X]$ with $\sigma'_Z+r_Z+\beta_Z +\gamma_Z \le \alpha_Z-1$. \\
\end{enumerate}

\item The cubic terms
\begin{eqnarray*}
\frac{1}{t^q } P(C)^{k, r_Z,s_X} w \cdot  K^\gamma(h) \partial_{t,x} K^\beta(h) \widehat{K}^\sigma,
\end{eqnarray*}
where $q+ r_Z+s_X+|\beta|+|\sigma|+|\gamma| \le |\alpha|+1$, $|\beta|, |\gamma|, |\sigma| \le |\alpha|$, $r_Z+s_X \le |\alpha|-1$  
and $k$ verifies either $\mathrm{C1'}$ or $\mathrm{C2'}$.

\item The source terms
$$
K^\sigma T_g,
$$
with $|\sigma| \le |\alpha|-1$ with $\sigma_Z\le \alpha_Z-1$.
\end{itemize}
\end{lemma}
\begin{remark}
The factor $\frac{1}{t^q}$ results from $Y$ being applied to any function in $\mathcal{F}_{x,v}$. For instance, $Y( \frac{x}{t} )= \widehat{Z}\frac{x}{t}+ C\cdot X\frac{x}{t}$ and the second term is a linear combination of $\frac{C}{t}$. Thus, this increases $q$ and $k$ by $1$.
\end{remark}
\begin{remark}
Note that in the above formula, each new $C$ coefficient typically comes with an extra $X$ derivative hitting either $h$ or one of the frame coefficients. The $C$ coefficients gives a $u^{1/2} \rho^{\delta/2}$ growth each, which are compensated by the fact that the $X$ derivatives provide extra $u$ decay. Let us discuss conditions $\mathrm{C1}$ and $\mathrm{C2}$ on $k$ more precisely ($\mathrm{C1'}$ and $\mathrm{C2'}$ are then similar).
\begin{itemize}
\item The condition $k \le \beta_X+q$ implies in particular that $k/2-\beta_X-q \le 0$. All terms coming from the commutator $[T_g, X]$ have this property. In this case, the $u$ growth of the $C$ coefficients is always compensated (at least when pointwise estimates on $C$ are allowed) by the extra $u$ decay.
\item The condition $k =\beta_X+q+1$ and $\widehat{K}^\delta=[\widehat{K}^{\delta'} X]$ with $\delta'_Z+r_Z+\beta_Z \le \alpha_Z-1$ can only occur if $\alpha_Z\ge 1$. These terms typically stem from the product $C^\alpha [T_g, X_\alpha]$ contained in the commutator $[T_g, Y]$. Note that if $\beta_X\ge 1$, then we have again $k/2-\beta_X -q\le 0$ if $k =\beta_X+q+1$. The condition $\delta'_Z+r_Z+\beta_Z \le \alpha_Z-1$ means that for $k =\beta_X+q+1$ to occur, some of the commutation by $Y$ vector fields must have generated some $X$ derivatives. This typically occurs when we distribute $Y$ on a product $C \partial K^\beta (h) \cdot \widehat{K} (f)$. Indeed, writing $Y=Z+C\cdot X$, we have schematically
\eq{\alg{
Y \left( C \partial K^\beta (h) \cdot \widehat{K} (f)\right) &= Y(C)\cdot  \partial K^\beta(h) \cdot \widehat{K} (f) + C \partial K^\beta (h) \cdot Y \widehat{K} (f) \\
&\quad+ C Z (\partial  K^\beta(h)) \cdot \widehat{K} (f) + C^2 \partial^2 K^\beta(h) \cdot \widehat{K} (f).
}}
Compared with the original quantity $C \partial K^\beta (h) \cdot \widehat{K} (f)$, the new term $C^2 \partial^2 K^\beta(h) \cdot \widehat{K} (f)$ has one more $C$ and $\beta_X$ has also been increased by one by this operation. On the other hand, the number of $Y$ or $Z$ vector fields hitting all the terms has not increased.
\end{itemize}
\end{remark}

\begin{proof}
We do an induction in $|\alpha|$. The formula holds for $|\alpha|=1$ in view of the first order commutator formulas. Assume it holds for some multi-index $\alpha$ and consider a differential operator of the form $\widehat{K} \widehat{K}^\alpha$. We use that
$$
[T_g, \widehat{K} \widehat{K}^\alpha ] = [T_g, \widehat{K} ]\widehat{K}^\alpha +  \widehat{K}  [T_g, \widehat{K}^\alpha ].
$$
The first term on the right-hand side of the previous equation clearly has the required structure in view of the first order commutator formulas. For the second, we must compute, for the main terms,
\begin{equation} \label{eq:kkacom}
\frac{1}{t^q}\widehat{K} P(C)^{k, r_Z,s_X} w \cdot \partial_{t,x} K^\beta(h) \widehat{K}^\delta
\end{equation}
where the indices are as in the statement of the lemma and similarly, for cubic and frame terms. We only treat the main terms below, since the computation for the cubic and frame terms are similar.
We simply distribute the $\widehat{K}$ in \eqref{eq:kkacom}. When it hits the $C$ coefficients or the final $\widehat{K}^\delta$, then the resulting terms are clearly of the required form.
When it hits $\partial_{t,x} K^\beta(h)$ or $\frac{1}{t^q}$, it has the required form if $\widehat{K}$ is a translation (preserving the number of $C$ coefficients) and if $\widehat{K}$ is a modified vector field, we simply write
$$
\widehat{K} = \widehat{Z} + C\cdot X
$$
and apply each terms on the right-hand side to $\partial_{t,x} K^\beta(h)$ or $\frac{1}{t^q}$. The terms coming from $C\cdot X \partial_{t,x} K^\beta(h)$ or $C\cdot X \frac{1}{t^q}$ then increases $k$ by $1$ and $\beta_X$ or $q$ by $1$, as required.
\end{proof}

We start the analysis of the error terms with a proposition that describes how integrals of error terms are estimated when the metric perturbation cannot be estimated pointwise.
\begin{proposition} \label{prop:eesfp}
Let $\alpha$ be a multi-index of length $|\alpha|=N_0$.

Consider an error term of the form
$$
Q:=P(C)^{k, r_Z,s_X} w \cdot \partial_{t,x} K^\beta(h) \widehat{K}^\mu f
$$
where $r_Z+s_X+|\beta|+|\mu| \le |\alpha|+1$, $|\beta| \le |\alpha|$, $1 \le |\mu| \le |\alpha|$, $r_Z+s_X \le |\alpha|-1$ and $k$ satisfies either $\mathrm{C1}$ or $\mathrm{C2}$.


Assume that we have the pointwise bounds

\begin{eqnarray*}
\int_v \left| \widehat{K}^\gamma (f)\right| w^0 dv &\lesssim& \epsilon \rho^{M_{N_0} \delta} \frac{1}{t^3}, \quad |\gamma| \le |\alpha|, \\
| \widehat{K}^\rho (C) | &\lesssim& \epsilon^{1/2} u^{1/2}\rho^{M_{N_0}\delta}, \quad |\rho| \le |\alpha|-1,
\end{eqnarray*}
for some constant $M_{N_0}$ depending only on $N_0$.

Then,
\begin{itemize}
\item if $k$ satisfies $\mathrm{C1}$, we have
$$\int_{H_\rho} \int_v |Q| w^0 d\mu_v d\mu_\rho \lesssim \epsilon^{3/2} \rho^{-3/2+M'_{N_0}\delta}.$$
\item if $k$ satisfies $\mathrm{C2}$, we have
$$\int_{H_\rho} \int_v |Q| w^0 d\mu_v d\mu_\rho \lesssim \epsilon^{2} \rho^{-1+M'_{N_0}\delta},$$ 
where $M'_{N_0}$ depends only on $N_0$.
\end{itemize}
\end{proposition}
\begin{proof}

Assume first that $k$ verifies $\mathrm{C1}$.
Then, all the growth in the $C$ coefficients can be absorbed by the extra $u$ or $t$ decay coming from the $\mathrm{C1}$ condition (cf~Appendix \ref{se:dect}) and as a result, we have the estimate

\begin{align*}
\int_{H_\rho} \int_v & | P(C)^{k, r_Z,s_X} w \cdot \partial_{t,x} K^\beta(h) \widehat{K}^\delta w^0\vert  d\mu_v \\
&\lesssim \rho^{M_{N_0} \delta k }\epsilon^{k/2+1} \left( \int_{H_\rho} \frac{\rho}{t} |\partial K^{\beta'} h|^2 d\mu_{H_\rho} \right)^{1/2} \left( \int_{H_\rho} \frac{t}{\rho} \frac{\rho^{2M_{N_0} \delta}} {t^6} d\mu_{H_\rho} \right)^{1/2} \\
&\lesssim \rho^{M_{N_0} \delta k }\epsilon^{k/2+3/2}  \rho^{( M_{N_0} +1/2) \delta-3/2},
\end{align*}
using the bootstrap assumptions to bound the integral involving the metric.

If instead $k$ verifies $\mathrm{C2}$, we have one more $\epsilon^{1/2} u^{1/2} \rho^\delta$ that cannot be absorbed leading to the extra growth.
\end{proof}
\begin{remark}
By the Klainerman-Sobolev inequality, see Section \ref{se:KSf}, we will be able to prove decay estimates for the velocity averages of $\widehat{K}^\mu(f)$ for any $|\mu| \le N-3$. However, the above proposition cannot be directly used to close the energy estimate up to $N-3$. This is because the Klainerman-Sobolev estimates loses in powers of $w^0$. Thus, we must first prove energy estimates for $(w^0)^2 \widehat{K}^\mu(f)$ in order to prove decay estimates for $\int_v w^0|\widehat{K}^\mu(f)| dv$. This is one of the purpose of the weighted norm \eqref{def:enq}.
\end{remark}

\subsection{Higher order commutators and the null structure}

We now investigate the null structure of the higher order commutators.

First, for the cubic terms or the frame terms, no null structure will be required as they will naturally have enough decay. Thus, we focus only the main terms \eqref{eq:mthc}.

\begin{definition}
We say that an expression
$$
P(C)^{k, r_Z,s_X} w \cdot \partial_{t,x} K^\beta(h) \widehat{K}^\mu
$$
has \emph{the null structure} provided it has one of the following forms
\begin{itemize}
\item $P(C)^{k, r_Z,s_X} w\cdot \delu_{x^i} K^\beta(h) \widehat{K}^\mu,$
\item $P(C)^{k, r_Z,s_X} \vu_a  \partial_{t,x} K^\beta(h) \widehat{K}^\mu,$
\item $P(C)^{k, r_Z,s_X} w\cdot  \partial_{t,x} K^\beta(\hu^{00}) \widehat{K}^\mu,$
\item $P(C)^{k, r_Z,s_X} w\cdot \frac{1}{t}\partial_{t,x} K^\beta(h) \widehat{K}^\mu.$
\end{itemize}
\end{definition}

In view of this definition, note that
\begin{itemize}
\item All the main terms in the commutator $[T_g, X_i]$ have the null structure.
\item The only terms in the first order commutator formulas that does not have the null structure (and which are neither cubic or frame terms) are the terms
\begin{itemize}
\item $w\cdot \partial h \cdot X_a$ coming from the commutator $[T_g, \partial_t]$.
\item $C\cdot w \partial h \cdot X_a$ coming from $C^\alpha [T_g, X_\alpha]$  in the commutator $[T_g, Y]$.
\end{itemize}
\item The last term in the above list, namely, $P(C)^{k, r_Z,s_X} w\cdot \frac{1}{t}\partial_{t,x} K^\beta(h) \widehat{K}^\mu$ does not occur in the first order commutator formula, but will arise because of the commutator $[\partial_t, \delu_{x^i}]=- \frac{x^i}{t^2}$.

\end{itemize}

One has easily

\begin{proposition}
Consider an error term $Q_\alpha$ of the form
\begin{eqnarray}
Q_\alpha= \frac{1}{t^q}P(C)^{k, r_Z,s_X} w \cdot \partial_{t,x} K^\beta(h) \widehat{K}^\sigma
\end{eqnarray}
where $q+r_Z+s_X+|\beta|+|\sigma| \le |\alpha|+1$, $|\beta| \le |\alpha|$, $|\sigma| \le |\alpha|$, $r_Z+s_X \le |\alpha|-1$ and $k$ verifying $\mathrm{C1}$ or $\mathrm{C2}$. 

Then, for $\widehat{K}=X_i$, $\partial_t$ or $Y$,
\begin{enumerate}
\item  $\widehat{K}\left( Q_\alpha \right):=Q_{\alpha'}$ can be written as
\begin{eqnarray*}
Q_{\alpha'}=\frac{1}{t^{q'}}P(C)^{k', r'_Z,s'_X} w \cdot \partial_{t,x} K^{\beta'}(h) \widehat{K}^{\sigma'}
\end{eqnarray*}
where $q'+r'_Z+s'_X+|\beta'|+|\sigma'| \le |\alpha'|$, $|\beta'| \le |\alpha'|$, $|\sigma'| \le |\alpha'|$, $r'_Z+s'_X \le |\alpha'|-1$,  
$k'$ verifying $\mathrm{C1}$ or $\mathrm{C2}$, with $|\alpha'|=|\alpha|+1$ and $\alpha'_Z=\alpha_Z+1$ if $\widehat{K}=Y$ and $\alpha'_Z=\alpha_Z$ otherwise.
\item Moreover, if $Q_\alpha$ has the null structure, then $Q_{\alpha'}$ has the null structure.
\end{enumerate}


\end{proposition}

We now prove bounds on the higher order energy up to $N-2$ (for all energies where we can bound pointwise all the $\widehat{K}^\rho(C)$ and the $K^\beta (h)$ appearing in the error terms.)

\begin{proposition}\label{prop:eesvfl}  Let $N_0\le N-2$.
Assume the pointwise bounds on the $C$ coefficients
$$
| \widehat{K}^\rho (C) | \lesssim \epsilon^{1/2} u^{1/2}\rho^{M_{N_0}\delta}, \quad |\rho| \le N_0-1.
$$

Then, we have the estimate
$$
E_{N_0} [f] \lesssim \epsilon \rho^{M'_{N_0}\delta},
$$
for some constant $M'_{N_0}$ depending only on $N_0$.
\end{proposition}

\begin{proof}
We have already proven the proposition for $N_0=0,1$. To close the estimate for $N_0=1$, we exploited a certain hierarchy in the commutators (the $X_i$ had good commutor properties, $\partial_t$ generated borderline terms, etc..)
The strategy is to exploit the same hierarchy, doing first the commutation by $X_i$, then $\partial_t$ and then $Y$, eventually.
Let $N_0 \ge 1$ and assume that the proposition holds for $N_0-1$. For simplicity in the exposition, we make the additional bootstrap assumption
$$
E_{N_0} [f] \lesssim \epsilon \rho^L,
$$
where $\delta << L << 1$ and in particular $M_{N_0}\delta<L$ for some constant $M_{N_0}$ being some positive constant.

For simplicity in the exposition, we will write $E[ X_i^{N_0} f ]$ to denote the sum of all energies $E[ X^\alpha f ]$, with $|\alpha|=N_0$ and where the $X^\alpha$ is composed only of $X_i$ vector fields. Similarly, when we write $E[ X_i^{N_0-1} \partial_t f ]$ (respectively $E[ X_i^{N_0-1} Y f ]$) below, we actually mean the sum over all possible combinations\footnote{Recall that $\partial_t$ commutes with $X_i$, while the $Y$ vector fields enjoy good commutation properties with $\partial_t$ or $X_i$, see Section \ref{se:camvf}.} containing one $\partial_t$ (respetively $Y$) and $N_0-1$ $X_i$ vector fields.

First, commute $N_0$ times by $X_i$ vector fields. Each error term is of the form
\begin{eqnarray*}
P(C)^{k, 0 ,s_X} w \cdot \partial_{t,x} K^\beta(h) \widehat{K}^\sigma
\end{eqnarray*}
where $s_X+|\beta|+|\sigma| \le N_0+1$, $|\beta| \le |\alpha|$, $|\sigma| \le N_0$, $s_X \le N_0-1$ and $k$ verifying $\mathrm{C1}$ (and even $k \le 1$). Moreover, from the previous proposition and the first order commutator for $X_i$, each term in the $X_i^{N_0}$ commutator formula verifies the null condition. It follows that each error term is integrable and thus one obtains that
\begin{eqnarray*}
E[ X_i^{N_0} f ](\rho) &\lesssim& E[ X_i^{N_0} f ](2)
+ 
\int_2^\rho \epsilon^{1/2} \rho'^{-3/2+\delta D} E[ \widehat{K}^\delta f ](\rho') d \rho' \\
&\lesssim& \epsilon
+ 
\int_2^\rho \epsilon^{3/2} \rho'^{-3/2+\delta D+L} d \rho' \\
&\lesssim& \epsilon, 
\end{eqnarray*}
where $D$ is some constant.

Then, we commute once with $\partial_t$ and $N_0-1$ times by $X_i$. Again, no error term will have a number of $C$ coefficients satisfying $\mathrm{C2}$. Using that for $\ell\leq N_0-1$
$$
[T_g , X_i^{N_0-1-\ell} \partial_t X_i^{\ell}]= X_i^{N_0-1-\ell}[T_g, \partial_t] X_i^{\ell}+ X_i^{N_0-1-\ell} \partial_t [T_g, X_i^{\ell}]+ [T_g, X_i^{N_0-1-\ell}]\partial_t X_i^\ell
$$
and that only $[T_g, \partial_t]$ generates a term which does not satisfies the null condition, we obtain
\begin{eqnarray*}
E[X_i^{N_0-1} \partial_t](\rho) &\lesssim& E[ X_i^{N_0-1} \partial_t  f ](2)+ \int_2^\rho \epsilon^{1/2} \rho'^{-3/2+\delta D} E[ \widehat{K} X_i^{N_0-2} \partial_t f] d \rho'  \\
&&\hbox{}+\int_2^\rho \epsilon^{1/2} \rho'^{-1+\delta/2} E[ X_i^{N_0} f] d \rho'+ \epsilon \rho^{M_{N_0-1}\delta}, \\
&\lesssim& \epsilon + \epsilon^{3/2}\delta^{-1} \rho^{M_{N_0-1}\delta+ \delta/2},  \\
&\lesssim& \epsilon \rho^{M_{N_0-1}\delta+ \delta/2}.
\end{eqnarray*}

We then commute once with $Y$ and $N_0-1$ times by $X_i$, using similarly that

\begin{equation} \label{eq:yxnm1}
\alg{
&[T_g , X_i^{N_0-1-\ell} Y X_i^{\ell}]\\
&= X_i^{N_0-1-\ell}[T_g, Y ] X_i^{\ell}+ X_i^{N_0-1-\ell} Y [T_g, X_i^{\ell}]+ [T_g, X_i^{N_0-1-\ell}]Y X_i^\ell
}
\end{equation}

The terms coming from $X_i^{N_0-1-\ell} Y [T_g, X_i^{\ell}]$ or $[T_g, X_i^{N_0-1-\ell}]Y X_i^\ell$ all verifies the null condition and have a number of $C$ coefficients satisfying $\mathrm{C1}$, so they are integrable.  Using the first order commutator formula for $Y$, the only term coming from the first term on the right-hand side of \eqref{eq:yxnm1} not satisfying the null condition is of the form
$$X_i^{N_0-1-\ell}\left( w\cdot \partial h C \cdot X_i^{\ell}f\right)$$
Note also that this term corresponds to $k=1$ verifying $\mathrm{C2}$.
We can distribute the $X_i^{N_0-1-\ell}$ on each term.
We estimate below the case when $X_i^{N_0-1-\ell}$ hits $X_i^{\ell}f$, the rest can be estimated using the estimate on $E_{N_0-1}[f]$. We have
\begin{eqnarray*}
|w\cdot \partial h C \cdot X_i^{N_0}f| &\lesssim& w_0\frac{t}{\rho} \frac{\rho}{t} \rho^{\delta/2} \epsilon^{1/2} \frac{1}{ t u^{1/2}} \epsilon^{1/2}  u^{1/2}\rho^{\delta/2} |X_i^{N_0}f| \\
&\lesssim&\epsilon \rho^{\delta-1} w_0 \frac{\rho}{t} |X_i^{N_0}f|.
\end{eqnarray*}

Since we already control the energy of $E[X_i^Nf]$, this term will lead to a small growth, as in Lemma \ref{es:Y1e}.
The other error terms with $k$ not satisfying $\mathrm{C2}$ can be estimated similarly by $ \rho^{\delta-1} \epsilon w_0 \frac{\rho}{t} |X^{N_0}f|$, where $X=X_i, \partial_t$ or using the bound on $E_{N_0-1}[f]$

We thus obtain
\begin{eqnarray*}
E[X_i^{N_0-1} Y](\rho) &\lesssim& E[ X_i^{N_0-1} Y  f ](2)+ \int_2^\rho \epsilon^{1/2} \rho'^{-3/2+\delta D} E[YX_i^{N_0-2}  \widehat{K} ] d \rho'  \\
&&\hbox{}+
\int_2^\rho \epsilon^{1/2} \rho'^{-1+\delta} E[X_i^{N_0-1}X ] d \rho'+ \epsilon \rho^{M_{N_0-1}\delta},
\end{eqnarray*}
where $X_i^{N_0-1}X$ denotes any differential operator composed of at least $N_0-1$ $X_i$ vector fields and one vector field among the $X_i$ or $\partial_t$ vector fields.
This gives
\begin{eqnarray*}
E[X_i^{N_0-1} Y](\rho) &\lesssim & \epsilon  + \epsilon^{3/2} \delta^{-1}\left( \rho^{M_{N_0-1}\delta}+ \rho^{\delta}\right) \\
&\lesssim & \epsilon \left( \rho^{M_{N_0-1}\delta}+ \rho^{\delta}\right).
\end{eqnarray*}

%
We now consider commuting with $X_i^{N_0-2} \partial_t^2$. We use that
\begin{eqnarray}
[T_g, X_i^{N_0-2}\partial_t^2]&=& [T_g , X_i^{N_0-2}] \partial_t^2+ X_i^{N_0-2} [T_g, \partial_t^2] \nonumber \\
&=& [T_g , X_i^{N_0-2}] \partial_t^2+ X_i^{N_0-2} [T_g, \partial_t] \partial_t + X_i^{N_0-2} \partial_t [T_g, \partial_t] \label{com:2dt}.
\end{eqnarray}
Again, all possible error terms have $k$ satisfying $\mathrm{C1}$.

The first term on the right-hand side contains only terms having the null structure, since there are no commutators with $\partial_t$ or a modified vector field $Y$. The terms not satisfying the null structure coming from the second term on the right-hand side arise all from expressions of the form
$$
 X_i^{N_0-2} \left(  w\cdot \partial h X_a \partial_t f\right).
$$
Distributing the $X_i$s and using that $X_i=\frac{Z_i}{t}+ \frac{\vu_i}{w^0} \partial_t$, the only terms coming from the above expression not satisfying the null condition are then of the form
$$
 w\cdot \partial h X_i^{N_0-1} \partial_t f,
$$
which can be estimated by the energy $E[ X_i^{N_0-1} \partial_t ]$.


 The terms not satisfying the null structure coming from the last term on the right-hand side of \eqref{com:2dt} arise all from expressions of the form 
$$
 X_i^{N_0-2} \partial_t \left(  w\cdot \partial h X_a f\right).
$$
Distributing the operators in $X_i^{N_0-2} \partial_t$ and using that $X_i=\frac{Z_i}{t}+ \frac{\vu_i}{w^0} \partial_t$, the only terms coming from the above expression not satisfying the null condition are then of the form
$$
  w\cdot \partial h X_i^{N_0-1} \partial_t f,
$$
which we already encountered above and
$$
  w\cdot \partial_t (\partial h) X_i^{N_0-1}  f,
$$
which are controlled by the $E_{N_0-1}[f]$ energy.
This leads to the energy estimate

\begin{eqnarray*}
E[ X_i^{N_0-2} \partial_t^2](\rho) &\lesssim& E[ X_i^{N_0-2} \partial_t^2](2)+\epsilon^{1/2}\int_2^\rho (\rho')^{-1+\delta/2} E[X_i^{N_0-1}\partial_t ](\rho') d\rho' \\
&&\hbox{}+ \epsilon^{1/2}\int_2^\rho (\rho')^{-3/2+\delta} E[X_i^{N_0-2} \widehat{K} \partial_t ](\rho') d\rho' + \epsilon \rho^{ M_{N_0 -1} \delta}\\
&\lesssim & \epsilon \rho^{ \left(M_{N_0 -1}+ 3/2\right) \delta}.
\end{eqnarray*}

Similarly, we can then commute with $X_i^{N_0-2} \partial_t Y$, $X_i^{N_0-2} Y^2$. We then consider commuting by $X_i^{N_0-3} \partial_t^3$, ... Iterating, we obtain the statement of the lemma. At each iteration, we pick up some $\rho^{\delta/2}$ losses but the number of such losses only depends on the number of iterations and therefore only on $N_0$. We obtain that $E_N[f] \lesssim  \epsilon \rho^{M_{N_0} \delta}$, which improves the additional bootstrap assumption and ends the proof of the proposition.
\end{proof}

We then consider energy estimates up to order $N-2$ with additional $v$ weights.

%

\begin{proposition}
Under the same assumptions as in Proposition \ref{prop:eesvfl}, we have,
$$
E_{N_0, q+2} [f] \lesssim \epsilon \rho^{M_{N_0}\delta},
$$
where $E_{N_0, q+2}$ is the weighted energy norm defined in \eqref{def:enq}.
\end{proposition}
\begin{proof}
The proof is similar to the proof of Proposition \ref{prop:eesvfl}. One simply needs to use the additional formula for multiplication by powers of $v$ given in Section \ref{se:multivz}, which only adds additional integrable terms, using the bootstrap assumptions.
\end{proof}

\begin{corollary}
Under the same assumptions as in Proposition \ref{prop:eesvfl}, we have the pointwise estimates, for any $|\alpha| \le N_0-3$,
$$
\int_v w^0 |\widehat{K}^\alpha (f)| dv \le \epsilon \rho^{M'_{N_0}\delta} t^{-3}.
$$
\end{corollary}
\begin{proof}
This follows from the Klainerman-Sobolev inequality with modified vector fields of Section \ref{se:KSf} and the above bounds on $E_{N_0, q}$, with $q \ge 2$ (the Klainerman-Sobolev inequality loses two powers of $v$).
\end{proof}

Before considering the $N$th order energy estimate for $f$, it will be useful to prove estimates for products of type\footnote{For $|\gamma|\le N-3$, we will have pointwise bounds on $\widehat{K}^\gamma (C)$ depending only on the bootstrap assumptions, and thus the lemma below does not bring any new information in that case.} $\widehat{K}^\gamma (C) \widehat{K}^\alpha (f)$, for $|\gamma| \le N-1$ and $|\alpha| \le N-3$.
\begin{proposition} \label{prop:pee}
Consider a product of the type $\widehat{K}^\gamma (C) \widehat{K}^\alpha (f)$, for $|\gamma| \le N-1$ and $|\alpha| \le N-3$.
Assume that
$$
| \widehat{K}^\mu (C) | \lesssim \epsilon^{1/2} u^{1/2}\rho^{M_{N-3}\delta}, \quad |\mu| \le N-3.
$$
as well as the initial data assumption
\begin{eqnarray} \label{ass:dataC}
E\left[ u^{-1/2} \widehat{K}^{\gamma'} (C) \widehat{K}^{\alpha'} (f) \right](2)] \le \epsilon^{3/2},\quad  |\gamma'| \le |\gamma|, \quad |\alpha'|\le |\alpha|.
\end{eqnarray}
Then, we have

$$
E\left[ (1+u)^{-1/2} \widehat{K}^\gamma (C) \widehat{K}^\alpha (f) \right](\rho) \lesssim \epsilon^{3/2} \rho^{M_N \delta}.
$$

\end{proposition}
\begin{remark} \label{rem:dataC}
Note that since $C(\rho=2)=0$, if $\underline{D}$ is a differential operator tangential to $\rho=2$, then $\underline{D}(C)=0$. In particular, we have $Z_i(C)=\partial_v(C)=\widehat{Z}_i(C)=\delu_{x^i}(C)=0$ initially. From the equation for $C$, $T_g(C)=-F_{ZB}$, one can then compute the normal derivative to $\rho=2$ in terms of the initial data and verify that \eqref{ass:dataC} actually holds initially.
\end{remark}

\begin{proof}
Again, we will do a small bootstrap argument, assuming that

$$E\left[ u^{-1/2} \widehat{K}^\gamma (C) \widehat{K}^\alpha (f) \right](\rho) \lesssim \epsilon^{3/2} \rho^{L}, $$
for some $M_N \delta<< L$ and indices $\gamma$  and $\alpha$ as in the statement of the proposition.

For $|\gamma'|\le N-3$, we have access to pointwise estimates on all $\widehat{K}^{\gamma'}(C)$ coefficients so that from the previous proposition 

$$
E\left[ (1+u)^{-1/2} \widehat{K}^{\gamma'} (C) \widehat{K}^\alpha (f) \right](\rho) \lesssim \epsilon^{1/2} \rho^{M''_{N-3} \delta},
$$
where $M''_{N-3}$ is some constant.

We now consider $N-3 \le N'\le N-2$ and assume that the proposition holds for all multi-indices $|\gamma'| \le N'$, so that
$$
E\left[ (1+u)^{-1/2} \widehat{K}^{\gamma'} (C) \widehat{K}^\alpha (f) \right](\rho) \lesssim \epsilon^{3/2} \rho^{M_N \delta}.
$$

Let $\gamma$ be a multi-index of length $|\gamma|=|\gamma'|+1$.

We start by computing
\begin{eqnarray*}
T_g((1+u)^{-1/2} \widehat{K}^\gamma (C) \widehat{K}^\alpha (f) )&=& (1+u)^{-1/2} \widehat{K}^\gamma ( T_g(C))  \widehat{K}^\alpha (f) \\
&&\hbox{}+ (1+u)^{-1/2} [T_g, \widehat{K}^\gamma](C). \widehat{K}^\alpha (f)  \\
&&\hbox{}+ (1+u)^{-1/2} \widehat{K}^\gamma (C) [T_g, \widehat{K}^\alpha] f \\
&&\hbox{} - \frac{1}{2}T_g(u) (1+u)^{-3/2} \widehat{K}^\gamma (C) \widehat{K}^\alpha (f) \\
&=& I_1+I_2+I_3+I_4,
\end{eqnarray*}
where
\begin{eqnarray*}
I_1&=& (1+u)^{-1/2} \widehat{K}^\gamma ( T_g(C))  \widehat{K}^\alpha (f), \\
I_2&=& (1+u)^{-1/2} [T_g, \widehat{K}^\gamma](C)\cdot \widehat{K}^\alpha (f),  \\
I_3&=& (1+u)^{-1/2} \widehat{K}^\gamma (C) [T_g, \widehat{K}^\alpha] (f), \\
I_4&=&- \frac{1}{2}T_g(u) (1+u)^{-3/2} \widehat{K}^\gamma (C) \widehat{K}^\alpha (f).
\end{eqnarray*}
$I_4$ can be estimated as in the proof of the pointwise estimate for $C$ (cf~Lemma \ref{es:ccu}).

For $I_1$, we have
\begin{eqnarray*}
I_1&=&(1+u)^{-1/2} \widehat{K}^\gamma ( T_g(C))  \widehat{K}^\alpha (f) \\
&=& -(1+u)^{-1/2} \widehat{K}^\gamma (F_{ZB} ) \widehat{K}^\alpha (f),
\end{eqnarray*}
using that $T_g(C)= -F_{ZB}$.

Consider a source term in $F_{ZB}$ of the form
$$
w\cdot K^\mu (h),  \quad |\mu| \le 1.
$$

Then, using that $Y (K^\mu (h))=Z K^\mu(h)+ C\cdot X K^\mu (h)$,  we have, for $|\mu| \le 1$,

$$
 \widehat{K}^\gamma (w\cdot K^\mu (h)) = P^{k, r_Z, s_X}(C)w\cdot K^\sigma (h),
$$
where $k \le \sigma_X$, $|\sigma|+r_Z+s_X \le |\gamma|+1$ and $r_Z+s_x \le |\gamma|-1$.

Note that at most one $C$ coefficient in $P^{k, r_Z, s_X}(C)$ can have a high-number of derivatives and that all other $C$ coefficients can be estimated pointwise. Moreover, from the condition $k \le \sigma_X$, for each $C$ coefficient there is one $1/u$ decay factor. For instance, the overall contribution to $I_1$ of the terms with $\sigma_X=1$ can be estimated as

\begin{eqnarray*}
(1+u)^{-1/2} \left| w\cdot \widehat{K}^{\sigma'}(C) \partial K(h) \widehat{K}^\alpha (f)\right| &\lesssim& w^0 \frac{\rho}{t}\frac{t}{\rho} \epsilon^{1/2} \frac{\rho^{\delta/2}}{t(1+u)}\left| \widehat{K}^{\sigma'}(C) \widehat{K}^\alpha (f)\right| \\
&\lesssim& w^0 \frac{\rho}{t}  \frac{\epsilon^{1/2}\rho^{\delta/2}}{\rho}\cdot\\
&&\quad (1+u)^{-1/2} \left| \widehat{K}^{\sigma'}(C) \widehat{K}^\alpha (f)\right|,
\end{eqnarray*}
with $|\sigma'|=|\gamma|$, so that we can estimate this term (and similarly any term with $k \ge 1$) using the induction hypothesis.
On the other hand, the terms with $k=0$ in $I_1$ can be estimated as
\begin{eqnarray*}
(1+u)^{-1/2} w\cdot\left| K^\gamma (h) \widehat{K}^\alpha (f)\right| &\lesssim& (1+u)^{-1/2} w^0\frac{\rho}{t} \frac{t}{\rho} \epsilon^{1/2} \frac{(1+u)^{1/2}\rho^{\delta/2}}{t}\left| \widehat{K}^\alpha (f)\right| \\
&\lesssim& w^0\frac{\rho}{t} \epsilon^{1/2} \frac{\rho^{\delta/2}}{\rho}\left| \widehat{K}^\alpha (f)\right|,
\end{eqnarray*}
and since there are no more $C$ coefficients on the right-hand side, we can estimate the overall contribution of this term using the energy estimate for $\widehat{K}^\alpha (f)$. All the other source terms in $F_{ZB}$ can be estimated similarly.



Thus, it remains to estimate $I_2$ and $I_3$. For those, we first assume that $\widehat{K}^\gamma(C) \widehat{K}^\alpha (f)= X_i^{|\gamma|}(C) X_i^{|\alpha|}(f)$, i.e.~there are only $X_i$ vector fields. Then each error term in the right-hand side satisfies the null structure, so that their contributions are integrable. Then, consider the case with only one $\partial_t$ vector field. As in the previous proposition, the only term not satisfying the null condition comes from $[T_g, \partial_t]$ and is controlled by the energy estimate with only $X_i$ vector fields. The remainder of the proof follows as in the previous proposition, exploiting the same hierarchy.
\end{proof}
Finally, we prove the higher order energy estimates for the Vlasov field up to order $N$.
\begin{proposition} \label{prop:eevfn}
Assume that the following pointwise estimates hold.

\begin{eqnarray*}
\int_v |\widehat{K}^\gamma (f)| w^0 dv &\lesssim& \epsilon \rho^{M_{N-4} \delta} \frac{1}{t^3}, \quad |\gamma| \le N-5, \\
| \widehat{K}^\mu (C) | &\lesssim& \epsilon^{1/2} u^{1/2}\rho^{M_{N-2}\delta/2}, \quad |\mu| \le N-3.
\end{eqnarray*}

Then, we have
$$
E_N[ f ] \lesssim \epsilon^{1/2} \rho^{M_N \delta}.
$$
\end{proposition}

\begin{proof}

Recall the general structure of the main terms appearing in the commutator $[T_g, \widehat{K}^\alpha](f).$
\eq{
\frac{1}{t^q}P(C)^{k, r_Z,s_X} w \cdot \partial_{t,x} K^\beta(h) \widehat{K}^\delta(f)
}
where $r_Z+s_X+|\beta|+|\delta| \le |\alpha|+1$, $q+|\beta| \le |\alpha|$, $|\delta| \le |\alpha|$, $r_Z+s_X \le |\alpha|-1$  
and $k$ satisfying either $\mathrm{C1}$ or $\mathrm{C2}$. 

As before, we first consider the energy estimate for $E[ X_i^N f]$. Then all error terms above have the null structure property.
\begin{itemize}
\item If there is one $Y^\gamma(C)$ coefficients with $|\gamma| \ge N-4$, we can use the estimate of the previous section, together with pointwise estimates on the term containing the $h$ coefficient. These terms are then all integrable.
\item Thus, we consider only the case where the $C$ coefficients can be estimated pointwise. If the wave term $\partial_{t,x} K^\beta(h)$ can be estimated pointwise (which occurs if $|\beta| \le N-2$), then we can estimate the error term as before. Finally, if $\beta \ge N-2$, then we have access to pointwise estimates on the velocity averages of $f$, and can just apply Proposition \ref{prop:eesfp}.
\end{itemize}

We can then follow the same hierarchy as before. Each term either is integrable, a borderline term depending on the previous energy estimate, or a term such we can apply Proposition \ref{prop:eesfp}.
%
%
\end{proof}
%
%
%
Let us also state the following proposition, which we will use to establish decay estimates for the high derivatives of the energy momentum tensor.
\begin{proposition} \label{prop:p2ee}
\begin{enumerate}
\item Consider a product of the type $| \widehat{K}^\gamma (C)|^2 \widehat{K}^\alpha (f)$, for $N-1 \ge |\gamma|$ and $|\alpha| \le N-3$. Then, we have
$$
E\left[ (1+u)^{-1} |\widehat{K}^\gamma (C)|^2 \widehat{K}^\alpha (f) \right](\rho) \lesssim \epsilon^{1/2} \rho^{M_N \delta}.
$$
\item
Consider a product of the type $|\widehat{K}^\gamma (XC)|^2 \widehat{K}^\alpha (f)$, for $N-2 \ge |\gamma|$ and $|\alpha| \le N-3$. Then, we have
$$
E\left[  | \widehat{K}^\gamma (XC)|^2 \widehat{K}^\alpha (f) \right](\rho) \lesssim \epsilon^{1/2} \rho^{M_N \delta}.
$$
\end{enumerate}
\end{proposition}
The proof is similar to that of Proposition \ref{prop:pee} and therefore omitted.

\section{Higher order pointwise estimates for the $C$ coefficients} \label{se:hocc}
Recall that the $C$ coefficients satisfy an equation of the form $T_g(C)= F$, where $F$ is a linear combination of the terms
\begin{eqnarray*}
&&w\cdot K^\sigma (h),  \quad |\sigma| \le 1, \\
&& w\cdot K^\sigma (h)h , \quad |\sigma| \le 1, \\
&&\widehat{K}^\sigma (w_\alpha) w_\beta \delu_{x^k} K^\gamma (h^{\alpha \beta}) \frac{t}{w^0}  , \quad  |\sigma|+|\gamma| \le 1, \\
&&\widehat{K}^\sigma(w_\alpha) w_\beta \delu_{x^k} K^\gamma(h^{\alpha \beta}) \frac{ t \vu_a}{(w^0)^2} , \quad  |\sigma|+|\gamma| \le 1, \\
&&\widehat{K}^\sigma (w_\alpha) w_\beta \del_{t} K^\gamma(h^{\alpha \beta})\frac{|x|^2-t^2}{t} \frac{1}{w^0}, \quad |\sigma|+|\gamma| \le 1, \\
&&\widehat{K}^\sigma (w_\alpha) w_\beta \partial_{t,x} K^\gamma (h^{\alpha \beta}) \frac{t \vu_i }{(w^0)^2}  , \quad  |\sigma|+|\gamma| \le 1, \\
&&\widehat{K}^\sigma (w_\alpha) w_\beta \partial_{t,x} K^\gamma(h^{\alpha \beta}) \frac{ t \vu_i \vu_a}{(w^0)^3}  , \quad  |\sigma|+|\gamma| \le 1, \\
&&w\cdot K^\sigma(h)\cdot \partial_{x} K^\gamma (h)\cdot t   \quad |\sigma|+|\gamma| \le 1.
\end{eqnarray*}
Using that $t \delu_{x^k}=Z_k$, that $(t-|x|) \partial_t$ can be written as a linear combination of the $Z$ vector fields and using the decomposition on the semi-hyperboloidal frame together with
$t\cdot \partial(\Phi\cdot\Phi) \in \mathcal{F}_x$, it follows that we can simplify the above list to
\begin{eqnarray*}
1. \quad &&w\cdot K^\sigma (h),  \quad |\sigma| \le 2, \\
2. \quad&& w\cdot K^\sigma (h)h, \quad |\sigma| \le 1, \\
3. \quad&&\partial_{t,x} K^\gamma (\hu) \frac{t \vu_i \vu_j }{w^0}  , \quad |\gamma| \le 1, \\
4. \quad&&\partial_{t,x} K^\gamma (\hu^{00}) t\cdot w^0  , \quad  |\gamma| \le 1, \\
5. \quad&& w\cdot K^\sigma(h)\cdot \partial_{x} K^\gamma (h)\cdot t,   \quad |\sigma|+|\gamma| \le 1.
\end{eqnarray*}

One has easily

\begin{lemma} \label{lem:hoccs}
Let $C$ be one of the $C$ coefficients and $F$ be such that
$$
T_g(C)=F.
$$
Let $|\alpha| \le N-1$ be a multi-index.
Then, $\widehat{K}^\alpha(F)$ can be written as a linear combination of
\begin{enumerate}
\item $w\cdot P(C)^{k,r_Z,s_X} K^\sigma (h)$, with the range of indices
$$ |\sigma|+r_Z+s_X \le 2+ |\alpha|, \quad r_Z+s_X \le |\alpha|-1, \quad k \le \sigma_X.$$ 
\item $w\cdot P(C)^{k,r_Z,s_X} K^\beta(h) K^\sigma (h)$, with the range of indices $$|\beta|+|\delta|+r_Z+s_X \le 1+ |\alpha|, \quad r_Z+s_X \le |\alpha|-1, \quad k \le \sigma_X+\beta_X.$$
\item $P(C)^{k,r_Z,s_X} \partial_{t,x} K^\sigma (h)\frac{t \vu_i \vu_j }{w^0}$, with the range of indices
$$ |\sigma|+r_Z+s_x \le 1+ |\alpha|, \quad r_Z+s_X \le |\alpha|-1, \quad k \le \sigma_X.$$ 
\item $P(C)^{k,r_Z,s_X} \partial K^\sigma (\hu^{00}) t. w^0 $, with the range of indices
$$ |\sigma|+r_Z+s_X \le 1+ |\alpha|, \quad r_Z+s_X \le |\alpha|-1, \quad k \le \sigma_X.$$
\item $w\cdot t\cdot P(C)^{k,r_Z,s_X} K^\beta(h) \partial_{t,x} K^\sigma (h)$, with the range of indices $$|\beta|+|\sigma|+r_Z+s_x \le 1+ |\alpha|, \quad r_Z+s_X \le |\alpha|-1, \quad k \le \sigma_X+\beta_X.$$
\end{enumerate}
\end{lemma}

Consider now the equation satisfied by $\widehat{K}^\alpha (C).$

We have
\begin{equation} \label{eq:higCeq}
T_g (\widehat{K}^\alpha(C) )= [ T_g, \widehat{K}^\alpha](C)+ \widehat{K}^\alpha(F).
\end{equation}
Moreover, from the equation $T_g(C)=F$, $C(\rho=2)=0$, one infers easily, cf.~Remark \ref{rem:dataC}, that
$$|\widehat{K}^\alpha(C)(\rho=2)| \lesssim \epsilon^{1/2}.$$

We can write the solution to \eqref{eq:higCeq} as
$$\widehat{K}^\alpha(C) = C_{h,\alpha}+ C_{com, \alpha}+ C_{inh, \alpha},
$$
where
\begin{enumerate}
\item $C_{h,\alpha}$ solves $T_g(C_{h,\alpha})=0$ with data $\widehat{K}^\alpha(C)(\rho=2)$. We have immediately $|C_{h,\alpha}| \lesssim \epsilon^{1/2}$ since the transport equation preserves the $L^\infty$-norm.
\item $C_{com, \alpha}$ solves $T_g(C_{com,\alpha})=[ T_g, \widehat{K}^\alpha](C)$ with $0$ data. 
\item $C_{inh, \alpha}$ solves $T_g(C_{inh,\alpha})=\widehat{K}^\alpha(F)$ with $0$ data.
\end{enumerate}

In view of the above decomposition of $\widehat{K}^\alpha(F)$, we have easily,
\begin{lemma} Let $N_0 \le N-3$.
Assume that for all  $|\alpha|\le N_0$, we have  $$|\widehat{K}^\alpha(C)| \lesssim \epsilon^{1/2} \rho^{\delta M_{N_0} /2}u^{1/2}, $$
for some constant $M_{N_0}$ depending only on $N_0$.

Then, we have, for all $|\alpha'|=N_0+1$,

$$
|C_{inh, \alpha'}| \lesssim \epsilon^{1/2} \rho^{\delta  M_{N_0} /2}u^{1/2},
$$
where $M_{N_0}$ is a constant depending only on $N_0$.
\end{lemma}
\begin{proof}We use Lemma \ref{lem:hoccs} to rewrite $\widehat{K}^\alpha(F)$. Since $N_0 \le N-3$, we have access to pointwise estimates on all the source terms, and the proof is then similar to that of Lemma \ref{lem:ccu}.
\end{proof}

We also have the following improvements.
\begin{lemma} \label{lem:iexc}
\begin{itemize}
\item Let $C_{inh,X}$ be the solution to $T_g(C_{inh,X})=X(F)$, with vanishing initial data. Then, we have the estimate
$$
|C_{inh,X} | \lesssim \epsilon^{1/2} \rho^{\delta/2}.
$$
\item Let $N_0 \le N-4$.
Assume that for all $|\alpha|\le N_0$, we have  $$|[ \widehat{K}^\alpha X] (C)| \lesssim \epsilon^{1/2} \rho^{\delta M_{N_0} /2}, $$
for some constant $M_{N_0}$ depending only on $N_0$ and where $[ \widehat{K}^\alpha X]$ denotes a differential operator composed of $N_0+1$ vector fields with one of them being $\partial_t$ or $X_i$.
Let $\widehat{K}^\beta$ be a differential operator of the form
$$\widehat{K}^\beta  = [\widehat{K}^{\alpha'}X],$$ with $|\alpha'|=N_0+1$.
Then, we have
$$
|C_{inh, \beta}| \lesssim \epsilon^{1/2} \rho^{\delta  M_{N_0} /2},
$$
where $M_{N_0}$ is a constant depending only on $N_0$.
\end{itemize}
\end{lemma}
\begin{proof}
This follows as in the proof of the pointwise estimates for the $C$ coefficients, but with each source term having an additional $u^{-1}$ decay coming from the application of the $X$ vector field.
\end{proof}

Thus, it follows that we only need to prove bounds on $C_{com,\alpha}$ and we obtain bounds on $\widehat{K}^\alpha(F)$ by the previous lemma and an easy induction.

Recall that the basic structure of the main terms of the commutator $[ T_g, \widehat{K}^\alpha](C)$ is of the form
$$
\frac{1}{t^q}P(C)^{k, r_Z,s_X} w \cdot \partial_{t,x} K^\beta(h) \widehat{K}^\mu,
$$
where $q+r_Z+s_X+|\beta|+|\mu| \le |\alpha|+1$, $|\beta| \le |\alpha|$, $1 \le |\mu| \le |\alpha|$, $r_Z+s_X \le |\alpha|-1$ and $k$ satisfies either $\mathrm{C1}$ or $\mathrm{C2}$.

We have already proven the pointwise estimates $|C| \le \epsilon^{1/2} \rho^{\delta/2}u^{1/2}$. Again, we go through the hierarchy of vector fields which we repeat here. For simplicity, we assume the weak bounds
$$
|\widehat{K}^\alpha(C)| \lesssim  \epsilon^{1/2} \rho^{\delta L}  u^{1/2},
$$
for some large $L$ verifying $\delta L << 1$.

We first commute once with an $X_i$ vector field.  As before, we can estimate the resulting error term coming from the commutator as
$$
\rho^{-3/2+\delta} |v^\rho \widehat{K}^\mu(C)|.
$$

From the weak bounds and the estimate on the source term $X(F)$, it follows that
$$|X_i(C)| \lesssim \epsilon^{1/2}.$$

We then commute with $\partial_t$. The error term can be estimated by
$$
\rho^{-3/2+\delta} |v^\rho \widehat{K}^\mu(C)| + \rho^{-1+\delta} |v^\rho \widehat{X}^\mu(C)|,
$$
which improves the weak bounds and give $\partial_t(C) \lesssim \epsilon^{1/2}$. We then commute with $Y$, then $X_i^2$ etc.. At each step in the iteration, we obtain an additional finite loss $\rho^{\delta/2}$. The total number of losses is then only proportional to the number of iterations, i.e.~the length of $|\alpha|$. We have thus proven

\begin{proposition}\label{prop:lowcp}
Let $|\alpha| \le N-3$. Then,
\begin{equation} \label{es:bce}
|\widehat{K}^\alpha (C)| \lesssim \rho^{M_{N} \delta} u^{1/2}.
\end{equation}
\end{proposition}

As above, we have the improvement

\begin{proposition}
Let $|\alpha| \le N-4$. Then,
\begin{eqnarray} \label{es:ixc}
|[ \widehat{K}^\alpha X ] (C)| \lesssim \rho^{M_{N} \delta}.
\end{eqnarray}
\end{proposition}

\section{Commuting the energy momentum tensor} \label{se:comt}
Let $\alpha$ be a multi-index and consider the expression
$$
K^\alpha\left( T_{\mu \nu}[f] \right),
$$
that naturally appears in the commuted wave equations.

The aim of this section is to explain how the above expression can be expressed in terms of $\widehat{K}^\beta(f)$, $|\beta| \le |\alpha|$. The Klainerman-Sobolev inequalities for velocity averages of Vlasov fields of Section \ref{se:KSf} combined with the energy estimates then automatically provide decay estimates for such quantities, at least when all $C$ and $h$ coefficients appearing can be estimated pointwise.

Recall the definition of the energy momentum tensor
$$
T_{\alpha \beta} [f]= \int_v f v_\alpha v_\beta \sqrt{-g^{-1}} \frac{dv}{g^{\alpha 0}v_{\alpha}}.
$$

We can write

$$
T_{\alpha \beta} [f]= \int_v f w_{\alpha}w_{\beta} \mathcal{L}_{\alpha,\beta}(v,h) \frac{dv}{\sqrt{1+|v|^2}}= \int_v f w_{\alpha}w_{\beta} \mathcal{L}_{\alpha,\beta}(v,h) \frac{dv}{w^0},
$$
where the functions\footnote{Note that the $\mathcal{L}_{\alpha,\beta}$ are not tensorial in $\alpha$ or $\beta$.} $\mathcal{L}_{\alpha, \beta}(v,h)$ are given by
\begin{eqnarray*}
\mathcal{L}_{0,0}(v,h)&=& \sqrt{-g^{-1}} \frac{v_0 v_0}{w^0 w^0} \frac{w^0}{g^{\alpha 0}v_{\alpha}}, \\
\mathcal{L}_{0,i}(v,h)&=& \sqrt{-g^{-1}} \frac{v_0}{w^0} \frac{w^0}{g^{\alpha 0}v_{\alpha}}, \\
\mathcal{L}_{i,j}(v,h)&=& \sqrt{-g^{-1}}  \frac{w^0}{g^{\alpha 0}v_{\alpha}}.
\end{eqnarray*}

In view of the definition \eqref{def:sf} of the tensor field $S[f]$, we also define ${\mathcal{L}'_{\alpha \beta} }=g_{\alpha \beta} \sqrt{-g^{-1}}  \frac{w^0}{g^{\alpha 0}v_{\alpha}}$, so that
$$
\int_v f {\mathcal{L}'}_{\alpha \beta} (v,h) \frac{dv}{\sqrt{1+|v|^2}}=g_{\alpha \beta} \int_v f \sqrt{-g^{-1}} \frac{dv}{g^{\alpha 0}v_{\alpha}}.
$$

We have easily
\begin{lemma}\label{lem:elf}
Let $L$ be any of the above $\mathcal{L}$ or $\mathcal{L'}$ functions.
For any multi-index $\alpha$,
$$
 \left|K^\alpha L \right| \lesssim \frac{1}{(1+u)^{\alpha_X}} \sup_{ |\beta| \le |\alpha|}| K^\beta(h) |,
$$
and
$$
\left|\partial_{v^i}K^\alpha L \right| \lesssim  \frac{1}{w^0(1+u)^{\alpha_X}} \sup_{ |\beta| \le |\alpha|}| K^\beta(h)|,
$$
where $\alpha_X$ is the number of $X$ vector fields in $K^\alpha$.
\end{lemma}
Now, for any multi-index $|\alpha|$, we have
$$K^\alpha T_{\mu \nu} [f]= \sum_{|\beta|+|\gamma| \le |\alpha|}
A_{\beta \gamma} \int_v K^\beta(f) w_{\mu} w_{\nu} K^\gamma \left( L_{\mu,\nu}\right) \frac{dv}{\sqrt{1+|v|^2}},
$$
for some constants $A_{\beta \gamma}$.

\begin{remark}
All the terms above with $|\gamma| \ge 1$ will have strong decay properties and can be considered as cubic terms.
\end{remark}
Let us introduce the notation

$$\mathcal{T}_{\mu \nu}[f]:=\int_{v} f w_\alpha w_\beta \frac{dv}{\sqrt{1+|v|^2}},$$
for the energy momentum tensor of a Vlasov field corresponding to the flat Minkowski metric $\eta$. From the above, we can write
$$
T_{\mu \nu}[f]=\Tcal_{\mu \nu}[f L],
$$
dropping the indices on the $L$ functions. Since we can always write $\partial_{x^i}$ in terms of $X_i$ and $\partial_t$, commuting with translations causes no difficulty. We consider in the following lemma the case of the homogeneous vector fields $Z$.
\begin{lemma} \label{lem:comtp}
Let $Z$ be a Lorentz boost $Z_i$ or the scaling vector field $S$. Let $\widehat{Z}$ be its complete lift for $Z_i$ a Lorentz boost and $\widehat{Z}=S$ if $Z=S$. Let $Y=\widehat{Z}+C.X$ be the modified version of $\widehat{Z}$.

Then,
\begin{eqnarray*}
Z  T_{\mu \nu} [f]&=& \Tcal_{\mu \nu}[f Z(L)]+\Tcal_{\mu \nu}[Y(f) L]-\Tcal_{\mu \nu}[C\cdot X(f) L] \\
&&\hbox{}+c_s \int_v f \partial_v( w_\mu w_\nu L) dv.
\end{eqnarray*}
where $c_s=0$ for the scaling vector field.
\end{lemma}
\begin{proof}
This is a simple calculation. We do the computation in the case of a Lorentz boost $Z_i$.
First, we distribute $Z$ and obtain
$$
Z  T_{\mu \nu} [f]= \Tcal_{\mu \nu}[f Z(L)]+\Tcal_{\mu \nu}[Z(f) L].
$$
For the second term on the right-hand side, we then write $Z=Y-c_s w^0 \partial_v -C\cdot X$, where $c_s=0$ for the $Z=S$, and we integrate by parts in $v$ (note that the $w^0$ in $w^0 \partial_v$ cancels with the $\frac{1}{w^0}$ of the measure) which gives
$$
\Tcal_{\mu \nu}[Z(f) L]=\Tcal_{\mu \nu}[Y(f) L]-\Tcal_{\mu \nu}[C\cdot X(f) L]+c_s \int_v f \partial_v( w_\mu w_\nu L) dv.
$$
\end{proof}
Note that for the term $\int_v f \partial_v( w_\mu w_\nu L) dv$, we can distribute the $\partial_v$, producing three terms which behave essentially like $T_{\mu \nu} [f]$ or better.
On the other hand, since the coefficients $C$ a priori have a $\rho^{\delta/2} u^{1/2}$ growth, we see that the term $\Tcal_{\mu \nu}[C\cdot X(f) L]$ behaves a priori worse by a factor of $\rho^{\delta/2} u^{1/2}$ than the other terms.

To improve upon the above commutation relation, we will make use the improved estimates \eqref{es:ixc} on $X(C)$. More precisely, we will use the following improved decomposition.

\begin{lemma} \label{lem:forcomt}
For any regular distribution function $k:=k(t,x,v)$, $\Tcal_{\mu \nu}[C\cdot X(k) L]$ can be expressed as a linear combination\footnote{It can happen that the component $\Tcal_{\mu \nu}[C\cdot X(k) L]$ is expressed in terms of others components $\Tcal_{\alpha \beta}$[F] with $(\alpha, \beta) \neq (\mu, \nu)$, hence the cumbersome sentence.} of terms of the form $\Tcal_{\alpha \beta}[F]$, where $F$ is a linear combination of
\begin{eqnarray*}
&&C k X(L), \\
&&X(C) k L, \\
&&\frac{1}{1+u}Y(C k L), \\
&&\frac{1}{1+u}X(C k L), \\
&&\frac{1}{1+u}C\cdot X(C k L), \\
&&\frac{1}{1+u}Ck L.
\end{eqnarray*}

\end{lemma}
\begin{proof}
We compute
$$
\Tcal_{\mu \nu}[C\cdot X(k) L]= \Tcal_{\mu \nu}[X\left(Ck L\right)]-\Tcal_{\mu \nu}[X(C) k L]-\Tcal_{\mu \nu}[C k X(L)].
$$
For the first term on the right-hand side, we write $X$ as a linear combination of the homogeneous vector fields (cf~Appendix \ref{se:dect}) $X= \frac{1}{u}a_\alpha Z^\alpha$, which we can also write
$$
X= \frac{1}{1+u}a_\alpha Z^\alpha + \frac{X}{1+u}.
$$
Finally, for any $Z$ vector fields in the above decomposition of $X$, we replace it in terms of its modified vector field version using as usual that $Z=Y-c_sw^0 \partial_{v}-C\cdot X$ and we integrate by parts in $v$.
\end{proof}

Let us define by $Q$ any of the vector fields $X$, $Z$ or $\partial_v$. Again, for a muli-index $\beta$ we will write $Q^\beta$ for a combination of $|\beta|$ such vector fields. We then have the following higher order commutator formula.
\begin{lemma} \label{lem:comtlw}
For any multi-index $\alpha$,
$\widehat{K}^\alpha T_{\mu \nu}[f]$ can be written as a linear combination of the terms $\Tcal_{\alpha \beta}[F]$, where $F$ are, modulo multiplications by a function in $\mathcal{F}_v$, of the form
\begin{eqnarray*}
 \frac{1}{(1+u)^q}  P_1(C)^{l_1+l_1',r_{Z,1},s_{X,1}} P_2(X(C))^{l_2,r_{Z,2},s_{X,2}} \widehat{K}^\beta(f) Q^\gamma(L) ,
\end{eqnarray*}
where $r_{Z,2}+s_{X,2} \le |\alpha|-1$, $r_{Z,1}+s_{X,1}+r_{Z,2}+s_{X,2}+l_2 +|\gamma|+|\beta| \le |\alpha|$ and $2 q  \ge  l_1$, $q \le |\alpha|$, $\gamma_X \ge l_1'$, $l_2 \le |\alpha|$.  
\end{lemma}

\begin{proof}
We have already proven the lemma for $|\alpha|=1$, in view of Lemmas \ref{lem:comtp} and \ref{lem:forcomt}, as can be checked easily. For instance,
\begin{enumerate}
\item The terms arising from $\frac{1}{1+u}C\cdot X(C) f L$ (see Lemma \ref{lem:forcomt}) correspond to $l_1=1$, $q=1$, $l_2=1$, $l_1'=0$ and $r_{Z,1}+s_{X,1}+r_{Z,2}+s_{X,2}+|\gamma|+|\beta|=0$.
\item The terms arising from $\frac{1}{1+u}C^2 X(f) L$ (see Lemma \ref{lem:forcomt}) correspond to $l_1=2$, $q=1$, $l_2=0$ , $l_1'=0$ and $r_{Z,1}+s_{X,1}+r_{Z,2}+s_{X,2}+|\gamma|+|\beta|=0$.
\item The terms arising from $CkX(L)$ (see Lemma \ref{lem:forcomt}) correspond to $l_1'=1$, $\gamma_X=1=|\gamma|$, $q=0$, $l_1=l_2=0$ and $r_{Z,1}+s_{X,1}+r_{Z,2}+s_{X,2}+|\beta|=0$.
\end{enumerate}

Assume that the statement holds for some multi-index $\alpha$ and consider a differential operator of the form $K\cdot K^\alpha$. When $K$ is a translation, we simply distribute $K$ on each of the three factors. For $K$ a Lorentz boost or the scaling vector field, we simply apply the first order commutator formula and verifies that each resulting term is of the above form.
\end{proof}

When $|\alpha|=N$, we cannot apply the previous formula, as it would involve derivatives of the $C$ coefficients of order $N$, which themselves would demand a control on $\partial Z^{N+1} h$. Instead, we rely on the following argument. First, we commute once the energy momentum tensor, using only Lemma \ref{lem:comtp}. According to the above first order formula, we obtain for $K T_{\mu \nu}[f]$
\begin{itemize}
\item The good commutation terms, denoted $GT$, are of the form
\begin{eqnarray*}
&&\hbox{}\Tcal_{\mu \nu}[f K(L)], \quad \Tcal_{\mu \nu}[\widehat{K}(f) L],\quad \int_v f \partial_v(w_\mu w_\nu L) dv.
\end{eqnarray*}
\item The bad commutation terms, denoted $BT$, of the form  $$\Tcal_{\mu \nu}[C\cdot X(f) L].$$
\end{itemize}
For the $GT$ terms, we can then commute $N-1$ extra times, using Lemma \ref{lem:comtlw}.

The error terms that we then obtain are listed in the following lemma\footnote{In the following lemmas, one should keep in mind that the multi-index $\alpha$ has length at most $N-1$, so that the maximum number of commutations for $f$ never exceeds $N$ and the maximum number of commutations for any $C$ coefficient never exceeds $N-1$.}.

\begin{lemma} \label{lem:comhgt}
 Let $\alpha$ be a multi-index with $|\alpha|\le N-1$. Then, for any $GT$ term, $K^\alpha GT$ can be written as a linear combination of the terms
$\Tcal_{\alpha \beta}[F]$, where $F$ are, modulo multiplication by a function in $\mathcal{F}_{x,v}$, of the form
\begin{enumerate}
\item Form 1 (comes from the GT term $\Tcal_{\mu \nu}[\widehat{K}(f) L]$)
\begin{eqnarray*}
 \frac{1}{(1+u)^q}  P_1(C)^{l_1+l_1',r_{Z,1},s_{X,1}} P_2(X(C))^{l_2,r_{Z,2},s_{X,2}} \widehat{K}^\beta(\widehat{K} f) Q^\gamma(L) ,
\end{eqnarray*}
where $r_{Z,2}+s_{X,2} \le |\alpha|-1$, $r_{Z,1}+s_{X,1}+r_{Z,2}+s_{X,2}+l_2 +|\gamma|+|\beta| \le |\alpha|$ and $2 q\ge  l_1$, $q \le |\alpha|$, $\gamma_X \ge l_1'$, $l_2 \le |\alpha|$.
\item Form 2 (comes from the GT terms $\Tcal_{\mu \nu}[f Q(L)]$ )
\begin{eqnarray*}
 \frac{1}{(1+u)^q}  P_1(C)^{l_1+l_1',r_{Z,1},s_{X,1}} P_2(X(C))^{l_2,r_{Z,2},s_{X,2}} \widehat{K}^\beta(f) Q^\gamma(Q L) ,
\end{eqnarray*}
where $r_{Z,2}+s_{X,2} \le |\alpha|-1$, $r_{Z,1}+s_{X,1}+r_{Z,2}+s_{X,2}+l_2 +|\gamma|+|\beta| \le |\alpha|$ and $2 q \ge  l_1$, $q \le |\alpha|$, $\gamma_X \ge l_1'$, $l_2 \le |\alpha|$.
\item Form 3 (comes from the GT term $\int_v f \partial_v(w_\mu w_\nu) L dv$)
\begin{eqnarray*}
 \frac{1}{(1+u)^q}  P_1(C)^{l_1+l_1',r_{Z,1},s_{X,1}} P_2(X(C))^{l_2,r_{Z,2},s_{X,2}} \widehat{K}^\beta(f) Q^\gamma( L) ,
\end{eqnarray*}
where $r_{Z,2}+s_{X,2} \le |\alpha|-1$, $r_{Z,1}+s_{X,1}+r_{Z,2}+s_{X,2}+l_2 +|\gamma|+|\beta| \le |\alpha|$ and $2 q\ge  l_1$, $q \le |\alpha|$, $\gamma_X \ge l_1'$, $l_2 \le |\alpha|$.
\end{enumerate}
\end{lemma}



These terms do not pose any regularity issue since only at most $N-1$ derivatives of the $C$ coefficients are involved and they also do not pose any decay issue, since all the $C$ coefficients in $P_2$ are hit by at least one $X$ derivative and all the $C$ coefficients in $P_1$ are compensated either by the prefactor in $1/(1+u)$ or from the extra decay coming from $\gamma_X$.

For the $BT$ terms, we can also commute $N-1$ extra times, which gives similarly,

\begin{lemma} \label{lem:btt}
 Let $\alpha$ be a multi-index with $|\alpha|\le N-1$. Then, for any $BT$ term, $K^\alpha BT$ can be written as a linear combinations of the terms $\Tcal_{\mu \nu}[F]$, where $F$ is of the form

\begin{eqnarray*}
 \frac{1}{(1+u)^q}  P_1(C)^{l_1+l_1',r_{Z,1},s_{X,1}} P_2(X(C))^{l_2,r_{Z,2},s_{X,2}}\widehat{K}^\beta(C) \widehat{K}^\sigma(X f) Q^\gamma(L) ,
\end{eqnarray*}
where $r_{Z,2}+s_{X,2} \le |\alpha|-1$, $r_{Z,1}+s_{X,1}+r_{Z,2}+s_{X,2}+l_2 +|\gamma|+|\beta|+|\sigma| \le |\alpha|$ and $2 q \ge  l_1$, $q \le |\alpha|$, $\gamma_X \ge l_1'$, $l_2 \le |\alpha|$.
\end{lemma}


If we naively estimate the error terms coming from the previous lemma, they contain a priori an extra power of $\widehat{K}^\beta(C)$. At least when we can estimate them pointwise,  this would lead to a loss of $\rho^{\delta M} u^{1/2}$ and we would not be able to improve the top order energy estimates\footnote{Note that later we will obtain improved estimates on the $\widehat{K}^\beta(C)$ for $|\beta|$ small enough. However, to obtain the improved estimates, we also need improved decay of the source terms, so they cannot be immediately used here.}.

Instead, we separate all the terms above into two sets, the terms for which one of the $C$ coefficients is hit by $|\alpha|$ vector fields (which arises only if $|\beta| =|\alpha|$ or $r_{Z,1}+s_{X,1}=|\alpha|$ or $r_{Z,2}+s_{X,2}=|\alpha|-1$) denoted $BT_{|\alpha|}$ and the terms for which no $C$ coefficients is hit by $|\alpha|$ vector fields, denoted $BT_{<|\alpha|}$.

From the above definition, we have
\begin{lemma} \label{lem:bta}
\begin{itemize}
\item The $BT_{|\alpha|}$ terms are all of the form $\Tcal_{\mu \nu}[F]$, where $F$ are linear combinations of
\vspace{0.3cm}
\begin{enumerate}
\item[1.] $\frac{1}{(1+u)^q }C^{l_1} K^\beta(C) X(f) L$, with $|\beta| =|\alpha|$ and $2q\ge l_1$,
\item[2.]  $C K^\beta(X C) X(f) L$,
with $|\beta| =|\alpha|-1$.
\end{enumerate}
\vspace{0.3cm}
\item The $BT_{<|\alpha|}$ terms are all of the form $\Tcal_{\alpha \beta}[F]$, where $F$ are, modulo multiplications by a function in $\mathcal{F}_v$, of the form
\begin{eqnarray*}
 \frac{1}{(1+u)^q}  P_1(C)^{l_1+l_1',r_{Z,1},s_{X,1}} P_2(X(C))^{l_2,r_{Z,2},s_{X,2}}\widehat{K}^\beta(C) \widehat{K}^\sigma(X f) Q^\gamma(L) ,
\end{eqnarray*}
where $r_{Z,2}+s_{X,2} \le |\alpha|-1$, $r_{Z,1}+s_{X,1}+r_{Z,2}+s_{X,2}+l_2 +|\gamma|+|\beta|+|\sigma| \le |\alpha|$, $|\beta| \le |\alpha|-1$, and $2 q \ge  l_1$, $q \le |\alpha|$, $\gamma_X \ge l_1'$ $l_2 \le |\alpha|$ and no $C$ coefficients in the formula is hit by more than $|\alpha|-1$ vector fields\footnote{For the $C$ coefficients in $P_2$, this means that a top order term is of the form $\widehat{K}^\sigma X(C)$, with $|\sigma|=|\alpha|-2$.}.
\end{itemize}
\end{lemma}
In the above terms, the worst terms are the $BT_{|\alpha|}$ terms as in 1., with $q=0$, which are the terms leading to the eventual extra\footnote{There is also an independent source of growth coming from the pure Einstein non-linearities.} growth for the top order energy $\Ecal_N[h]$.
The other $BT_{|\alpha|}$ terms, always contain at least a $C$ coefficient with low (no) derivatives and we will be able to use the improved decay (cf~\eqref{es:iecc}) to estimate those.

However, the $BT_{<|\alpha|}$ terms as written above are not good enough, in that they a priori contain one extra power of $\widehat{K}^\beta(C)$. We divide them into two sets, $BT_{<|\alpha|,|\sigma| < |\alpha| }$ and $BT_{<|\alpha|, |\sigma| = |\alpha| }$, according to the value of $\sigma$. We note that
\begin{itemize}
\item The $BT_{<|\alpha|,|\sigma| < |\alpha|}$ terms involve only $|\widehat{K}^\sigma(Xf)|$ for $|\sigma|<  |\alpha| \le N-1$ (so $|\sigma|+1 \le N-1$) and $\widehat{K}^\beta(C)$ for $|\beta| \le |\alpha|-1\le N-2$. This implies that controlling the $BT_{<|\alpha|,|\sigma| < |\alpha|}$ can be achieved using only estimates on $\Ecal_{N-1}[h]$ and $E_{N-1}[f]$.
\item The $BT_{<|\alpha|, |\sigma| = |\alpha| }$ terms are of the form
$$
 \frac{1}{(1+u)^q}  C^{l_1+1} \widehat{K}^\sigma(X f) L ,
 $$
 where $2 q \ge l_1$. Here, there are no derivatives acting on the $C$ coefficients, so we will be able to estimate them using the improved decay \eqref{es:iecc}.
\end{itemize}


For the terms $BT_{<|\alpha|, |\sigma| = |\alpha| }$, we can also use here the fact that we can transfer the $X$ derivative of $f$ as in the first order commutator formula, leading to the following lemma.

\begin{lemma} \label{lem:bts} Any $BT_{<|\alpha|, |\sigma| = |\alpha|}$ term can be rewritten as a linear combination of
\begin{itemize}
\item $BT_{<|\alpha|, |\sigma| < |\alpha|}$ terms and
\item Good terms of the form of the form $\Tcal_{\mu \nu}[F]$, with $F$ given as a linear combination of terms of the form
\begin{eqnarray*}
 \frac{1}{(1+u)^q}  P_1(C)^{l_1+l_1',r_{Z,1},s_{X,1}} P_2(X(C))^{l_2,r_{Z,2},s_{X,2}}\widehat{K}^\sigma( f) Q^\gamma(L) ,
\end{eqnarray*}
where $r_{Z,2}+s_{X,2} \le |\alpha|$, $r_{Z,1}+s_{X,1}+r_{Z,2}+s_{X,2}+l_2 +|\gamma|+ |\sigma| \le |\alpha|+1$, $2 q\ge  l_1$, $\gamma_X \ge l_1'$, $q \le |\alpha|+1$, $l_2 \le |\alpha|+1$ and no $C$ coefficients in the formula is hit by more than $|\alpha|$ vector fields.
\end{itemize}
\end{lemma}
\begin{proof}
Consider a $BT_{<|\alpha|, |\sigma| = |\alpha|}$ term, denoted $T$, of the form
$$
T=  \frac{1}{(1+u)^q}  C^{l_1+1} \widehat{K}^\sigma(X f) L.
$$

We first commute the $X$ and the $\widehat{K}^\sigma$. We obtain
\begin{eqnarray*}
T&=&  \frac{1}{(1+u)^q}  C^{l_1+1}X \widehat{K}^\sigma(f) L+  \frac{1}{(1+u)^q}  C^{l_1+1}[\widehat{K}^\sigma, X] (f) L \\
&=& T_{1}+T_{2}.
\end{eqnarray*}

We consider $T_1$ first. As in the first order case, we write it as
\eq{
T_1= T_{1,1}+ T_{1,2},
}
where
$$
T_{1,1}:= X \left( \frac{1}{(1+u)^q}  C^{l_1+1} \widehat{K}^\sigma( f) L \right)\\
$$
and
$$
T_{1,2}:= -   \widehat{K}^\sigma( f) \cdot X \left( \frac{1}{(1+u)^q}  C^{l_1+1}  L \right).
$$
For $T_{1,2}$, we simply distribute the $X$ derivatives.

For $T_{1,1}$, we rewrite the first $X$ in terms of the $Z$ vector fields, schematically as $X=\frac{1}{u}Z$, and finally complete the $Z$ vector fields. Thus, we can rewrite $T_{1,1}$ in terms of
\begin{eqnarray*}
&&\hbox{}\frac{1}{(1+u)} Y \left( \frac{1}{(1+u)^q}  C^{l_1+1} \widehat{K}^\sigma( f) L \right) \\
&&\hbox{} \frac{1}{(1+u)} C\cdot X\left( \frac{1}{(1+u)^q}  C^{l_1+1} \widehat{K}^\sigma( f) L  \right) \\
&&\hbox{} \frac{1}{(1+u)}c_s w^0\cdot \partial_v \left( \frac{1}{(1+u)^q}  C^{l_1+1} \widehat{K}^\sigma( f) L  \right).
\end{eqnarray*}
For the first and second error terms, we can simply distribute the $Y$ and $X$ derivatives. For the last error term, we integrate by parts in $v$. All the terms then take the form of the lemma.   

We now consider $T_2$. Using Lemma \ref{lem:commvf}, $[\widehat{K}^\sigma, X]$ can be written as a linear combination of terms of the form
$$
P(X(C))^{k,r_Z,s_X}\cdot K^{\alpha'},
$$
where $|\alpha'|+r_Z+s_X \le |\sigma|$, $|\alpha'| \le |\sigma|$, $\alpha'_X \ge 1$, $k \le |\sigma|$, $r_Z+s_X \le |\sigma|-1$.
\end{proof}

Finally, we note that in the special case when $K^\alpha=X K^\beta$, with $|\beta|=|\alpha|-1$, then there are no  dangerous terms, even at the top order when $|\alpha|=N$. More precisely, we have

\begin{lemma} \label{lem: Xtopt}
Let $\beta$ be a multi-index with $|\beta|\le N-1$ and let $K^\alpha=X K^\beta$, where $X=\partial_{x^\gamma}$. Then, $K^\alpha T_{\alpha \beta}[f]$ can be written as a linear combination of the terms $\Tcal_{\mu \nu}[F]$, with $F$ given as a linear combination of terms of the form
\begin{eqnarray*}
 \frac{1}{(1+u)^q}  P_1(C)^{l_1+l_1',r_{Z,1},s_{X,1}} P_2(X(C))^{l_2,r_{Z,2},s_{X,2}}\widehat{K}^\sigma( f) Q^\gamma(L) ,
\end{eqnarray*}
where $r_{Z,2}+s_{X,2} \le |\alpha|-2$, $r_{Z,1}+s_{X,1} \le |\alpha|-1$, $r_{Z,1}+s_{X,1}+r_{Z,2}+s_{X,2}+l_2 +|\gamma|+ |\sigma| \le |\alpha|$, and $2q\ge l_1$, $q \le |\alpha|$, $\gamma_X \ge l_1'$, $l_2 \le |\alpha|$.
\end{lemma}
\begin{proof}We use first Lemma \ref{lem:btt} and then commute once by $X$. Again, we can write $X=\frac{Y-CX-c_s w^0 \partial_v}{1+u}$. It follows that the resulting terms have the required structure.
\end{proof}

\section{The analysis of the reduced Einstein equations}

\label{sec:aree}

\subsection{Structure of the reduced Einstein equations}

Recall the basic structure of the reduced Einstein equations
$$
\widetilde{\square}_g h_{\alpha \beta} = F_{\alpha \beta}(h, \partial h) -  S_{\alpha \beta}[f]
$$
First, we recall the structure of the semi-linear terms $F_{\alpha \beta}(h, \partial h)$ (see for instance \cite[Lemma 3.1]{Lindblad:2005ex}).

\begin{lemma}
We have $F_{\alpha \beta} := P_{\alpha \beta} + Q_{\alpha \beta}$, where
\begin{eqnarray*}
Q_{\alpha \beta} & = &  g^{\lambda \lambda'}g^{\delta \delta'} \del_{\delta}g_{\alpha \lambda'} \del_{\delta'}g_{\beta \lambda}
-g^{\lambda \lambda'}g^{\delta \delta'} \big
(\del_{\delta}g_{\alpha \lambda'} \del_{\lambda}g_{\beta \delta'} - \del_{\delta}g_{\beta \delta'} \del_{\lambda}g_{\alpha \lambda'} \big)
\\
&&\hbox{}+ g^{\lambda \lambda'}g^{\delta \delta'}
\big(\del_{\alpha}g_{\lambda'\delta'} \del_{\delta}g_{\lambda \beta} - \del_{\alpha}g_{\lambda \beta} \del_{\delta}g_{\lambda'\delta'} \big)\\
&&\hbox{}+\frac{1}{2}g^{\lambda \lambda'}g^{\delta \delta'}
\big(\del_{\alpha}g_{\lambda \beta} \del_{\lambda'}g_{\delta \delta'} - \del_{\alpha}g_{\delta \delta'} \del_{\lambda'}g_{\lambda \beta} \big)
\\
&&\hbox{} + g^{\lambda \lambda'}g^{\delta \delta'} \big(\del_{\beta}g_{\lambda'\delta'} \del_{\delta}g_{\lambda \alpha} - \del_{\beta}g_{\lambda \alpha} \del_{\delta}g_{\lambda'\delta'} \big)\\
&&\hbox{}+\frac{1}{2}g^{\lambda \lambda'}g^{\delta \delta'}
\big(\del_{\beta}g_{\lambda \alpha} \del_{\lambda'}g_{\delta \delta'} - \del_{\beta}g_{\delta \delta'} \del_{\lambda'}g_{\lambda \alpha} \big),
\end{eqnarray*}
and
$$
P_{\alpha \beta} = - \frac{1}{2}g^{\lambda \lambda'}g^{\delta \delta'} \del_{\alpha}g_{\delta \lambda'} \del_{\beta}g_{\lambda \delta'}
+\frac{1}{4}g^{\delta \delta'}g^{\lambda \lambda'} \del_{\beta}g_{\delta \delta'} \del_{\alpha}g_{\lambda \lambda'}.
$$
\end{lemma}
The $Q$ terms are all null forms and therefore will enjoy strong decay properties. The $P$ terms are not null forms. To control them, one uses the wave coordinate conditions, together with a hierarchy in the estimates: this is the weak null condition (cf \cite{Lindblad:2005ex}).

\subsection{Classification and structure of the pure Einstein non-linearities}

This section is concerned with the analysis of the nonlinear terms of the Einstein equations non-interacting with the Vlasov field $f$, such as those arising from $K^\gamma F_{\alpha \beta}$ and commutator terms $[ K^\gamma,h^{\mu\nu} \del_\mu\del_\nu]h_{\alpha \beta}$.

First, we provide a classification of all possible such non-linearities. Our terminology follows closely that of \cite{lm:gsmkg}, Section 4.3. Below $Z$ denotes any of the homogeneous vector fields $Z_i$ or $S$, while $\partial$ denotes any of the $\partial_{x^\gamma}$ vector fields. Note that in \cite{lm:gsmkg}, the scaling vector field was not part of the algebra of commutator vector fields, since it does not commute well with the Klein-Gordon equation. Nonetheless, most computations and conclusions of \cite{lm:gsmkg} still remain true in our case even with $S$ as a commutator, since in our setting, it behaves similarly to any other homogeneous vector field. We will highlight any important differences below.

In the notation below for the non-linear terms, the $p$ index refers to the total number of vector fields while the index $k$ only refers to the original number of homogeneous vector fields.


\begin{itemize}

\item
The basic semi-linear terms $QS(p,k)$. They are linear combinations of the following terms
$$
\partial^I Z^J\big(\del_{\mu}h_{\alpha \beta} \partial_{\nu}h_{\alpha'\beta'} \big),
$$
with $|I|+|J| \le p$, $|J| \leq k$.

\item The basic quasi-linear terms $Q(p,k)$ arising from $[ \partial^I Z^J,h^{\mu\nu} \del_\mu\del_\nu]h_{\alpha \beta}$

They are linear combinations of the following terms 
$$
\partial^{I_1}Z^{J_1}h_{\alpha'\beta'} \del^{I_2}Z^{J_2} \del_{\mu} \del_{\nu}h_{\alpha \beta},
\qquad
h_{\alpha'\beta'} \del_{\mu} \del_{\nu} \del^I Z^{J'}h_{\alpha \beta},
$$
with $|I_1|+|I_2| \leq p-k$, $|J_1|+|J_2| \leq k$ and $|I_2|+|J_2| \leq p-1$ and $|J|<|J|$.
\item The cubic terms $Cub(p)$. They are linear combinations of the following terms
$$
K^\alpha (h) . K^\beta(h) . \partial K^\gamma(h), \quad K^\alpha (h) . \partial K^\beta(h) . \partial K^\gamma(h)
$$
with $|\alpha| + |\beta| +|\gamma| \le p$.
\item The good semi-linear terms $GQS(p,k)$. They are linear combinations of the following terms and their derivatives of order $p+k$
$$
\del^IZ^J\big(\delu_a h_{\alpha \beta} \delu_{\gamma}h_{\alpha'\beta'} \big), \quad
(s/t)^2\del^IL^J\big(\del_th_{\alpha \beta} \del_th_{\alpha'\beta'} \big),
$$
with $|I|+|J| \leq p$ and $|J| \leq k$.
\item The good quasi-linear terms $GQQ(p,k)$ arising from $[ \partial^I Z^J,h^{\mu\nu} \del_\mu\del_\nu]h_{\alpha \beta}$. They are linear combinations of the following terms 
\begin{eqnarray*}
\del^{I_1}Z^{J_1}h_{\alpha'\beta'} \del^{I_2}Z^{J_2} \delu_a \delu_{\mu}h_{\alpha \beta}, \quad
 &&\del^{I_1}Z^{J_1}h_{\alpha'\beta'} \del^{I_2}Z^{J_2} \delu_{\mu} \delu_bh_{\alpha \beta},
\\
h_{\alpha'\beta'} \del^IZ^{J'} \delu_a \delu_{\mu}h_{\alpha \beta}, \quad
 &&h_{\alpha'\beta'} \del^IZ^{J'} \delu_{\mu} \delu_bh_{\alpha \beta},
\end{eqnarray*}
with $|I_1|+|I_2| \leq p-k$, $|J_1|+|J_2| \leq k$ and $|I_2|+|J_2| \leq p-1$, $|J'|<|J|$.
\item The frame terms $Com(p,k)$. 
These terms are linear combinations of the following terms 
\begin{eqnarray*}
&&t^{-1}{QS}(p,k), \\
&&t^{-1} \del^{I_1}Z^{J_1}h_{\mu\nu} \del^{I_2}Z^{J_2} \del_{\gamma}h_{\mu'\nu'}, \\
&&t^{-2} \del^{I_1}Z^{J_1}h_{\mu\nu} \del^{I_2}Z^{J_2}h_{\mu'\nu'},
\end{eqnarray*}
where $|I| \leq p-k, |J| \leq k$ and $|I_1|+|J_1| \leq p-1$, $|I_1|+|I_2| \leq p-k,|J_1|+|J_2| \le k$.

\end{itemize}




We then recall how the non-linearities arising from $F_{\alpha \beta}$ and $[\del^IZ^J,h^{\mu\nu} \del_\mu\del_\nu]$ can be decomposed, using the above notations.

We start with the commutator terms.
\begin{lemma}\cite[Lemma 4.4]{lm:gsmkg}.
\label{lem:dqlt}

For $|I|=p-k$ and $|J|=k$, 
the commutator $[\del^IZ^J,h^{\mu\nu} \del_\mu\del_\nu]h_{\alpha \beta}$ can be written as a linear combination of the following terms 
\begin{eqnarray*}
{GQQ}(p,k), &\quad& t^{-1} \del^{I_3}Z^{J_3}h_{\mu\nu} \del^{I_4}Z^{J_4} \del_{\gamma}h_{\mu'\nu'},
\\
\del^{I_1}Z^{J_1} \hu^{00} \del^{I_2}L^{J_2} \del_t\del_t h_{\alpha \beta}, &\quad& L^{J_1'} \hu^{00} \del^IZ^{J_2'} \del_t\del_t h_{\alpha \beta}, \\
&&\hbox{}\hu^{00} \del_{\gamma} \del_{\gamma'} \del^IL^{J'}h_{\alpha \beta},
\end{eqnarray*}
where $I_1+I_2=I,J_1+J_2=J$ with $|I_1|\geq 1$, $J_1'+J_2'=J$ with $|J_1'|\geq 1$ and $|J'|<|J|$,  $|I_3|+|I_4| \leq |I|, |J_3|+|J_4| \leq |J|$.
\end{lemma}

The $Q_{\alpha \beta}$ terms in $F_{\alpha \beta}$ are all null forms and those pose no threat. For $P_{\alpha \beta}$, we need

\begin{lemma} \cite[Lemma 4.10]{lm:gsmkg}. 
Let $\Pu$ denotes the components of $P$ in the semi-hyperboloidal frame.
Then
\begin{itemize}
\item $\Pu_{00}$ can be written as a linear combinations of the following terms
$$
GQS(0,0), \quad {Cub}(0,0), \quad {Com}(0,0), \quad  \underline{\eta}^{\gamma \gamma'} \underline{\eta}^{\delta \delta'} \del_t\hu_{\gamma \gamma'} \del_t\hu_{\delta \delta'}, \quad
\underline{\eta}^{\gamma \gamma'} \underline{\eta}^{\delta \delta'} \del_t\hu_{\gamma \delta} \del_t\hu_{\gamma'\delta'}.
$$
\item $\Pu_{a \beta}$ can be written as a linear combinations of terms of type ${GQS}(0,0)$ and ${Cub}(0,0)$.
\end{itemize}
\end{lemma}

So, the only problematic terms in $P_{\alpha \beta}$ are $\underline{\eta}^{\gamma \gamma'} \underline{\eta}^{\delta \delta'} \del_t\hu_{\gamma \gamma'} \del_t\hu_{\delta \delta'}$ and $\underline{\eta}^{\gamma \gamma'} \underline{\eta}^{\delta \delta'} \del_t\hu_{\gamma \delta} \del_t\hu_{\gamma'\delta'}$. They will be controlled using the wave gauge condition.


\subsection{Metric components in the semi-hyperboloidal frame}
In this section, we recall the structure of the wave equations for the components of $h$ in the semi-hyperboloidal frame.

Recall that $$S_{\alpha \beta}[f]=\big(\Phi_{\alpha}^{\alpha'} \Phi_{\beta}^{\beta'} \big) t_{\alpha'\beta'} =  \int_{v} f(2\underline{v}_{\alpha}\underline{v}_{\beta} + \underline{g}_{\alpha \beta}) d\mu_v$$
and denote by $\Su$, its components in the semi-hyperboloidal frame.

Then, we easily compute (cf~\cite[Section 4.6]{lm:gsmkg})
\begin{eqnarray*}
\widetilde {\Box}_g\hu_{00} & = & \Phi_0^{\alpha'} \Phi_0^{\beta'}Q_{\alpha'\beta'} + \Pu_{00} - \Su_{00} + { Cub}(0,0),
\\
\widetilde {\Box}_g\hu_{0a} & = & \Phi_0^{\alpha'} \Phi_a^{\beta'}Q_{\alpha'\beta'} + \Pu_{0a} - \Su_{a0}
+ \frac{2}{t} \delu_a h_{00} - \frac{2x^a}{t^3}h_{00} + {Cub}(0,0),
\\
\widetilde {\Box}_g\hu_{aa} & = & \Phi_a^{\alpha'} \Phi_a^{\beta'}Q_{\alpha'\beta'} + \Pu_{aa} - \Su_{aa}
  +  \frac{4x^a}{t^2} \delu_ah_{00} + \frac{4}{t} \delu_ah_{0a}- \frac{4x^a}{t^3}h_{0a} \\
&&\hbox{}  + \bigg(\frac{2}{t^2} - \frac{6|x^a|^2}{t^4} \bigg)h_{00} + { Cub}(0,0),
\\
\widetilde {\Box}_g\hu_{ab} & = & \Phi_a^{\alpha'} \Phi_b^{\beta'}Q_{\alpha'\beta'} + \Pu_{ab} - \Su_{ab}
+  \frac{2x^b}{t^2} \delu_ah_{00}  + \frac{2x^a}{t^2} \delu_bh_{00}  + \frac{2}{t} \delu_ah_{0b} \\
 &&\hbox{}+ \frac{2}{t} \delu_bh_{0a} - \frac{6x^ax^b}{t^4}h_{00} - \frac{2x^a}{t^3}h_{0b} - \frac{2x^b}{t^3}h_{0a}+ {Cub}(0,0), \quad (a \neq b).
\end{eqnarray*}

\subsection{Estimates of the metric nonlinearities}

We state here decay estimates which are direct consequences of the bootstrap assumption \ref{eq:bsm}. In particular, the estimates that follow are independent from those concerning Vlasov fields.

We consider first decay estimates for the "good" nonlinear terms.
\begin{lemma}\cite[lemmas 6.1 and 6.2]{lm:gsmkg}
For any of the following nonlinear terms, we have

\begin{itemize}
\item The pointwise estimates
\begin{eqnarray*}
||GQS(N-2, k)||_{L^{\infty}} &\lesssim& \epsilon \frac{\rho^{\delta}}{t^2 \rho}, \\
||GQQ(N-2, k)||_{L^{\infty}} &\lesssim& \epsilon \frac{\rho^{\delta}}{t^3}, \\
||Com(N-2, k)||_{L^{\infty}} &\lesssim& \epsilon \frac{\rho^{\delta}}{t^{5/2} \rho}, \\
||Cub(N-2, k)||_{L^{\infty}} &\lesssim& \epsilon \frac{\rho^{3/2\delta}}{t^{5/2} \rho}.
\end{eqnarray*}
\item The $L^2$ estimates
\begin{eqnarray*}
||GQS(N, k)||_{L^{2}(H_\rho^\star)} &\lesssim& \epsilon \frac{\rho^{\delta}}{\rho^{3/2} }, \\
||GQQ(N-2, k)||_{L^{\infty}} &\lesssim& \epsilon \frac{\rho^{\delta}}{\rho^{3/2}}, \\
||Com(N, k)||_{L^{2}(H_\rho^\star)} &\lesssim& \epsilon \frac{\rho^{\delta}}{\rho^{5/2} }, \\
||Cub(N, k)||_{L^{2}(H_\rho)} &\lesssim& \epsilon \frac{\rho^{3/2\delta}}{\rho^{3/2} }.
\end{eqnarray*}
\end{itemize}
\end{lemma}

Recall that it follows from the gauge condition that $\partial \hu^{00}$ is controlled by the good derivatives $h$ (Lemma \ref{lem:wc}). Thus, the $L^2$ estimates of $\partial K^\alpha h$ contained in the bootstrap assumptions can be translated into improved $L^2$ estimates for $\partial \hu^{00}$. More specifically, we have
\begin{lemma}\cite[lemmas 7.5]{lm:gsmkg}
Under the bootstrap assumptions, for any $|\alpha| \le N$,
$$
||K^\alpha \partial_{t,x} \hu^{00}||_{L^2(H_\rho^\star)} \lesssim \epsilon^{1/2} \rho^{\delta/2}.
$$
\end{lemma}

Let now $\hu_m:=\gu_m-\underline{\eta}$, where $\gu_m$ is the Schwarzschid metric in the semi-hyperboloidal frame and define $\hu_1$ as
$$
\hu:= \chi(r/t) \hu_m + \hu_1=\hu_0 + \hu_1,
$$
where $\hu_0=\chi(r/t) \hu_m$ and $\chi$ is smooth real function verifying $0 \le \chi \le 1$, $\chi(x)=0$ for $x \in [0,1/3]$, $\chi(x)=1$ for $1 \in [2/3,1]$. Note that since $m$ is small and $t \ge 2$, $\hu_1$ is regular up to $r=0$.

By explicit calculations,
\begin{eqnarray*}
|\partial_{t,x}^\alpha \hu_0^{00}| &\lesssim& m t^{-1+|\alpha|},\\
||\partial_{t,x} \hu^{00}_0||_{L^2(H_\rho)} &\lesssim& m \lesssim \epsilon^{1/2}.
\end{eqnarray*}

We then state
\begin{lemma}\cite[Lemma 7.6]{lm:gsmkg}
Under the bootstrap assumptions, we have
\begin{eqnarray*}
\left\Vert\left( \frac{\rho}{t} \right)^{-1+\delta} \rho^{-1} K^\alpha \hu_1^{00}\right\Vert_{L^2(H_\rho^\star)} \lesssim \epsilon^{1/2} \rho^{\delta/2}.
\end{eqnarray*}
\end{lemma}
Note that the proof of the above lemma only makes use of the bootstrap assumptions for $h$ and functional inequalities on each hyperboloid $H_\rho$. In particular, it is independent of the equations satisfied by $h$ and therefore of the Vlasov field.

We also recall here the estimates on $[ K^\alpha, h^{\mu \nu} \partial_{x^\mu} \partial_{x^\nu}]h$.

\begin{lemma}\cite[Lemmas 7.7 and 7.8]{lm:gsmkg}
Under the bootstrap assumptions,
\begin{itemize}
\item For any $|\alpha| \le N-2$, and any $(t,x) \in \Kcal \cap \{ \rho \ge 2 \}$,
\begin{eqnarray*}
\left| [ K^\alpha, h^{\mu \nu} \partial_{x^\mu} \partial_{x^\nu}]h\right|(t,x) &\lesssim& \epsilon t^{-2} \rho^{-1+\delta}\\
&&\hbox{} + \epsilon^{1/2}\left( t^{-1}+ (\frac{\rho}{t})^2 t^{-1/2} \rho^\delta \right) \sum_{|\beta| < |\alpha|, \beta_Z < \alpha_Z, \beta_X=\alpha_X} \left| \partial_t^2 K^\beta h \right|.
\end{eqnarray*}
\item For any $|\alpha| \le N$ and any $\rho \ge 2$,
\begin{eqnarray*}
\left|\left| \rho  [ K^\alpha, h^{\mu \nu} \partial_{x^\mu} \partial_{x^\nu}]h\right|\right|_{L^2(H_\rho^\star)} &\lesssim& \epsilon \rho^{\delta} + \epsilon^{1/2} \rho^{\delta/2} \sum_{\beta_X=\alpha_X, \beta_Z \le 1} || \rho^2 \left(\frac{\rho}{t}\right)^{1-\delta} K^\beta \partial_t^2 h ||_{L^\infty} \\
&&\hbox{}+ \epsilon^{1/2} \rho^{1/2+ \delta/2} \sum_{\beta_X=\alpha_X, \beta_Z < \alpha_Z} ||�\left(\frac{\rho}{t}\right)^{5/2} \partial_t^2 K^\beta  h ||_{L^\infty}.
\end{eqnarray*}
\end{itemize}
\end{lemma}

\subsection{Second derivatives of the metric and consequences}
In \cite{lm:gsmkg}, an important difficulty comes from the fact that the scaling vector field is not part of the algebra of commuting vector fields. Since we do here commute with $S$, the estimates on the second derivatives of the metric are a direct consequence of the basic decay estimates and the usual decomposition $\partial_{t,x}= \frac{1}{u}a_\alpha Z^\alpha$, see Appendix \ref{se:dect}.

Thus, we have directly,
\begin{lemma}
\begin{itemize}
\item For any $|\alpha|\le N-3$ and any $(t,x) \in \Kcal \cap \{ \rho \ge 2 \}$,
\begin{eqnarray*}
|\partial_{t,x}^2 K^\alpha h|(t,x) &\lesssim& \epsilon \rho^{\delta/2} \frac{1}{t (1+u)^{-3/2}}\\
&\lesssim& \epsilon \rho^{\delta/2} \frac{t^{1/2}}{t \rho^{3/2}}.
\end{eqnarray*}
\item For any $|\alpha|\le N-1$ and any $\rho \ge 2$,
$$
||\frac{\rho^3}{t^2}\partial_{t,x}^2 K^\alpha h||_{L^2(H_\rho^\star)} \lesssim \epsilon \rho^{\delta/2}.
$$
\end{itemize}
\end{lemma}

This leads to the improved estimates for the commutators $[ K^\alpha, h^{\mu \nu} \partial_{x^\mu} \partial_{x^\nu}]h$.

\begin{lemma}\cite[Lemma 8.6]{lm:gsmkg}
Under the bootstrap assumptions,
\begin{itemize}
\item For any $|\alpha| \le N-4$,
$$
\left| [ K^\alpha, h^{\mu \nu} \partial_{x^\mu} \partial_{x^\nu}]h\right| \lesssim \epsilon^{1/2} t^{-2} \rho^{-1+\delta} + \epsilon^{1/2} t^{-1/2} \rho^{-3+\delta}.
$$
\item For any $|\alpha| \le N$,
\begin{eqnarray*}
\left|\left| \rho  [ K^\alpha, h^{\mu \nu} \partial_{x^\mu} \partial_{x^\nu}]h\right|\right|_{L^2(H_\rho^\star)} &\lesssim& \epsilon^{1/2} \rho^{-1/2+ 3/2\delta} + \epsilon^{1/2}
\sum_{\beta_X=\alpha_X, \beta_Z < \alpha_Z} ||�\rho^3 t^{-2} \partial_t^2 K^\beta  h ||_{L^2(H_\rho^\star)}.
\end{eqnarray*}
\end{itemize}
\end{lemma}

\subsection{Improved $L^\infty$ estimates for transversal derivatives}
As in \cite[Proposition 9.1]{lm:gsmkg}, we have
\begin{lemma}
We have the improved $L^\infty$ estimates
\begin{eqnarray*}
| \partial_u K^\alpha \partial_{t,x} \hu_{a \beta}| &\lesssim& \epsilon^{1/2} t^{-1+C\epsilon}, \quad |\alpha| \le N-4, \\
| \partial_u \partial_{t,x}^\gamma \hu_{a \beta}| &\lesssim& \epsilon^{1/2} t^{-1}, \quad |\gamma| \le N-4.
\end{eqnarray*}
\end{lemma}
\begin{proof}
The proof is similar to that of \cite[Proposition 9.1]{lm:gsmkg}, replacing the estimates on the Klein-Gordon field $\phi$ by those on the Vlasov field, since these terms have the exact same decay, according to Proposition \ref{prop:dssfl}.
\end{proof}

These improved estimates imply
\begin{lemma}\cite[Lemma 8.6]{lm:gsmkg}
For any multi-index $\alpha$ with $|\alpha| \le N$, we have
$$
|| K^\alpha P ||_{L^2(H_\rho^\star)} \lesssim \epsilon^{1/2} \rho^{-1} \Ecal^\star_{N}(h)^{1/2} +\epsilon^{1/2} \rho^{-1+C \epsilon^{1/2}}\Ecal^\star_{N-1}(h)^{1/2} + \epsilon \rho^{-3/2+\delta}.
$$
\end{lemma}

\section{Energy estimate for the metric coefficients} \label{se:eemc}

The purpose of this section is to perform the energy estimate for the Einstein equation in the wave gauge
\begin{equation}
\widetilde {\Box}_g h_{\alpha \beta} = F_{\alpha \beta}(h, \partial h) -  \int_v f \left(2v_\alpha v_\beta + g_{\alpha \beta}\right) d \mu_v.
\end{equation}
We introduce here the notation
$$
S_{\alpha \beta} =  \int_v f \left(2v_\alpha v_\beta + g_{\alpha \beta}\right) d \mu_v.
$$
Dropping the $\alpha, \beta$ indices, we rewrite the wave equations as
$$
\widetilde {\Box}_g h = F-  S. 
$$

\subsection{Basic energy estimate}
The first step is the basic energy estimate (see \cite[Proposition 3.1]{lm:gsmkg}).
\begin{proposition}
\label{es:eesw}
For any multi-index $|\alpha| \le N$, the following energy estimate holds :
\begin{eqnarray*} \label{es:eesw}
\Ecal[K^\alpha h ](\rho) & \lesssim & \Ecal_g[K^\alpha h ](2) +  m^2
+ \int_2^\rho \Ecal[K^\alpha h ](\rho')^{1/2}  || K^\alpha F ||_{L^2(H^\star_{\rho'})}d\rho'
\\
&&\hbox{} +  \int_2^\rho \Ecal[K^\alpha h ](\rho')^{1/2} \|[K^\alpha ,H^{\mu\nu} \del_\mu\del_\nu]h \|_{L^2(H_{\rho'}^\star)}d\rho' \\
&&\hbox{}+  \int_2^\rho \Ecal[K^\alpha h ](\rho')^{1/2} M[K^\alpha h](\rho') \, d\rho'
\\
&&\hbox{} +  \int_2^\rho
\Ecal[K^\alpha h ](\rho')^{1/2} \sum_{|\beta| \le |\alpha|} || K^\alpha (S) \|_{L^2(H_{\rho'}^\star)} d\rho' \\
&&\hbox{}+ \sum_{|\beta| < |\alpha|}\bigg( \int_2^\rho \Ecal[K^\alpha h ](\rho')^{1/2} || K^\beta F ||_{L^2(H^\star_{\rho'})}d\rho'
\\
&&\hbox{} +  \int_2^\rho \Ecal[K^\alpha h ](\rho')^{1/2} \|[K^\beta ,H^{\mu\nu} \del_\mu\del_\nu]h \|_{L^2(H_{\rho'}^\star)}d\rho' \\
&&\hbox{} +  \int_2^\rho
\Ecal[K^\alpha h ](\rho')^{1/2} || K^\beta(S) \|_{L^2(H_{\rho'}^\star)} d\rho'\bigg).
\end{eqnarray*}
in which $M [K^\alpha h](\rho)$ is a positive function such that
\begin{align*}
\int_{H_\rho^\star}(\rho/t) \big|\del_{\mu}g^{\mu\nu} \del_{\nu} \big(K^\alpha h_{\alpha \beta} \big) \del_t\big(K^\alpha h_{\alpha \beta} \big)- \frac{1}{2} \del_tg^{\mu\nu}& \del_{\mu} \big(K^\alpha h  \big) \del_{\nu} \big(K^\alpha h  \big) \big| \, dx
\\
\quad &\le M [K^\alpha h](\rho). \Ecal[K^\alpha h ](\rho')^{1/2} .
\end{align*}
\end{proposition}
\begin{remark}
The above lemma implies in particular, the slightly weaker estimate
\begin{eqnarray*} \label{es:eesw}
\Ecal[K^\alpha h ]^{1/2}(\rho) & \lesssim & \Ecal_g[K^\alpha h ](2)^{1/2} +  m
+ \sum_{|\beta| \le |\alpha|} \bigg( \int_2^\rho  || K^\beta F ||_{L^2(H^\star_{\rho'})}d\rho'
\\
&&\hbox{} +  \int_2^\rho \|[K^\beta ,H^{\mu\nu} \del_\mu\del_\nu]h \|_{L^2(H_{\rho'}^\star)}d\rho' \\
&&\hbox{} +  \int_2^\rho || K^\beta (S) \|_{L^2(H_{\rho'}^\star)} d\rho' \bigg) \\
&&\hbox{}+  \int_2^\rho  M[K^\alpha h](\rho') \, d\rho'.
\end{eqnarray*}
The reason why we keep the above estimate is that at top order, we will need to be a bit more careful to apply a Gr\"onwall type inequality.
\end{remark}

As in \cite{lm:gsmkg}, Lemma 7.3, we have, as a direct consequence of the bootstrap assumptions \eqref{eq:bsm},
\begin{lemma}[Estimates for $M$]\label{lem:esM}
For any multi-index $|\alpha| \le N$, we have
$$
M[K^\alpha h]\lesssim \epsilon \rho^{-3/2 +\delta}. 
$$
\end{lemma}

\subsection{Improved estimates for $\Ecal^\star_{N-4}[h]$}
For $|\alpha|\le N-3$, it follows from Proposition \ref{prop:dssfl}
that the contribution of $K^\alpha S[f]$ to the above energy estimates is integrable, since its $L^2$ norm decays like $\rho^{D\delta-3/2}$. Moreover, we have
\begin{lemma}\cite[Lemmas 10.1 and 10.2 ]{lm:gsmkg}
\begin{itemize}
\item For any $\alpha$ with $|\alpha| \le N$,
$$
|| K^\alpha F ||_{L^2(H^\star_\rho)} \lesssim \epsilon \rho^{-3/2+\delta}+ \epsilon \rho^{-1} \Ecal^\star_{|\alpha|}(\rho)^{1/2}+ \epsilon \rho^{-1+D \epsilon^{1/2}} \Ecal^\star_{|\alpha|-1}.(\rho)^{1/2}.
$$
\item For any $\alpha$ with $|\alpha| \le N-4$,
$$
|| [K^\alpha, h^{\mu \nu}] \partial_{\mu}\partial_ {\nu} h  ||_{L^2(H^\star_\rho)} \lesssim \epsilon \rho^{-3/2+\delta}+ \epsilon \rho^{-1} \Ecal^\star_{|\alpha|}(\rho)^{1/2}+ \epsilon \rho^{-1+D \epsilon^{1/2}} \Ecal^\star_{|\alpha|-1}.(\rho)^{1/2}.
$$
\end{itemize}
\end{lemma}

As a consequence,

\begin{proposition}\cite[Proposition 10.3]{lm:gsmkg}
We have the improved estimates
$$
\Ecal_{N-4}[h](\rho) \le D/2 \epsilon \rho^{D \epsilon^{1/2}}.
$$
\end{proposition}

\section{Improved decay estimates for the inhomogenous wave equation} \label{se:ideiwe}

First, in view of the improved energy estimates for the metric coefficients up to $N-1$, we have, by standard Klainerman-Sobolev inequality, the following improvement upon \eqref{es:bde}-\eqref{es:bde2}.
 \begin{proposition}
One has the decay estimates, for all $|\alpha | \le N-2$, and $\rho \ge 2$,
\begin{eqnarray} \label{es:ibde}
| K^\alpha \partial h |(t,x)+| \partial Z^\alpha  h |(t,x) &\lesssim& \epsilon^{1/2} \rho^{D\epsilon^{1/2}} \frac{1}{t(1+u)^{1/2} }, \\
| K^\alpha h |(t,x) &\lesssim& \epsilon^{1/2} \rho^{D\epsilon^{1/2}} \frac{(1+u)^{1/2}}{t}. \label{es:ibde2}
\end{eqnarray}
\end{proposition}

We now rewrite the wave equation for the metric components in terms of the flat wave operator, under the form
$$
\square K^\alpha h= \sum_{|\beta| \le |\alpha|} C^\alpha_\beta \left( -K^\beta \left( H^{\mu\nu} \right)+ K^\beta F - K^\beta S[f]\right),
$$
where the sum appears because of the possible commutation with the scaling vector field and the idendity $[\square, S]=2 \square$.

From standard $L^\infty$ estimates for solutions to the inhomogeneous wave equation (see for instance \cite[Proposition 3.10]{lm:gsmkg}), we have

\begin{proposition}
Assume that we have the pointwise estimates
\begin{eqnarray*}
 \sum_{|\beta|\le |\alpha| }\left( | K^\beta (H^{\mu \nu} \partial_{\mu \nu} h)| + | K^\beta F| + | K^\beta S[f]| \right) \lesssim \frac{\epsilon}{\mu|\nu|}t^{-2-\nu}(t-r)^{-1+\mu}.
\end{eqnarray*}

Then, for $0<\mu\leq 1/2$ and $0<\nu\leq 1/2$, one has
\begin{equation}
| K^\alpha h (t, x)| \leq \frac{\epsilon}{\mu|\nu|}t^{-1}(t-r)^{\mu- \nu} + \epsilon^{1/2} t^{-1},
\end{equation}
while, for $0<\mu\leq 1/2$ and $-1/2\leq \nu<0$,
\begin{equation}
| K^\alpha h (t, x)| \leq \frac{\epsilon}{\mu|\nu|}t^{-1- \nu}(t-r)^{\mu} + \epsilon^{1/2} t^{-1}.
\end{equation}
\end{proposition}

As a consequence, we have

\begin{proposition} \label{prop:idmc}
For any $\alpha$ with $|\alpha| \le N-3$, we have


 $$
 |K^\alpha h| \lesssim  \dfrac{\rho^{D\delta}\epsilon}{\delta ^2 t} +  \dfrac{\epsilon^{1/2}}{t},
  $$
  and
  $$
 |\partial K^\alpha h| \lesssim \dfrac{\rho^{D\delta}\epsilon}{\delta ^2(1+|u|) t}+\dfrac{\epsilon^{1/2}}{t(1+u)}.
  $$
  \end{proposition}
  \begin{proof}
  It suffices to apply the above proposition with $\mu=D\delta/4 $, $\nu=-D\delta/4$. Note that the required decay assumptions on the metric terms follows from the bootstrap assumptions, while those on the Vlasov term follows from Proposition \ref{prop:dssfl}.
  \end{proof}
Finally, we note that from the above proposition with $\delta=\epsilon^{1/2}$, we have

\begin{proposition} \label{prop:idmc2}
For any $\alpha$ with $|\alpha| \le N-3$, we have


 $$
 |K^\alpha h| \lesssim  \dfrac{\rho^{D\epsilon^{1/2}}}{ t},
  $$
  and
  $$
 |\partial K^\alpha h| \lesssim \dfrac{\rho^{D\epsilon^{1/2}}}{(1+|u|) t}.
  $$
  \end{proposition}

\section{Improved estimates for $\Ecal^\star_{N-1}[h]$. } \label{se:ien1}
From Proposition \ref{prop:dssfl2}, the error term coming from the Vlasov field can be estimated as
\begin{equation} \label{es:sfd}
|| K^\alpha(S) ||_{L^2(H_\rho^\star)} \lesssim \rho^{-3/2+D\delta}
\end{equation}
and is therefore integrable.

Moreover, we have
\begin{lemma}\cite[Lemma 12.2]{lm:gsmkg}
\begin{itemize}
\item For any $\alpha$ with $|\alpha| \le N-1$,
\begin{eqnarray*}
|| [K^\alpha, h^{\mu \nu} \partial_{\mu} \partial_{\nu}]h ||_{L^2(H_\rho^\star)} &\lesssim& \epsilon^{1/2} \rho^{-1} \Ecal_{N-1}^\star[h](\rho)+ \epsilon^{1/2} \rho^{-1+D\epsilon} \Ecal_{N-2}^\star[h](\rho)\\
&&\hbox{}+ \epsilon^{1/2} \rho^{-3/2+D\delta}.
\end{eqnarray*}
\item For any $\alpha$ with $|\alpha| \le N$,
\begin{eqnarray*}
|| [K^\alpha, h^{\mu \nu} \partial_{\mu} \partial_{\nu}]h ||_{L^2(H_\rho^\star)} &\lesssim& \epsilon^{1/2} \rho^{-1} \Ecal_{N}^\star[h](\rho)+ \epsilon^{1/2} \rho^{-1+D\epsilon} \Ecal_{N-1}^\star[h](\rho) \\
&&\hbox{}+ \epsilon^{1/2} \rho^{-3/2+D\delta}+ ||K^\alpha S[f] ||_{L^2(H_\rho^\star)}.
\end{eqnarray*}
\end{itemize}
\end{lemma}
\begin{proof}
The proof is similar to that of \cite[Lemma 12.2]{lm:gsmkg}, replacing the terms coming from the Klein-Gordon field by the Vlasov field. When $|\alpha| \le N-1$, then the contribution of the Vlasov field is controlled from \eqref{es:sfd}, hence the statement of the Lemma.
\end{proof}
\begin{remark}
We will use the second inequality of the lemma for only in the top order estimate, but we stated it here since the proof is identical expect that we kept the Vlasov terms since we do not have yet estimates on them.
\end{remark}

This leads to the following improved estimate for $\Ecal_{N-1}^\star[h]$.

\begin{proposition}
We have
$$
\Ecal_{N-1}^\star[h](\rho) \lesssim \epsilon \rho^{D \epsilon^{1/2}}.
$$
\end{proposition}
\begin{remark}The analogue in \cite{lm:gsmkg} is contained in their Lemma 12.7. However, their proof is different because of the estimates on the Klein-Gordon field. Note that in \cite[Lemma 12.7]{lm:gsmkg}, the estimate hold up to the top order, but here, we stop at $N-1$, since above the contribution of the Vlasov field must be estimated differently.
\end{remark}
\begin{proof}
This is a simple consequence of Gronwall lemma together with the above estimates on each of the non-linear terms and the energy estimate.
\end{proof}

\section{Improved estimates for the Vlasov field} \label{se:ievf}
In this section, we prove improved estimates for the Vlasov field.

\subsection{Multiplications by powers of $v$ and the $\zz$ weights}  \label{se:multivz}

We consider here the effects of adding additional $v$ or $\zz$ weights. First, we compute

\begin{lemma}[Multiplication by $v$ weights] \label{lem:mvp}
Let $w^0=\sqrt{1+|v|^2}$.
We have
\begin{eqnarray*}
T_g(w^0)&=& - 1/2 v_\alpha v_\beta h^{\alpha \beta}_{,i} \frac{v^i}{w^0} \\
&=& - 1/2 v_\alpha v_\beta \partial_{x^i} \left(\Phi\cdot \Phi \right)\hu^{\alpha \beta}  \frac{v^i}{w^0} \\
&&\hbox{} - 1/2 \vu_0 \vu_0 h^{00}_{,i} \frac{v^i}{w^0}
  -  \vu_a \vu_0 h^{a0}_{,i} \frac{v^i}{w^0}
   - 1/2 \vu_a \vu_b h^{ab}_{,i} \frac{v^i}{w^0}.
\end{eqnarray*}
In particular,
\begin{equation}\label{eq:tgw0}
|T_g(w^0)| \lesssim \epsilon^{1/2} \rho^{\delta-3/2}  \left( w^0 \frac{\rho}{t}+\frac{t}{\rho}\frac{|\vu_a|^2}{w^0} \right).
\end{equation}
\end{lemma}

Since the above error is integrable, it follows that for any distribution function $k$ for which we can prove the energy estimate, we can also prove weighted energy estimate, replacing $k$ by $(1+|v|^2)^{q/2}k$ (provided the initial data decays fast enough in $v$ naturally). In other words, we have, using equation \eqref{eq:tgw0} and Gr\"onwall's lemma,
\begin{lemma}For any regular distribution function $k$, satisfying $T_g(k)=F[k]$.
\begin{eqnarray}
E[(1+|v|^2)^{q/2}k](\rho_2)&\lesssim& E[(1+|v|^2)^{q/2}k](\rho_1)\nonumber \\
&&\hbox{}+ \int_{\rho_1 \le \rho \le \rho_2}d\rho \int_{H_\rho} \left( \int_v (1+|v|^2)^{q/2}|F[k]|dv\right)\frac{\rho}{t}r^2 dr d\omega_{S^2}.  \label{es:eemvf}
\end{eqnarray}
\end{lemma}

Recall that in Section \ref{se:hovf}, we proved bounds on $E_{N}[f]$ using only the pointwise estimates \eqref{es:bde}, \eqref{es:ide1}, \eqref{es:ide2}, for the $h$ metric coeffcients and the estimates \eqref{es:bce} for the $C$ coefficients as well as estimates on $E_{N-3,2}[f]$. In view of the previous lemma, the following corollary holds.


\begin{corollary}
\begin{eqnarray*}
E_{N-2, q+2}[f](\rho)&\lesssim& \epsilon \rho^{M_N \delta}, \\
E_{N,q}[f](\rho) &\lesssim& \epsilon \rho^{M_N \delta}.
\end{eqnarray*}
\end{corollary}



Eventhough we do not need them in order to prove our main result, let us also explain below how we can consider weighted energy norms with $\zz$ weights.
Recall first the definition of the $\zz$ weights,
$$\zz_i:= v_it-x^i w^0= t \vu_i.$$

We have
\begin{lemma}[Multiplication by $\zz$ weights.]
Let $\zz$ denote one of the weights $$\frac{\zz_i}{w^0}=x_i - \frac{t v_i}{w^0}.$$
Then,
$$
|T_g[ \frac{\zz}{w^0}]| \le  \frac{u^{1/2}}{t} w^0.
$$
Moreover, we have the decomposition \begin{eqnarray*}
T_g\left(\frac{\zz}{w^0}\right)&=& (v_0-w_0) g^{0\beta} \partial_{x^\beta} \zz + \wu_\alpha \Hu^{\alpha \beta} \delu_{x^\beta} (\zz) \\
&=&(v_0-w_0) g^{0\beta} \partial_{x^\beta} \zz + \wu_0 \Hu^{00} \delu_{0}(\zz)+ \wu_0 \Hu^{0a}\delu_a \zz + \wu_a \Hu^{a\beta} \delu_\beta \zz
\end{eqnarray*}
and the improved estimate
\eq{
\left\vert T_g\left(\frac{\zz}{w^0}\right)\right\vert \lesssim |v_0 -w_0| +  w^0 \rho^{\delta} \frac{u^{3/2}}{t^2} + w^0 \rho^{\delta} \frac{u^{1/2}}{t^2} + |\wu_a| \rho^{\delta} \frac{u^{1/2}}{t}.
}
\end{lemma}

As for the $v$ weights, we can prove weighted energy estimates, using $\zz$ as a weight. However, note that $T_g(\zz)$ is not integrable. Thus, we consider instead $\frac{\zz}{(1+u)^{1/2+\delta}}$. One can then prove boundedness of any weighted energy, replacing $\widehat{K}^\alpha(f)$ by $(\frac{|\zz|}{(1+u)^{1/2+\delta}})^{q}\widehat{K}^\alpha(f)$.

\subsection{Improved estimates for the $C$ coefficients}.

Recall the basic structure of the equations satisfied by $C$
$$
T_g(C)= -F_{ZB}.
$$

In Section \ref{se:hocc} concerning higher order estimates for $C$, we proved that $F_{ZB} u^{-1/2}$ was almost integrable (cf.~\eqref{es:bce}).
In view of the improved estimates on the metric coefficients in Proposition \ref{prop:idmc}, we have a gain of $u^{-1/2}$ compared to the estimates of Section \ref{se:hocc}. For instance, repeating the previous on the $C$ coefficients with the improved estimates in Proposition \ref{prop:idmc} replacing the decay estimates \eqref{es:bde}-\eqref{es:bde2}, we have
$$
\frac{| F_{ZB} |}{v^\rho} \lesssim \frac{\rho^{D\delta}\epsilon}{\delta^2 \rho}+ \frac{\epsilon^{1/2}}{\rho},
$$
to be compared with the previous estimate
$$
\frac{| F_{ZB} |}{v^\rho} \lesssim \frac{\rho^{D\delta}\epsilon^{1/2}u^{1/2}}{\rho}.
$$

Thus, we immediately obtain
\begin{lemma}
 For any multi-index $|\alpha| \le N-4$, we have the improved estimates
\begin{eqnarray}
\label{es:iecc}
|\widehat{K}^\alpha(C)| \lesssim\epsilon^{1/2}  \rho^{\delta D}
\end{eqnarray}
for some positive constant $D$.
\end{lemma}

To improve the estimate \eqref{es:ixc} for $\widehat{K}^\alpha(XC)$, we need to make sure that the gain of $u^{-1/2}$ decay can be used effectively. This means that either we need to use the null structure of the equations (which actually gives the strongest results) or we can simply recall that
$$\frac{t}{\rho w^0} \le v^\rho,$$
which implies that we can convert a $u^{-1/2}$ decay into a $t^{1/2}$ decay provided we lose in powers of $v$. For simplicity, we use this estimate, leading to

\begin{lemma}
For any multi-index $|\alpha| \le N- 5$, we have the improved estimates
\begin{eqnarray*}
|\widehat{K}^\alpha(XC)| \lesssim\epsilon^{1/2} (w^0)^2.
\end{eqnarray*}
\end{lemma}
Note that the loss of $w^0$ coming from the right-hand side can be absorbed in any product involving $f$ thanks to the weighted norms $E_{N,q'}[f]$. The improved estimate on $\widehat{K}^\alpha(XC)$ will only be used in the top order energy estimates for the metric coefficients and in particular so as to not get any loss coming from the application of the Klainerman-Sobolev inequality of Lemma \ref{lem:KSf} to the energy-momentum tensor components of $f$.

Finally, we can also use the improved estimates of Proposition \ref{prop:idmc2}, leading to

\begin{lemma}
 For any multi-index $|\alpha| \le N-4$, we have the improved estimates
\begin{eqnarray}
\label{es:iecc2}
|\widehat{K}^\alpha(C)| \lesssim  \rho^{D \epsilon^{1/2}}.
\end{eqnarray}
\end{lemma}
\begin{remark}The above estimate will only be used in the proof of the top order energy estimate for the metric coefficients to estimate velocity averages of terms of the form
$$
\widehat{K}^\alpha(C) \widehat{K}^\beta(C) \widehat{K}^\sigma(f),
$$
where $|\alpha| \le N-4$, so that the above estimate applies, while $|\beta| > N-4$. Since such a term contains two powers of $C$ coefficients, we can afford to lose an $\epsilon^{1/2}$ in the estimate.
\end{remark}

\subsection{Estimates for the Vlasov field without loss}
In view of the improved estimates for the metric coefficients, the stronger weighted $v$ norms for $f$ and the improved estimates on the $C$ coefficients, we can revisit the energy estimates for $K^\alpha (f)$ for $|\alpha|$ sufficiently small so that the improved estimates on $h$ and $K^\beta(C)$, $|\beta| \le |\alpha|-1$ hold, i.e.~for $|\alpha| \le N-4$. One can repeat the previous proof, except that there are no borderline terms anymore.

\begin{lemma}\label{lem:evfi1}
We have the improved estimates
\begin{equation} \label{es:fwlo}
E_{N-4, q+1}[f] \lesssim \epsilon.
\end{equation}
\end{lemma}
\begin{proof}
Recall that we already proved that $$E[X_i f] \lesssim \epsilon.$$ For the $\partial_t$ vector field, recall that the only borderline term in the commutator formula $[T_g, \partial_t]$ was of the form
$w \partial h \cdot X_i(f)$ and that this term lead to a small growth because it did not satisfy the null condition. This term can now be estimated losing one power of $w^0$ as in
$$
| w^0 \partial h \cdot X_i(f)|  \lesssim  \epsilon^{1/2} \rho^{\epsilon/2} w^0 \frac{1}{\rho^{-3/2}  } v^\rho |X_i(f) |. 
$$
This leads to the estimate
\begin{eqnarray*}
E[\partial_t f] &\lesssim& \epsilon+ \epsilon^{1/2}\int_2^\rho (\rho')^{-3/2+D\delta}E[w^0 Xf](\rho') d\rho' \\
&\lesssim& \epsilon.
\end{eqnarray*}

Similarly, we can obtain an estimate on $E[Y f]$ for $Y$ a modified vector field without a loss. In that case, the borderline terms came from the terms

$C^\alpha \partial h \cdot Xf$, but these are now integrable in view of the improved estimates for the $C^\alpha$ coefficients and again possibly losing a power of $w^0$ if necessary.

Using the higher order commutator formula of Section \ref{se:hovf}, we can also prove estimates without a loss for $\widehat{K}^\alpha(f)$, for any $\alpha$ for which we have access to the improved estimate on the metric coefficients. Finally, we can also consider extra $v$ weights, using Lemma \ref{lem:mvp}, losing at most one power of $v$ as in the estimate for $\partial_t$, as in \eqref{es:fwlo}.
\end{proof}

In particular, combined with the Klainerman-Sobolev inequalities for Vlasov fields of Section \ref{se:KSf}, we obtain

\begin{corollary}

For any multi-index $|\alpha| \le N-7$,

$$
\int_v |\widehat{K}^\alpha | (f) (w^0)^{q-3} dv \lesssim \frac{\epsilon}{t^3}.
$$
\end{corollary}

\subsection{Improved estimates for the Vlasov field up to order $N-1$.}
In view of the improved estimates for the metric and the improved estimates for the $C$ coefficients \eqref{es:iecc}, we can revisit the proof of the higher order estimates for $f$, expect that each loss in $\rho^{\delta}$ can now be replaced by a loss in $\rho^{\epsilon^{1/2}}$. Thus, we have

\begin{lemma}\label{lem:evfi2} The following improved estimates hold.
\begin{itemize}
\item The first estimate reads $$
E_{N-1,q}[f](\rho) \lesssim \epsilon \rho^{D \epsilon^{1/2}}.
$$
\item Consider a product of the type $| \widehat{K}^\gamma (C)|^2 \widehat{K}^\alpha (f)$, for $N-2 \ge |\gamma|$ and $|\alpha| \le N-3$. Then, we have
$$
E\left[(w^0)^{q}(1+u)^{-1} |\widehat{K}^\gamma (C)|^2 \widehat{K}^\alpha (f) \right](\rho) \lesssim \epsilon^{1/2} \rho^{D \epsilon^{1/2}}.
$$
\item
Consider a product of the type $|\widehat{K}^\gamma (XC)|^2 \widehat{K}^\alpha (f)$, for $N-3 \ge |\gamma|$ and $|\alpha| \le N-3$. Then, we have
$$
E\left[ (w^0)^q | \widehat{K}^\gamma (XC)|^2 \widehat{K}^\alpha (f) \right](\rho) \lesssim \epsilon^{1/2} \rho^{D \epsilon^{1/2}}.
$$
\end{itemize}
\end{lemma}

\section{The top order estimate for the metric coefficients} \label{se:toemc}
Recall the basic structure of the evolution equation for the metric coefficients,

\begin{eqnarray*}
\widetilde{\square}_g K^\alpha (h)&=&  - [K^\alpha ,H^{\mu\nu} \del_\mu\del_\nu]h + K^\alpha F- K^\alpha S[f]  \\
&&\hbox{}+\sum_{|\beta| < |\alpha|}C^\alpha_\beta \left(  - [K^\beta ,H^{\mu\nu} \del_\mu\del_\nu]h + K^\beta F- K^\beta S[f] \right),
\end{eqnarray*}
where $F=F(h, \partial h)$ denotes non-linearities depending only on $h$ and where the Vlasov terms $K^\beta S[f]$ can be written using the commutation formulae of Section \ref{se:comt}.

Recall also that the energy estimates take the form
\begin{eqnarray}
\Ecal[K^\alpha h ](\rho) & \lesssim & \Ecal_g[K^\alpha h ](2) +  m^2+ \bigg( \int_2^\rho \Ecal[K^\alpha h ](\rho')^{1/2} || K^\alpha F ||_{L^2(H^\star_{\rho'})}d\rho' \nonumber
\\
&&\hbox{} +  \int_2^\rho \Ecal[K^\alpha h ](\rho')^{1/2} \|[K^\alpha ,H^{\mu\nu} \del_\mu\del_\nu]h \|_{L^2(H_{\rho'}^\star)}d\rho' \nonumber\\
&&\hbox{} +  \int_2^\rho
\Ecal[K^\alpha h ](\rho')^{1/2} || K^\alpha(S) \|_{L^2(H_{\rho'}^\star)} d\rho' \nonumber \\
&&\hbox{} +   \int_2^\rho \Ecal[K^\alpha h ](\rho')^{1/2} M[K^\alpha h](\rho') \, d\rho' \bigg) \nonumber \\
&&\hbox{}+ \sum_{|\beta| < |\alpha|}\bigg( \int_2^\rho \Ecal[K^\alpha h ](\rho')^{1/2} || K^\beta F ||_{L^2(H^\star_{\rho'})}d\rho' \nonumber
\\
&&\hbox{} +  \int_2^\rho \Ecal[K^\alpha h ](\rho')^{1/2} \|[K^\beta ,H^{\mu\nu} \del_\mu\del_\nu]h \|_{L^2(H_{\rho'}^\star)}d\rho' \nonumber\\
&&\hbox{} +  \int_2^\rho
\Ecal[K^\alpha h ](\rho')^{1/2} || K^\beta(S) \|_{L^2(H_{\rho'}^\star)} d\rho'\bigg).\label{es:eeswto2}
\end{eqnarray}

In this section, we close the energy estimates for $K^\alpha h$ when $|\alpha|=N$ and in particular improve the top order bootstrap assumption on $h$.

Let us define $S\Ecal_\alpha(\rho)$ as
$$
S\Ecal_\alpha(\rho):= \sup_{\rho' \in [2,\rho],\,\, |\beta| \le |\alpha|}\left[\Ecal[K^\beta h ](\rho') \right].
$$

The results of sections \ref{sec:aree}, \ref{se:eemc} and \ref{se:ien1} already imply that, for $|\beta| \le |\alpha|\le N$,
\begin{eqnarray*}
|M[K^\alpha h](\rho)| &\lesssim& \epsilon^{1/2} \rho^{-3/2+\delta}, \\
\int_2^\rho \Ecal[K^\alpha h ](\rho')^{1/2} || K^\beta F ||_{L^2(H^\star_{\rho'})}d\rho' &\lesssim& \epsilon^{1/2} S\Ecal_\alpha(\rho)^{1/2}  \rho^{D \epsilon^{1/2}}\\
&&\hbox{}+\epsilon^{1/2}  \sum_{|\beta| \le |\alpha|} \int_2^\rho \frac{1}{\rho'} \Ecal[K^\beta h ](\rho') d\rho',
\end{eqnarray*}
\eq{\nonumber\alg{
\int_2^\rho \Ecal[K^\alpha h ](\rho')^{1/2} \|[K^\alpha ,H^{\mu\nu} \del_\mu\del_\nu]h \|_{L^2(H_{\rho'}^\star)}d\rho' &\lesssim \epsilon^{1/2} S\Ecal_\alpha(\rho)^{1/2} \rho^{D \epsilon^{1/2}}\\
&\hbox{}+\epsilon^{1/2} \sum_{|\beta| \le |\alpha|} \int_2^\rho \frac{1}{\rho'} \Ecal[K^\beta h ](\rho') d\rho'.
}}

In order to apply the energy estimate for $K^\alpha h$, it thus remains to estimate $|| K^\alpha S ||_{L^2(H_{\rho})}$. For simplicity, we write $K^N$ (respectively $K^{N-1}$) to denote a differential operator of the form $K^\alpha$ with $|\alpha|=N$ (respectively $|\alpha|=N-1$) below.

We start with the case where $Z^N=X Z^{N-1}$, with $X$ a translation.

According to the top order commutation formula of Lemma \ref{lem: Xtopt}, each term in $X K^{N-1} S$ can be written as $\Tcal [F]$, where $F$ are, modulo multiplications by a function in $\mathcal{F}_{x,v}$, of the form
\begin{eqnarray*}
F= \frac{1}{(1+u)^q}  P_1(C)^{l_1+l_1',r_{Z,1},s_{X,1}} P_2(X(C))^{l_2,r_{Z,2},s_{X,2}}\widehat{K}^\sigma( f) Q^\gamma(L) ,
\end{eqnarray*}
where $r_{Z,2}+s_{X,2} \le N-2$, $r_{Z,1}+s_{X,1}+r_{Z,2}+s_{X,2}+l_2 +|\gamma|+ |\sigma| \le N$, and $2 q\ge   l_1$, $\gamma_X \ge l_1'$, $q \le N$, $l_2 \le N$.

We estimate the above error terms according to the following cases.
\begin{enumerate}
\item Case $|\gamma| \ge N-5$. All quantities apart from $Q^\gamma(L)$ can be estimated pointwise. We then have
\begin{eqnarray*}
| F | &\lesssim& \rho^{D' \delta} | \widehat{K}^\sigma( f)| |K^{\gamma'}(h) |,
\end{eqnarray*}
with $|\gamma'|=|\gamma|$.
After integration in $v$ and in $L^2_{H_\rho}$, we can estimate the contribution of this term by $\epsilon \rho^{D'' \delta-2}$, which is integrable for some constant $D''$.
\item Case $|\sigma| \ge N-5$. All quantities apart from $\widehat{K}^\sigma (f)$ can then be estimated pointwise. We then have
\begin{eqnarray*}
| F | &\lesssim& \rho^{D' \delta} | \widehat{K}^\sigma( f)|.
\end{eqnarray*}
The $L^2$-decay estimates for this term are then given by Proposition \ref{prop:dssfl2}. 
\item Case $r_{Z,1}+s_{X,1} \ge  N-5$. In that case, there is at most one $C$ coefficient in $P_1(C)^{l_1,r_{Z,1},s_{X,1}}$ which can not be estimated pointwise. We then have an estimate of the form
$$
| F |  \lesssim  \rho^{D' \delta} | (1+u)^{-1/2} \widehat{K}^\gamma(C) \widehat{K}^\sigma( f)|.
$$
After integration in $v$ and and application of the Cauchy-Schwarz inequality, we estimate its contribution by
$$
\int_v |F| w^0 dv \lesssim \left( \int_v  | \widehat{K}^\gamma(C)|^2  |\widehat{K}^\sigma( f)| w^0 dv \right)^{1/2} \epsilon^{1/2} \frac{1}{(1+u)^{-1/2}}\cdot \frac{1}{t^{3/2}}.
$$
After integration in $L^2_{H_\rho}$, we can thus estimate this term by
$$
E \left [(1+u)^{-1} |\widehat{K}^\gamma(C)|^2 \widehat{K}^\sigma (f) \right](\rho)^{1/2} \frac{\epsilon}{\rho^{3/2}} \lesssim \epsilon^{3/2} \rho^{-3/2+D\delta},
$$
where we have used Proposition \ref{prop:p2ee}. The contribution of this term is therefore integrable.
\item Case $r_{Z,2}+s_{X,2} \ge  N-5$. This case is similar to the above, using again Proposition \ref{prop:p2ee}.
\end{enumerate}

We thus have the estimate

$$|| X K^{N-1} S  ||_{L^2(H_\rho)} \lesssim \epsilon \rho^{-3/2+D\delta }, $$
using only on the bootstrap assumptions. It follows that this term is integrable, and we have obtained 
\begin{lemma} For any multi-index $|\alpha| \le N-1$ and any translation $X$, we have
$$
\Ecal [ X K^{\alpha} h ] \lesssim \epsilon \rho^{D\epsilon^{1/2}}.
$$
\end{lemma}

We then consider the case where $K^N$ is of the form $Z^N$ and focus again on the source term $Z^{N}S$. In that case, we use the commutator formula for top order given by lemmas \ref{lem:comhgt}, \ref{lem:bta} and the discussion before Lemma \ref{lem:bts}. According to these lemmas, the error terms then split into the $GT$ terms, the $BT_{|\alpha|}$ terms, the $BT_{< |\alpha|, |\sigma| < |\alpha|}$ and the $BT_{< |\alpha|, |\sigma|=|\alpha|}$, where
\begin{enumerate}
\item The $GT$ can be estimated as above. In particular, their contribution is integrable.
\item The $BT_{|\alpha|}$ are given by
\begin{enumerate}
\item[1.] $\frac{1}{(1+u)^q }C^{l_1} K^\beta(C) X(f) L$, with $|\beta| =|\alpha|$ and $2q\ge l_1$,
\item[2.]  $C K^\beta(X C) X(f) L$,
with $|\beta| =|\alpha|-1$,
\end{enumerate}
and where only the terms in 1. ,with $q=0$, give borderline terms.
\item The $BT_{< |\alpha|,  |\sigma| < |\alpha|}$ terms are given by
\begin{eqnarray*}
 \frac{1}{(1+u)^q}  P_1(C)^{l_1+l_1',r_{Z,1},s_{X,1}} P_2(X(C))^{l_2,r_{Z,2},s_{X,2}}\widehat{K}^\beta(C) \widehat{K}^\sigma(X f) Q^\gamma(L) ,
\end{eqnarray*}
where $r_{Z,2}+s_{X,2} \le |\alpha|-1$, $r_{Z,1}+s_{X,1}+r_{Z,2}+s_{X,2}+l_2 +|\gamma|+|\beta|+|\sigma| \le |\alpha|$, $|\beta| \le |\alpha|-1$, and $2 q\ge  l_1$, $q \le |\alpha|$, $\gamma_X\ge l_1'$, $l_2 \le |\alpha|$, $|\sigma| < |\alpha|$ and no $C$ coefficients in the formula is hit by more than $|\alpha|-1$ vector fields.
\item The $BT_{< |\alpha|,  |\sigma| = |\alpha|}$ terms are of the form
$$
 \frac{1}{(1+u)^q}  C^{l_1+1} \widehat{K}^\sigma(X f) L ,
 $$
 where $2 q \ge l_1$. Since these terms do not contain any high derivatives on the $C$ coefficients, we can estimate them using the improved decay \eqref{es:iecc} and these terms are therefore integrable.
\end{enumerate}

For the $BT_{< |\alpha|,  |\sigma| < |\alpha|}$, we can estimate them without analysing their structure in further detail as follows.
\begin{lemma} The $BT_{< |\alpha|,  |\sigma| < |\alpha|}$ terms verify
$$
|| BT_{< |\alpha|,  |\sigma| < |\alpha|} ||_{L^2(H_\rho)} \lesssim \epsilon \rho^{D \epsilon^{1/2}-1}.
$$
\end{lemma}
\begin{proof}
First, when all indices are sufficiently small, then we have access to pointwise estimates on all quantities and the resulting terms are integrable in view of the improved estimates on the $C$ coefficients and on the metric coefficients.

Now, if
\begin{itemize}
\item $|\gamma| > N-5$. Then we have access to pointwise estimates on each term apart from $Q^\gamma(L)$. Thus, we can estimate the term by
$$
\rho^{D \delta} |\widehat{K}^\sigma(X f)|\cdot | Z^{\gamma'} (h)|,
$$
where $|\gamma'| \le |\gamma|$ and $D$ is some constant.
 One can then estimate the resulting terms crudely, using that $Z^{\gamma'} (h)$ can either be estimated pointwise for $|\gamma'| < N -2$, or otherwise that
 $  |Z^{\gamma'} (h)| \lesssim t | \partial Z^{\gamma''}(h)|$, where $|\gamma''|=|\gamma'|-1$, together with
 \begin{eqnarray*}
\left\Vert\int_v w^0 \widehat{K}^\sigma(X f) Z^{\gamma'} (h) \right\Vert_{L^2(H_\rho)} &\lesssim& \left\Vert \left \Vert \int_v w^0 |\widehat{K}^\sigma(X(f)| dv \right\Vert_{L^\infty} |Z^{\gamma'} (h)| \right \Vert_{L^2(H_\rho)} \\
&\lesssim& \left\Vert \epsilon\rho^{D'\delta} t^{-2} \left(t^{-1}Z^{\gamma'} (h)\right) \right \Vert_{L^2(H_\rho)} \\
&\lesssim& \epsilon\rho^{D'\delta -2} \left\Vert   \partial Z^{\gamma''} (h) \right \Vert_{L^2(H_\rho)} \\
 &\lesssim& \epsilon^{3/2} \rho^{D\delta-2}.
 \end{eqnarray*}
 \item The $P_1(C)$ term contains $Y^\beta(C), |\beta| \ge N-5$. For any $C$ coefficient that can be estimated pointwise, we use \eqref{es:iecc2}, so that we can estimate the error term as
 $$
 \rho^{D \epsilon^{1/2}} |Y^\beta(C)| |\widehat{K}^\sigma(X f)|,
 $$
 with $|\sigma| \le 5$. By Cauchy-Schwarz and the Klainerman-Sobolev inequality,
 \begin{align*}
\Big\Vert\int_v w^0 |Y^\beta(C)|& |\widehat{K}^\sigma(X f)| dv \Big\Vert_{L^2(H_\rho)}\\
 &\lesssim \epsilon^{1/2} \rho^{-1+D\epsilon} E[(1+u)^{-1} |Y^\beta(C)|^2 |\widehat{K}^\sigma(X f)| ]^{1/2}(\rho) \\
&\lesssim \epsilon \rho^{-1+D\epsilon^{1/2}} .
\end{align*}
  \item The $P_2(C)$ term contains $Y^\beta(X C), |\beta| \ge N-4$. This can be estimated similarly to the above, except that we can use the improved estimated to bound all the other $C$ coefficients, so that this term is integrable.
  \item $|\sigma| \ge N-5$. In this case, all the $C$ coefficients can be estimated pointwise using the improved estimates. Applying the $L^2$-decay estimates for the Vlasov field, the contribution of this term is then integrable.
   \end{itemize}

\end{proof}

Thus, we are left with estimating the $BT_{|\alpha|}$ terms with $q=0$, i.e.~we need to estimate the $L^2$-norm of $\int_v \vert\widehat{K}^{N-1}(C)\cdot X(f) w^0\vert dv$. We proceed as follows.
First,
$$
 \int_v \left\vert\widehat{K}^{N-1}(C)\cdot X(f) w^0 \right\vert dv \lesssim \epsilon^{1/2} \left( \int_v | \widehat{K}^{N-1}(C)|^2 |X(f)| \frac{\rho}{t}w^0 \right)^{1/2} \frac{1}{\rho^{3/2}},
$$
using the improved decay estimates for the velocity averages of $X(f).$

This gives

\begin{equation}\label{eq:KXf1}
\left|\left | \int_v \widehat{K}^{N-1}(C)\cdot X(f) \right|\right|_{L^2_{H_\rho}}\lesssim \epsilon^{1/2} E[|\widehat{K}^{N-1}(C)|^2|X(f)|]^{1/2} \frac{1}{\rho^{3/2}}.
\end{equation}

To estimate the right-hand side of this equation, we prove below
\begin{lemma} \label{lem:com2nl}
$$
E[\widehat{K}^{N-1}(C)|^2X(f)](\rho) \lesssim \epsilon \rho \sup_{ \rho' \in [2, \rho]} \Ecal [ Z^N(h)(\rho')] + \epsilon \rho^{1+D\epsilon^{1/2}}, 
$$
for some universal constant $D\ge 0$.
\end{lemma}

Before proving the above lemma, let us show how it implies the improved estimate at top order.

Combined with the inequality \eqref{eq:KXf1} (and $\sqrt{a+b}\le \sqrt{a}+\sqrt{b}$, for  $a,b\ge 0$), the lemma implies that

$$
\left| \left| \int_v \widehat{K}^{N-1}(C)\cdot X(f) w^0 dv \right| \right|_{L^2_{H_\rho}}\lesssim \frac{\epsilon}{\rho^{1-D\epsilon^{1/2}}} +\epsilon \rho^{-1} \sup_{ \rho' \in [2, \rho]} \Ecal [ Z^N(h)(\rho')]^{1/2}.
$$

Combined with the energy estimate \eqref{es:eeswto2} for $Z^N(h)$, we obtain

\begin{eqnarray*}
\Ecal [Z^N(h)] (\rho) &\lesssim& \epsilon^{1/2} \rho^{D\epsilon^{1/2}}\left(1+\sup_{ \rho' \in [2, \rho]} \Ecal [ Z^N(h)(\rho')]^{1/2} \right)  \\
&&\hbox{}+  \sup_{ \rho' \in [2, \rho]} \Ecal [ Z^N(h)(\rho')]^{1/2} \int_2^\rho \epsilon (\rho')^{-1} \Ecal [Z^N(h)(\rho')]^{1/2}]  d\rho'.
\end{eqnarray*}

This implies that
$$
\sup_{\rho' \in [2, \rho]} \Ecal [Z^N(h)(\rho')]^{1/2}  \lesssim \epsilon^{1/2} \rho^{D'\epsilon^{1/2}} + \int_2^\rho \epsilon (\rho'^{-1})  \sup_{ \rho' \in [2, \rho]} \Ecal [ Z^N(h)(\rho')]^{1/2} d\rho',
$$
and hence, by Gronwall's lemma,
$$
\sup_{\rho' \in [2, \rho]} \Ecal [Z^N(h)(\rho')]^{1/2}  \lesssim \epsilon^{1/2} \rho^{D''\epsilon^{1/2}}, 
$$
for some universal constant $D''\ge 0$, which improves the bootstrap assumptions for $Z^N(h)$.\\

It remains to prove Lemma \ref{lem:com2nl}.

\begin{proof}Let $\beta$ be a multi-index with $|\beta|=N-1$. We consider the energy estimate \eqref{es:eevf} for $|\widehat{K}^\beta(C)|^2 X(f)$. The initial data at $\rho=2$ verify (see Remark \ref{rem:dataC})
$$
E\left[|\widehat{K}^\beta(C)|^2 X(f)\right](2) \le \epsilon^{2}.
$$
We compute

\begin{eqnarray*}
T_g \left ( |\widehat{K}^\beta(C)|^2 X(f) \right)&=& 2T_g\left( \widehat{K}^\beta(C) \right)\cdot \widehat{K}^\beta(C) X(f)+ |\widehat{K}^\beta(C)|^2 T_g(X(f)) \\
&=&2 [T_g, \widehat{K}^\beta](C)\cdot \widehat{K}^\beta(C) X(f)\\
&&\hbox{}+ 2\widehat{K}^\beta \left(T_g(C) \right) \cdot  \widehat{K}^\beta(C)\cdot X(f)+ |\widehat{K}^\beta(C)|^2 T_g(X(f)) \\
&=& H1+H2+H3,
\end{eqnarray*}
where
\begin{eqnarray*}
H1&=& 2[T_g, \widehat{K}^\beta](C)\cdot \widehat{K}^\beta(C) X(f),\\
H2&=&  2\widehat{K}^\beta \left(T_g(C) \right) \cdot  \widehat{K}^\beta(C)\cdot X(f), \\
H3&=& |\widehat{K}^\beta(C)|^2 T_g(X(f)).
\end{eqnarray*}

For $H3$, we use that, in view of the improved decay estimates for the metric coefficients and the $C$ coefficients, we have
$$
|T_g(X(f))| \lesssim \epsilon \rho^{-2-D\delta} |Y(f)|\left( w^0 \frac{\rho}{t} +\frac{t}{\rho} \frac{|\vu_i|^2}{w^0} \right).
$$
Since moreover, it follows from Proposition \ref{prop:p2ee} that
$$E\left( (1+u)^{-1}  |\widehat{K}^\beta(C)|^2 Y(f) \right) \le \epsilon \rho^{M \delta},$$
we see that the contribution of $H3$ is almost integrable and can therefore be discarded, since we can allow any growth less than $\rho$.

For $H1$, we use the higher order commutator formula of lemma \ref{lem:mthc}. The main terms, according to \eqref{eq:mthc}, are then of the form
\begin{eqnarray*}
P(C)^{k, r_Z,s_X} w \cdot \partial_{t,x} K^\sigma(h) \widehat{K}^\gamma(C)\cdot \widehat{K}^\beta(C) X(f),
\end{eqnarray*}
where $r_Z+s_X+|\sigma|+|\gamma| \le |\beta|+1$, $|\sigma| \le |\beta|$, $1 \le |\gamma| \le |\beta|$, $r_Z+s_X \le |\beta|-1$  
and the number of $C$ coefficients, $k$ satisfies either $C1$ or $C2$. Recall that in view of the improved decay for low derivatives of Lemmas \ref{lem:evfi1} and \ref{lem:evfi2} (and the strong decay of velocity averages of $X(f)$), we can neglect here the null structure. As usual, we consider different cases depending on the range of the indices.

\begin{enumerate}
\item Case $|\sigma|\le N-3$.  We have access to the improved decay for $\partial K^\sigma(h)$ of Proposition \ref{prop:idmc}. These terms can therefore by estimated by
$$
\epsilon^{1/2} \rho^{D \epsilon^{1/2} } \frac{1}{t(1+u)^{1-D \delta+\sigma_X}}P(C)^{k, r_Z,s_X} w \cdot\widehat{K}^\gamma(C)\cdot \widehat{K}^\beta(C) X(f).
$$
If now $k$ verifies $C1$, using the weak bounds from Proposition \ref{prop:p2ee}, after integration in $v$ and on $H_\rho$, we can control this error term by $\epsilon \rho^{D\delta+M\delta } \le \epsilon \rho^{1+D\epsilon}$, for $\delta$ small enough.

Similarly, if $k$ verifies $C2$, using the weak bounds from Proposition \ref{prop:p2ee}, after integration in $v$ and on $H_\rho$, we can control this error term by $\epsilon \rho^{D\delta+1/2+M\delta} \le \epsilon \rho^{1+D\epsilon^{1/2}}$. 

\item Case $|\sigma| \ge N-3$, then we have access to pointwise estimates of Proposition \ref{prop:lowcp} on the $C$ coefficients apart from the $\widehat{K}^\beta(C)$.


Using, Proposition Proposition \ref{prop:p2ee}, this implies that these terms can be estimated by
\begin{eqnarray*}
&&\int_{H_{\rho}} \epsilon^{1/2} (\rho')^{M\delta} \Ecal[\partial K^\sigma(h)]^{1/2}\cdot \left|\left |\int_v w^0 |\widehat{K}^\beta(C)|^2 X(f)) dv \right|\right|_{L^2_{H_\rho} } d\rho'\\
&&\lesssim  \epsilon \int_\rho (\rho')^{M'\delta} E[|\widehat{K}^\beta(C)|^2 X(f) ]^{1/2} \frac{1}{\rho'{}^{3/2}}d\rho' \\
&&\lesssim  \epsilon^{3/2} \rho^{M'\delta-1/2}\\
 &&\lesssim  \epsilon \rho^{1+D\epsilon^{1/2}}. 
\end{eqnarray*}

\end{enumerate}

The other terms in $H1$ (the cubic and the frame terms) have similar regularity and better decay, so their contribution can be estimated again by $\epsilon \rho^{1+D\epsilon^{1/2}}$.

It thus remains to consider the terms coming from $H2$. Recall that we can write schematically,
$$
T_g( C) = t \partial Z(h) w,
$$
so that $K^\beta (T_g(C))$ (see Lemma \ref{lem:hoccs}) can be written as a sum of terms of the form
$$
w\cdot P(C)^{k,r_Z,s_X}t \partial K^\mu (h),$$
with the range of indices
$$
 |\mu|+r_Z+s_X \le 1+ |\beta|, \quad r_Z+s_X \le |\beta|-1, \quad k \le \mu_X, \quad \mu \le |\beta|+1.
$$

\begin{enumerate}
\item Case $|\mu|< N$. First if $k>0$, then it follows that $\mu_X>0$. In that case, we can use the extra $u$ decay coming from replacing translations by $Z$ vector fields, and the resulting terms are then subleading and easily seen to be integrable. On the other hand,  the contribution of the terms with $k=0$ can be estimated by
\begin{eqnarray*}
&&\int_2^\rho\int_{H_{\rho'}}\int_v  t \partial K^\mu (h) \widehat{K}^\beta(C)X(f) w dv d\mu_{H_\rho} d\rho'  \\
&&\hbox{}\lesssim \epsilon^{1/2} \int_2^\rho \Ecal_{N-1}(h)^{1/2}(\rho') E\left( |\widehat{K}^\beta(C)|^2 X(f) \right)^{1/2} (\rho')^{-1/2} d\rho \\
&&\hbox{}\lesssim \epsilon^{1/2}\int_2^\rho E\left( |\widehat{K}^\beta(C)|^2 X(f) \right)^{1/2} (\rho')^{-1/2+D\epsilon^{1/2}} d\rho.
\end{eqnarray*}
Since this term is sublinear in $E\left( |\widehat{K}^\beta(C)|^2 X(f) \right)$, it can then be absorbed on the left-hand side.
\item Case $|\mu|=N$.
Again, we focus only the case where $k=0$, since otherwise the extra $u$ decay coming from $\mu_X>0$ makes these terms subleading.
Thus, we need to estimate the contribution of
$$
D:= w\cdot t\cdot \partial Z^N(h)\cdot K^\beta(C) \cdot X(f).
$$
We proceed as above and after integration, we have
\begin{eqnarray*}
 \int_\rho d\rho' \int_{H_\rho'}d\mu_{H_\rho} \int_v dv |D|\lesssim \epsilon^{1/2} \int_2^{\rho}   \Ecal(Z^N(h))^{1/2} E[K^\beta(C)^2 X(f) ]^{1/2}(\rho')^{-1/2} d\rho',
 \end{eqnarray*}
where we have used the Cauchy-Schwarz inequality in $v$ again as above and the decay estimates for the velocity averages of $X(f)$.
\end{enumerate}
 Let $y(\rho)=E[K^\beta(C)^2 X(f) ](\rho)$. Then, it follows from the above that $y$ verifies the inequality
 \begin{eqnarray*}
 y(\rho) &\lesssim& \epsilon \rho^{1+D\epsilon^{1/2}} + \epsilon^{1/2} \sup_{\rho'\in[0,\rho]} \left[ \Ecal(Z^N(h))^{1/2} \right] \int_2^\rho (\rho')^{-1/2} y^{1/2}(\rho') d\rho' \\
&\lesssim& \epsilon \rho^{1+D\epsilon^{1/2}} + \epsilon^{1/2} \sup_{\rho'\in[0,\rho]} \left[ \Ecal(Z^N(h))^{1/2} \right] \cdot \sup_{\rho'\in [0,\rho]} y^{1/2}(\rho)\cdot\, \rho^{1/2}
 \end{eqnarray*}
Let $A=\sup_{\rho'\in [0,\rho]} y^{1/2}(\rho)$, $B= \epsilon^{1/2} \sup_{\rho'\in[0,\rho]} \left[ \Ecal(Z^N(h))^{1/2} \right] \rho^{1/2}$ and $D= \epsilon \rho^{1+D\epsilon^{1/2}}$. Then, from the above inequality, we have
$$
A^2- BA-D \le 0,
$$
which implies that
$$
A^2 \leq B^2+ 2D,
$$
concluding the proof of the lemma.

\end{proof}
To summarize, we have proven
\begin{proposition} Let $\alpha$ be a multi-index with $|\alpha| \le N$. Then, we have the improved energy estimate
$$
\Ecal_N(h) \lesssim \epsilon \rho^{D \epsilon^{1/2}},
$$
for some sufficiently large universal constant $D$.

\end{proposition}
This ends the proof of the improved estimates \eqref{prop:ibs} and hence of the proof of Theorem \ref{th:main}.









\section{Klainerman-Sobolev estimates for the Vlasov field} \label{se:KSf}
We have the following decay estimates for velocity averages of Vlasov fields.
\begin{lemma}\label{lem:KSf}For any $k:=k(t,x,v)$,
\begin{eqnarray*}
\int_v |k|(t,x,v) dv
&\lesssim& \frac{1} {t^{3}}\sum_{|\alpha| \le 3} E[ p_3 K^\alpha(k) ](\rho)
\end{eqnarray*}
where $p_3:=p_3\left(\| \partial C\|_{\infty}, u^{-1}\|Y C\|_{\infty}, u^{-1}\|C\|_{\infty}, u^{-1}\| C\|_\infty^2\right)$ is a polynomial of third order and by $C$ we denote any of the correction coefficients.
\end{lemma}
\begin{proof}
We follow the analogous proof contained in \cite{fjs:savn}, Section 11,  until the estimate of the following integral
\eq{\alg{
&\int_v Z |f|(y^0,x^1+ty^1,x^2,x^3,v) dv\\
&=\int (Z+w^0\partial_{v_1}+C^\alpha X_{\alpha})|f|(y^0,x^1+ty^1,x^2,x^3,v) dv\\
&\quad-\int (w^0\partial_{v_1}+C^a X_{a}+C^0\partial_t)|f|(y^0,x^1+ty^1,x^2,x^3,v) dv\\
&=\int Y |f|(y^0,x^1+ty^1,x^2,x^3,v) dv\\
&\quad-\int (w^0\partial_{v_1}+C^a X_{a}+C^0\partial_t)|f|(y^0,x^1+ty^1,x^2,x^3,v) dv \label{eq:kste}
}}
The first term on the right-hand side of the last line only contains $Y$ derivative of $f$ and therefore has the right structure. We discuss the remaining terms.

The term
\eq{
\int w^0\partial_{v_1}|f|(y^0,x^1+ty^1,x^2,x^3,v) dv
}
can be integrated by parts and directly estimated.

We now consider the two last terms on the right-hand side of \eqref{eq:kste}.

For clarity, we only the discuss the term $\int_v C^0\partial_t |f|(y^0,x^1+ty^1,x^2,x^3,v) dv$, since the other ones can be treated similarly. We compute
\eq{\label{KS-2}\alg{
&\int C^0\partial_t |f|(y^0,x^1+ty^1,x^2,x^3,v) dv\\
&=\int \partial_t(C^0 |f|)(y^0,x^1+ty^1,x^2,x^3,v) dv\\
&\quad-\int (\partial_tC^0) |f|(y^0,x^1+ty^1,x^2,x^3,v) dv
}}
The second term yields a term that can be directly estimated and we continue with the first term. Writing $\partial_t$ in terms of the $Z$ vector fields, we have
\eq{\alg{
&\int \partial_t(C^0 |f|)(y^0,x^1+ty^1,x^2,x^3,v) dv\\
&=\int (u)^{-1}a^\alpha Z_\alpha(C^0 |f|)(y^0,x^1+ty^1,x^2,x^3,v) dv\\
&=\int (u)^{-1}a^\alpha (Z_\alpha+[v\partial_v]_\alpha+C^\beta_{\alpha} X_\beta)(C^0 |f|)(y^0,x^1+ty^1,x^2,x^3,v) dv\\
&-\int (u)^{-1}a^\alpha ([v\partial_v]_\alpha+C^\beta_{\alpha} X_\beta)(C^0 |f|)(y^0,x^1+ty^1,x^2,x^3,v) dv\\
}}
The first line yields terms with integrands of the form
\eq{
u^{-1} a^\alpha Y_\alpha( C^0) |f| \mbox{ and } u^{-1} a^\alpha C^0 Y_\alpha|f|.
}
After integration by parts in the corresponding $v$-variable, the first term in the second line yields a term with an intgrand of the form
\eq{
u^{-1} \frac{v}{w^0}a^{\alpha} C^0 |f|.
}
The last term from the second line yields a term with the integrand
\eq{
a^\alpha\frac{C^0C^\beta_\alpha}{u} X_\beta |f|.
}

If we repeat the procedure above for the two remaining variables we obtain similar terms with at most two additional factors in the integrands, which are of the types above. In total, all factors can be estimated by cubic terms with factors of the form
\eq{\alg{
\| \partial C\|_{\infty},\quad u^{-1}\|Y C\|_{\infty},\quad u^{-1}\|C\|_{\infty} \mbox{ and } u^{-1}\| C\|_\infty^2.
}}
These are precisely the terms as claimed by the lemma.
\end{proof}

As a consequence, we have

\begin{proposition} \label{prop:dssfl}
Assume the bootstrap assumption \ref{eq:bsm} holds. Let $S_{\mu \nu}[f]$ denotes the components of the tensor field $S[f]$. For any $|\alpha| \le N-3$, we have, for any $(t,x) \in \Kcal \cap \{ \rho \ge 2\}$,
$$
| K^\alpha S_{\mu \nu}[f] |(t,x) \lesssim \epsilon \rho^{D_N \delta} \frac{1}{t^3}.
$$
\end{proposition}
\begin{proof}
We use the commutator formula of Lemma \ref{lem:comtlw}. We must therefore estimate each term of the form $\mathcal{T}_{\alpha \beta} F$, where $F$ is given by

\begin{eqnarray*}
 \frac{1}{(1+u)^q}  P_1(C)^{l_1+l_1',r_{Z,1},s_{X,1}} P_2(X(C))^{l_2,r_{Z,2},s_{X,2}} \widehat{K}^\beta(f) Q^\gamma(L)
\end{eqnarray*}
where $r_{Z,2}+s_{X,2} \le |\alpha|-1$, $r_{Z,1}+s_{X,1}+r_{Z,2}+s_{X,2}+l_2 +|\gamma|+|\beta| \le |\alpha|$ and $q  \ge 2 l_1$, $q \le |\alpha|$, $\gamma_X \ge l_1'$, $l_2 \le |\alpha|$.

Note that since $|\alpha | \le N-3$, we have access to pointwise estimates on each of the above $C$ coefficients as well as on $Q^\gamma(L)$.
For this, we recall that
\begin{itemize}
\item Any $K^\beta X(C)$ coming from $P_2$ only gives $\rho^{\delta D}$ growth
\item Any $K^\beta(C)$ coming from $P_1$ only gives $\rho^{\delta D}u^{1/2}$ growth, which is compensated either by the factor of  $\frac{1}{(1+u)^q}$ or by the extra $u$ decay coming from $Q^\gamma(L)$ depending on the value of $\gamma_X$.
\end{itemize}
It follows that we can estimate pointwise each of the $F$ term by $\rho^{\delta D} |\widehat{K}^\beta(f)|$ and the result then follows from the application of the Klainerman-Sobolev inequalities and the energy estimates for $f$.
\end{proof}


\section{$L^2$ estimates for the transport equation}
\label{se:l2vf}

The proof of $L^2$-decay estimates for the transport equation is based on the strategy already developed in \cite{fjs:vfm,fjs:savn}.

Thus, following \cite{fjs:vfm,fjs:savn}, we
\begin{itemize}
  \item first summarize the set of commutators in a system of equations
  \item and then exploit the form of the system to obtain a representation of the solutions from which $L^2$-estimates can be derived.
  \end{itemize}


Let $\kappa$ be a constant such
\eq{
M_{N} \delta = \kappa,
}
for some constant $M_N$ depending on the maximal number of derivatives that will be fixed later in this section.

Define the vector $F^h=(F^h_{\alpha, s_X,r_Z,k,P})$ whose elements are of the form
\eq{
F^h_{\alpha, s_X,r_Z,k,P}:=   \rho^{-\kappa (r_{Y,\partial_t}+\alpha_{Y,\partial_t})-k\kappa}(1+u)^{-k/2} P^{k,r_Z, s_X}(C) \widehat{K}^{\alpha} f ,
}
where $\alpha_{Y,\partial_t}$ and $r_{Y,\partial_t}$ count the number of fields $Y$ and $\partial_t$ in the multi-index $\alpha$ and acting on the $C$ coefficient, and where the indices verify the condition
\begin{equation}\label{eq:C21} r_Z+s_X+ |\alpha| \leq N,\quad r_Z+s_X \le N-1,\quad k\in \{0,1\}
\end{equation}
and the index $P$ in $F^h_{\alpha, s_X,r_Z,k,P}$ labels the different combinations for $P^{k,r_Z, s_X}(C)$.
\begin{remark}
The prefactor in the definition of $F^h$ in terms of $\rho^{-\kappa}$ and $(1+u)^{-1/2}$ factors is chosen to compensate the growth coming from commutators of $Y$ or $\partial_t$ fields or the $C$ coefficients and to assure an evolution equation for $F^h$ similarly to the cases treated in \cite{fjs:vfm,fjs:savn}.
\end{remark}
The main result of this section is the following $L^2$-estimate for $F^h$.
\begin{proposition}\label{prop:l2est} There exists a constant $M'_N$, depending solely on the number of derivatives $N$ of $f$ such that
\eq{
\int_{H_\rho} \dfrac{t}{\rho} \left(\int_v |F^h| w^0 dv \right)^2 d \mu_{H_\rho} \lesssim \epsilon^2 \rho^{M'_N \delta} \rho^{-3}.
  }
\end{proposition}
\begin{remark}Note that for an element $F^h_{\alpha, s_X,r_Z,k,P}$ such that $|\alpha| \le N-3$, we can use the Klainerman-Sobolev estimates together with the Cauchy-Schwarz inequality in $v$ to obtain a decay estimate similar to the one of the proposition. Thus, the aim of the proposition is to estimate the elements of $F^h$ for which $|\alpha| > N-3$.
\end{remark}

The proof of Proposition \ref{prop:l2est} is divided into a derivation of a system of evolution equations for $F^h$ and a subsequent proof of $L^2$-estimates based on that system. 

\subsection{System of evolution equations}

We  define the vector $F^l$ containing the lower order derivatives of $f$ as being the vector whose elements are given by
\eq{
\widehat{K}^{\alpha} f  \mbox{ with } |\alpha| \leq N-3.
}

Then $F^h$ fulfills an equation of the following form.
\begin{lemma}\label{lem:basiceql2} There exist matrices $\mathbf{A}$, $\mathbf{B}$ and $\mathbf{I}$ such that
\eq{\label{te-Fh}
T_g F^h + \mathbf{A} F^h+\mathbf I F^h =  \mathbf{B}F^l,
 } where the matrices satisfy
  \begin{itemize}
    \item
   \eq{
|\mathbf A| \lesssim \sqrt{\epsilon} w^{\rho} \rho^{-1}
    }
and such that for any sufficiently regular distribution function $k$,
    \eq{
    \int_{H_\rho}\int_v |\mathbf A_{ij}k|dv d\mu_{H_\rho}\lesssim \sqrt{\epsilon}  \rho^{-1} \int_{H_\rho} \chi_\eta(|k|)d\mu_{H_\rho}.
    }

 \item
\eq{
|\mathbf B|\lesssim \max\left(\sqrt{\rho}, t/\rho\right) \Psi,
}
for some $\Psi$ with
 \eq{
 \Vert \Psi \Vert_{L^2(H_\rho)} \lesssim  \epsilon\rho^{D\delta}
 }
 for some constant $D$.
 \item The matrix $\mathbf I$ is diagonal and the diagonal terms are either zero,
$$   -\kappa\frac{w^\rho}{\rho}, \quad
  \frac12(1+u)^{-1} (w^0-w_i\frac{x^i}{r})\, \mbox{ or }\,
     \kappa(r_{Y,\partial_t}+\alpha_{Y,\partial_t})\frac{w^\rho}{\rho},
$$
 where the constants $r_{Y,\partial_t}$, $\alpha_{Y,\partial_t}$ are the corresponding exponents in $F^h$.

\end{itemize}
\end{lemma}
\begin{remark} An immediate consequence of this lemma is the fact that $F^l$ satisfies a similar equation as $F^h$; the only difference is the fact that all the derivatives of $h$ can now be estimated pointwise. In other terms, the equation satisfied by $F^l$ has no source term. Furthermore, there is no matrix $\mathbf{I}$ appearing in the equation since the components of $F^l$ are not rescaled. In particular, there exists a matrix $\hat{\mathbf{A}}$, satisfying the same properties as $\mathbf A$ above, such that:
\begin{eqnarray} \label{eq:fl}
T_g F^l  +  \hat{\mathbf{A}}F^l  = 0.
\end{eqnarray}
\end{remark}
\begin{proof} The proof relies on revisiting the corresponding formulae of Section \ref{se:hovf} where the $L^1$-estimates are proven.

First, we compute the action of the operator $T_g$ onto the components of $F^h$
\begin{eqnarray}
&&T_g \left( (1+u)^{-k/2}\rho^{-\kappa (r_{Y,\partial_t}+\alpha_{Y, \partial_t})-k\kappa} P^{k,r_z, s_X}(C) \widehat{K}^{\alpha} f \right)= \nonumber \\
  &&\qquad T_g\left((1+u)^{-k/2}\right)\left(\rho^{-\kappa (r_{Y,\partial_t}+\alpha_{Y, \partial_t})-k\kappa} P^{k,r_z, s_X}(C) \widehat{K}^{\alpha} f \right) \label{eq:term11}\\
  &&\qquad+ (1+u)^{-k/2}T_g {\left( \rho^{-\kappa (r_{Y,\partial_t}+\alpha_{Y, \partial_t})-k\kappa} \right)} P^{k,r_z, s_X}(C) \widehat{K}^{\alpha} f \label{eq:term1}\\
  &&\qquad+  (1+u)^{-k/2}\rho^{-\kappa (r_{Y,\partial_t}+\alpha_{Y, \partial_t}))-k\kappa}   T_g \left(P^{k,r_Z, s_X}(C)\right)  \widehat{K}^{\alpha} f \label{eq:term30}\\
    &&\qquad+ (1+u)^{-k/2} \rho^{-\kappa (r_{Y,\partial_t}+\alpha_{Y, \partial_t})-k\kappa}  P^{k,r_z, s_X}(C)  \left [T_g, \widehat{K}^{\alpha}  \right]f.\label{eq:term3}
\end{eqnarray}
We discuss in the following the four terms on the right-hand side above and evaluate their contribution to the different terms, $\mathbf A F^h$, $\mathbf I F^h$ and $\mathbf B F^l$, respectively.\\

Recall that, from Lemma \ref{lem:weghtingwithu}, the following identity holds,
$$
T_g(1+u)=w^0-w_i\frac{x^i}{|x|}+E_u ,
$$
where $E_u$ satisfies
$$
|E_u|\lesssim | w_0|  \rho^{\delta/2} \epsilon^{1/2} \frac{1+u}{t^{3/2}}  +  \frac{|\vu_a|^2}{ | w_0|} \frac{t}{\rho} \rho^{\delta/2} \epsilon^{1/2} \frac{(1+u)^{1/2}}{t} +  \dfrac{\epsilon^{1/2} w^0}{t},
$$
and
$$
w^0-w_i\frac{x^i}{|x|} \geq 0.
$$
The latter term generates a first term originating from \eqref{eq:term11}, which is
$$
-\frac{k}{2}(1+u)^{-1}\left(w^0-w_i\frac{x^i}{r}\right)(1+u)^{-k/2}\left(\rho^{-\kappa (r_{Y,\partial_t}+\alpha_{Y, \partial_t})-k\kappa} P^{k,r_z, s_X}(C) \widehat{K}^{\alpha} f \right),
$$
and which contributes to the $\mathbf I F^h$ term. The other term can be estimated by
$$
-\frac{k}{2}
(1+u)^{-1}
\left |E_u (1+u)^{-k/2}\left(\rho^{-\kappa (r_{Y,\partial_t}+\alpha_{Y, \partial_t})-k\kappa} P^{k,r_z, s_X}(C) \widehat{K}^{\alpha} f \right)\right|.
$$
In the latter term, the first factor contributes to the matrix $\mathbf A$, since
\begin{eqnarray*}
|(1+u)^{-1} E_u|& \lesssim &| w_0|  \rho^{\delta/2} \epsilon^{1/2} \frac{1}{t^{3/2}}  +  \frac{|\vu_a|^2}{ | w_0|} \frac{t}{\rho} \rho^{\delta/2} \epsilon^{1/2} \frac{(1+u)^{-1/2}}{t}+  \dfrac{\epsilon^{1/2} w^0}{(1+u)t}\\
&\lesssim& w^\rho\left( \dfrac{t}{\rho} \rho^{\delta/2} \epsilon^{1/2}\frac{1}{t^{3/2}} +  \rho^{\delta/2} \epsilon^{1/2} \frac{(1+u)^{-1/2}}{t}\right)+\dfrac{\epsilon^{1/2} w^0}{(1+u)t}\\
&\lesssim& w^\rho \rho^{-1}.
\end{eqnarray*}
This finishes the discussion of the first term. The second term, \eqref{eq:term1}, can be handled accordingly using
\eq{\label{adfhasf}\alg{
T_g(\rho)&=\frac{-v_\alpha g^{\alpha0} x_0-v_\alpha g^{\alpha b} x_b}{\rho}\\
&=w^\rho-\frac{v_\alpha H^{\alpha\beta} x_\beta}{\rho}.
}}
Using the previous decomposition to evaluate
\eq{ \nonumber
T_g(\rho^{-\kappa(r_{Y,\partial_t}+\alpha_{Y, \partial_t})-k\kappa})=-(\kappa(r_{Y,\partial_t}+\alpha_{Y, \partial_t})+k\kappa)\rho^{-1}T_g(\rho)(\rho^{-\kappa(r_{Y,\partial_t}+\alpha_{Y, \partial_t})-k\kappa})
}
we obtain the first term from \eqref{eq:term1}
\eq{ \nonumber
-(\kappa(r_{Y,\partial_t}+\alpha_{Y, \partial_t})+k\kappa)\frac{w^\rho}{\rho}\left((1+u)^{-k/2}\rho^{-\kappa(r_{Y,\partial_t}+\alpha_{Y, \partial_t})-k\kappa} P^{k,r_Z, s_X}(C) \widehat{K}^{\alpha} f \right),
}
which directly contributes to the term $\mathbf I F^h$. The term arising from the second term in \eqref{adfhasf} can be bounded by
\eq{ \nonumber
(r_{Y,\partial_t}+\alpha_{Y, \partial_t})+k\kappa)\frac{1}{\rho}\left|\frac{v_\alpha H^{\alpha\beta}x_\beta}{\rho}\right|\left|(1+u)^{-k/2}\rho^{-\kappa(r_{Y,\partial_t}+\alpha_{Y, \partial_t})-k\kappa} P^{k,r_Z, s_X}(C) \widehat{K}^{\alpha} f \right|,
}
where we estimate analogously to \eqref{wHx}
\eq{ \nonumber
\left|\frac{v_\alpha H^{\alpha\beta}x_\beta}{\rho^2} \right| \lesssim \frac{\sqrt{\varepsilon}}{\rho} v^\rho.
}
Hence the corresponding term contributes to the term $\mathbf A F^h$. We proceed with the evaluation of the third term \eqref{eq:term30}. We expand the action of $T_g$ in the corresponding case $k=1$ by
\eq{\alg{
T_g \left( P^{1,r_Z,s_X}(C) \right )&= T_g(\widehat K^{\mu}C)\\
&=[T_g,\widehat K^{\mu}]C+ \widehat K^{\mu} (T_gC),
}}
where $\mu$ denotes the corresponding multi-index. This yields two different types of terms, which we discuss separately, beginning with the first type. This type is a sum of terms of the form
\eq{\alg{
 \/[T_g,\widehat K^{\mu}]C&=\sum_{\nu_1,\nu_2} \widehat K^{\nu_1} [T_g,\widehat K] \widehat K^{\nu_2} C
}}
for suitable pairs of multi-indices $\nu_1,\nu_2$.
The number of different types of fields $X$, $\partial_t$ and $Y$ appearing on the RHS can change from the multi-index $\mu$ only through the commutator $[T_g,\widehat K]$. Here it is important to distinguish the three different cases $\widehat K=X, \partial_t$ or $Y$ and keep track of the corresponding weights in $\rho$. The total term to estimate, which arises from the previous type of terms, is a sum of terms of the form
\eq{
(1+u)^{-1/2}\rho^{-\kappa(r_{Y,\partial_t}+\alpha_{Y, \partial_t})-\kappa}\left( \sum \widehat K^{\nu_1} [T_g,\widehat K] \widehat K^{\nu_2} C\right)\widehat K^\alpha f.
}
In the case $\widehat K=X_i$, we conclude from the analysis of the higher order commutators that the resulting term can be written as a sum of terms of the form
\eq{
(1+u)^{-1/2}\rho^{-\kappa(r_{Y,\partial_t}+\alpha_{Y, \partial_t})-k\kappa}\left( \frac{P(C)^{\tilde k,\tilde r_Z,\tilde s_X}}{(1+u)^{\tilde k}}w \partial_{t,x}K^{\beta}(h)\widehat K^{\gamma}C\right)\widehat K^\alpha f
}
and terms with stronger decay for one of the factors, which we ignore. Note that, the commutation with $X$ may generate one extra $Y$ or $\partial_t$ field which then, if in the following $\widehat K^{\gamma}C$ is estimated in energy, does not have the appropriate $\rho^{-\kappa}$ factor to be consistent with the definition of $F^h$. However, we recall that the coefficient resulting from the commutator has the null condition and the resulting additional decay in the coefficent can be used to multiply and divide by $\rho^{-\kappa}$ as $\kappa$ is still small w.r.t.~1/2, which is the extra $\rho$-decay from the null structure. In particular, in both cases, when $f$ or $h$ is estimated in energy, the coefficient contains enough decay to fulfill the conditions for $\mathbf A$ or $\mathbf B$, which follows as in the $L^1$-case before (see the proof of Proposition \ref{prop:eesvfl}, and following).

For the analysis of both other cases $\widehat K=\partial_t$ or $Y$ the factor $\rho^{-\kappa r_{Y,\partial_t}}$ is important. This can be seen as follows. The corresponding terms to estimate then contain in particular a term where the $Y$ or $\partial_t$ field generates an $X_i$ field, however, with a coefficient that does not have the null structure. These terms are of the form

\eq{
(1+u)^{-1/2}\rho^{-\kappa(r_{Y,\partial_t}+\alpha_{Y, \partial_t}-1)}\rho^{-\kappa}\left( \frac{P(C)^{\tilde k,\tilde r_Z,\tilde s_X}}{(1+u)^{\tilde k}}w \partial_{t,x}K^{\beta}(h)[\widehat K^{\gamma-1} X]C\right)\widehat K^\alpha f.
}

We emphasize here that due to the reduction of the number of non-$X_i$ fields acting on $C$ by one, we can separate a $\rho^{-\kappa}$ factor, which then, using the relation
\eq{
\kappa=M_N \delta>D\delta
}
where $D$ denotes the constant bounding all constants from the small $\delta$ growth of the terms,
can be used to absorb all those small $\rho^{D\delta}$-factors and yields the bounds as claimed.

The next term to discuss results from $\widehat K^\nu(T_g C)$, which generates the source term of the evolution equation for $C$. However, then the $(1+u)^{-1/2}\rho^{-k\kappa}$ factor is used to improve the decay of that source term, which then provides the required $\rho^{-1}$ decay without loss.  This is shown analogously to the corresponding section on the estimates of the $C$ coefficients (see Proposition \ref{prop:lowcp}).

Finally, the last term, where the commutator $[T_g,\widehat K^\alpha]$ acts on $f$ can be analyzed similarly. If from this commutator a term including $C$ arises with a high number of derivatives then this term will be interpreted as contributing to the term $\mathbf A F^h$. In particular the weights that come with a $C$ term in the definition of need to be generated. This can be seen to work following the analysis as for instance in Proposition \ref{prop:eevfn}.
\end{proof}

\subsection{Generic $L^2$-estimates} This section is devoted to the proof of $L^2$-estimates for $F^h$. The equation satisfied by $F^h$ is a linear inhomogeneous transport equation. The first step is to perform a decomposition between the homogeneous part and the inhomogeneous part. For the inhomogeneous part, we follow the strategy of \cite{fjs:vfm}. To treat the homogeneous part in \cite{fjs:vfm}, we simply performed three more commutations and use the Klainerman-Sobolev. However this is no longer possible in the current case since differentiating three more times the correction terms $Y^\alpha C$ (for $|\alpha| >N-4$) is not possible and these terms naturally appears even in the homogeneous part. Hence, we also perform a decomposition of the initial data for the homogeneous part between regular and less regular terms.

We now prove
\begin{lemma}\label{prop:l2basis} Let $F^h:\mathcal{P} \rightarrow \mathbb{R}^N$ be a solution to the Cauchy problem
  $$
T_g F^h + (\mathbf{A}+\mathbf{I}) F^h =  \mathbf{B} F^l  , \mbox{ with } E_{{N+3,4}}[f](\rho= 2)\lesssim \epsilon.
  $$
and where the matrices $\mathbf{A}$, $\mathbf{B}$ and $\mathbf{I}$ satisfy the same properties as in Lemma \ref{lem:basiceql2}. Then, there exists a constant $M'_N$ such that
    $$
\int_{H_\rho} \dfrac{t}{\rho}\left(\int_v| F^h| w^0 dv\right)^2 d\mu_{H_{\rho}, \eta} \lesssim \epsilon^{2} \rho^{M'_N\delta} \rho^{-3},
    $$
\end{lemma}
\begin{proof} We have already proven that the vector $F^h$ satisfies the inhomogeneous equation
$$
T_g F^h  + (\mathbf{A}+ \mathbf{I})F^h = \mathbf{B} F^l.
$$
We start by performing the decomposition
$$
F^h = F_{0}+F_i \mbox{ with }
$$
\begin{itemize}
  \item the homogeneous part of the equation satisfying
  $$
  T_g F_0  + \mathbf{A}F_0 = 0 \mbox{ with } F_0|_{H_2} =  F^h|_{H_2};
  $$
  \item and the inhomogeneous equation
  $$
  T_g F_i  + \mathbf{A}F_i = \mathbf{B} F^l - \mathbf{I} F^h   \mbox{ with } F_i|_{H_2} = 0.
  $$
\end{itemize}

To deal with the homogeneous equation, we need to perform another decomposition of the initial data. We need to distinguish between the initial data that allow for three more commutations, and those who cannot. Recall here the definition of the vector $F^h$ whose  coefficients are
$$
F^h :=  \rho^{-\kappa(r_{Y,\partial_t}+\alpha_{Y, \partial_t})}\rho^{-k\kappa}(1+u)^{-k/2} P^{k,r_Z, s_X}(C) \widehat{K}^{\alpha} f.
$$
When $r_Z + s_X >N-4$, it is not possible to differentiate these terms three more times. Consider now the case when $r_Z + s_Z \leq N-4$. Then, by definition, all the coefficients in front of $\widehat{K}^{\alpha} f$ can be estimated pointwise, and can be differentiated three more times. Since $r_Z + s_Z \leq N-4$, then $|\alpha|$ can (possibly) take values up to $N$. Hence, the term $\widehat{K}^{\alpha} f$ may also contain high derivatives of the $C$ coefficients, since $\widehat{K}=\widehat{Z}+C.X$. In the following lemma, we expand the expression of $\widehat{K}^\alpha f$ to identify within these terms the derivatives of $C$ which cannot be estimated pointwise. Recall that the value of the derivatives of $C$ on $H_2$ can be computed explicitly in terms of the data for $h$, see Remark \ref{rem:dataC}.
\begin{lemma}\label{lem:decinitialdata} The term $\widehat{K}^\alpha f$ can be written as
  $$
P^{k', r'_Z, s'_X}(C) \widehat{Z}^{\alpha'} f
  $$
  where
\begin{equation}\label{eq:C22}
  r'_Z + s'_X +  |\alpha'| \leq |\alpha|,\quad r'_Z + s'_X \leq |\alpha|-1, \quad k' \leq |\alpha|.
\end{equation}
\end{lemma}
\begin{proof} The proof of this fact is a straightforward application of the Leibniz rule.
\end{proof}
Hence, altogether, the restriction of $F^h$ to the initial hyperboloid $H_2$ is of the form:
\begin{equation}
w^0 \rho^{-\kappa(r_{Y,\partial_t}+\alpha_{Y, \partial_t})}\rho^{-k\kappa}(1+u)^{-k/2} P^{k+k',r_Z+r'_Z, s_X +s'_X }(C) Z^{\alpha'}f,
\end{equation}
where the conditions \eqref{eq:C21} and \eqref{eq:C22} are satisfied.

Using Lemma \ref{lem:decinitialdata}, we perform a decomposition of the initial data of $F^h|_{H_2}$. Let $F_{01}$ be the solution of the equation
\begin{equation}\label{eq:f01}
  T_g F_{01}  + \mathbf{A}F_{01} = 0,
\end{equation}
 where the initial data are taken to be
 $$
 F_{01}|_{H_2} =\rho^{-\kappa(r_{Y,\partial_t}+\alpha_{Y, \partial_t})}\rho^{-k\kappa}(1+u)^{-k/2} P^{k+k',r_Z+r'_Z, s_X +s'_X }(C) Z^{\alpha'}f|_{H_2},
 $$
 with the conditions \eqref{eq:C21} and \eqref{eq:C22} and $r_Z+r'_Z + s_X +s'_X \leq N-4$, that is to say that $F_{01}$ contains the terms that can be differentiated three times more. 
We define $F_{02}$ to be the solution to
\begin{equation}\label{eq:f02}
  T_g F_{02}  + \mathbf{A}F_{02} = 0
\end{equation}
  where the initial data are taken to
  $$
  F_{02}|_{H_2} =\rho^{-\kappa(r_{Y,\partial_t}+\alpha_{Y, \partial_t})}\rho^{-k\kappa}(1+u)^{-k/2} P^{k+k',r_Z+r'_Z, s_X +s'_X }(C) Z^{\alpha'}f|_{H_2},
  $$
  with the conditions \eqref{eq:C21} and \eqref{eq:C22}, and either $r_Z+r'_Z + s_X +s'_X > N-4$, that is to say when one cannot differentiate $C$ three times more.

 We notice that, since we control initially $3$ derivatives more in energy for the distribution function $f$, and since $\mathbf{A}$ contains at most $N-3$ derivatives of the metric, it is possible to commute again $3$ more times Equation \eqref{eq:f01} and obtained pointwise estimates for $F_{01}$. This argument has been presented already in our previous works and we consequently do not present it here again (see \cite[p. 56]{fjs:vfm} and \cite[Lemma 9.4]{fjs:savn}). Hence, $F_0$ satisfies the required pointwise estimates. Since it will be useful for the following, we actually prove that $(w^0)^2F_0$ satisfies pointwise estimates
\begin{lemma} \label{lem:L21}The solution to the equation
  $$
  T_g ((w^0)^2F_{01})  + \left(\mathbf{A} -2 \frac{T_g(w^0)}{w^0}\mathbf{Id}\right) (w^0)^2F_{01} = 0
  $$
  $$
  F_{01}|_{H_2} =\rho^{-\kappa(r_{Y,\partial_t}+\alpha_{Y, \partial_t})}\rho^{-k\kappa}(1+u)^{-k/2} P^{k+k',r_Z+r'_Z, s_X +s'_X }(C) Z^{\alpha'}f|_{H_2},
  $$
  with the conditions \eqref{eq:C21} and \eqref{eq:C22}, and $r_Z+r'_Z + s_X +s'_X \leq N-4$ and $|\alpha| \leq N-3$,
  satisfies the pointwise estimates:
  $$
\int_v (w^0)^2|F_{01}| w^0 dv \lesssim  \epsilon\rho^{D\delta}t^{-3},
  $$
  where $D$ is a constant. As a consequence, the following $L^2$-estimates holds:
  $$
\int_{H_\rho} \dfrac{t}{\rho}\left(\int_v| (w^0)^2F_{01}|w^0 dv\right)^2 d\mu_{H_{\rho}} \lesssim \epsilon^{2} \rho^{2\delta} t^{-3}.
  $$
\end{lemma}
\begin{proof} The proof of this fact relies solely on commuting three times, and exploiting the estimates on $T_g(w^0)$ of Lemma 15.1.
\end{proof}

We consider now $F_{02}$, for which the previous method cannot be applied. We perform a representation of the solution as follows. Let us consider the solution $F_{02}$ to the equation
$$
T_gF_{02} + \mathbf{A}F_{02}
 = 0.
 $$
 where the initial data are
 $$
  F_{02}|_{H_2} = \rho^{-\kappa(r_{Y,\partial_t}+\alpha_{Y, \partial_t})}\rho^{-k\kappa}(1+u)^{-k/2} P^{k+k',r_Z+r'_Z, s_X +s'_X }(C) Z^{\alpha'}f|_{H_2},
  $$
  with the conditions \eqref{eq:C21} and \eqref{eq:C22}, and either $r_Z+r'_Z + s_X +s'_X > N-4$ or $|\alpha| > N-3$.

Let $G_2$ be the solution to 
$$
T_g \left(  G_2 \right)+ \mathbf{A}G_2=0, 
$$
with data $G_{2|H_2}=Z^{\alpha'}f|_{H_2}$. 

Note that we can differentiate $G_2$ three more times, and hence prove decay estimates via Klainerman-Sobolev inequality for $G_2$. 

Let $K_2$ be the solution to 
$$
T_g(K_2)+K_2 \mathbf{A}-\mathbf{A}K_2=0,
$$
with data $\rho^{-\kappa(r_{Y,\partial_t}+\alpha_{Y, \partial_t})}\rho^{-k\kappa}(1+u)^{-k/2} P^{k+k',r_Z+r'_Z, s_X +s'_X }(C)$. 

Then, by uniqueness of the Cauchy problem, we have 
$$
F_{02}=K_2 G_2.
$$

Then, we use the usual manipulation, based on the Cauchy-Schwarz inequality:
\begin{eqnarray*}
\int_{H_{\rho}}\left(\int_v  | K_2 G_2 |  w^0 dv\right)^2  d\mu_{H_\rho}
 &\le& \int_{H_{\rho}}\left( \int_{v}\vert K_2^2 G_2 \vert w^0 dv \int_v  |G_2| w^0 dv  \right) d\mu_{H_\rho}.
\end{eqnarray*}

Since the first term is bounded in $L^1$, and the second decays strongly by the Klainerman-Sobolev estimates, we obtain, for $F_{02}$ the following lemma.
\begin{lemma}\label{lem:f02} There exists a constant $M_N'$ such that
  $$
\int_{H_\rho} \dfrac{t}{\rho}\left(\int_v| F_{02}| w^0 dv\right)^2 d\mu_{H_{\rho}} \lesssim \epsilon^{2} \rho^{M'_N\delta} \rho^{-3}.
  $$
\end{lemma}
Moreover, we can repeat the above analysis and add $w^0$ weights, so that, we also have

\begin{lemma}\label{lem:f02} There exists a constant $M_N'$ such that
  $$
\int_{H_\rho} \dfrac{t}{\rho}\left(\int_v (w^0)^2| F_{02}| w^0 dv\right)^2 d\mu_{H_{\rho}} \lesssim \epsilon^{2} \rho^{M'_N\delta} \rho^{-3}
  $$
and thus
$$
\int_{H_\rho} \dfrac{t}{\rho}\left(\int_v (w^0)^2| F_{0}| w^0 dv\right)^2 d\mu_{H_{\rho}} \lesssim \epsilon^{2} \rho^{M'_N\delta} \rho^{-3}.
  $$

\end{lemma}

Consider now the inhomogeneous equation
$$
T_g F_i  + \mathbf{A}F_i = \mathbf{B} F^l - \mathbf{I} F^h   \mbox{ with } F_i|_{H_2} = 0.
$$
In this equation, we can still perform the decomposition $
F^h =  F_0 + F_i
$
in the source term so that the equation becomes:
$$
T_g F_i  + (\mathbf{A}+ \mathbf{I})F_i = \mathbf{B} F^l - \mathbf{I} F_0   \mbox{ with } F_i|_{H_2} = 0.
$$
\begin{remark} We need in what follows to exploit the positivity of the matrix $\mathbf{I}$. This positivity can nonetheless be exploited only when $F_i$ is positive. We need to consider to establish the estimates the vector made out of the modulus of the components of $F_i$, which satisfies an equation having the same characteristics as $F_i$. Hence, to avoid introducing new notations, we will still write $F_i$ below instead of ``the matrix whose components are given by the absolute values of the components of $F_i$''.
\end{remark}
It is necessary to perform a second decomposition. Define $F_1$ and $F_2$ as the solutions to
\begin{subequations}
\begin{equation}\label{eq:l2eq1}
T_g F_1  + (\mathbf{A}+ \mathbf{I})F_1 = \mathbf{B} F^l   \mbox{ with } F_1|_{H_1} = 0
\end{equation}
and
\begin{equation}\label{eq:l2eq2}
T_g F_2  + (\mathbf{A}+ \mathbf{I})F_2 = -\mathbf{I}F_0   \mbox{ with } F_2|_{H_1} = 0.
\end{equation}
\end{subequations}

To deal with $F_1$, we follow the strategy already used in our previous work to prove the $L^2$-estimates; the following lemma is a variation of \cite[Section 4.5.7]{fjs:savn}.
\begin{lemma}\label{lem:L22} Let $F_1:\mathcal{P} \rightarrow \mathbb{R}_+^L$ be a solution to the Cauchy problem
  $$
T_g F_1  + (\mathbf{A}+ \mathbf{I})F_1 = \mathbf{B} F^l   \mbox{ with } F_1|_{H_1} = 0,
  $$
 and where the matrices $\mathbf{A}$, $\mathbf{B}$ and $\mathbf{I}$ satisfy the same properties as in Lemma \ref{lem:basiceql2}. There exists a constant $M^1_N$, depending on $N$ such that
    $$
\int_{H_\rho} \dfrac{t}{\rho}\left(\int_v| F_1| w^0 dv\right)^2 d\mu_{H_{\rho}} \lesssim \epsilon^{2} \rho^{M^1_N\delta} \rho^{-3}
    $$
\end{lemma}
\begin{proof} In the absence of the matrix $\mathbf{I}$, this proposition has been proved in \cite[Section 4.5.7]{fjs:savn}, and exploited to proved the $L^2$-estimates in that situation. The only difference with respect to \cite[Section 4.5.7]{fjs:savn} is the matrix $\mathbf{I}$. Since all the elements of $F_1$ are non negative, and since $\mathbf{I}$ is bounded, positive, diagonal, and appears solely in the potential of the equations satisfied by $F^h$ and $F^l$, all the estimates derived in \cite[Section 4.5.7]{fjs:savn} remain valid in this situation.
\end{proof}

We finally consider the solution $F_2$ to Equation \eqref{eq:l2eq1}. This part is new with respect to our previous work, but to deal with it, we follow the same strategy of representation. Let $\mathbf{K}$ be the solution to the equation
$$
T_g \mathbf{K} + (\mathbf{A}+ \mathbf{I})\mathbf{K} - \mathbf{K}\left(\mathbf{A} - \dfrac{2T_g(w^0)}{w^0}\mathbf{Id}\right) = -\dfrac{1}{(w^0)^2}\mathbf{I}  \mbox{ with } \mathbf{K}|_{H_1} = 0.
$$
An immediate algebraic manipulation provides, by uniqueness of solutions of the Cauchy problem,
$$
F_2 = \mathbf{K} (w^0)^2F_0.
$$
We know already that $(w^0)^2F_0$ satisfies the pointwise estimates:
$$
\int_{v} \vert (w^0)^2F_0 \vert w^0 dv  \lesssim \epsilon \rho^{D\delta}t^{-3},
$$
for some constant $D$. It is then sufficient to check that the matrix $\mathbf{K}$ has bounded components.
\begin{lemma} The solution $\mathbf{K}$ to the equation
  $$
  T_g \mathbf{K} + (\mathbf{A}+ \mathbf{I})\mathbf{K} - \mathbf{K}\left(\mathbf{A} - \dfrac{2T_g(w^0)}{w^0}\mathbf{Id}\right) = -\dfrac{1}{(w^0)^2}\mathbf{I}  \mbox{ with } \mathbf{K}|_{H_1} = 0
  $$
  satisfies:
  $$
\vert \mathbf{K} \vert \lesssim \rho^{M\sqrt{\epsilon}}.
  $$
  for some constant $M$.
\end{lemma}
\begin{proof}
The proof of this fact relies on the use of the representation formula for solution to the transport equation, and exploit the diagonal form of $\mathbf{I}$ as well as the positivity of the components of $\mathbf{I}$. The components of $\mathbf{K}$ satisfies the equation:
$$
T_g(K_{a}{}^b) = -\sum_{c}\left((A_a{}^c + I_a^c)K_{c}{}^b\right) - \dfrac{1}{(w^0)^2}I_a^b
$$
Hence, $|K_{a}{}^b|$ satisfies
\begin{eqnarray}
T_g(|K_{a}{}^b|) &=& -\sum_{c}\left((A_a{}^c + I_a^c)K_{c}{}^b \dfrac{K_{a}{}^b}{|K_{a}{}^b|}\right) - \dfrac{1}{(w^0)^2}I_a^b\dfrac{K_{a}{}^b}{|K_{a}{}^b|}\\
&=& -\sum_{c}\left(A_a{}^cK_{a}{}^b\dfrac{K_{a}{}^b}{|K_{a}{}^b|}\right) - I_a^b |K_{a}{}^b| - \dfrac{1}{(w^0)^2}I_a^b\dfrac{K_{a}{}^b}{|K_{a}{}^b|},
\end{eqnarray}
since $\mathbf{I}$ is diagonal. Using the representation formula, since $\mathbf{K}$ has trivial initial data, one obtains
$$
|K_{a}{}^b| =  \int_{2}^\rho \left(\dfrac{1}{v^\rho}\sum_{c}\left(-A_a{}^cK_{a}{}^b\dfrac{K_{a}{}^b}{|K_{a}{}^b|}\right) - I_a^b |K_{a}{}^b| - \dfrac{1}{(w^0)^2}I_a^b\dfrac{K_{a}{}^b}{|K_{a}{}^b|}\right)ds.
$$
Since $\mathbf{I}$ contains solely non-negative components, one obtains:
\begin{eqnarray*}
|K_{a}{}^b| &\leq &  \int_{2}^\rho\dfrac{1}{v^\rho} \left(\left(-\sum_{c}\left(A_a{}^cK_{a}{}^b\dfrac{K_{a}{}^b}{|K_{a}{}^b|}\right) - \dfrac{1}{(w^0)^2}I_a^b\dfrac{K_{a}{}^b}{|K_{a}{}^b|}\right)\right)ds \\
&\leq &  \int_{2}^\rho \dfrac{|\mathbf{A}|}{|v^\rho|} |\mathbf{K}| + \dfrac{|\mathbf{I}|}{(w^0)^2v^\rho} ds.
\end{eqnarray*}
Using the assumption of Lemma \ref{prop:l2basis}, one obtains:
$$
|\mathbf{K}| \lesssim \int_2^\rho \dfrac{w^\rho}{v^{\rho}}\left( \sqrt{\epsilon} \dfrac{|\mathbf{K}|}{s} + \dfrac{1}{(w^0)^2}\cdot \dfrac{1}{1+\sqrt{s^2+r^2}-r}\right) ds.
$$
We estimate the last integral by
$$
\int_2^\rho\dfrac{1}{(w^0)^2}\cdot \dfrac{1}{1+\sqrt{s^2+r^2}-r} ds  \leq  \int_2^\rho \dfrac{w^{\rho}}{w^0}\cdot \dfrac{s}{t(s)} \dfrac{t(s)}{s^2} ds \lesssim \log(\rho).
$$
Recall that $(w^0)^{-1} \leq 2 \rho t^{-1} w^{\rho}$ (see Lemma 6.2 in \cite{fjs:savn}).
Finally, applying Gr\"onwall's lemma, one obtains
$$
|\mathbf{K}| \lesssim \log(\rho) \rho^{M\sqrt{\epsilon}}\lesssim \rho^{M\sqrt{\epsilon}},
$$
where $M$ is a constant that is changing from a line to the other.
\end{proof}

Hence, exploiting the pointwise estimates for $F_0$, we obtain immediately:
\begin{lemma}\label{lem:L23}$F_2$ satisfies the $L^2$ estimates: there exists a constant $M^2_N$ depending on $N$ such that
  $$
\int_{H_\rho} \dfrac{t}{\rho} \left(\int_v |F_2| w^0 dv \right)^2 d \mu_{H_\rho} \lesssim \epsilon^2 \rho^{M^2_N\delta} \rho^{-3}.
  $$
\end{lemma}

The proof of the $L^2$-estimates is obtained by combining the previous Lemmas \ref{lem:L21}, \ref{lem:f02}, \ref{lem:L22} and \ref{lem:L23}.
\end{proof}

As a consequence, we have
\begin{proposition} \label{prop:dssfl2}
Let $S_{\mu \nu}[f]$ denotes the components of the tensor field $S[f]$. For any $|\alpha| \le N-1$, there exists a constant $D_N$ such that
$$
\int_{H_\rho} \dfrac{t}{\rho}
 \left(\int_v | K^\alpha S_{\mu \nu}[f] | dv\right)^2
  d\mu_{H_\rho}\lesssim \epsilon^2
   \frac{ \rho^{D_N \delta}}{\rho^3}.
$$
\end{proposition}

\begin{proof}
We use the commutator formula of Lemma \ref{lem:comtlw}. Recall as well that the term $Q^\gamma (L)$ can always be evaluated by the corresponding metric components, see Lemma \ref{lem:elf}. We must therefore estimate each term of the form $\mathcal{T}_{\alpha \beta} [F]$, where $F$ is given by
\begin{eqnarray*}
 \frac{1}{(1+u)^q}  P_1(C)^{l_1+l_1',r_{Z,1},s_{X,1}} P_2(X(C))^{l_2,r_{Z,2},s_{X,2}} \widehat{K}^\beta(f) Q^\gamma(L)
\end{eqnarray*}
where $r_{Z,2}+s_{X,2} \le |\alpha|-1$, $r_{Z,1}+s_{X,1}+r_{Z,2}+s_{X,2}+l_2 +|\gamma|+|\beta| \le |\alpha|$ and $2q  \ge  l_1$, $q \le |\alpha|$, $\gamma_X \ge l_1'$, $l_2 \le |\alpha|$.

We consider first the case when we have access to pointwise estimates on each of the above $C$ coefficients as well as on $Q^\gamma(L)$,

For this, we recall that
\begin{itemize}
\item any $K^\beta X(C)$ coming from $P_2$ only gives $\rho^{\delta D}$ growth;
\item any $K^\beta(C)$ coming from $P_1$ only gives $\rho^{\delta D}u^{1/2}$ growth, which is compensated either by the factor of  $\frac{1}{(1+u)^q}$ or by the extra $u$ decay coming from $Q^\gamma(L)$ depending on the value of $\gamma_X$.
\end{itemize}
It follows that we can estimate pointwise each of the $F$ term by $\rho^{\delta D} |\widehat{K}^\beta(f)|$ and the result then follows from the application of the above $L^2$-estimates (with $k=0$).

The second case we need to consider the cases when either $K^\beta C$ (for $|\beta| >N-4$) or $\widehat{K}^\beta f$ (for $|\beta| >N-3$) cannot be estimated pointwise.  In that situation $|\gamma|<4$, and $Q^\gamma(L)$
 can be estimated pointwise. Applying the $L^2$-estimates of Proposition \ref{prop:l2est}, we obtain the result.

If there is one $K^\beta(C)$ coefficient coming from $P_1$ which cannot be estimated pointwise (for $|\beta| >N-3$). In this situation $Q^{\gamma}(L)$ (with $|\gamma| <4 $) can be estimated pointwise, and  we proceed similarly, as before, noticing that Proposition \ref{prop:l2est} allows for up to one $C$ to be absorbed in the $L^2$-estimates.

If there is only $K^\beta(X(C))$ coefficient from $P_2$ which cannot be estimated pointwise, we proceed similarly as previously, relying on Proposition \ref{prop:l2est}.

Finally, if $Q^\gamma(L)$ cannot be estimate pointwise (that is to say when $|\gamma|> N-2$,  by Lemma \ref{lem:elf}), then we have that
$$
r_{Z,1}+s_{X,1}+r_{Z,2}+s_{X,2}+l_2 +|\gamma|+|\beta| \le 2.
$$
Hence, all other terms can be estimated pointwise, and by the means of the Klainerman-Sobolev estimates for $f$ (see Section \ref{se:KSf}), one obtains there exists a constant $K$ such that
\begin{gather*}
 \left|\frac{1}{(1+u)^q}  P_1(C)^{l_1+l_1',r_{Z,1},s_{X,1}} P_2(X(C))^{l_2,r_{Z,2},s_{X,2}} \widehat{K}^\beta(f) Q^{\gamma}(L)\right| \\
 \lesssim  \epsilon \cdot \rho^{K\delta}\cdot (1+u)^{l_1+l'_1 -q-\gamma_X} |(1+u)^{\gamma_X}Q^{\gamma}(L)|,
\end{gather*}
from which the $L^2$-estimates follow.
\end{proof}
\appendix
\section{Decomposition of the translations} \label{se:dect}
First, we recall the basic formula
\begin{eqnarray*}
\partial_t &=& \frac{tS -x^i Z_i}{t^2-|x|^2}.
\end{eqnarray*}
Since $\partial_{x^i}= \delu_i -\frac{x^i}{t}\partial_t$, and $\delu_i=\frac{Z_i}{t}$, it will be sufficient to only consider the decomposition of the time translation. In view of the formula for $\partial_t$, we automatically have
$$
|\partial_t h| \le \sum_{|\alpha|\le 1} \frac{1}{u} | K^\alpha(h) |.
$$
In fact, we have
$$
\partial_t =\sum_{|\alpha|\le 1}a_\alpha \frac{1}{u} K^\alpha,
$$
with $a_\alpha \in \mathcal{F}_x$ (homogeneous of degree $0$).
Since the function $u$ is not regular at $|x|=0$, let $\tilde{u}$ be defined as follows.
$$
\tilde{u}=t. \chi\left( \frac{|x|^2}{t^2}\right) + \left(1-\chi\left( \frac{|x|^2}{t^2}\right) \right). u,
$$
where $0 \le \chi \le 1$ is a smooth cut-off function with $\chi(y)=1$ for $y \le 1/2$ and $\chi(y)=0$ for $y\ge 1/2$.
We then have easily
$$
\partial_t =\sum_{|\alpha|\le 1} \frac{\tilde{a}_\alpha}{\tilde{u}} K^\alpha,
$$
and moreover, since $(t,x) \rightarrow \chi\left( \frac{|x|^2}{t^2}\right)$ belongs to $\mathcal{F}_x$, the above formula can be commuted with $K^\beta$ and further translations gain additional $\tilde{u}$ decay.

\begin{remark}
In the paper, for simplicity, we have written our decay in terms of $u$ instead of $\tilde{u}$.
\end{remark}

\printbibliography

\end{spacing}

\end{document}